\documentclass[11pt]{article}

\usepackage{xcolor}
\definecolor{ForestGreen}{rgb}{0.1333,0.5451,0.1333}
\definecolor{DarkRed}{rgb}{0.80,0,0}
\definecolor{Red}{rgb}{1,0,0}
\usepackage[linktocpage=true,
pagebackref=true,colorlinks,
linkcolor=DarkRed,citecolor=ForestGreen,
bookmarks,bookmarksopen,bookmarksnumbered]
{hyperref}

\usepackage{fullpage}
\usepackage[utf8]{inputenc}
\usepackage[american]{babel}
\usepackage[normalem]{ulem}
\usepackage{amsmath, amssymb, cases, amsthm}
\usepackage{thmtools}
\usepackage[shortlabels]{enumitem}
\usepackage{mdframed}
\usepackage{bbm}
\usepackage{bm}
\usepackage{microtype}
\usepackage{xcolor}
\usepackage{makecell}
\usepackage{mathtools}
\usepackage{float}
\usepackage{multirow}
\usepackage{soul}

\usepackage{easyReview} 
\usepackage{graphics}
\usepackage{mathrsfs}
\usepackage[ruled,vlined,linesnumbered]{algorithm2e}
\SetFuncSty{textsc}
\usepackage[capitalize,noabbrev,nameinlink]{cleveref}

\Crefname{claim}{Claim}{Claims}
\Crefname{condition}{Condition}{Conditions}
\Crefname{assumption}{Assumption}{Assumptions}
\Crefname{scenario}{Scenario}{Scenarios}

\usepackage{caption}
\usepackage{subcaption}

\usepackage{tcolorbox}

\declaretheorem[numberwithin=section,refname={Theorem,Theorems},Refname={Theorem,Theorems}]{theorem}

\declaretheorem[numberlike=theorem]{lemma}

\declaretheorem[numberlike=theorem]{corollary}
\declaretheorem[numberlike=theorem,style=definition]{definition}

\declaretheorem[numberlike=theorem,style=definition]{condition}
\declaretheorem[numberlike=theorem,style=definition]{scenario}
\declaretheorem[numberlike=theorem,style=definition]{assumption}
\declaretheorem[numberlike=theorem]{claim}
\declaretheorem[numberlike=theorem,style=remark]{remark}
\declaretheorem[numberlike=theorem,refname={Fact,Facts},Refname={Fact,Facts},name={Fact}]{fact}

\declaretheorem[numberlike=theorem, refname={Observation,Observations},Refname={Observation,Observations},name={Observation}]{observation}

\newcommand{\vol}{\mathrm{vol}}
\renewcommand{\deg}{\mathrm{deg}}

\newcommand{\polylog}{\mathrm{polylog}}
\newcommand{\poly}{\mathrm{poly}}
\newcommand{\defeq}{\stackrel{\mathrm{{\scriptscriptstyle def}}}{=}}

\newcommand{\caug}{c^{\mathrm{augment}}}
\renewcommand{\cong}{\mathrm{cong}}

\newcommand{\CMG}{\text{CMG}}
\newcommand{\eps}{\varepsilon}
\renewcommand{\P}{\mathcal{P}}

\newcommand{\SCC}{\mathrm{SCC}}
\newcommand{\rev}[1]{\overleftarrow{#1}}
\newcommand{\forward}[1]{\overrightarrow{#1}}
\newcommand{\backward}[1]{\overleftarrow{#1}}
\newcommand{\alg}[2]{\textsc{#1(}#2\textsc{)}}
\newcommand{\Size}{\textsc{Size}}
\newcommand{\Time}{\textsc{Time}}
\newcommand{\expect}{\mathop{\mathbb{E}}}

\newcommand{\Ba}{\boldsymbol{a}}
\newcommand{\Bb}{\boldsymbol{b}}
\newcommand{\Bc}{\boldsymbol{c}}
\newcommand{\Bw}{\boldsymbol{w}}
\newcommand{\Bd}{\boldsymbol{d}}
\newcommand{\Bf}{\boldsymbol{f}}
\newcommand{\Br}{\boldsymbol{r}}
\newcommand{\Bx}{\boldsymbol{x}}
\newcommand{\BB}{\boldsymbol{B}}
\newcommand{\Bsource}{\boldsymbol{\Delta}}
\newcommand{\Bsink}{\boldsymbol{\nabla}}
\newcommand{\Bgamma}{\boldsymbol{\gamma}}
\newcommand{\Bzero}{\boldsymbol{0}}
\newcommand{\Bone}{\boldsymbol{1}}
\newcommand{\Bfout}{\boldsymbol{f}^{\mathrm{out}}}
\newcommand{\Btau}{\boldsymbol{\tau}}
\newcommand{\Bnu}{\boldsymbol{\nu}}
\newcommand{\abs}{\boldsymbol{\mathrm{abs}}}
\newcommand{\ex}{\boldsymbol{\mathrm{ex}}}
\newcommand{\dist}{\mathrm{dist}}

\newcommand{\Bell}{\boldsymbol{\ell}}

\newcommand{\N}{\mathbb{N}}
\newcommand{\R}{\mathbb{R}}
\newcommand{\Z}{\mathbb{Z}}
\newcommand{\Q}{\mathbb{Q}}

\newcommand{\cA}{\mathcal{A}}

\newcommand{\cD}{\mathcal{D}}
\newcommand{\cE}{\mathcal{E}}

\newcommand{\cH}{\mathcal{H}}
\newcommand{\cI}{\mathcal{I}}
\newcommand{\cJ}{\mathcal{J}}
\newcommand{\cK}{\mathcal{K}}

\newcommand{\cM}{\mathcal{M}}

\newcommand{\cP}{\mathcal{P}}

\newcommand{\cR}{\mathcal{R}}
\newcommand{\cS}{\mathcal{S}}
\newcommand{\cT}{\mathcal{T}}
\newcommand{\cU}{\mathcal{U}}

\newcommand{\tO}{\widetilde{O}}

\def\final{0}
\def\cameraready{0}
\newcommand{\para}[1]{\paragraph{#1.}}

\ifnum\final=0

\newcommand{\blikstad}[1]{{\color{purple}[\textbf{Joakim}: #1]}}
\newcommand{\tawei}[1]{{\color{red}[\textbf{Ta-Wei}: #1]}}
\newcommand{\aaron}[1]{{\color{blue}[\textbf{Aaron}: #1]}}

\def\thatchaphol#1{\marginpar{$\leftarrow$\fbox{TS}}\footnote{$\Rightarrow$~{\sf\textcolor{purple}{#1 --Thatchaphol}}}}

\else

\newcommand{\blikstad}[1]{}
\newcommand{\tawei}[1]{}
\newcommand{\aaron}[1]{}

\def\thatchaphol#1{}

\fi

\makeatletter
\newcommand\footnoteref[1]{\protected@xdef\@thefnmark{\ref{#1}}\@footnotemark}
\makeatother

\title{Maximum Flow by Augmenting Paths in $n^{2+o(1)}$ Time}
\author{Aaron Bernstein\thanks{
        New York University,
        \texttt{bernstei@gmail.com}. Work done in part while at Rutgers University. Supported by Sloan Fellowship, Google Research Fellowship,  NSF Grant 1942010, and Charles S. Baylis endowment at NYU.
    } \and Joakim Blikstad\thanks{
        KTH Royal Institute of Technology \& Max Planck Institute for Informatics,
        \texttt{blikstad@kth.se}.
        Supported by the Swedish Research Council (Reg. No. 2019-05622) and the Google PhD Fellowship Program.
    }  \and Thatchaphol Saranurak\thanks{
        University of Michigan,
        \texttt{thsa@umich.edu}.
        Supported by NSF Grant CCF-2238138.
    }  \and Ta-Wei Tu\thanks{
        Stanford University,
        \texttt{taweitu@stanford.edu}.
        Supported by a Stanford School of Engineering Fellowship.
    } }

\begin{document}

\begin{titlepage}
  \maketitle \pagenumbering{gobble}
  \begin{abstract}
We present a combinatorial algorithm for computing exact maximum flows in directed graphs with $n$ vertices and edge capacities from $\{1,\dots,U\}$ in $n^{2+o(1)}\log U$ time, which is almost optimal in dense graphs. Our algorithm is a novel implementation of the classical augmenting-path framework; we list augmenting paths more efficiently using a new variant of the push-relabel algorithm that uses additional edge weights to guide the algorithm, and we derive the edge weights by constructing a directed expander hierarchy. 

Even in unit-capacity graphs, this breaks the long-standing $O(m\cdot\min\{\sqrt{m},n^{2/3}\})$ time bound of the previous combinatorial algorithms by Karzanov (1973) and Even and Tarjan (1975) when the graph has $m=\omega(n^{4/3})$ edges. Notably, our approach does not rely on continuous optimization nor heavy dynamic graph data structures, both of which are crucial in the recent developments that led to the almost-linear time algorithm by Chen et al.~(FOCS 2022). 
Our running time also matches the $n^{2+o(1)}$ time bound of the independent combinatorial algorithm by Chuzhoy and Khanna~(STOC 2024) for computing the maximum bipartite matching, a special case of maximum flow.

\end{abstract}

  \setcounter{tocdepth}{3}
  \newpage
  \phantom{a} %
  
  \vspace{-1.4cm}
  \tableofcontents
  \newpage
\end{titlepage}
\newpage
\pagenumbering{arabic}

\newpage

\section{Introduction}

Fast algorithms for computing maximum flows have played a central role in algorithmic research, motivating various algorithmic paradigms such as graph sparsification, dynamic data structures, and the use of continuous optimization in combinatorial problems. These algorithms also have numerous applications in problems like bipartite matching, minimum cuts, and Gomory-Hu trees~\cite{gomory1961multi,LiP20,LiNPSY21,Cen0NPSQ21,CenH0P23,Abboud0PS23}. Below, we summarize the development of fast maximum flow algorithms following Dantzig's introduction of the problem \cite{dantzig1951application}. The input graph for this problem is a directed graph with $n$ vertices and $m$ edges. For convenience, only in this introduction, we assume that edge capacities range from $\{1,\dots,\text{poly}(n)\}$.

\para{Augmenting Paths}
Ford and Fulkerson \cite{ford1956maximal} first introduced the \emph{augmenting path} framework. In this framework, algorithms repeatedly find an augmenting path (or a collection of them) in the residual graph. They then augment the flow along this path with a value equal to the bottleneck of the augmenting path. Over the next four decades, this simple framework led to the development of several influential techniques, including shortest augmenting paths \cite{EdmondsK72}, blocking flows \cite{karzanov1973finding,dinic1970algorithm,GoldbergR98}, push-relabel \cite{GoldbergT88}, and sparsification \cite{KargerL15}. The best time bound within this framework, $O(m\cdot\min\{m^{1/2},n^{2/3}\})$, was given by Karzanov \cite{karzanov1973finding} and independently by Even and Tarjan \cite{EvenT75} for unit-capacity graphs. Goldberg and Rao \cite{GoldbergR98} later matched this time bound in capacitated graphs up to poly-logarithmic factors.

\para{Continuous Optimization and Dynamic Graph Data Structures}
In the 2000s, Spielman and Teng \cite{SpielmanT04} introduced a completely different framework based on \emph{continuous optimization}, and solved the electrical flow problem in near-linear time.\footnote{Algorithms for \emph{approximating} maximum flow in undirected graphs using gradient descent and multiplicative weight update also follow this framework \cite{ChristianoKMST11,KelnerMP12,Sherman13,KelnerLOS14,Peng16,Sherman17,SidfordT18}.} Then, Daitch and Spielman \cite{DaitchS08}  showed a reduction from maximum flow to electrical flow using an interior point method. This motivated further research on advanced interior point methods \cite{Madry13,LeeS14,Madry16,LiuS20,KathuriaLS20} that minimize the number of iterations of calling electrical flow or related problems. As a result, a $\widetilde{O}(m\sqrt{n})$-time\footnote{In this paper, we use $\widetilde{O}(\cdot)$ to hide poly-logarithmic factors in $n$ and $\widehat{O}(\cdot)$ for subpolynomial factors.} maximum flow algorithm \cite{LeeS14} and, for unit-capacity graphs, a $\widetilde{O}(m^{4/3})$-time algorithm \cite{KathuriaLS20} were developed. 

Since 2020, the focus has shifted from minimizing the number of iterations in interior point methods to instead minimizing the cost per iteration using \emph{dynamic graph data structures}. Building upon dynamic data structures for sparsifiers \cite{DurfeeGGP19,BernsteinBGNSS022,ChenGHPS20}, a series of impressive works \cite{BrandLNPSS0W20,BrandLLSS0W21,GaoLP21,BrandGJLLPS22} finally led to the breakthrough by Chen \emph{et al.}~\cite{ChenKLPGS22,Brand0PKLGSS23} who showed an $m^{1+o(1)}$-time algorithm for maximum flow and its generalization.

\para{Combinatorial Approaches}
Although the recent developments have achieved an almost optimal time bound, these algorithms are not simple in either conceptual or technical sense. Continuous optimization approaches update the flow solutions without a clear combinatorial interpretation, and the required dynamic data structures are still highly involved. With this motivation, Chuzhoy and Khanna recently \cite{ChuzhoyK24SODA} showed a conceptually simpler algorithm for maximum bipartite matching, a special case of maximum flow, that runs in $\widetilde{O}(m^{1/3}n^{5/3})$ time; in very recent independent work, they improved the running time to $n^{2+o(1)}$ \cite{ChuzhoyK24STOC}.\footnote{The algorithm of \cite{ChuzhoyK24STOC} was submitted several months before ours (STOC 2024), but we consider it independent because it was not publicly available when we submitted our paper to FOCS 2024.}

These algorithms update the flow (i.e., the fractional matching) in a more intuitive manner: they repeatedly increase the flow value along paths listed by dynamic shortest-path data structures. Moreover, the flow on each edge is a multiple of $\frac{1}{\Theta(\log n)}$, i.e., it is almost integral. Unfortunately, the algorithms do not extend to exact maximum flow, even in unit-capacity graphs.

So, can one hope for an optimal and combinatorial maximum flow algorithm? We make significant progress in this direction by showing that the classical augmenting path framework can provide an almost-optimal algorithm for dense graphs.

\begin{theorem}
\label{thm:main}
There is an augmenting-path-based randomized algorithm that, given a directed graph with $n$ vertices and edge capacities from $\{1,\dots,U\}$, with high probability computes a maximum $s$-$t$ flow in $n^{2+o(1)}\log U$ time.
\end{theorem}

Our algorithm strictly follows the augmenting path framework, i.e., it maintains an \emph{integral} flow that is repeatedly increased along augmenting paths in the residual graph. When $m=\omega(n^{4/3})$, \Cref{thm:main} improves upon the long-standing time bound of $O(m\cdot\min\{m^{1/2},n^{2/3}\})$ presented in \cite{karzanov1973finding,EvenT75,GoldbergR98}, in the context of previous augmenting-path-based algorithms. Additionally, our running time matches the $n^{2+o(1)}$-time bound in the independent work by Chuzhoy and Khanna \cite{ChuzhoyK24STOC} who gave combinatorial algorithms for computing the maximum bipartite matching (and consequently unit-vertex-capacitated maximum flow), a special case of maximum flow. 
As far as we know, their techniques are based on dynamic shortest-path data structures and multiplicative weights update and, hence, are different from ours.

\para{A Bird's-Eye View of Our Algorithm}
We now give a basic outline of our algorithm; we provide a more detailed overview in \cref{sec:overview}.
To list augmenting paths, we introduce the \emph{weighted push-relabel} algorithm, a new and simple variant of the well-known push-relabel algorithm \cite{GoldbergT88} guided by an additional edge weight function. Edges with higher weight are relabeled less often, allowing for more efficiency. Given a ``good'' weight function, %
the weighted push-relabel algorithm will list augmenting paths in $\widetilde{O}(n^{2})$ total time and return an $O(1)$-approximate maximum flow. By repeating the algorithm on the residual graph, this immediately gives a maximum flow algorithm.

Our starting observation is that on a DAG with a topological order $\Btau$, the simple function $\Bw(u,v)=|\Btau_u-\Btau_v|$ for each edge $(u,v)$ is in fact a good weight function. %
The question is now to figure out what ``good'' weight function to use on general graphs.

We introduce the \emph{directed expander hierarchy} and show that it induces a natural vertex ordering $\Btau$ such that $\Bw(u,v)=|\Btau_u-\Btau_v|$ is a good weight function. We remark that while there are several successful variants of expander hierarchies in undirected graphs \cite{Racke02,PatrascuT07,RackeST14,GoranciRST21}, we believe ours is the first paper to successfully apply them to directed graphs. Unfortunately, all known approaches for the hierarchy construction are either too slow \cite{PatrascuT07}, assume a maximum flow subroutine itself \cite{RackeST14}, or are specific to undirected graphs \cite{GoranciRST21}.

Therefore, we show a new bottom-up construction based on our weighted push-relabel algorithm. The basic idea is that we repeatedly use weighted push-relabel to construct the next level of the hierarchy, which in turn gives us a better weight function, allowing us to compute yet one more level. To be a bit more concrete, let $X_{i}$ be the candidate edge set for  level-$i$ of the hierarchy and suppose that we have already built a directed expander hierarchy of $G\setminus X_{i}$, which consists of all edges below level $i$. If we can certify that $X_i$ is expanding (i.e. well-connected in some sense), then we can leave $X_i$ as the last level of the hierarchy; on the other hand, if $X_i$ if not expanding, then we need to find a sparse cut with respect to $X_i$ and elevate those cut edges to the next level $X_{i+1}$. As is standard, we solve this problem using the cut-matching game, which requires computing flow between subsets of $X_{i}$. The challenge lies in solving this flow problem efficiently.

The crucial observation is that by setting the weight of edges in $G\setminus X_{i}$ according to its expander hierarchy  (which we already computed) and setting the weight of all edges in $X_{i}$ to be $n$, our weighted push-relabel algorithm will solve this flow problem in $\widetilde{O}(n^{2})$ time. We note that standard push-relabel (without weights) can only solve this problem for the bottom level of the hierarchy, i.e., when $X_{i}$ is the whole edge set.

We emphasize that so far, all of our algorithmic components
\ifnum\cameraready=0
(from \Cref{sec:push-relabel}~to~\ref{sec:sparse-cut})
\else
(from \Cref{sec:push-relabel}~to~\cite[Section 6]{BernsteinBST24})
\fi
are very implementable. The weighted push-relabel algorithm (\Cref{alg:push-relabel}) simply increments vertex labels
and, in capacitated graphs, uses the link-cut tree \cite{SleatorT83} to push flow.
\ifnum\cameraready=0
To find sparse cuts (\Cref{alg:sparse-cut}),
\else
To find sparse cuts (\cite[Section 6]{BernsteinBST24}),
\fi
we also call Dijkstra's algorithm to produce levels and return the sparsest level cut.

The novelty is in the analysis. To show that the directed expander hierarchy gives a good weight function, we prove a new trade-off between length and congestion for rerouting flow on expanders. To show that there exists a sparse level cut, we show a novel angle to the \emph{directed expander pruning} problem studied in \cite{BernsteinGS20, HuaKGW23, SulserP24}. 
In the standard version of pruning, we update a few edges in a directed $\phi$-expander $G$, and the goal is to prune away a small set of vertices $P$ so that $G \setminus P$ is still an expander. We extend pruning to work with \emph{path-reversal} updates, which reverse the direction of \emph{every} edge on a given path; this kind of update is very natural in the residual graph (reversing augmenting paths). We show that, somewhat surprisingly, reversing a whole path has approximately the same impact on pruning as updating a single edge. This allows us to show that the directed expander hierarchy is robust under flow augmentation 
\ifnum\cameraready=0
(\Cref{lemma:low-diameter-expander-new}).
\else
(\cite[Lemma 6.5]{BernsteinBST24}).
\fi
For our purposes, we only ever need an existential version of path-reversal pruning, but the algorithmic version should be plausible and useful.

There is unfortunately one challenge that adds a huge amount of complexity. 
When computing a sparse cut with respect to $X_i$, the sparse cut can ``cut through'' components of the expander hierarchy in the lower levels. This forbids us from building the hierarchy bottom-up in one go; instead, our algorithm needs to regularly move up edges at different levels of the hierarchy.
To modularize the analysis, we employ a data structure point-of-view that models these interactions between levels. We note that our approach is not inherently \emph{dynamic}, in that we are not aiming for fast or sublinear update times. In fact, each operation of the data structure requires $n^{2+o(1)}$ time. The usefulness of this perspective lies instead in showing that only $n^{o(1)}$ sequential updates are needed.
This is by far the most complicated part of our algorithm
\ifnum\cameraready=0
(essentially all of \cref{sec:nested-expander-decomposition}),
\else
(essentially all of \cite[Section 7]{BernsteinBST24}),
\fi
and is also the only reason our algorithm is randomized and requires an inherent $n^{o(1)}$-factor in the running time. We believe this step can be simplified once tools related to directed expanders are as developed as their undirected counterparts \cite{RackeST14,SaranurakW19,GoranciRST21}.

To summarize, in contrast to recent developments, \Cref{thm:main} does not rely on continuous optimization or heavy dynamic data structures. It also paves the way to a very implementable $\widetilde{O}(n^{2})$-time deterministic algorithm once a better construction of directed expander hierarchy is shown.
We note that our paper is quite self-contained: the black boxes we assume only include basic graph algorithms (e.g., topological sort and Dijkstra's algorithm), link-cut trees~\cite{SleatorT83}, and Louis's cut-matching game \cite{Louis10}.
Lastly, we believe and hope that some novel tools we developed in this paper---including the \emph{weighted push-relabel algorithm}, \emph{directed expander hierarchy}, and \emph{expander pruning under path-reversals}---will find future applications.

\para{Future Work}
The natural next step is to show a simple algorithm for constructing a directed expander hierarchy in $\widetilde{O}(n^{2})$. Combined with our new approach, this would yield a much simpler $\widetilde{O}(n^{2})$-time max-flow algorithm, and we believe it would also prove a powerful tool for other directed problems.

A more challenging goal is to achieve an $m^{1+o(1)}$ time bound via simple combinatorial algorithms. One takeaway of our paper is that the main bottleneck seems to be a fast $n^{o(1)}$-approximation for DAGs. On the one hand, it seems quite plausible that our tools would allow an improvement for DAGs to be generalized to all directed graphs; but on the other hand, the current toolkit for DAGs is quite limited. DAGs also seem to capture the hardness of other fundamental problems in directed graphs, such as dynamic shortest paths \cite{BernsteinGS20}, parallel reachability \cite{Fineman18,LiuJS19}, and diameter-reducing shortcuts \cite{KoganP22}.%

\para{Organization}
The rest of the paper is organized as follows.
In \cref{sec:overview}, we give a comprehensive overview of the technical components of our algorithm.
In \cref{sec:prelim} we provide necessary preliminaries.
We develop our weighted push-relabel algorithm in \cref{sec:push-relabel}. In \cref{sec:weight}, we show that a directed expander hierarchy induces a ``good'' weight function, and hence, when combined with our weighted push-relabel algorithm, solves maximum flow.
\ifnum\cameraready=0
In \cref{sec:sparse-cut}, we show how to leverage the weighted push-relabel algorithm to compute sparse cuts, which is a crucial subroutine in how we construct the expander hierarchy in \cref{sec:nested-expander-decomposition}.
In \cref{sec:capacitated,appendix:capacity-scaling,appendix:omitted-proofs} we provide details omitted from the main body of the paper.
\else
We defer the rest of the content to the full version~\cite{BernsteinBST24}.
In \cite[Section 6]{BernsteinBST24}, we show how to leverage the weighted push-relabel algorithm to compute sparse cuts, which is a crucial subroutine in how we construct the expander hierarchy in \cite[Section 7]{BernsteinBST24}.
In \cite[Appendix A-C]{BernsteinBST24} we provide details omitted from the main body of the paper.
\fi

\newcommand{\fopt}{\Bf^{*}}
\newcommand{\fstar}{\Bf^*}
\newcommand{\otil}{\widetilde{O}}
\newcommand{\Otil}{\otil}
\newcommand{\fin}{\Bf^{\mathrm{in}}}
\newcommand{\Ell}{\Bell}
\newcommand{\fsink}{\Bf^{\mathrm{sink}}}
\newcommand{\pset}{\mathcal{P}}
\newcommand{\psetshort}{\mathcal{P}_{\textrm{short}}}
\newcommand{\psetlong}{\mathcal{P}_{\textrm{long}}}
\newcommand{\pfirst}{P_{\textrm{early}}}
\newcommand{\plast}{P_{\textrm{late}}}
\newcommand{\Bellmax}{\Bell_{\textrm{max}}}
\newcommand{\xsupply}{X_{\Bsource}}
\newcommand{\xsink}{X_{\Bsink}}
\newcommand{\fearly}{\Bf_{\textrm{early}}}
\newcommand{\flate}{\Bf_{\textrm{late}}}
\newcommand{\Bdp}{\Bd'}
\newcommand{\Bwp}{\Bw'}

\section{Technical Overview}\label{sec:overview}

In this section we give a high-level overview our maximum flow algorithm.
For simplicity of presentation, we assume during this overview that the input graph is unit-capacitated.
Note that it suffices to design a constant- or even $1/n^{o(1)}$-approximate flow algorithm for directed graphs, as the exact algorithm then follows by repeating the approximate algorithm $n^{o(1)}$ times on the residual graph.
This is in contrast to undirected graphs: although efficient approximations are known here~\cite{Sherman13,KelnerLOS14,Peng16,Sherman17,SidfordT18}, the residual graph of the found approximate flow is no longer undirected, so an approximate flow algorithm cannot be bootstrapped to an exact one.
Although we assume unit capacities, in the analysis we will sometimes refer to a flow $\Bf$ that disobeys these capacities; we define the \emph{congestion} of a flow $\Bf$, denoted $\cong(\Bf)$, to be $\max_{e \in E} \Bf(e)$. %

\subsection{Weighted Push-Relabel Algorithm}
\label{sec:overview:push-relabel}

The starting point of our algorithm is a weighted variant of the classic push-relabel algorithm.

\para{Summary of Classic Push-Relabel}
Let us recall the classic push-relabel algorithm in unit-capacitated graphs. Suppose we have a flow instance with integral source vector $\Bsource$ and sink vector $\Bsink$, and let us assume for simplicity that this flow instance is feasible. The push-relabel algorithm will always maintain a pre-flow,\footnote{A \emph{pre-flow} is an intermediate flow that has not yet sent all units of demands to sink vertices.} where every vertex $v$ might have excess flow $\ex_{\Bf}(v) \defeq  \max\{ \Bsource(v) - \Bfout(v) - \Bsink(v), 0 \}$ where $\Bfout(v)$ denotes the \emph{net} flow going out from $v$; note that initially, all positive excess are on the source vertices.
The algorithm also maintains an integral label $\Bell(v)$ on every $v \in V$, which gradually increases over time; initially $\Bell(v) = 0$ for all $v \in V$.

Informally, the main loop of the push-relabel algorithm repeatedly finds a vertex $v$ with $\ex_{\Bf}(v) > 0$ and attempts to \emph{push} a unit of flow along an edge $(v,w)$ with $\Bell(v) \geq \Bell(w) + 1$; we refer to such edges as \emph{admissible}. Following the standard operation of residual graphs, this edge $(v,w)$ is then removed and replaced with the reverse edge $(w,v)$. If a vertex $v$ has $\ex_{\Bf}(v) > 0$, but there are no admissible edges $(v,w)$, then the algorithm performs operation \textsc{Relabel}($v$), which increases $\Bell(v)$ by $1$. The sequence of push operations effectively traces augmenting paths from the source vertices to the sink vertices. 

\para{Analysis of Classic Push-Relabel}

The analysis rests on the following admissibility invariant: for any edge $(u,v)$ in the residual graph, we have $\Bell(u) \leq \Bell(v) + 1$. This easily follows from the fact that if $\Bell(u) = \Bell(v) + 1$, then $(u,v)$ is an admissible edge, so the algorithm will not relabel $u$ as long as $(u,v)$ remains in the residual graph.

We now sketch the proof that push-relabel successfully finds a flow that routes all the demands. In fact, we show something stronger: at termination, we have $\Bell(v) \leq n$ for all $v \in V$. Say, for contradiction, that the algorithm relabels a vertex $v$ from $\Bell(v) = n$ to $\Bell(v) = n+1$. This implies that $\ex_{\Bf}(v)>0$, so since we assumed the original flow instance is feasible, there must exist some path in the residual graph from $v$ to an unsaturated sink vertex $t$. It is easy to see that $\Bell(t) = 0$; since $t$ is still a sink, it never had an excess, and so was never relabeled. This $(v,t)$-path has at most $n-1$ edges, so by the admissibility invariant above, $\Bell(v) \leq \Bell(t) + n - 1 = n-1$, contradicting the assumption that $\Bell(v) = n$.

For the running time analysis, it is not hard to check that any edge $(u,v)$ can undergo at most one push operation as long as the level $\Ell(u)$ is fixed; similarly, the admissibility status of $(u,v)$ can only change when $u$ or $v$ is relabeled. Since $u$ and $v$ undergo at most $O(n)$ relabel operations, the total running time is $O(mn)$.

\para{Motivating Our Weighted Push-Relabel}
Consider the following simplified scenario: we are told in advance that a certain subset of the edges is \emph{infrequent}, meaning that there exists some approximate maximum flow $\Bf$, where every flow path in $\Bf$ uses at most $k$ infrequent edges (think of $k$ as small).  

We then modify the classic push-relabel algorithm as follows. An infrequent edge $(u,v)$ only counts as admissible if $\Bell(u) \geq \Bell(v) + n/k$. This means that the algorithm might need to perform more relabel operations, and yet we can still show that if any label ever exceeds $10n$, this means that push-relabel has already routed a constant fraction of the demand. To see this, assume, for contradiction, that the algorithm relabels some $v$ from $\Bell(v) = 10n$ to $\Bell(v) = 10n + 1$. As before, there must exist some $(v,t)$-path $P$ in the residual graph to an unsaturated sink $t$ with $\Bell(t) = 0$. Intuitively, the path $P$ contains at most $k$ infrequent edges: this is not technically true, and the full analysis is slightly more involved.
However, this is close to being true, so we make the simplifying assumption here that this path contains at most $9k$ infrequent edges.\footnote{The reason it is not technically true is that even though flow paths in the original graph use at most $k$ infrequent edges, this is no longer true in the residual graph. But letting $\fearly$ be the flow already computed by push-relabel,  note that the residual graph only differs from the original one by edges in $\fearly$, so if paths in the residual graph use significantly more infrequent edges than paths in the original graph, this implies that $\fearly$ is itself using many infrequent edges, and hence has a large value. In the technical exposition, we show that either $\fearly$ sends a constant fraction of the supply (so the algorithm can terminate),  or paths in the residual graph are relatively similar to those in the original graph and hence have few infrequent edges.} 
The natural generalization of the admissibility invariant then implies that
\ifnum\cameraready=0
\begin{align*}
  \Bell(v)
    &\leq \Bell(t) + [\text{\# infrequent edges on } P] \cdot (n/k) + [\text{\# frequent edges on } P] \\
    &\leq 0 + 9k\cdot (n/k) + n-1 = 10n-1,
\end{align*}
\else
\begin{alignat*}{2}
  \Bell(v)
    &\leq \Bell(t) &&+ [\text{\# infrequent edges on } P] \cdot (n/k) \\ &\phantom{\leq \Bell(t)} &&+ [\text{\# frequent edges on } P] \\
    &\leq 0 &&+ 9k\cdot (n/k) + n-1 = 10n-1,
\end{alignat*}
\fi
contradicting the assumption that $\Bell(v) = 10n$. 
For the running time analysis, note that an infrequent edge can only change status every $n/k$ relabels, so the new runtime is $O([\text{\# infrequent edges}] \cdot n/k + [\text{\# frequent edges}] \cdot n)$, which is significantly smaller than $O(mn)$ if most edges are infrequent.

\para{Weighted Push-Relabel}
Imagine a generalization of the above scenario where we have a different frequency promise for every edge. We represent these promises with a weight function $\Bw \in \N^E$. An edge $(u,v)$ is defined as admissible in the push-relabel algorithm if $\Ell(u) \geq \Ell(v) + \Bw(u,v)$. 
Following the logic of the above paragraph, suppose we have a promise that there exists a flow where every flow path has $\Bw$-weight at most $h=n^{1+o(1)}$, then we can guarantee that, when running the algorithm with maximum vertex label of $10h$, the algorithm will find a flow that routes a constant fraction of the demands.
This yields the following theorem:%

\begin{theorem}[Informal version of \cref{thm:push-relabel-main-theorem}]
\label{thm:push-relabel-main-informal}
Given edge weights $\Bw \in \N^E$ and parameter $h$, the weighted push-relabel algorithm return a flow in $\otil(m+ h \cdot \sum_{e \in E} 1/\Bw(e))$ with the following guarantee:
if there exists a flow $\Bf^{*}$ such that every flow path $P$ in $\Bf$ has $\sum_{e \in P} \Bw(e) \le h$, then the returned flow has value $\Omega(|\Bf^{*}|)$.
In particular, if $\Bf^{*}$ is an $\alpha$-approximate maximum flow, then the returned flow is a $O(\alpha)$-approximate maximum flow.
\end{theorem}

The general idea of using a weight function to limit how often the algorithm touches various edges is inspired by a similar weighted variant of the Even-Shiloach trees \cite{Bernstein17,GutenbergW20a} that has been applied to dynamic shortest paths.
A more detailed comparison of our push-relabel algorithm and the standard version and a discussion of possible future improvements for sparse graphs are given in \Cref{sec:push-relabel-implementation}. 

\para{The Maximum Flow Algorithm}
To this end, we say that a weight function $\Bw$ satisfies the \emph{path-weight requirement} if there exists an $1/n^{o(1)}$-approximate maximum flow $\Bf$ such that every flow path $P$ in $\Bf$ has $\sum_{e \in P} \Bw(e) \le h= n^{1+o(1)}$.
Our main technical contribution is showing how to compute a weight function $\Bw$ that satisfies the path-weight requirement and has $\sum_{e \in E} 1/\Bw(e) = n^{1+o(1)}$.
Given this, by applying \Cref{thm:push-relabel-main-informal}, we immediately obtain a maximum flow algorithm with running time $n^{2+o(1)}$. (Recall that in directed graphs, an approximate flow algorithm immediately implies an exact algorithm.)
In the remainder of this overview we explain how to get this weight function.

\subsection{Examples of Good Weight Functions}

\para{Directed Acyclic Graphs}
Let us consider the simplest directed graph: a directed acyclic graph (DAG).
We know that a DAG admits a topological order $\Btau \in [n]^V$ such that $\Btau_v > \Btau_u$ for each edge $(u, v)$.
This topological order also gives us the desired weight function: if we set $\Bw(u, v) \defeq \Btau_v - \Btau_u$, then not only flow paths on the maximum flow, but any path in the DAG will have weight at most $n$.
Moreover, it is easy to see that $\sum_{e \in E} 1/\Bw(e) = O(n\log(n))$, because the sum of weights incident to a specific vertex $v$ forms a harmonic series and is hence $O(\log(n))$. Plugging this into \cref{thm:push-relabel-main-informal} yields a remarkably simple $\otil(n^2)$-time algorithm for computing a $O(1)$-approximate flow in a DAG using only classical flow techniques (\cref{cor:dag-approx}).

\begin{remark}
The above simple algorithm for DAGs is inherently approximate. Whereas in general graphs there is a standard reduction from exact to approximate max flow, this does not apply to DAGs: the reduction involves recursively calling the approximate flow algorithm on the residual graph, but even if the original graph is a DAG, the residual graph might not be. Thus, on its own, our approximate max flow algorithm on DAGs has no implication for general graphs. By contrast, a reduction of Ramachadran \cite{Ramachandran87} shows that an exact algorithm for DAGs would imply an exact algorithm for general graphs as well.
\end{remark}

\para{General Graphs Given Maximum Flow}
The analysis of the DAG case also shows the \emph{existence} of a good weight function in general graphs: take any integral maximum flow, the support of which after cycle cancellation forms a DAG, and then assign weights as above to this support and assign large weight (e.g. $100n$) to all other edges. Of course, this weight function requires computing a maximum flow and, hence, is not useful for us. We will show another construction of good weight function based on a directed expander hierarchy.

\subsection{Basic Facts About Expanders}
In order to describe the directed expander hierarchy, 
we review some basic properties of expanders.

\begin{definition}[Directed expander]
Consider a directed, unweighted graph $G=(V, E)$. For any set of vertices $S \subseteq V$, we define $\vol(S) \defeq \sum_{v \in S} \deg(v)$, where $\deg(v)$ counts both in- and out-edges incident to $v$. We say that cut $\emptyset \neq S \subsetneq V$ is \emph{$\phi$-sparse} if $\min\{|E(S,\overline{S})|, |E(\overline{S}, S)|\} < \phi \cdot \min\{\vol(S), \vol(\overline{S})\}$, where $\overline{S} \defeq V \setminus S$. We say that a graph $G$ is a \emph{$\phi$-expander} if it contains no $\phi$-sparse cuts.
\end{definition}

One should think of the $\phi$ parameter above as being $1/n^{o(1)}$.
We also modify the above definitions to apply with respect to an edge set $F \subseteq E$, often referred to as \emph{terminal} edges. In particular, define $\deg_F(v)$ to be the number of edges in $F$ incident to $v$ and $\vol_F(S) = \sum_{v \in S} \deg_F(v)$; we say a cut $S$ is \emph{$\phi$-sparse} with respect to $F$ if  $\min\{|E(S,\overline{S})|, |E(\overline{S}, S)|\} < \phi \cdot \min\{\vol_F(S), \vol_F(\overline{S})\}$; we say that $G$ is a \emph{$\phi$-expander} with respect to $F$
if $G$ contains no $\phi$-sparse cuts with respect to $F$.

To handle graphs that are not strongly connected, it is useful to define a notion of expansion that applies separately to every strongly connected component (SCC). Given a set of terminal edges $F$, we say that $F$ is \emph{$\phi$-expanding} in $G$ if every SCC of $G$ is a $\phi$-expander with respect to $F$.\footnote{More precisely, each SCC is a $\phi$-expander with respect to the volume induced by $F$.}

Expanders are nice to work with in the context of flow problems because they admit a low-congestion flow between any sets of sources/sinks. This also generalizes to a terminal set $F$.

\begin{fact}[Proved in \Cref{lemma:expander-routing-respecting}]
\label{fact:expander-routing-overview}
Let $G = (V,E)$ be a $\phi$-expander with respect to a terminal set $F \subseteq E$. Consider any flow-instance with supply/demand $\Bsource,\Bsink$ such that $\|\Bsource\|_1 = \|\Bsink\|_1$ with all supply/demand on terminal edges; formally, this means that for every vertex $v$, $\Bsource(v) \leq \deg_F(v)$ and $\Bsink(v) \leq \deg_F(v)$. Then, there exists a flow $\Bf$ in $G$ that routes all the supply/demand and has the following properties: 
\textbf{1)} $\Bf(e) = O(\log(n)/\phi)$ for every $e \in E$ and \textbf{2)} Every flow path in $\Bf$ uses at most $O(\log(n)/\phi)$ edges in $F$.
\end{fact}

A standard approach to dealing with a general undirected graph $G = (V,E)$ is to decompose it into a hierarchy of expanders~\cite{PatrascuT07,GoranciRST21}; in this paper, we propose an analogous hierarchy for directed graphs. Let us first consider a single expander decomposition of a directed graph $G$, formalized in \cite{BernsteinGS20}. In any directed graph $G$, it is possible to find a set of ``back edges'' $B$ such that every strongly connected component (SCC) of $G \setminus B \defeq (V, E \setminus B)$ is a $\phi$-expander and $|B| = \otil(\phi m)$.\footnote{To see this existentially, imagine the algorithm that repeatedly finds $\phi$-sparse cuts in the graph, adds the (the sparser direction of) cut edges to $B$, and recurses on both sides of the cut. This clearly results in a desired expander decomposition. The size of $B$ can be bounded by a simple charging argument: Every time we find a sparse cut, we can charge the cut edges to the smaller side of the cut. Since a vertex can be in the smaller side of the cut at most $O(\log n)$ times, the bound of $\widetilde{O}(\phi m)$ follows.}
If we imagine a topological sort of the SCCs in $G \setminus B$, then the above partition effectively decomposes $E$ into three edges types: 
\begin{enumerate}
\item Edges inside SCCs of $G \setminus B$, which we denote as $X_1.$
\item Edges $(u,v)$ between different SCCs of $G \setminus B$ that go forward in the topological ordering. We will denote these as $D$, which stands for DAG edges.
\item Edges in $B$, which may go backward in the topological ordering. We denote these as $X_2$.
\end{enumerate}

Put succinctly, $X_1$ is $\phi$-expanding in $G\setminus X_2$.
If $X_2$ happens to be $\phi$-expanding in $G$, then the expander hierarchy is complete; if not, we need to add a level to the hierarchy. We can again perform expander decomposition with respect to $X_2$ to compute an even smaller set of edges $X_3$ such that 
every SCC of $G \setminus X_3$ is a $\phi$-expander with respect to $X_2 \setminus X_3$. Let us assume, for simplicity, that $X_3 \subseteq X_2$; then, to maintain a partition, we replace $X_2$ with $X_2 \setminus X_3$, and we now have a partition of the edge set $E = D \cup X_1 \cup X_2 \cup X_3$. If $X_3$ is $\phi$-expanding in $G$ then the expander hierarchy is complete; otherwise we define a new set $X_4$ in the same manner. We now define the hierarchy more formally (see also \cref{fig:hierarchy example} in \cref{sec:weight}).

\begin{definition}
\label{dfn:overview-expander-hierarchy}
A partition $\cH = (D,X_1, \ldots, X_\eta)$ of the edges is a \emph{$\phi$-expander hierarchy} if $D$ is acyclic and for every $i \in [\eta]$, $X_i$ is $\phi$-expanding in $G\setminus X_{>i}$; that is, all SCCs of $G\setminus X_{>i}$ are $\phi$-expanders with respect to $X_i$, where $X_{>i} = X_{i+1} \cup \cdots \cup X_\eta$. Note that $X_{\eta}$ must be $\phi$-expanding in $G$.
\end{definition}

While several variants of expander hierarchies have been previously used in \emph{undirected} graphs \cite{Racke02,PatrascuT07,GoranciRST21}, we believe ours is the first paper to apply them to directed graphs. 
The \emph{existence} of the directed expander hierarchy below follows from generalizing the construction by \cite{PatrascuT07} in undirected graphs. As we will discuss later, however, our construction is entirely different from previous approaches in undirected graphs.

\begin{fact}
\label{fact:hierarchy-existence-overview}
Given any directed graph $G = (V,E)$ and $\phi \leq 1/\polylog(n)$, there exists an expander hierarchy $\cH = (D,X_1, ..., X_\eta)$ such that
\begin{enumerate}
\item $|X_i| = \otil(m\phi^{i-1})$.
\item The total number of levels is around $\log_{1/\phi}(m) = O(\log(n))$.
\end{enumerate}
Note that the first property implies the second. %
\end{fact}

\subsection{Directed Expander Hierarchy Implies Good Weight Function}
The primary technical challenge lies in computing the expander hierarchy of \cref{fact:hierarchy-existence-overview}. But first, in this section, we will show that once we compute such a hierarchy, it implies a good weight function that we can plug into \cref{thm:push-relabel-main-informal}.

\para{Simple Expander}
Let us first consider the very simple case that the entire graph $G$ is a $\phi$-expander (for some $\phi = 1/n^{o(1)}$). In this case, we simply set $\Bw(e) = n$ for all $e \in E$. Note that $\sum_e 1/\Bw(e) = O(n)$; all that remains is to show that $\Bw$ satisfies the path-weight requirement. Let $\fopt$ be the actual maximum flow. The flow $\fopt$ itself may have long flow paths, but we will use the expansion of $G$ to \emph{shortcut} $\fopt$ while only paying a small overhead in congestion. By \Cref{fact:expander-routing-overview}, there exists a flow $\Bf$ such that $\Bf$ routes the same supply/demand as $\fopt$, $\Bf$ has congestion $\otil(1/\phi)$, and every flow path in $\Bf$ contains $\otil(1/\phi)$ edges. Scaling $\Bf$ down by a $\widetilde{O}(\phi)$-factor thus yields a $\widetilde{\Omega}(1/n^{o(1)})$-approximate flow where every flow path $P$ has weight $\Bw(P) = \otil(n/\phi) = n^{1+o(1)}$. Note that the algorithm never explicitly computes $\Bf$; rather, we simply use its existence to argue that $\Bw$ is a good weight function.

\para{DAG of Expanders}
We now consider a slightly more general case, where every SCC of $G$ is a $\phi$-expander, but there can be DAG edges between the SCCs. This corresponds to a one-level expander hierarchy $\cH = (D,X_1)$, where $X_1$ contains the edges inside SCCs of $G$, and $D$ contains the inter-component edges. We say that a topological order $\Btau$ \emph{respects} the SCCs of $G$ if it has the following properties:
\begin{itemize}
\item  For every edge $(u,v) \in D$, we have $\Btau_u < \Btau_v$.
\item For every SCC $C$ of $G$, the set $\Btau(C) \defeq \{\Btau_v: v \in C\}$ is contiguous; in other words, it contains precisely the set of numbers between  $\Btau_{\mathrm{min}}(C) \defeq \min_{v \in C}\Btau_v$ and $\Btau_{\mathrm{max}}(C) \defeq \max_{v \in C}\Btau_v$.
\end{itemize}
It is easy to see that such a respecting $\Btau$ exists. We now define $\Bw(u,v) = |\Btau_v - \Btau_u|$. Not that if $u,v$ are in the same SCC $C$, then $\Bw(u,v) \leq |C|$.
Since the weight function is defined by a topological ordering
 we have that $\sum_e 1/\Bw(e) = O(n\log(n))$. The analysis is exactly the same as for the case when $G$ is a DAG.

We now show that $\Bw$ satisfies the path-weight requirement. Let $\fopt$ be the maximum flow. As before, we start by shortcutting $\fopt$ inside each expander. Formally, for every component $C$ of $X_1$, we apply \cref{fact:expander-routing-overview} to the following flow instance: for every flow path $P$ in $\fopt$, we add one unit of supply to the first vertex in $P \cap C$ and one unit of demand to the last vertex in $P \cap C$. Let $\Bf$ be the flow resulting from shortcutting $\fopt$ inside every SCC $C$. Note that $\Bf$ incurs a congestion of $1/\phi$ and that for every flow path $P$ in $\Bf$, $|P \cap C| = \otil(1/\phi)$.

We now argue that every flow path $P$ in $\Bf$ has $\Bw(P) = \otil(n/\phi)$. First, consider the weight of $X_1 \cap P$, i.e. the intra-component edges. For any component $C$, $P \cap C$ contains $\otil(1/\phi)$ edges, each of weight at most $|C|$, so $\Bw(P \cap C) = \otil(|C|/\phi)$; summing over all components yields weight $\otil(n/\phi)$. For the inter-component edges on $P$, since the topological labels on these edges are monotonically increasing, it is easy to see that their total weight contribution is $O(n)$.

\para{Two-Level Expander Hierarchy}
The next slightly more general case is when the edges of $G$ can be partitioned into a two-level expander hierarchy $(D,X_1,X_2)$: $X_1$ contains edges inside SCCs of $G \setminus X_2$, and each of these SCCs is a $\phi$-expander; $D$ contains edges between SCCs of $G \setminus X_2$; finally, $X_2$ is expanding in $G$. %
This two-level hierarchy is far from the general case because of our assumption that $X_2$ is expanding in $G$; nonetheless, this special case will already contain all of our main ideas for proving that an expander hierarchy implies a good weight function.

The weight function is exactly the same as the previous one: we compute a topological order~$\Btau$ that respects the SCCs of $G \setminus X_2$ and we set $\Bw(u,v) = |\Btau_v - \Btau_u|$. Since $\Bw$ is still based on a topological ordering, we again get $\sum_{e \in E} 1/\Bw(e) = O(n\log(n))$. All that remains is to show that $\Bw$ satisfies the path-weight requirement. To do so, we use the following claim:

\begin{claim}
\label{claim:overview-one-level-hierarchy}
Let $\fopt$ be the optimal maximum flow. There exists a flow $\Bf$ routing the same supply/demand as $\fopt$ does such that:
\begin{enumerate}
    \item\label{item:overview-one-level-hierarchy:congestion} $\Bf$ has congestion $\otil(1/\phi)$. (Actually we get $\cong(f) = 1+\frac{1}{\log(n)}$, but $\otil(1/\phi)$ is good enough.)
    \item Every flow path in $\Bf$ contains $\otil(1/\phi)$ edges from $X_2$.
    \item\label{item:overview-one-level-hierarchy:short} Let $P$ be any flow path in $\Bf$. For every SCC $C$ of $G \setminus X_2$, we have $|P \cap C| = \otil(1/\phi)$. (Recall that the SCCs of $G \setminus X_2$ are precisely the SCCs in which $X_1$ is $\phi$-expanding.)
\end{enumerate}
\end{claim}

Before proving this claim, let us see why it implies that $\Bw$ satisfies the path-weight requirement. Scaling $\Bf$ down by $\cong(\Bf) = \otil(1/\phi)$ we get a feasible approximate flow, as desired. Consider any flow path $P$ in $\Bf$. The path $P$ contains at most $\otil(1/\phi)$ edges from $X_2$, each with weight at most $n$, so the total weight contribution of $X_2 \cap P$ is $\otil(n/\phi)$. For edges that belong to an SCC of $G \setminus X_2$, the analysis is exactly the same as for a DAG of expanders: each component $C$ contributes $\otil(|C|/\phi)$ weight to path $P$, for a total of $\otil(n/\phi)$. Finally, consider the DAG edges in $D$. For any subpath of $P$ that is disjoint from $X_2$, all the edges in $D$ are increasing in terms of the $\Btau$ values, so the total weight of $D$-edges in such a subpath is $O(n)$. Every edge in $X_2$ can then go back to the beginning of the topological order, but since there are only $\otil(1/\phi)$ edges in $P \cap X_2$, the total contribution of $P \cap D$ is $\otil(n/\phi)$.
We now sketch a proof of \cref{claim:overview-one-level-hierarchy}.

\begin{proof}[Proof Sketch of \cref{claim:overview-one-level-hierarchy}]

Recall that we are assuming a two-level hierarchy where $X_2$ is expanding in $G$. By \cref{fact:expander-routing-overview}, we can thus reroute $\fopt$ to a new flow $\Bf_2$ such that $\Bf_2$ has congestion $c_2 = \otil(1/\phi)$ and every flow path in $\Bf_2$ contains at most $\otil(1/\phi)$ edges from $X_2$.

\para{A Na\"ive Approach}
We now need to further shortcut $\Bf_2$ so that it satisfies Property \ref{item:overview-one-level-hierarchy:short}. Consider any SCC $C$ of $G \setminus X_2$, and recall that by definition of expander hierarchy, $C$ is a $\phi$-expander with respect to $X_1$. The na\"ive way to shortcut the flow inside $C$ is to repeat the procedure above: reroute flow from the first vertex in $P \cap C$ to the last, for every flow path $P$. There is, however, a subtle but significant issue with this approach. We are rerouting the flow $\Bf_2$ and not the original flow $\fopt$. Whereas $\fopt$ has congestion $1$, the flow $\Bf_2$ already has congestion $c_2 = \otil(1/\phi)$. For this reason, there could be a vertex $v \in C$ such that for every edge $e$ entering $v$ has a flow $\Bf_2(e) = c_2$ on it. As a result, the flow instance that we used to reroute $C$ could have $\Bsource(v) \approx c_2\deg(v) \approx \frac{1}{\phi}\deg(v)$, which exceeds the maximum specified by \cref{fact:expander-routing-overview}.
We can still apply a scaled version of this fact, but the resulting edge congestion will then be $\widetilde{O}(c_2/\phi) = \widetilde{O}(1/\phi^2)$, instead of $\widetilde{O}(1/\phi)$. At first glance this might seem acceptable, since $1/\phi^2 = n^{o(1)}$. But for general graphs, the hierarchy might have as many as $\log_{1/\phi}(n)$ levels (see \cref{fact:hierarchy-existence-overview}), and the na\"ive shortcutting approach above will multiply the congestion by $1/\phi$ per level, leading to an unacceptably high congestion of $\Omega(n)$.

\para{All-to-All Rerouting With Less Demand Per Edge}
To overcome this issue, we need a more careful shortcutting procedure. Consider again the flow $\Bf_2$ with congestion $c_2 = \otil(1/\phi)$. We will show how to reroute $\Bf_2$ so that Property \ref{item:overview-one-level-hierarchy:short} of the claim is satisfied, while the congestion of the flow only increases to $c_2 \cdot (1+1/\log(n))$. Consider any SCC $C$ of $G \setminus X_2$. Let $\psetshort$ contain all flow paths $P$ of $\Bf_2$ for which $|P \cap C| \leq 2k$, where $k$ is a parameter we will later set to $\otil(1/\phi)$.\footnote{Note that $P$ might enter and leave $C$ multiple times, but we can still consider the first (or last) $k$ vertices of $P \cap C$.} Let $\psetlong$ contain all other flow paths. Note that there is no need to reroute the flow paths of $\psetshort$, as they already satisfy Property \ref{item:overview-one-level-hierarchy:short}. 

We define the following flow instance for rerouting $\psetlong$. For every $P \in \psetlong$, let $\pfirst$ contains the first $k$ vertices of $P$ and $\plast$ the last $k$. 
We reroute from all of $\pfirst$ to all of $\plast$, which will allow us to place less supply/demand on every individual vertex. Formally, we add supply $1/k$ to every vertex in $\pfirst$ and demand $1/k$ to every vertex in $\plast$. Since $\cong(\Bf_2) = c_2$, the supply/demand on every vertex is now at most $c_2/k$.
Therefore, applying (the scaled version of) \cref{fact:expander-routing-overview}, we get a flow $\Bf'$ with short flow paths and congestion $(c_2/k) \cdot \otil(1/\phi)$. To reroute $\psetlong$, we must combine flow $\Bf'$ with the flow along $\pfirst$ and $\plast$, as the new flow must use $\pfirst$ and $\plast$ to reach all the sources and sinks on these segments. We now bound the overall congestion of the resulting flow $\Bf$. Any edge $e \in C$ includes at most $\Bf_2(e) \leq c_2$ from the parts of $\Bf_2$ that have not been rerouted, which includes all the flow from $\psetshort$, as well all the flow from the early and late segments of each path in $\psetlong$.
Also, $e$ gets an additional $(c_2/k) \cdot \otil(1/\phi)$ units of flow from the rerouting.
Together, this results in $\Bf(e) \leq c_2(1 + 1/k \cdot \otil(1/\phi))$.
Setting $k$ to a large enough $\polylog(n)/\phi$ yields $\Bf(e) \leq c_2(1+1/\log(n))$. Since there are $O(\log(n))$ levels in the expander hierarchy of \cref{fact:hierarchy-existence-overview}, the congestion at the final level will still be $O(c_2) = \otil(1/\phi)$.
\end{proof}

\para{Generalizing to a Multi-Level Expander Hierarchy} Let us now consider a general graph $G$, which we know admits a multi-level expander hierarchy $\cH$ as in \cref{fact:hierarchy-existence-overview}. We can obtain a good weight function $\Bw$ using the same tools as in the simpler two-level hierarchy above. 

First, let us say a topological order $\Btau$ respects $\cH = (D,X_1, ..., X_{\eta})$ if for every $i$, the $\Btau$ labels are contiguous in every SCC of $G \setminus X_{> i}$, and for every DAG edge $(u,v)\in D$ we have $\Btau_u < \Btau_v$. It is easy to construct such a topological order by going from the top to the bottom of the hierarchy, and computing SCCs in each $G \setminus X_{>i}$. We then define our weight function as $\Bw(u,v) = |\Btau_u - \Btau_v|$. Since $\Bw$ is defined by a topological order, we again have $\sum_{e} 1/\Bw(e) = O(n\log(n))$.

To prove that $\Bw$ satisfies the path-weight requirements we prove that there exists an approximate maximum flow $\Bf$ such that for every $X_i$ and every SCC $C$ of $G \setminus X_{>i}$, the flow $\Bf$ uses $\otil(1/\phi)$ edges in $X_{i}\cap C$. It is easy to show that such a flow $\Bf$ satisfies the path-weight requirement, and we can show that such a flow $\Bf$ exists by using the careful rerouting procedure above starting at the top level, then the second-highest level, and so on.%

\begin{remark} The flow rerouting above is needed for analysis only. The algorithm never computes $\Bf$; instead, its existence is enough to prove that $\Bw$ satisfies the path-weight requirement. All the algorithm does is to find a hierarchy-respecting topological order $\Btau$ and set $\Bw(u,v) = |\Btau_v - \Btau_u|$.
\end{remark}

\subsection{Constructing the Directed Expander Hierarchy}

As far as we know, our paper is the first to define a directed expander hierarchy. We first contrast our hierarchy with existing work in \emph{undirected} graphs.

\para{Previous Work: Undirected Expander Hierarchy}
The definition of our hierarchy can be thought of as a generalization of the undirected hierarchies of \cite{PatrascuT07}. The problem, however, is that \cite{PatrascuT07} relies on a slow polynomial-time algorithm for constructing the hierarchy. It is possible to use techniques from \cite{RackeST14} to efficiently construct the hierarchy in a top-down manner, but this requires solving a max-flow problem at each level, which we cannot afford to do as we are trying to develop our own efficient combinatorial max-flow algorithm. We thus develop an entirely different bottom-up construction.

A more recent paper of~\cite{GoranciRST21} shows a different undirected expander hierarchy that admits a very efficient bottom-up construction. Their construction is based on boundary-linked expanders, which allow for low-congestion routing between the boundary edges. In \emph{undirected} graphs, we observed that we could have naturally defined a good weight function from their hierarchy. 
Unfortunately, a decomposition into boundary-linked expanders does not exist for \emph{directed graphs}. In undirected graphs, an expander decomposition has a small number (i.e., $\widetilde{O}(\phi m)$) of boundary edges, which is why boundary-linkedness is possible. In directed graphs there can be arbitrarily many boundary edges, because even if a cut $E(S,\overline{S})$ is sparse, there may still be $\Omega(m)$ edges in the other direction $E(\overline{S},S)$.

\para{Our Construction}
We now give an overview of our framework for constructing the expander hierarchy $\cH(D,X_1, ..., X_\eta)$ of \cref{fact:hierarchy-existence-overview}. We proceed in a bottom-up fashion. The first step is to compute a set of edges $X_2$ such that $|X_2| = \widetilde{O}(\phi m)$ and all SCCs of $G \setminus X_2$ are $\phi$-expanders; the edges inside these SCCs then become $X_1$. Loosely speaking, we can compute $X_2$ by repeatedly computing a $\phi$-sparse cut and recursing on both sides (more details below). To construct the next level $X_3$ of the hierarchy, we again need to repeatedly find sparse cuts, but this time they need to be sparse with respect to $X_2$. Here, however, we encounter a potential issue: we may find a sparse cut $E(S, \overline{S})$ which is not a subset of $X_2$. As a result, when we move $E(S, \overline{S})$ to $X_3$, we will end up disturbing lower levels of the hierarchy. Unfortunately, there is no way to avoid this issue; in fact, depending on the choice of $X_2$, there might not even \emph{exist} a cut that is sparse with respect to $X_2$ and whose crossing edges are contained in $X_2$. This lack of nestedness poses a huge technical challenge which we discuss later, but let us bypass it for now and make the following unrealistic assumption:%

\begin{assumption}[Unrealistic Nestedness Assumption]
Whenever we compute a cut $(S,\overline{S})$ that is $\phi$-sparse with respect to some $X_i$, we are in the lucky case where $E(S,\overline{S}) \subseteq X_i$.
\label{assumption:nested}
\end{assumption}

Given the assumption above, we can proceed to construct the whole hierarchy in a bottom-up fashion. The challenge now is to do so efficiently. %
As suggested above, finding the next edge set $X_{i+1}$ requires repeatedly finding cuts that are sparse with respect to $X_i$. In fact, the whole construction can effectively be reduced to the following subroutine:

\para{Sparse-Cut Subroutine}
Given a graph $G = (V,E)$, a set of terminal edges $X_i \subseteq E$, a set of sources $\xsupply \subseteq X_i$, and a set of sinks $\xsink \subseteq X_i$ with $|\xsink| = |\xsupply|$, the algorithm must either:
\begin{enumerate}
\item find a flow $\Bf$ of congestion $\otil(1/\phi)$, where all vertices in $\xsupply$ (resp.,~$\xsink$) have one unit of supply (resp.,~demand) and $\Bf$ routes at least $|\xsupply|/2$ units of flow, or
\item find a cut that is $\phi$-sparse with respect to $X_i$.
\end{enumerate}

If we allow the nestedness assumption above, and have an efficient algorithm for the sparse-cut subroutine, then we can apply the standard approach of combining the subroutine with the celebrated cut-matching game framework~\cite{KhandekarRV06,Louis10} to either locate a sparse cut (without being given as input the $(\xsupply, \xsink)$ pair) in the graph or certify that it is an expander.
By recursing on both sides of the sparse cut, we get an expander decomposition algorithm that computes $X_{i+1}$, and we can further ensure that the total number of invocations of the sparse-cut subroutine is $n^{o(1)}$ using ideas developed in~\cite{NanongkaiS17,Wulff-Nilsen17,NanongkaiSW17}.\footnote{Similar ideas were previously applied to \emph{directed} expander decomposition/pruning in \cite{BernsteinGS20,HuaKGW23}. We make particular use of the algorithmic framework established by \cite{HuaKGW23} later in the paper to handle the unrealistic nestedness assumption.}

\para{Sparse-Cut Subroutine: Level One}
We now describe our implementation of the sparse-cut subroutine, which uses entirely new techniques. Let us start on the bottom level, where the set of terminal edges is $X_i = E$. At this level, the subroutine can easily be done in $\widetilde{O}(m/\phi)$ time using existing techniques (see e.g.~\cite{HenzingerRW17,SaranurakW19}), which we quickly review. The algorithm is quite simple: we run regular (non-weighted) push-relabel to send flow from $\xsupply$ to $\xsink$, except that we allow edges to have capacity up to $\polylog(n)/\phi$, and we impose a maximum vertex label of $h = \otil(1/\phi)$; this artificial maximum might prevent push-relabel from finding a maximum flow. Let $\Bf$ be the flow computed by push-relabel. There are two cases to consider. If $\Bf$ sends at least $|\xsupply|/2$ flow, we are done. If not, let $G_{\Bf}$ be the remaining residual graph, and note that since $\Bf$ has small value, there must exist some $s \in \xsupply$ and some $t \in \xsink$ with $\Bell(s) = h = \otil(1/\phi)$ and $\Bell(t) = 0$. 

Rather than working directly with the labels $\Bell$, our algorithm computes a new labelling $\Bd$, where $\Bd(v) \defeq \dist_{G_{\Bf}}(v,t)$.
By the admissibility property of push-relabel, we have $\Bd(s) \geq h$. Now, for any $k$, define $V_{k} \defeq \{v \in V \mid \Bd(v) = k\}$ and $V_{\geq k} \defeq \{v \in V \mid \Bd(v) \geq k\}$. We refer to cuts $(S,\overline{S})$ of the form $S = V_{\geq k}$ as a \emph{level cut}. We will show that one of the level cuts is a sparse cut in $G_{\Bf}$. For the full proof we need to show that one of the level cuts is sparse in the original graph $G$, but the proof is essentially the same: loosely speaking, since we set edge capacities to be $\polylog(n)/\phi$, the flow $\Bf$ saturates at most $O(\phi|\xsupply|/\polylog(n))$ edges, which is so few that they have minimal effect on the sparseness of a level cut.

To see that one of the level cuts in $G_{\Bf}$ is sparse, we use a so-called ball-growing argument. Consider some level cut $S = V_{\geq k}$, and note that since $\Bd$ corresponds to distances, every edge in $G'$ leaving $V_{\geq k}$ goes to $V_{k-1}$. So if cut $S$ is non-sparse, then there are many edges from $V_{\geq k}$ to $V_{k-1}$, and in particular $\vol(V_{\geq k-1}) \geq \vol(V_{\geq k}) \cdot (1 + \phi)$. Thus, there can be at most $\log_{1 + \phi}(m) = \otil(1/\phi)$ non-sparse layers, so as long as we set $h$ large enough, we can ensure that over half the level cuts are sparse.

\para{Sparse-Cut Subroutine: Level Two}
Let us now consider the case where we have already constructed the first level of the hierarchy $\cH = (D, X_1, X_2)$. To construct the next layer, we need to solve the sparse-cut subroutine with respect to terminal edges $X_2$. This simple case will once again contain most of our main ideas for finding a sparse cut with respect to a general $X_i$.

We can no longer directly use a ball-growing argument. In the simple case above, the crux of the argument was that the edges crossing any non-sparse level cut $S = V_{\geq k}$ get added to the volume of $V_{\geq k-1}$, which guarantees that $\vol(V_{\geq k})$ increased multiplicatively as we move from $k = h$ to $k = 0$. 
The problem is that for the second level of the hierarchy, sparseness is defined with respect to $\vol_{X_2}$, but the edges crossing a non-sparse cut $S = V_{\geq k}$ might not belong to $X_2$, so $\vol_{X_2}(V_{\geq k})$ might not change at all across levels. In order to use ball-growing to argue that there exists a sparse level cut, we will need to reassign vertex levels in such a way that there exist many level cuts whose edges come primarily from $X_2$. %

The key idea is to use the weighted push-relabel algorithm, where the weight $\Bw$ will be based on the incomplete hierarchy we have already built. In particular, let $\Btau$ be a topological order that respects the SCCs of $G \setminus X_2$ and let $\Bw(u,v) = |\Btau_{u} - \Btau_{v}|$. We will now run the weighted push-relabel algorithm up to maximum label $h = \otil(n/\phi)$; by \Cref{thm:push-relabel-main-informal} the runtime is $O(m+h\sum_{e \in E}1/\Bw(e)) = O(m+hn\log(n)) = n^{2+o(1)}$. Let $\Bf$ be the flow computed by weighted push-relabel. If $\Bf$ sends at least $|\xsupply|/2$ flow, then we are done. Otherwise, we once again have vertices $s \in \xsupply$ and $t \in \xsink$ with $\Bell(s) = h = \otil(n/\phi)$ and $\Bell(t) = 0$. As before, define $\Bd(v) \defeq \dist^{\Bw}_{G_{\Bf}}(v,t)$, where $\dist^{\Bw}_{G_{\Bf}}$ is the shortest distance according to $\Bw$ in the residual graph. We know that $\Bd(s) \geq h$.

Now, for the sake of intuition, consider the simplistic case where $\Bf$ is empty, so the residual graph $G_{\Bf} = G$. As discussed above, to argue that there exists a level cut 
that is sparse with respect to $X_2$, we need there to be many level cuts whose edges come primarily from $X_2$. This might not be true under the current labelling $\Bd$ because of the presence of DAG edges, so we define a new weight function $\Bwp$, which is the same as $\Bw$ except that it sets the weight of all DAG edges to $0$. We then define labeling $\Bdp(v) \defeq  \dist^{\Bwp}_{G_{\Bf}}(v,t)$. Even though $\Bdp(s) < \Bd(s)$, we argue that $\Bdp(s) \geq \Bd(s)/2 - n = \Omega(h)$, so we still have many levels. This follows from the fact that under the original weight function $\Bw$, the DAG edges in $D$ always increase $\Btau$, so except for the initial increase from $\Btau = 0$ to $\Btau = n$, any further weight-contribution from $D$ must be balanced by edges in $X_1$ and $X_2$ that move backward in the topological ordering; as a result, $D$ can only account for around half the total weight of a path under $\Bw$, so $\Bdp(s) \gtrsim \Bd(s)/2$.

We thus have a distance labeling $\Bdp$ such that $\Bdp(s) = \Omega(h)$ and none of the level cuts contain any DAG edges (because they have weight $0$). We now argue that most of the level cuts also do not contain any edges from $X_1$.
Recall that the SCCs of $X_1$ are $\phi$-expanders. Consider any SCC $C$ of $X_1$; since for any edge $(u,v) \in C$ we have $\Bwp(u,v) = \Bw(u,v) \leq |C|$, the diameter of $C$ under $\Bwp$ is at most $\otil(|C|/\phi)$, so edges inside $C$ are present in at most $\otil(|C|/\phi)$ different level cuts. Therefore, in total there are at most $\otil(n/\phi)$ level cuts containing edges from $X_1$, and if we set $h = \otil(n/\phi)$ large enough then there will be $\Omega(h)$ level cuts that contain exclusively edges from $X_2$. We can now use a standard ball-growing argument to argue that one of these level cuts is sparse with respect to $X_2$.

Recall that we made the simplifying assumption that $G_{\Bf} = G$. In reality, weighted push-relabel might compute some initial flow $\Bf$, so $G_{\Bf} \neq G$. We now argue that we can still find a cut in the residual graph $G_{\Bf}$ that is sparse with respect to $X_2$, which as already discussed, also yields a sparse cut in $G$. We start by again setting $\Bwp$ to have weight $0$ on all edges in $D$ except the residual edges of flow $\Bf$, and we define distance function $\Bdp$ accordingly. The residual edges of flow $\Bf$ have a small contribution,\footnote{The \emph{residual edges} are those edges which we sent flow along, and are then reversed in the residual graph. Our weighted push-relabel algorithm will guarantee that each augmenting path it finds is of $\Bw$-length $O(h)$, so the contribution of these edges to the level cuts is not too much.} so we ignore them for this overview; as a result, we again have that $\Bdp(s) = \Omega(h)$ and level cuts that contain no edges in $D$.%

Dealing with the edges of $X_1$ is trickier. Consider a SCC $C$ of $X_1$. The problem is that if the flow $\Bf$ saturated some edges in $C$, then those edges are reversed in $G_{\Bf}$, so $C$ might no longer be an expander, and hence might have high diameter. The crux of our analysis is to argue that, as the value of flow $\Bf$ is relatively small, it does not impact the average expander $C$ by too much. 

To argue this, a natural idea is to apply the expander pruning argument %
(see e.g.,~\cite{SaranurakW19,BernsteinGS20,HuaKGW23}). 
In particular, if $\Bf$ saturates $\sigma$ \emph{edges} in $C$, then there exist a pruned set $P_C \subseteq C$ such that $P_C$ has small size $\otil(\sigma/\phi)$ and $C \setminus P_C$ is still a $\Omega(\phi)$-expander. Thus, since $C \setminus P_C$ has small diameter, its edges are once again present in only a minority of level cuts, so the remaining level cuts only contain edges from $X_2$ and from the pruned parts $P_C$. As long as the pruned parts are small we can argue that their impact is minimal, and thus most level cuts contain edges primarily from $X_2$. Again, standard ball-growing techniques prove that one of the remaining level cuts is sparse with respect to $X_2$. Unfortunately, the standard expander pruning technique does not give a small enough pruned set $P_C$.

\para{Technical Highlight: Path-Reversal Pruning}
The remaining challenge is in arguing that the pruned set $P_C$ is small. Let $R$ be some flow path of $\Bf$ that goes through $C$ ($R$ for reversed path). Since $\Bf$ is relatively small, the number of such flow paths is also small. The problem is that $|R|$ can contain many edges, so if we apply standard expander pruning by simply deleting all of $R$, the resulting pruned set $P_C$ will be too large. 

To overcome this challenge, we introduce a new technique---\textit{path-reversal pruning}---which we believe might find other applications. Note that $R$ is not actually deleted from the residual graph $G_{\Bf}$; instead, its edges are reversed. Reversing an entire path only changes the size of any directed cut by at most $1$, and so intuitively it should not affect expansion by too much. We are able to show that from the perspective of pruning, reversing an entire path (no matter the length) has approximately the same impact as deleting a single edge. In particular, we prove that if we reverse $\sigma$ different paths $R_1,...,R_{\sigma}$, then there exists a pruned set $P_C$ such that $C \setminus P_C$ is still an expander and the size of $P_C$ is roughly $\sigma /\phi$, rather than $ \sum_{i=1}^{\sigma}|R_i|/\phi$ given by previous pruning guarantees. The technical details end up being quite different from standard pruning.

\begin{remark}
Note that the algorithm itself never performs any pruning. All it does is: compute a flow $\Bf$ using weighted push-relabel, change the weight of the DAG edges to $0$, compute new distance labels $\Bdp$ using Dijkstra's algorithm, and then check all the level cuts until it finds a sparse one.
Pruning is used only in the analysis to argue that one of the level cuts is indeed sparse.
\end{remark}

\para{Edge Capacities} All of the analysis and expander decomposition tools generalize almost seamlessly to capacitated graphs. To make our weighted push relabel algorithm still efficient in capacitated graphs we use dynamic trees \cite{SleatorT83}, similar to what is done for a standard push relabel \cite{GoldbergT88}.

\subsection{Removing the Unrealistic Nestedness Assumption}

Until now, we have assumed \cref{assumption:nested} that when we compute a $\phi$-sparse cut $S$ with respect terminal edge set $X_i$, we always have $E(S,\overline{S}) \subseteq X_i$, i.e., the cut edges consist only of the terminal edges. 
Unfortunately, there are many counterexamples showing this assumption is impossible.
Without \cref{assumption:nested}, the following issue occurs:
once a sparse cut $S$ in $G[U]$ is found, our algorithm needs to further recurse on both sides $S$ and $U \setminus S$, yet if the cut contains non-terminal edges, then we no longer have an expander hierarchy of $G[S]$ and $G[U\setminus S]$ from lower levels that our flow algorithm needs when performing the recursions.

\ifnum\cameraready=0
In \cref{sec:nested-expander-decomposition} we address this problem.
\else
In \cite[Section 7]{BernsteinBST24} we address this problem.
\fi
In particular, instead of fixing the $i$-th level expanding edges $X_i$ once it is computed, we allow edges to be moved between different levels in the hierarchy to ensure nestedness.
Similarly, our algorithm also moves between levels and may attempt to find further sparse cuts following edge movements.
To modularize the analysis, we employ a data structure point-of-view that models these interactions between levels.
We adapt the framework of \cite{HuaKGW23} to maintain a single-level expander decomposition when edges are moving between levels.
However, unlike the analysis of \cite{HuaKGW23}, our approach is not inherently \emph{dynamic} in the sense that we do not exploit any local property of the weighted push-relabel algorithm we developed.
Instead, our focus is on arguing that the total number of updates given to these data structures is small throughout the construction of the hierarchy.
That is, in contrast to achieving a local and sublinear update time as in the dynamic graph algorithm literature and previous maximum flow algorithms, our data structure spends  $n^{2+o(1)}$ time \emph{per update}, which when combined with the analysis that there are only $n^{o(1)}$ updates results in the final running time.
\ifnum\cameraready=0
We defer a more detailed overview of our approach to \cref{subsec:overview}.
\else
We defer a more detailed overview of our approach to \cite[Section 7.1]{BernsteinBST24}.
\fi

We acknowledge that our current construction
\ifnum\cameraready=0
(spanning more than 40 pages in \cref{sec:nested-expander-decomposition})
\else
(spanning more than 40 pages in \cite[Section 7]{BernsteinBST24})
\fi
seems overly involved (unlike the otherwise relatively simple algorithm parts of our paper) and we believe that with future developments of directed expander-related techniques this can be greatly simplified.
We also emphasize that this step of avoiding non-nested cuts is the only reason why our algorithm is randomized\footnote{We also use a randomized cut-matching game from \cite{KhandekarRV06,Louis10}, but that can be easily replaced with a deterministic counterpart~\cite{BernsteinGS20}.} and has an inherent subpolynomial overhead.\footnote{Technically speaking, most current directed expander decomposition algorithms run in almost-linear instead of near-linear time. However, with the recent work of \cite{SulserP24} it seems promising that one can adopt their techniques in combination with our push-relabel algorithm to achieve a $\widetilde{O}(n^2)$ construction, at least if \emph{assuming the unrealistic nestedness assumption}.}%

\section{Preliminaries}\label{sec:prelim}

\para{General Notation}
We use $\N$ to denote the set of \emph{nonnegative} integers.
Let $[k]$ for $k \in \N$ be $\{1, \ldots, k\}$, and in particular $[0] \defeq \emptyset$.
For a collection of sets $\{S_i\}_{\ell \leq i \leq r}$ indexed by integers, let $S_{\leq j} \defeq \bigcup_{\ell \leq i \leq j}S_i$ and $S_{\geq j} \defeq \bigcup_{j \leq i \leq r}S_i$, and define $S_{<j}$ and $S_{>j}$ analogously.
We let $\Ba \leq \Bb$ for vectors $\Ba$ and $\Bb$ act entry-wise.
For a vector $\Bx \in \R^{U}$ we may write $\Bx(S) \defeq \sum_{u \in S}\Bx(u)$ for $S \subseteq U$.
We use $\Bzero$ and $\Bone$ to denote the all-zero and all-one vectors whose dimensions shall be clear from context.

We say an event happens \emph{with high probability} if it does with probability at least $1 - n^{-c}$ for an arbitrarily large (but fixed) constant $c > 0$.

\para{Graphs}
Graphs in this paper are assumed to be directed.
Unless explicitly stated to be \emph{simple}, multi-edges are allowed.
Let $G = (V, E)$ be a graph.
For disjoint subsets $A, B \subseteq V$, let $E_G(A, B) \defeq \{(u, v): u \in A, v \in B\}$.
Let $\delta_G^{+}(v) \defeq E_G(\{v\}, V \setminus \{v\})$ and $\delta_G^{-}(v) \defeq E_G(V \setminus \{v\}, v)$ be the \emph{outward} and \emph{inward} edges incident to $v$.
Let $\delta_G(v) \defeq \delta^{+}(v) \cup \delta^{-}(v)$.
When clear from context, let $\overline{S}$ for $S \subseteq V$ be $V \setminus S$.
For instance, we write $E_{G[U]}(S, \overline{S}) \defeq E_{G[U]}(S, U \setminus S)$ for $U \subseteq V$.
Let $\rev{G} = (V, \rev{E})$ be $G$ where all edges are reversed, i.e., $\rev{E} \defeq \{(v, u): (u, v) \in E\}$. The \emph{edge-vertex incidence matrix} $\BB_G \in \{-1, 0, 1\}^{E \times V}$ of $G$ is given by $\BB_G(e, u) = 1$ and $\BB_G(e, v) = -1$ for each $e = (u, v) \in E$ with all other entries set to zero.

A graph $G$ is \emph{strongly connected} if $E_G(S, \overline{S}) \neq \emptyset$ for every $\emptyset \neq S \subsetneq V$.
A \emph{strongly connected component} of $G$ is a maximal strongly connected subgraph of $G$.
Let $\SCC(G)$ denote the collection of strongly connected components of $G$.
An edge set $F \subseteq E$ is a \emph{separator} of $G$ if no edge in $F$ has both its endpoints in the same strongly connected components of $G \setminus F$.

\para{Capacitated Graphs}
We consider capacitated graphs $(G, \Bc)$ with capacities $\Bc \in \N^E$.
Unless stated otherwise, throughout this paper by standard capacity scaling
\ifnum\cameraready=0
(see \cref{appendix:capacity-scaling})
\else
(see \cite[Appendix B]{BernsteinBST24})
\fi
we assume $\Bc(e) \leq n^2$ for all $e \in E$.
For a subgraph $H \subseteq G$, we may overload notation and $(H, \Bc)$ to denote a capacitated graph with capacities $\Bc$ restricted $H$.
For $F \subseteq E$, let $\deg_{F,\Bc}^{+}(v) \defeq \sum_{e \in \delta^{+}(v) \cap F}\Bc(e)$, $\deg_{F,\Bc}^{-}(v) \defeq \sum_{e \in \delta^{-}(v) \cap F}\Bc(e)$, and $\deg_{F, \Bc}(v) \defeq \deg_{F,\Bc}^{+}(v) + \deg_{F,\Bc}^{-}(v)$ be the sum of capacities of edges in $F$ incident to $v$.
Let $\vol_{F,\Bc}(S) \defeq \sum_{v \in S}\deg_{F, \Bc}(v)$ for $S \subseteq V$.
When $G$ is clear from context, let $\deg_{\Bc}(v) \defeq \deg_{E, \Bc}(v)$ and $\vol_{\Bc}(S) \defeq \vol_{E, \Bc}(S)$.
When the graph is unit-capacitated, i.e., $\Bc = \Bone$,
we drop the subscript $\Bc$ in the above notation which recovers the standard definitions of degree and volume.
For analysis it is oftentimes simpler to work with unit-capacitated graphs.
Let $G^{\Bc}$ be $G$ where each edge $e$ is duplicated $\Bc(e)$ times.
For $F \subseteq E$, let $F^{\Bc} \subseteq E(G^{\Bc})$ be the multi-subset of $E(G^{\Bc})$ that contains precisely the duplicates of edges in $F$.
It is easy to see that the above definitions are equivalent in $(G, \Bc)$ and $G^{\Bc}$.
Let $G^{c}$ for $c \in \N$ be $G^{c\Bone}$ and $F^{c}$ be $F^{c\Bone}$.

\para{Flows}
A \emph{flow instance} $\cI$ is a tuple $\cI = (G, \Bc, \Bsource, \Bsink)$ where $G = (V, E)$ is a graph with edge capacities $\Bc \in \N^E$, $\Bsource \in \R_{\geq 0}^{V}$ is the source vector, and $\Bsink \in \R_{\geq 0}^{V}$ is the sink vector.
Without stated otherwise, we further assume $\|\Bsource\|_1 \leq \|\Bsink\|_1$, i.e., $\cI$ is a \emph{diffusion} instance.
When unspecified, we assume the graph is unit-capacitated, i.e., $\Bc = \Bone$.
Consider a vector $\Bf \in \Q_{\geq 0}^{E}$.\footnote{In general, flows in graphs can take real values on edges. However, our algorithms and analyses will always work with flows of rational values, and in particular restricting $\Bf$ to be in $\Q^E$ allows us to treat a fractional flow as an integral flow in the graph in which edges are duplicated, making our analyses cleaner.}
The \emph{absorption} of $\Bf$ is $\abs_{\Bf} \defeq \min\{-\BB_G^\top \Bf + \Bsource, \Bsink\}$, where the $\min$ operator is defined entry-wise.
The \emph{excess} of $\Bf$ is $\ex_{\Bf} \defeq -\BB_G^\top \Bf + \Bsource - \abs_{\Bf} = \max\{ -\BB_G^\top \Bf + \Bsource - \Bsink, \Bzero\}$.
The \emph{value} of $\Bf$ is $|\Bf| \defeq \abs_{\Bf}(V) = \|\Bsource\|_1 - \ex_{\Bf}(V)$.
Let $\Bfout \defeq \BB_G^\top \Bf$ and so $\Bfout(v)$ is the net flow going out of $v$.
The vector $\Bf$ is a \emph{$(\Bsource, \Bsink)$-flow}, or simply a \emph{flow}, if $\Bzero \leq \ex_{\Bf} \leq \Bsource$.
The \emph{congestion} of $\Bf$ is $\cong(\Bf) \defeq \|\Bf/\Bc\|_{\infty}$.
A flow is \emph{feasible} if $\Bf \leq \Bc$ or equivalently $\cong(\Bf) \leq 1$.
The flow $\Bf$ \emph{routes} $\cI$ (or routes the demand $(\Bsource, \Bsink)$) if $|\Bf| = \|\Bsource\|_1 = \|\Bsink\|_1$, and we say $\cI$ is \emph{routable} with congestion $\kappa$ if $\cong(\Bf) \leq \kappa$ for such a flow (or simply \emph{routable} if $\kappa \le 1$).
Two flows $\Bf_1$ and $\Bf_2$ are \emph{equivalent} if they route the same demand, i.e., $\BB_G^\top \Bf_1 = \BB_G^\top \Bf_2$.

\begin{fact}
  For any flow $\Bf$ and $S \subseteq V$ it holds that $\Bsource(S) = \abs_{\Bf}(S) + \Bfout(S) + \ex_{\Bf}(S)$.
  \label{fact:flow}
\end{fact}

Given a flow $\Bf$, the \emph{residual graph} $G_{\Bf}$ contains for each $e = (u, v) \in E$ a \emph{forward} edge $\overrightarrow{e} = (u,v)$ with capacity $\Bc_{\Bf}(\forward{e}) \defeq \Bc(e) - \Bf(e)$ if $\Bc_{\Bf}(\forward{e}) > 0$ and a \emph{backward} edge $\rev{e} = (v, u)$ with capacity $\Bc_{\Bf}(\backward{e}) \defeq \Bf(e)$ if $\Bc_{\Bf}(\backward{e}) > 0$.
For $F \subseteq E$, let $\forward{F} \defeq \{\forward{e}: e \in F\}$ and $\backward{F} \defeq \{\backward{e}: e \in F\}$.
Let $\Bsource_{\Bf} \defeq \ex_{\Bf}$ and $\Bsink_{\Bf} \defeq \Bsink - \abs_{\Bf}$ be the \emph{residual sources} and \emph{residual sinks}.
A source $s$ with $\Bsource(s) > 0$ is \emph{unsaturated} by $\Bf$ is $\Bsource_{\Bf}(s) > 0$; likewise, a sink $t$ with $\Bsink(t) > 0$ is \emph{unsaturated} if $\Bsink_{\Bf}(t) > 0$.
Together this defines the residual flow instance $\cI_{\Bf} = (G_{\Bf}, \Bc_{\Bf}, \Bsource_{\Bf}, \Bsink_{\Bf})$.
An \emph{augmenting path} is a path in $G_{\Bf}$ consisting of edges with positive residual capacities from an unsaturated source to an unsaturated sink.
The following standard fact justifies the use of residual graphs.

\begin{fact}
  For any feasible flow $\Bf$ of $\cI$, it holds that if $\Bf^\prime$ is a maximum flow of $\cI_{\Bf}$ then $\Bf + \Bf^\prime$ is a maximum flow of $\cI$.
  \label{fact:flow-in-residual-graph}
\end{fact}

A flow $\Bf$ is \emph{integral} if $\Bf \in \N^E$.
Otherwise, $\Bf$ is \emph{fractional} and \emph{$\frac{1}{z}$-integral} for $z \in \N$ such that $\Bf \in (\frac{1}{z}\cdot \N)^E$.
When $\Bf$ is $\frac{1}{z}$-integral, we often equivalently view it as an integral flow in the unit-capacitated $G^{(z \cdot \Bc)}$ and decompose it into a collection of flow paths through the following standard fact.
Let $\Bf_P$ for $P$ a path in $G$ be the flow that sends one unit of flow along $P$, i.e., $\Bf(e) = 1$ for all $e \in P$.
While the capacitated perspective allows for faster algorithms, the unit-capacitated one is sometimes easier to work with for analysis, as demonstrated by, e.g., the following standard fact.

\begin{fact}
  An integral flow $\Bf$ admits a path decomposition $\cP_{\Bf} \defeq \{P_1, \ldots, P_{|\Bf|}\}$ such that $\Bf^\prime \defeq \Bf_{P_1} + \cdots + \Bf_{P_{\Bf}}$ is equivalent to $\Bf$ and satisfies $\Bf^\prime \leq \Bf$.
  \label{fact:path-decomposition}
\end{fact}

The following equivalence between maximum flow and minimum cut is standard.

\begin{fact}[Max-flow min-cut theorem]
  \label{fact:maxflow-mincut}
  For a flow instance $\cI = (G, \Bc, \Bsource, \Bsink)$ the maximum flow value is equal to
  \[
    \min_{S \subseteq V}\Bc(E_G(S, \overline{S})) + \Bsource(\overline{S}) + \Bsink(S).
  \]
\end{fact}

The \emph{maximum $(s, t)$-flow} or simply the \emph{maximum flow} problem is to find a maximum $(\Bsource_{s}, \Bsink_{t})$-flow with $\Bsource_{s} \defeq \infty \cdot \Bone_s$ and $\Bsink_{t} \defeq \infty \cdot \Bone_t$.

\para{Weights and Distances}
Consider some edge weights $\Bw \in \N^E$.
Let $\dist_G^{\Bw}(s, t)$ be the shortest $(s, t)$-distance in $G$ with respect to $\Bw$.
This is also referred to as the \emph{$\Bw$-distance} between $s$ and $t$ in $G$.
Let $\dist_{G}^{\Bw}(S, T)$ for $S, T \subseteq V$ be $\min_{s \in S, t \in T}\dist_G^{\Bw}(s, t)$.
For any flow $\Bf$, we often extend $\Bw$ to assign the same weight $\Bw(e)$ to both $\forward{e}$ and $\backward{e}$ in $G_{\Bf}$ when referring to $\dist_{G_{\Bf}}^{\Bw}(s, t)$.
The \emph{weight} of a flow is $\Bw(\Bf) \defeq \sum_{e \in E}\Bw(e)\Bf(e)$.
The \emph{$\Bw$-length} of a path $P$ is $\sum_{e \in P}\Bw(e)$.
For $F \subseteq E$, the \emph{$F$-distance} and \emph{$F$-length} are defined as the $\Bw_F$-distance and $\Bw_F$-length for $\Bw_F(e) = 1$ for $e \in F$ and $\Bw_F(e)$ for $e \in E \setminus F$.

\para{Expanders}
Consider first a strongly connected capacitated graph $(G, \Bc)$ and vertex weights $\Bnu \in \R_{\geq 0}^{V}$.
A cut $\emptyset \neq S \subsetneq V$ is \emph{$\phi$-sparse} \emph{with respect to $\Bnu$ in $(G, \Bc)$} if $\min\{\Bc(E_G(S,\overline{S})), \Bc(E_G(\overline{S}, S))\} < \phi \cdot \min\{\Bnu(S), \Bnu(\overline{S})\}$.
We say that $\Bnu$ is \emph{$\phi$-expanding} in $(G, \Bc)$ if there is no $\phi$-sparse cut in $G$ with respect to $\Bnu$.
For $G$ that is not necessarily strongly connected, we say that $\Bnu$ is \emph{$\phi$-expanding} in $(G, \Bc)$ if $\Bnu$ restricted to $U \subseteq V$ is $\phi$-expanding in $(G[U], \Bc)$ for every strongly connected component $U$ of $G$.
An edge set $F \subseteq E$ is \emph{$\phi$-expanding} if $\deg_{F, \Bc}$ is $\phi$-expanding in $G$.
We may sometimes overload notation and say that $F$ is $\phi$-expanding in a subgraph $H \subseteq G$ if $\deg_{F,\Bc}$ restricted to $V(H)$ is $\phi$-expanding in $H$.
When the graph is unit-capacitated, i.e., when $\Bc = \Bone$, we may drop the vector $\Bc$ in the notation.
A \emph{$\phi$-(pure)-expander} is a $(G, \Bc)$ in which $E(G)$ is $\phi$-expanding.
For analysis of our algorithm, we often make use of the following equivalence between uncapacitated and capacitated expanders.

\begin{fact}
  \label{fact:equivalence}
  An edge set $F$ is $\phi$-expanding in $(G, \Bc)$ if and only if $F^{\Bc}$ is $\phi$-expanding in $G^{\Bc}$.
\end{fact}

\para{Embedding}
An \emph{embedding} $\Pi_{H \to G}$ from $(H, \Bc_H)$ to $(G, \Bc_G)$ where $V(H) \subseteq V(G)$ maps each $e = (u, v) \in E(H)$ to a $(u, v)$-path $\Pi_{H \to G}(e)$ in $G$.
The \emph{congestion} of $\Pi_{H \to G}$ is
\[ \cong(\Pi_{H \to G}) \defeq \max_{e_G \in E(G)}\frac{\sum_{e_H \in E(H): e_G \in \Pi_{H \to G}(e_H)}\Bc_H(e_H)}{\Bc_G(e_G)}. \]

\para{Cut-Matching Game}
The cut-matching game is a framework for constructing expanders from the interaction of two players: the cut player and the matching player.
Suppose we want to construct an expander over vertices $V$ starting from an initially empty graph.
The game proceeds in rounds, and in each round the cut player first computes a bisection $(A, B)$ of $V$, and then the matching player returns a (perfect) matching from $A$ to $B$, which is then added to the graph.
The goal of the cut player is to compute the bisections in such a way that after a small number of rounds, the resulting graph becomes an expander regardless of what perfect matchings the matching player returns.
\cite{Louis10} extended the randomized cut player for undirected graphs and its analysis from \cite{KhandekarRV06} to work in directed graphs.
This can be straightforwardly generalized to the capacitated case.
For vertex weights $\Bnu_A$ and $\Bnu_B$ with $\|\Bnu_A\|_1 \leq \|\Bnu_B\|_1$, a \emph{$(\Bnu_A, \Bnu_B)$-perfect (capacitated) matching} is an $(M, \Bc_M)$ such that $\deg^{+}_{M,\Bc_M}(v) = \Bnu_A(v)$ and $\deg^{-}_{M,\Bc_M}(v) \leq \Bnu_B(v)$ for all $v \in V$.

\begin{restatable}[\cite{KhandekarRV06,Louis10}]{theorem}{CutMatching}
  Given $n$ vertices $V$ and a vector $\Bnu \in \N^V$ with entries bounded by $U$, there is a randomized algorithm that computes in sequence $t_{\CMG} = O(\log^2 (nU))$ vector pairs $(\Bnu_A^{(i)}, \Bnu_B^{(i)})$ with $\Bnu_A^{(i)} + \Bnu_B^{(i)} \leq \Bnu$ and $\|\Bnu_A^{(i)}\|_1 \leq \|\Bnu_B^{(i)}\|_1$ such that if it is given $(\Bnu_A^{(i)}, \Bnu_B^{(i)})$-perfect capacitated matching $(M_i, \Bc_i)$ after it outputs each $(\Bnu_A^{(i)}, \Bnu_B^{(i)})$, then in the end it outputs a $\psi_{\CMG}$-expander $(W, \Bc_W)$ with edges $M_1 \cup \cdots \cup M_{t_{\CMG}}$ such that $\Bnu(v) \leq \deg_{W,\Bc_W}(v) \leq t_{\CMG} \cdot \Bnu(v)$ for all $v \in V$, where $\psi_{\CMG} = \Omega\left(\frac{1}{\log^2 (nU)}\right)$.
  The algorithm runs in $\widetilde{O}(n + |M_1| + \cdots + |M_{t_{\CMG}}|)$ time.
  \label{thm:directed-cut-matching-game}
\end{restatable}

\section{Push-Relabel Algorithm}\label{sec:push-relabel}

Suppose that $G = (V,E)$ is a directed graph, which we want to solve the maximum flow problem on.
In this section, we will also assume that we are given a \emph{weight function} $\Bw \in \N^{E}$
on the edges as additional input. This weight function will serve as a ``hint'' and will help us to find a good approximate flow more efficiently.
In an ideal world, we would want the weight function to satisfy the following properties:

\begin{itemize}
\item\label{item:weight-good-short-flow} There is some ``short'' flow $\Bf$ which is a good approximation to the optimal maximum flow. With ``short'', we mean
that the average $\Bw$-length $\frac{\Bw(\Bf)}{|\Bf|}$ is something like $\tO(n)$.
\item\label{item:weight-good-small-sum} The sum $\sum_{e \in E}\frac{1}{\Bw(e)}$ is ``small'', something like $\tO(n)$.
\end{itemize}

The goal of this section is to design a version of the \emph{push-relabel}\footnote{Also sometimes called \emph{preflow-push}, although our version will maintain proper flows and not preflows.} algorithm~\cite{GoldbergT88}, that, when the above properties are fulfilled, will find a constant-approximation to the maximum flow efficiently. Hence, given the following \cref{thm:push-relabel-main-theorem}, solving the maximum flow problem in $n^{2+o(1)}$
time reduces to efficiently finding a ``good'' weight function $\Bw$.

\begin{theorem}[Push-Relabel] \label{thm:push-relabel-main-theorem}
  Suppose we have a maximum flow instance $\cI = (G,\Bc,\Bsource,\Bsink)$ consisting of an $n$-vertex $m$-edge directed graph $G = (V, E)$, edge capacities $\Bc\in \N^E$, and integral source and sink vectors $\Bsource,\Bsink \in \mathbb{\N}^{V}$. Additionally, suppose we have a weight function $\Bw \in \N_{>0}^E$ and height parameter $h \in \N$. Then there is an algorithm---\cref{alg:push-relabel}: \emph{\alg{PushRelabel}{$G, \Bc, \Bsource, \Bsink, \Bw, h$}}---that in $\tO\left(m + n + \sum_{e \in E}\frac{h}{\Bw(e)} \right)$ time finds a feasible integral flow $\Bf$ such that

  \begin{enumerate}[(i)]
    \item\label{item:push-relabel:invariant} the $\Bw$-distance in the residual graph $G_{\Bf}$ between any unsaturated source $s$ $(\Bsource_{\Bf}(s) > 0)$ and any unsaturated sink $t$ $(\Bsink_{\Bf}(t) > 0)$ is at least $\dist^{\Bw}_{G_{\Bf}}(s,t) > 3h$,
    \item\label{item:push-relabel:short-flow} the average $\Bw$-length of the flow is $\frac{\Bw(\Bf)}{|\Bf|} \le 9h$, and
    \item\label{item:push-relabel:approximation} $\Bf$ is a $\frac{1}{6}$-approximation of $\Bf^{*}_{\Bw,h}$---the optimal (not necessarily integral) flow with average $\Bw$-length $\frac{\Bw(\Bf^{\star}_{\Bw,h})}{|\Bf^{\star}_{\Bw,h}|}\le h$.%
  \end{enumerate}
\end{theorem}

\subsection{Push-Relabel Finds an Approximate Short Flow}
Before proving \cref{thm:push-relabel-main-theorem} fully, we show how \labelcref{item:push-relabel:approximation} is implied by \labelcref{item:push-relabel:invariant,item:push-relabel:short-flow}.

\begin{lemma}
  Let $\Bf^\star$ be a (possibly fractional) feasible $(\Bsource,\Bsink)$-flow where $\Bw(\Bf^{\star}) \le |\Bf^{\star}|\cdot h$, and let $\Bf$ be a (possibly fractional) feasible $(\Bsource,\Bsink)$-flow which satisfies \labelcref{item:push-relabel:invariant,item:push-relabel:short-flow} of \cref{thm:push-relabel-main-theorem};
  then $|\Bf| \ge \frac{1}{6}|\Bf^{\star}|$.
  \label{lemma:push-relabel:approximation}
\end{lemma}

\begin{proof}
  We extend the graph $G$ to $G'$ by adding a super-source $s$ and super-sink $t$, and adding edges $(s,v)$ and $(v,t)$ with capacities $\Bc(s,v) = \Bsource(v)$ respectively $\Bc(v,t) = \Bsink(v)$ and weights $\Bw(s,v) = \Bw(v,t) = 0$. This lets us now consider the $(s,t)$-flow problem with $\Bsource' = \infty \cdot \Bone_{s}$ and $\Bsink' = \infty \cdot \Bone_{t}$. Similarly, the flow $\Bf$ and $\Bf^{\star}$ can be extended to the graph $G'$, by setting
  $\Bf(s,v) = \Bsource(v)-\ex_{\Bf}(v)$
  and $\Bf(v,t) = \abs_{\Bf}(v)$ and similarly for $\Bf^{\star}$.
  
  Assume for contradiction that $|\Bf| < \frac{1}{6}|\Bf^{\star}|$.
  Consider the flow $\Bf^{\prime}$ in the residual graph $G'_{\Bf}$ where we first send $\Bf$ backward, making the residual graph equal to $G'$, and then send $\Bf^{\star}$ forwards (i.e., $\Bf^{\prime} = \Bf^{\star}-\Bf$).
  We note that $\Bf'$ is a feasible flow in $G'_{\Bf}$, since
  $\ex_{\Bf'}(s) = \infty$ and $\ex_{\Bf'} (v) =0$ for all $v\neq s$, so we  have $\Bzero\le \ex_{\Bf'}\le \Bsource'_{\Bf} = \infty\cdot \Bone_{s}$.\footnote{In the original graph $G$, the flow $\Bf'$ would not necessarily be feasible in $G_{\Bf}$ as $\ex_{\Bf^{\star}}(v)$ could be less than $\ex_{\Bf}(v)$ for some vertex~$v$. This is the reason why we work in $G'$ instead.}
  We know that $|\Bf^{\prime}| = |\Bf^{\star}| - |\Bf| > \frac{5}{6} \Bf^{\star}$.
  We also know that, by definition, $\Bw(\Bf^{\prime}) \le \Bw(\Bf^{\star}) + \Bw(\Bf)$ with $\Bw(\Bf) \leq 9h|\Bf|$, by \labelcref{item:push-relabel:short-flow}. 
  This gives
  \[
    \frac{\Bw(\Bf^{\prime})}{|\Bf^{\prime}|} \leq
    \frac{|\Bf^{\star}| \cdot h + |\Bf^{}| \cdot 9h}{\frac{5}{6}|\Bf^{\star}|}
    <\frac{|\Bf^{\star}| \cdot h + |\Bf^{\star}| \cdot \frac{3}{2}h}{\frac{5}{6}|\Bf^{\star}|}
    = 3h
  \]
  meaning that, by an averaging argument, there must exist an $(s,t)$-path in the residual graph $G'_{\Bf}$ (and thus also a path in $G_{\Bf}$ between an unsaturated source and unsaturated sink) of length less than $3h$. However, this contradicts \labelcref{item:push-relabel:invariant}.
\end{proof}

\begin{remark}
  It is worth noting that the reference flow $\Bf^{\star}$ in \cref{lemma:push-relabel:approximation} needs not be integral.
  Nevertheless, this does not contradict the large integrality gap or the hardness-of-approximation results of the ``bounded-length flow polytope'' (see, e.g., \cite{GuruswamiKRSY03,BaierEHKKPSS10}), since the flow $\Bf$ we find will have length slightly larger than $h$ (i.e., the $9h$ term in \cref{thm:push-relabel-main-theorem}\labelcref{item:push-relabel:short-flow}).
\end{remark}

An immediate corollary of the \cref{lemma:push-relabel:approximation} is that the existence of short, possibly fractional flow implies the existence of short integral flow.

\begin{corollary}
  If there is a (possibly fractional) feasible $(\Bsource, \Bsink)$-flow $\Bf^{\star}$ where $\Bw(\Bf^{\star}) \leq |\Bf^{\star}| \cdot h$, then there is an integral feasible $(\Bsource, \Bsink)$-flow $\Bf$ with $|\Bf| \geq \frac{1}{6}|\Bf^{\star}|$ with $\Bw(\Bf) \leq |\Bf| \cdot O(h)$.
  \label{cor:fractional-implies-integral}
\end{corollary}

\begin{proof}
  Let $\Bf$ be the integral flow obtained by repeatedly finding augmenting paths $P$ of weight $\Bw(P) \leq 3h$ until no such paths exist.
  The flow $\Bf$ clearly satisfies \cref{thm:push-relabel-main-theorem}\labelcref{item:push-relabel:invariant,item:push-relabel:short-flow}.
  The corollary follows from \cref{lemma:push-relabel:approximation}.
\end{proof}

\subsection{Implementation}
\label{sec:push-relabel-implementation}

We now present the pseudocode in \cref{alg:push-relabel}. Our implementation differs from a textbook push-relabel algorithm in the following ways:

\begin{itemize}
  \item We restrict our algorithm to $9h$ levels; vertices $v$ with level $\Bell(v) > 9h$ are marked as \emph{dead}.
  \item Our algorithm allows for edge-length $\Bw(e)$ for each edge. While a textbook push-relabel algorithm can send flow on \emph{admissible} edges $(u,v)$ where the level $\Bell(u) = \Bell(v)+1$, we instead call an edge \emph{admissible} when $\Bell(u) \approx \Bell(v) + \Bw(e)$. This is useful to obtain the faster running time, since, as we will see, an edge $e$ only changes between being admissible/inadmissible $O\left(\frac{h}{\Bw(e)}\right)$ times.
  \item Our algorithm employs an aggressive relabeling rule: as long as some vertex (which has no unsaturated sink capacity) does not have any \emph{admissible} outgoing edge, we relabel it (even if it does not have any excess flow).
  \item The above point means that whenever we find an augmenting flow path, we can push a unit of flow all the way from a source to a sink directly, and that the length of this flow path is only $O(h)$ (allowing us to argue \labelcref{item:push-relabel:short-flow} and hence also \labelcref{item:push-relabel:approximation}). Another consequence is that the flow $\Bf$ maintained by the algorithm will always be a proper flow, and not a \emph{preflow}, as is usual in push-relabel implementations.
\end{itemize}

\begin{remark}
\label{remark:push-prioritized}
 Even on a directed path of length $n$, our push-relabel algorithm would require $\Omega(n^{2})$ time. This is unlike most variants of push-relabel that usually prioritize pushing instead of relabeling, which would take $O(n)$ time on a path.
 While our relabel-prioritized variant can compute an approximate maximum flow with small average length, which is crucial for us, it also becomes our bottleneck.
This raises the exciting question of whether a weighted version of the push-prioritized push-relabel algorithm can be devised so that, given a good weight function (or something similar), it runs in $m^{1+o(1)}$ time and computes a $(1/n^{o(1)})$-approximate maximum flow. In a graph with unit vertex capacities, a weight function induced by the DAG of the maximum flow will guide a push-prioritized algorithm to run in linear time, but as of now it is unclear how to identify such a ``good'' weight function without first computing the maximum flow.
\end{remark}

Recall that the push relabel algorithm runs in the residual graph $G_{\Bf} = (V, \forward{E}\cup \backward{E})$ which is defined as follows: for each edge $e = (u,v)\in E$, we have a forward edge $\forward{e} = (u,v) \in \forward{E}$ and a backward edge $\backward{e} = (v,u)\in \backward{E}$ with residual capacities $\Bc_{\Bf}(\forward{e}) = \Bc(e)-\Bf(e)$ and $\Bc_{\Bf}(\backward{e}) = \Bf(e)$. We will often use $e$ and $\forward{e}$ interchangeably (e.g., the flow $\Bf$ will be defined on $\forward{E}$), and often when referring to $G_{\Bf}$ as a graph we will ignore all edges with residual capacity $0$ (e.g., when talking about distances in $G_{\Bf}$).

\begin{algorithm}[!ht]
  \caption{\alg{PushRelabel}{$G,\Bc,\Bsource, \Bsink,\Bw,h$}} \label{alg:push-relabel}
  
  \SetEndCharOfAlgoLine{}
  \SetKwInput{KwData}{Input}
  \SetKwInput{KwResult}{Output}
  \SetKwProg{KwProc}{function}{}{}
  \SetKwFunction{Relabel}{Relabel}
   \SetKwFor{Loop}{main loop}{}{}

  Initialize $\Bf$ as the empty flow.\;
  Let $\Bell(v) = 0$ for all $v\in V$. \tcp{levels}
  Mark each edge $e \in \forward{E}\cup\backward{E}$ as \emph{inadmissible} and all vertices as \emph{alive}.\;
  
  \vspace{0.4em}
  
  \KwProc{\Relabel{$v$}}{
      Set $\Bell(v) \gets \Bell(v) + 1$.\;
      \If{$\Bell(v) > 9h$}{mark $v$ as \emph{dead} and \Return.}
      \For{each edge $e\ni v$ where $\Bw(e)$ divides $\Bell(v)$}{
        Let $(x,y) = e$.\;
        \lIf{$\Bell(x) - \Bell(y) \ge 2\Bw(e)$ and $\Bc_{\Bf}(e) > 0$} {
          mark $e$ as \emph{admissible}.
        }
        \lElse{
          mark $e$ as \emph{inadmissible}.
        }
      }
  }
  
  \vspace{0.4em}

  \Loop{}{
      \While{there is an alive vertex $v$ with $\Bsink_{\Bf}(v) = 0$ and without an admissible out-edge} {
          \Relabel{$v$}\;
      }
      \If{there is some alive vertex $s$ with $\Bsource_{\Bf}(s) > 0$} {
          \tcp{$P$ is an "augmenting path"}
          Trace a path $P$ from $s$ to some sink $t$, by arbitrarily following admissible out-edges.\;
          Let $\caug \gets \min\{\Bsource_{\Bf}(s), \Bsink_{\Bf}(t), \min_{e\in P} \Bc_{\Bf}(e)\}$. \;%
          \For(\tcp*[f]{Augment $\Bf$ along $P$}){$e\in P$}{
            \lIf{$e$ is a forward edge} {
            $\Bf(e) \gets \Bf(e) + \caug$.%
            }
            \lElse{
            $\Bf(e')\gets \Bf(e')-\caug$, where $e'$ is the corresponding forward edge to $e$.
            }
            Adjust residual capacities $\Bc_{\Bf}$ of $e$ and the corresponding reverse edge.\;
            \lIf{$\Bc_{\Bf}(e) = 0$}{
            mark $e$ as \emph{inadmissible}.
            }
          }
          \tcp{$\Bsource_{\Bf}(s)$ and $\Bsink_{\Bf}(t)$ goes down by $c^{\mathrm{augment}}$}
      }
      \lElse{ \Return{$\Bf$} }
  }
\end{algorithm}

\subsection{Proof of the Push Relabel Algorithm}

We begin by showing some helpful invariants.%
\begin{lemma} \label{lem:invariants}
  Throughout the run of \cref{alg:push-relabel}, the following invariants hold:
  \begin{enumerate}[label=\textup{(I-\arabic*)},widest=(I-3),itemindent=*]
    \item\label{inv:all-edges}
    $\Bell(u)-\Bell(v) < 3\Bw(e)$, 
    for all $e = (u,v)\in \forward{E}\cup\backward{E}$ with $\Bc_{\Bf}(e)>0$.
    \item\label{inv:admissible-edges}
    $\Bell(u)-\Bell(v) > \Bw(e)$,  for all $e = (u,v)\in \forward{E}\cup\backward{E}$ marked \emph{admissible},
    \item\label{inv:levels} $\Bell(v) \le 9h$ for each \emph{alive} vertex $v$, $\Bell(v)>9h$ for each \emph{dead} vertex $v$,
    and $\Bell(t) = 0$ for all unsaturated sinks $t$ ($\Bsink_{\Bf}(t)>0$).
  \end{enumerate}
\end{lemma}

\begin{proof} 
  It is easy to verify that all invariants hold initially.

  We begin with the invariants~\labelcref{inv:all-edges} and \labelcref{inv:admissible-edges}. Consider some edge $e = (u,v)$, and let $\Bell^{\mathrm{old}}(u), \Bell^{\mathrm{old}}(v)$ be the levels of $u$ and $v$ the last time edge $e$
  was marked as admissible or inadmissible.
  Note that $\Bell(u) \in [\Bell^{\mathrm{old}}(u), \Bell^{\mathrm{old}}(u)+\Bw(e)-1]$ and $\Bell(v) \in [\Bell^{\mathrm{old}}(v), \Bell^{\mathrm{old}}(v)+\Bw(e)-1]$, as if the levels of $u$ (or $v$) had increased by at least $\Bw(e)$, then there must have been a point where $\Bw(e)$ divided $\Bell(u)$ (or $\Bell(v)$).
  \begin{enumerate}
    \item If $e$ was marked as inadmissible and $\Bc_{\Bf}(e)>0$, we know $\Bell^{\mathrm{old}}(u)-\Bell^{\mathrm{old}}(v) < 2\Bw(e)$, and hence that $\Bell(u)-\Bell(v) < 2\Bw(e)+(\Bw(e)-1)$.
    \item If $e$ was marked as admissible, we know $\Bell^{\mathrm{old}}(u)-\Bell^{\mathrm{old}}(v) \ge 2\Bw(e)$, and hence that $\Bell(u)-\Bell(v) \ge 2\Bw(e)-(\Bw(e)-1)$. 
      Additionally, we note that as long as $e$ is admissible, the quantity $\Bell(u)-\Bell(v)$ cannot increase (since it only increases
      when $\Bell(u)$ goes up, which only happens if we relabel $u$, which in turn only happens when there is no admissible outgoing edge of $u$).
      Because at the last point when $e$ was inadmissible we had $\Bell(u)-\Bell(v) < 3\Bw(e)-1$, we know that $\Bell(u)-\Bell(v)< 3\Bw(e)$ now too.
  \end{enumerate}

  Invariant \ref{inv:levels} is easy to see, since any vertex of level $> 9h$ is marked dead from the graph, and the unsaturated sinks (i.e., those $t$ with $\Bsink_{\Bf}(t) > 0$) are never relabeled.
\end{proof}

Because the algorithm relabels all alive vertices (except unsaturated sources) until they have an admissible outgoing edge, we note that when the algorithm tries to trace a path $P$ by arbitrarily following admissible edges, the path $P$ must eventually end in an unsaturates sink. Indeed, at this time, all alive vertices which does not have admissible outgoing edges are exactly the unsaturated sinks. Moreover, traversing an admissible edge, by \labelcref{inv:admissible-edges}, decreases the level $\Bell$, so this process cannot go on forever and $P$ must eventually end in an unsaturated sink.

We will also need a bound on the number of augmentations the algorithm performs:

\begin{lemma} \label{lem:few-augmenting-paths}
Every edge $e$ (or its reverse) is only saturated (i.e., has $\Bc_{\Bf}(e) = \caug$) in at most $O(\frac{h}{\Bw(e)})$ augmenting paths.
Thus, there are at most $O(n+\sum_{e\in E} \frac{h}{\Bw(e)})$ many augmenting paths found by the algorithm.
\end{lemma}
\begin{proof}
Indeed, whenever edge $e = (u,v)$ is fully saturated as part of an augmenting path, it will be marked as inadmissible. At this point we have $\Bell(u)-\Bell(v) > \Bw(e)$ by \labelcref{inv:admissible-edges},
and it (or rather, its reverse $e' = (v,u)$) will only ever be marked as admissible when $\Bell(v)-\Bell(u) \ge 2\Bw(e)$. This means that the sum $\Bell(u)+\Bell(v)$ must have increased (since levels only ever increase) by
$\Theta(\Bw(e))$. Since $\Bell(u)+\Bell(v)\le 18h$, this can only happen $O(\frac{h}{\Bw(e)})$ times.

For the second part of the lemma, we note that for each augmenting path, either an edge is saturated, or a source/sink vertex gets saturated.
The former can happen at most $O(\sum \frac{h}{\Bw(e)})$ times for each edge $e$, and the latter can happen at most once for each vertex.
This bounds the number of augmenting paths.
\end{proof}

We now resume to prove the guarantees listed in \cref{thm:push-relabel-main-theorem} (recall that \labelcref{item:push-relabel:approximation} was already shown in \cref{lemma:push-relabel:approximation}).

\para{Returns Short Flows}
We begin by showing that \cref{alg:push-relabel} returns a flow $\Bf$ such that $\Bw(\Bf) \le 9h \cdot |\Bf|$ (thus proving \labelcref{item:push-relabel:short-flow}).
Indeed, each augmenting path $P = (s=v_0 , v_1, v_2, \ldots, v_k = t)$ our algorithm finds consists of admissible edges. 
This means that $\Bw(P) = \sum_{e\in P} \Bw(e) \le \sum_{i=0}^{k-1} \Bell(v_i)-\Bell(v_{i+1}) = \Bell(s)-\Bell(t) \le 9h$ (by \labelcref{inv:admissible-edges,inv:levels}).
Hence $\Bw(P)\le 9h$.
Now, by summing over all augmenting paths, observe that $\Bw(\Bf) \le \sum_P \Bw(P) \cdot  |\Bf_P|$ because the flow paths may only cancel in the final flow $\Bf$.
Therefore, we have $\Bw(\Bf) \le \sum_P 9h  \cdot |\Bf_P|  =  9h \cdot |\Bf|$.

\para{Source-to-Sink Distance is Large in Residual Graph}
We now argue that the shortest source-to-sink path in the residual graph $G_{\Bf}$ must have $\Bw$-length more than $3h$ (thus proving \labelcref{item:push-relabel:invariant}). Consider any $(s,t)$-path $P = (s=v_0 , v_1, v_2, \ldots, v_k = t)$ in the residual graph, where $s$ is an unsaturated source and $t$ an unsaturated sink.
Then $3\Bw(P) = \sum_{e\in P} 3\Bw(e) > \sum_{i=0}^{k-1} \Bell(v_i)-\Bell(v_{i+1}) = \Bell(s)-\Bell(t) > 9h$. The first inequality is by \labelcref{inv:all-edges}. The last inequality is because every unsaturated source $s$ has level $\Bell(s) > 9h$ at termination (indeed, $s$ must be \emph{dead}, as otherwise the algorithm would not have terminated), and every unsaturated sink $t$ always has level $\Bell(t) = 0$ by \labelcref{inv:levels}. %
Hence $\Bw(P) > 3h$, which is what we wanted.

\para{Running Time}
We now argue that we can implement \cref{alg:push-relabel} in running time bounded by $\tO\left(m + n + \sum_{e\in E} \frac{h}{\Bw(e)}\right)$. 
Here, we only argue that this is the case if the graph is unit-capacitated, i.e., when $\Bc(e) = 1$ for all $e \in E$.
The full argument on the running time bound in capacitated graphs requires keeping track of the admissible edges using dynamic trees (e.g., link-cut trees \cite{SleatorT83}) to speed up parts of the algorithm---the discussion of which we postpone
\ifnum\cameraready=0
to \cref{sec:capacitated-push-relabel}.
\else
to \cite[Appendix A]{BernsteinBST24}.
\fi
We proceed by analyzing the running time of the different parts of the algorithm.
\begin{itemize}
  \item The initialization steps takes $O(n+m)$ time.
  
  \item For each vertex $v$ we can keep track of its admissible out-edges in a linked list (to support addition and removal in constant time).
  Additionally, we can keep track of a list of all vertices $v$ which have no admissible out-edges (so that we can find such a vertex efficiently in constant time).
  
  \item Relabel:
    \begin{itemize}
      \item For each vertex $v$, the relabel operation is run at most $O(h)$ times. This would give an extra factor of $O(nh)$.\footnote{For our purposes this term is actually acceptable.}
      To avoid this, when we perform a relabel operation, we can increase the level of a vertex by more than one. Note that a vertex will only get a new admissible out-edge when some incident edge $e$ has $\Bw(e)$ which divides the new level $\Bell(v)$. Thus, for vertex $v$, we can immediately raise the level to the next multiple of $\Bw(e)$ for any of the adjacent edges.
      In total, vertex $v$ will thus only visit at most $O\left(\sum_{e\in \delta(v)}\frac{h}{\Bw(e)}\right)$ levels.
      In total, over all vertices, we will thus have at most $O\left(\sum_{e\in E}\frac{h}{\Bw(e)}\right)$ relabel operations.
      \item We also argue that the for-loop to mark edges as (in)admissible is efficient. 
        Consider an edge $e = (u,v)$. It will be considered $O\left(\frac{h}{\Bw(e)}\right)$ many times in the for-loop (since at most this many times, $\Bw(e)$ will divide $\Bell(v)$ or $\Bell(u)$).
        In total, this for-loop thus accounts for $O\left(\sum_{e\in E} \frac{h}{\Bw(e)}\right)$ running time. 
        Indeed, we can quickly identify which edges to loop over by storing, for each vertex, a dictionary, where entry $k$ maps to all edges whose weights divide $k$; such a dictionary can be populated as an initialization step.
    \end{itemize}
  
  \item Processing Augmenting Paths:
    \begin{itemize}
      \item By \cref{lem:few-augmenting-paths}, there are only $O\left(n+\sum_{e\in E} \frac{h}{\Bw(e)}\right)$ augmenting paths found. Using dynamic trees (as is a standard speed-up for push-relabel algorithms), we can in fact support each augmentation in $O(\log n)$ time by keeping track of trees where vertex $v$ has an arbitrary admissible out-edge as parent, and using a dynamic tree data structure to support
      ``add'' and ``find-min'' operations on a vertex-to-root path (note: the roots will exactly be the unsaturated sinks).
      \ifnum\cameraready=0
      However, we postpone this discussion to \cref{sec:capacitated-push-relabel}. 
      \else
      However, we postpone this discussion to \cite[Appendix A]{BernsteinBST24}. 
      \fi
      \item Here we instead argue the bound in the case of unit-capacitated graphs: The amount of work we do when we find an augmenting path is proportional to the length of this augmenting path. Thus
      we charge one unit of work to each edge $e_i$ on the path. Since the graph is of unit-capacity, \emph{all} edges $e_i$ on the path will be saturated. Hence, by \cref{lem:few-augmenting-paths}, we know that the edge $e_i$ will only be charged cost $O\left(\frac{h}{\Bw(e)}\right)$ throughout the run of the algorithm, for a total cost of $O\left(\sum_{e\in E} \frac{h}{\Bw(e)}\right)$.
    \end{itemize}
\end{itemize}
The above discussion concludes the proof of \cref{thm:push-relabel-main-theorem}.

\para{Additional Property of Finding Almost Shortest Paths}
The push relabel algorithm works by finding augmenting paths one by one. Say the paths, in order, are $P_1, P_2, \ldots, P_{|\Bf|}$. Let $\Bf_i$ be the flow induced by paths $P_1, \ldots, P_i$, and note that the path $P_{i+1}$ is a path in the residual graph $G_{\Bf_i}$. When bootstrapping our algorithm to find an expander decomposition, we will later need the additional property that the augmenting paths our algorithm finds cannot be ``shortcutted'' significantly. We prove the following lemma.
\begin{lemma}
\label{lem:pr-no-shortcut}
Consider the state of the algorithm just before the $i$-th path $P_i$ is augmented along. Then, for any vertices $s$ and $t$ we have $\Bell(s)-\Bell(t) \le 3 \dist^{\Bw}_{G_{\Bf_{i-1}}}(s,t)$. In particular, this means that any subpath $P'$ of $P_i$, between vertices $s$ and $t$, has weight at most $\Bw(P') \le 3\dist^{\Bw}_{G_{\Bf_{i-1}}}(s,t)$.
\end{lemma}
\begin{proof}
Consider the shortest $(s,t)$-path $Q = (v_1,v_2, \ldots, v_{|Q|})$ (with $v_1 = s$ and $v_{|Q|}=t$) in the residual graph $G_{\Bf_{i-1}}$. Then we have
\[ 3\Bw(Q) = \sum_{e\in Q}3\Bw(e) \ge \sum_{i=1}^{|Q|-1} (\Bell(v_{i})-\Bell(v_{i+1})) = \Bell(s)-\Bell(t) \]
by \ref{inv:all-edges}.
This proves the first part of the lemma.

The second part is similar, but now using \ref{inv:admissible-edges} and the fact that all edges on the path $P'$ are admissible. Suppose $P' = (u_1, \ldots, u_{|P'|})$ (with $u_1 = s$ and $u_{|P'|} = t$). Then we have
\ifnum\cameraready=0
\[ \Bw(P') = \sum_{e\in P'} \Bw(e) \le \sum_{i=1}^{|P'|-1} (\Bell(u_i)-\Bell(u_{i+1})) = \Bell(s)-\Bell(t) \le 3\dist^{\Bw}_{G_{\Bf_{i-1}}}(s,t).\qedhere{} \]
\else
\begin{align*}
    \Bw(P') &= \sum_{e\in P'} \Bw(e) \le \sum_{i=1}^{|P'|-1} (\Bell(u_i)-\Bell(u_{i+1})) \\&= \Bell(s)-\Bell(t) \le 3\dist^{\Bw}_{G_{\Bf_{i-1}}}(s,t).\qedhere{}
\end{align*}
\fi
\end{proof}

\subsection{Application: Approximate Max-Flow in DAGs}

While it is non-trivial to find a ``good'' weight function for arbitrary graphs\footnote{Finding a ``good'' weight function efficiently is exactly what we do in the remainder of our paper, by the use of expander decomposition.}, the special case of directed acyclic graphs (DAGs) turns out to be easy.
Therefore we can immediately obtain (by using our push-relabel algorithm \cref{alg:push-relabel}) a
relatively simple, combinatorial, linear-time-for-dense-graphs, constant-approximation algorithm for maximum flow in DAGs. 

\begin{corollary}
\label{cor:dag-approx}
One can find a $\Theta(1)$-approximate maximum flow in a DAG in $\tO(n^{2})$ time.
\end{corollary}
\begin{proof}
Let $\Btau \in [n]^V$ be the topological order of the vertices (which can be found in $O(n+m)$ time): that is $\Btau_u < \Btau_v$ for each edge $e = (u,v)$.
Then let $\Bw(e) = |\Btau_v-\Btau_u|$. Indeed, any flow path in the maximum flow $\Bf^{\star}$ will have $\Bw$-length at most $n$.
Thus, we can invoke \cref{alg:push-relabel,thm:push-relabel-main-theorem} with $h = n$, getting a $\frac{1}{6}$-approximation of the maximum flow. %
\end{proof}

\newcommand{\Pshort}{\mathcal{P}_{\text{short}}}
\newcommand{\Plong}{\mathcal{P}_{\text{long}}}
\newcommand{\Approx}{O(\log{n})} %

\section{Solving Maximum Flow}\label{sec:weight}

In this section, we show how to compute maximum flow exactly using the weighted push-relabel algorithm from \cref{sec:push-relabel} given an expander hierarchy defined below.

\begin{definition}
  \label{def:expander-hierarchy}
  Given a capacitated graph $(G, \Bc)$, a partition $\cH = (D, X_1, \ldots, X_\eta)$ of $E(G)$ is a \emph{$\phi$-expander hierarchy} of $(G, \Bc)$ with \emph{height} $\eta(\cH) = \eta$ if
  \begin{enumerate}
    \item $D$ is acyclic,
    \item each $e \in X_i$ is contained in some strongly connected component of $G\setminus X_{>i}$, and
    \item $X_i$ is a $\phi$-expanding
    in $(G \setminus X_{>i}, \Bc)$.\footnote{Recall that this means that $X_i\cap C$ is $\phi$-expanding in $(C,\Bc)$, for each strongly connected component $C$ of $G\setminus X_{>i}$.}
  \end{enumerate}
\end{definition}
When the graph is of unit-capacity (i.e., when $\Bc = \Bone$), we will leave $\Bc$ out from the notation.
We remark that our algorithm for constructing an expander hierarchy actually guarantees that $X_i$ is a \emph{separator} in $G\setminus X_{>i}$.
\ifnum\cameraready=0
See \cref{sec:nested-expander-decomposition} for more details.
\else
See \cite[Section 7]{BernsteinBST24} for more details.
\fi
As we will see, for our purposes, we should think of $\phi = n^{o(1)}$ and $\eta = O(\log n)$.

Consider a $\phi$-expander hierarchy $\cH$ of $(G,\Bc)$.
An edge $e \in D$ is a \emph{DAG edge}.
An edge $e \in X_i$ for some $i$ is an \emph{expanding edge}, and more specifically a \emph{level-$i$} expanding edge.
Let $G_i \defeq G \setminus X_{>i}$, in which each $C \in \SCC(G_i)$ is a \emph{level-$i$ expander}. Note that, by definition, the level-$i$ expanders $\SCC(G_i)$ form a refinement of the level-$(i+1)$ expanders $\SCC(G_{i+1})$. See \Cref{fig:hierarchy example} for illustration.

\ifnum\cameraready=0

\begin{figure}[!bt]
    \centering
    \includegraphics[width=\textwidth]{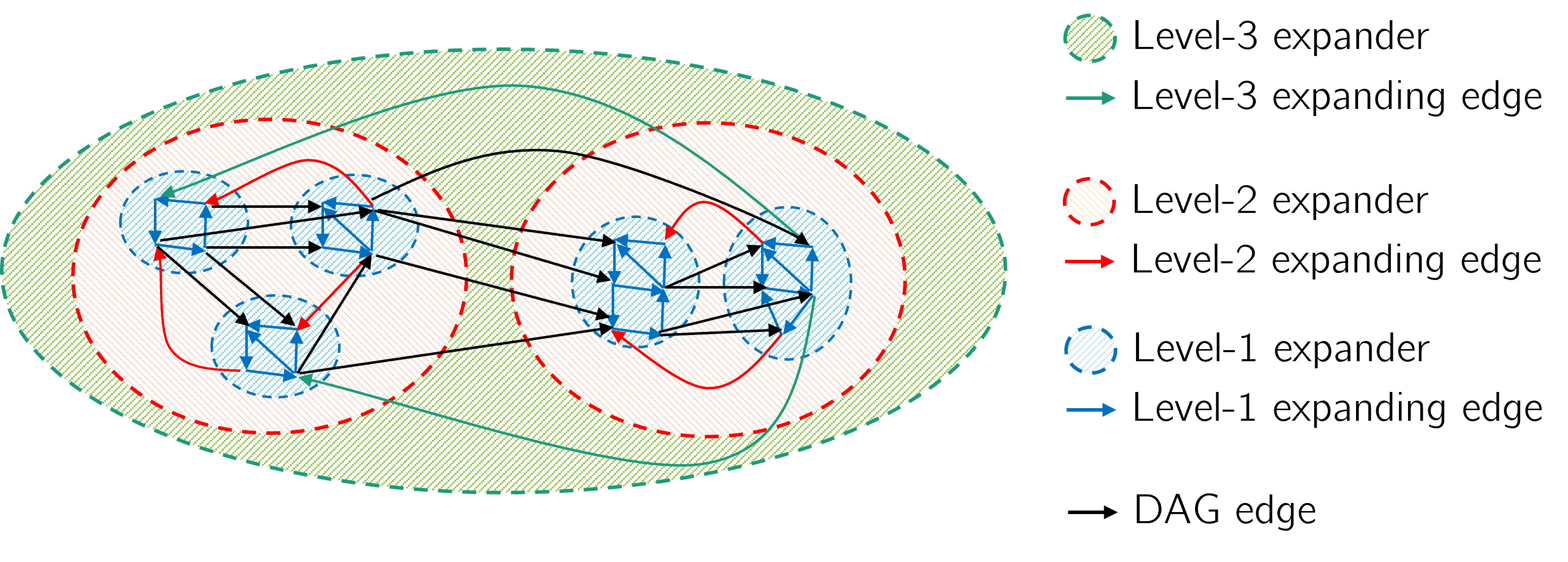}
    \caption{An example of an expander hierarchy with 3 levels}
    \label{fig:hierarchy example}
\end{figure}

\else

\begin{figure*}[!bt]
    \centering
    \includegraphics[width=\textwidth]{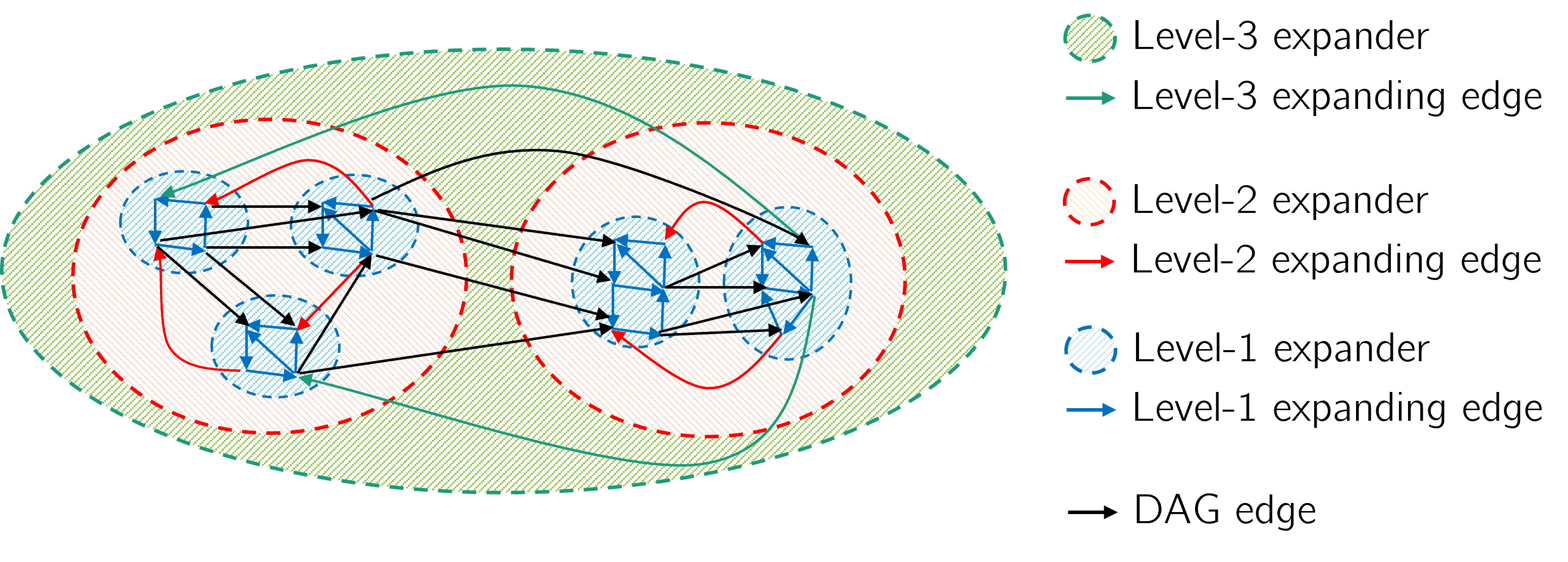}
    \caption{An example of an expander hierarchy with 3 levels}
    \label{fig:hierarchy example}
\end{figure*}

\fi

Note that \cref{def:expander-hierarchy} differs from the undirected expander hierarchy of \cite{GoranciRST21} in that we are not contracting strongly connected components as we go up in the hierarchy.
Indeed, it is impossible to ensure the boundary-linkedness of all inter-cluster edges as in \cite{GoranciRST21} due to the presence of DAG edges in directed graphs.
We also remark that our $\phi$-expander hierarchy precisely generalizes the expander hierarchy notion from \cite[Definition 2]{PatrascuT07} in undirected graphs to directed graphs. Note that in undirected graphs, we would always have $D=\emptyset$.

\subsection{Weight Function}

As alluded to in previous sections, computing maximum flow boils down to the design of a good weight function which we can run the push relabel algorithm from \cref{sec:push-relabel} on.
Recall that we want our weight function $\Bw$ to satisfy (1) that some approximate maximum flow is short with respect to $\Bw$, and (2) that the sum of inverses $\sum_{e\in E} \frac{1}{\Bw(e)}$ is small (for an efficient running time). We prove (1) in \cref{thm:good-weight-function} and (2) in \cref{claim:sum-edge-weights-inv}.

Let $\cH = (D,X_1,\dots,X_\eta)$ be  a $\phi$-expander hierarchy of $(G,\Bc)$
\ifnum\cameraready=0
(we show how to compute $\cH$ in \cref{sec:nested-expander-decomposition}).
\else
(we show how to compute $\cH$ in \cite[Section 7]{BernsteinBST24}).
\fi
We now define the weight function $\Bw_{\cH} \in \N_{>0}^E$ of $G$ induced by the hierarchy $\cH$. We will show in the remainder of the section that it indeed satisfies the desired properties.
The choice of our weight function is inspired by the topological order\footnote{Recall that a \emph{topological order} of an acyclic $D$ is a permutation $\Btau \in \N^V$ such that $\Btau_u < \Btau_v$ for every $(u, v) \in E(D)$. A topological order always exists and can be computed in $O(m)$ time~\cite{Tarjan72}.} of DAGs, and we consider the following notion of respecting topological order.

\begin{definition}
  A topological order $\Btau \in \N^V$ of $D$ is \emph{$\cH$-respecting} if for each level $i$ and each $C \in \SCC(G_i)$ the set $\Btau(C) \defeq \{\Btau_v: v \in C\}$ is contiguous
  In other words, it contains precisely the set of numbers between
  \[
    \Btau_{\mathrm{min}}(C) \defeq \min\Btau(C)\quad\text{and}\quad\Btau_{\mathrm{max}}(C) \defeq \max\Btau(C).
  \]
\end{definition}

Note that given a hierarchy $\cH$, an $\cH$-respecting $\Btau$ can be easily computed in $O(m\eta)$ time by the following lemma whose proof is deferred to
\ifnum\cameraready=0
\cref{appendix:omitted-proofs}.
\else
\cite[Appendix C]{BernsteinBST24}.
\fi

\begin{restatable}{lemma}{RespectingTopo}
\label{lemma:compute-respecting-topo}
  Given an expander hierarchy $\cH$, in $O(m\eta)$ time we can compute an $\cH$-respecting topological order $\Btau$.
\end{restatable}

As a result, in the remainder of this paper whenever there is a hierarchy $\cH$ we assume we also have a corresponding $\cH$-respecting topological order $\Btau$ (which might not be unique).
The weight function $\Bw_{\cH}$ induced by $\cH$ (and $\Btau)$ is simply defined as
\begin{equation}
  \Bw_{\cH}(e) \defeq
  |\Btau_v - \Btau_u|.
  \label{eq:weight-function}
\end{equation}

\begin{observation}
  A level-$i$ expanding edge $e = (u, v)$ has $\Bw_{\cH}(e) \leq |C|$, where $C$ is the level-$i$ expander in which $e$ is expanding.
  \label{obs:back-edge-larger-than-exp-edge}
\end{observation}

Recall that the running time of our push-relabel algorithm depends on the sum of the inverses of the edge weights, which we claim below is small.

\begin{claim}
  For a simple graph $G$ it holds that $\sum_{e \in E}\frac{1}{\Bw_\cH(e)} = O(n\log{n})$.
  \label{claim:sum-edge-weights-inv}
\end{claim}

\begin{proof}
    Since the graph is simple, there is at most two edges (one in each direction) between a pair of vertices $\{u,v\}$.
    Hence,
    \ifnum\cameraready=0
    \[
    \sum_{e\in E} \frac{1}{\Bw_{\cH}(e)}
    \le \sum_{\substack{\Btau_u, \Btau_v \in [n]\\\Btau_u\neq \Btau_v}}
    \frac{1}{|\Btau_u-\Btau_v|} = 2\sum_{i=1}^{n}\sum_{j=i+1}^{n}\frac{1}{j-i} = O(n\log n).\qedhere
    \]
    \else
    \begin{align*}
    \sum_{e\in E} \frac{1}{\Bw_{\cH}(e)}
    &\le \sum_{\substack{\Btau_u, \Btau_v \in [n]\\\Btau_u\neq \Btau_v}}
    \frac{1}{|\Btau_u-\Btau_v|} \\
    &= 2\sum_{i=1}^{n}\sum_{j=i+1}^{n}\frac{1}{j-i} = O(n\log n).\qedhere
    \end{align*}
    \fi
\end{proof}

\para{The Maximum Flow Algorithm}
The key structural lemma that our max-flow algorithm relies on is that it is without loss of optimality for us to focus only on flow paths that are relatively short with respect to the weight function $\Bw_\cH$.
More concretely, by restricting our attention to flows of average $\Bw_{\cH}$-length $\widetilde{O}\left(n/\phi\right)$, we only lose a near-constant factor in its value (compared to the optimal maximum flow).

\begin{theorem}
  Given a flow instance $\cI = (G, \Bc, \Bsource, \Bsink)$ and a $\phi$-expander hierarchy $\cH$ of $(G,\Bc)$ of height $\eta$, for any feasible integral $(\Bsource, \Bsink)$-flow $\Bf$ there exists a feasible (not necessarily integral) $(\Bsource, \Bsink)$-flow $\Bf^\prime$ with $|\Bf^\prime| \geq \frac{1}{\eta+1}|\Bf|$ such that $\Bw_{\cH}(\Bf^\prime) \leq |\Bf^\prime| \cdot O\left(n \cdot \frac{\eta^2\log n}{\phi}\right)$.
  \label{thm:good-weight-function}
\end{theorem}

We will prove \cref{thm:good-weight-function} in \cref{sec:short-flow-exists}.
It essentially reduces the maximum flow problem to constructing an expander hierarchy.
\ifnum\cameraready=0
In \cref{sec:nested-expander-decomposition} we prove such a theorem as follows.
\else
In \cite[Section 7]{BernsteinBST24} we prove such a theorem as follows.
\fi

\ifnum\cameraready=0
\begin{restatable*}{theorem}{NestedExpanderHierarchyCorollary}
  There is a randomized algorithm that, given an $n$-vertex capacitated simple graph $(G, \Bc)$, with high probability constructs a $1/n^{o(1)}$-expander hierarchy $\cH = (D, X_1, \ldots, X_{\eta})$ of $(G, \Bc)$ with $\eta = O(\log n)$ in $n^{2+o(1)}$ time.
  \label{cor:nested-expander-hierarchy}
\end{restatable*}
\else
\begin{restatable}{theorem}{NestedExpanderHierarchyCorollary}
  There is a randomized algorithm that, given an $n$-vertex capacitated simple graph $(G, \Bc)$, with high probability constructs a $1/n^{o(1)}$-expander hierarchy $\cH = (D, X_1, \ldots, X_{\eta})$ of $(G, \Bc)$ with $\eta = O(\log n)$ in $n^{2+o(1)}$ time.
  \label{cor:nested-expander-hierarchy}
\end{restatable}
\fi

We show how \cref{cor:nested-expander-hierarchy} in combination with \cref{thm:good-weight-function} proves our main theorem. %

\begin{restatable}[Restatement of \Cref{thm:main}]{theorem}{MaxFlow}
  There is a randomized augmenting-path-based algorithm that solves the maximum $(s, t)$-flow problem on an $n$-vertex capacitated simple graph with capacities bounded by $U$ in $n^{2+o(1)}\log U$ time with high probability.
  \label{thm:max-flow}
\end{restatable}

\begin{proof}
  Via standard capacity scaling techniques
  \ifnum\cameraready=0
  (see \cref{appendix:capacity-scaling}),
  \else
  (see \cite[Appendix B]{BernsteinBST24}),
  \fi
  by paying an $O(\log U)$ multiplicative overhead in the overall running time we can assume all capacities are bounded by $n^2$ (as in our definition of capacitated graphs in \cref{sec:prelim}).
  Therefore, the maximum $(s, t)$-flow can have value at most $n^4$.
  Let $\Bsource_s = n^4 \cdot \Bone_s$ and $\Bsink_t = n^4 \cdot \Bone_t$ be the $(s, t)$-flow demand.
  We first invoke \cref{cor:nested-expander-hierarchy} to construct a $\phi$-expander hierarchy $\cH$ of $(G, \Bc)$ for some $\phi = 1/n^{o(1)}$, with height $\eta = O(\log n)$.
  Then, using the weight function $\Bw_{\cH}$ induced by $\cH$, we run the push-relabel algorithm (\cref{thm:push-relabel-main-theorem}) on the flow instance with $\Bw_{\cH}$ and height $h = \Theta\left(\frac{n \eta^2 \log n}{\phi}\right)$, obtaining a flow $\Bf$.
  We then simply recurse on the residual instance $(G_{\Bf}, \Bc_{\Bf}, \Bsource_{s,\Bf}, \Bsink_{t,\Bf})$ until there is no augmenting path.
  By \cref{thm:good-weight-function} and \cref{thm:push-relabel-main-theorem}\labelcref{item:push-relabel:approximation}, the flow $\Bf$ is an $O(\log n)$-approximation to the maximum $(s, t)$-flow, which means that the maximum $(s, t)$-flow value decreases by a factor of $\left(1-\frac{1}{O(\log n)}\right)$ each time.
  As such, after $O(\log^2 n)$ iterations the maximum flow value will drop to zero.
  In each iteration we spend $n^{2+o(1)}$ time constructing $\cH$ by \cref{cor:nested-expander-hierarchy} and $\tO(\frac{n^2}{\phi})$ time (which is also $n^{2+o(1)}$ by our choice of $\phi = n^{-o(1)}$) in the push-relabel algorithm by \cref{thm:push-relabel-main-theorem} and \cref{claim:sum-edge-weights-inv}.
  This proves the theorem.
\end{proof}

\subsection{Existence of Short Flow} \label{sec:short-flow-exists}

To prove 
\cref{thm:good-weight-function}, we will instead show the existence of a flow $\Bf^{\prime}$ with low average path length that routes the \emph{same} demand as $\Bf$ does, at the cost of increasing   congestion by $O(\eta)$ factor.

\begin{restatable}{lemma}{GoodWeightFunctionCongested}
  Given a flow instance $\cI = (G, \Bc, \Bsource, \Bsink)$ and a $\phi$-expander hierarchy $\cH$ of $G$ (of height $\eta$), for any integral $(\Bsource, \Bsink)$-flow $\Bf$ of congestion $\kappa \in \N$ there exists an equivalent flow $\Bf^\prime$ of congestion $(\eta+1)\kappa$ such that $\Bw_{\cH}(\Bf^\prime) \leq |\Bf^\prime| \cdot O\left(n\cdot \frac{\eta^2\log n}{\phi}\right)$.%
  \label{lemma:good-weight-function-congested}
\end{restatable}

\cref{thm:good-weight-function} follows from \cref{lemma:good-weight-function-congested} by simply scaling down $\Bf'$ so that it becomes feasible.

\begin{proof}[Proof of \cref{thm:good-weight-function}]
  By \cref{lemma:good-weight-function-congested} with $\kappa \defeq 1$, there is an equivalent flow $\Bf^{\prime\prime}$ with congestion $(\eta+1)$.
  The flow $\Bf^\prime \defeq \frac{\Bf^{\prime\prime}}{\eta+1}$ is a feasible $(\Bsource, \Bsink)$-flow that satisfies $\Bw_{\cH}(\Bf^\prime) \leq |\Bf^\prime| \cdot O\left(\frac{n\eta^2 \log n}{\phi}\right)$.
\end{proof}

The rest of the section proves \cref{lemma:good-weight-function-congested}.
Before diving into the actual proof, we briefly outline the strategy here for better intuition. We start with the not-necessarily short flow $\Bf$, and then ``short-cut'' some parts of the flow, making the flow shorter at the cost of some congestion.
While our arguments here are somewhat algorithmic, we note that for our maximum flow algorithm we only need the \emph{existence} of a short flow, and that the maximum flow algorithm does not need to know the flow $\Bf$ to start with.

To make each flow path short, our goal is to start from the topmost level down, making sure that in each level $i$ and each level-$i$ expander $C$, the flow path only uses $\widetilde{O}(1/\phi)$ level-$i$ expanding edges in $C$. If this held for all levels and all expanders within those levels, each flow path would have length $\tO(n/\phi)$.
Therefore, for flow paths that use a large number of such expanding edges, we have to reroute and short-cut them, using the property of expanders, to reduce the length.
Rerouting inevitably incurs congestion in the resulting flow, and if done na\"ively the congestion will grow by a multiplicative factor of $1/\phi$ in each level. One key component in our analysis is showing that a more careful way of rerouting actually saves us from this congestion blow-up.
Note that we will first prove most of the statements for integral flows, as they admit decomposition into paths that are nice to work with. The statements are then easily extended to fractional flows by simply scaling the flow and capacities up to make them integral.

In the remainder of the section we consider a fixed $\phi$-expander hierarchy $\cH = (D, X_1, \ldots, X_{\eta})$ of $G$ given to us.
By the equivalence between the uncapacitated and capacitated definitions of $\phi$-expanding (see \cref{fact:equivalence}), \emph{we will assume in our analysis (without loss of generality) that $G$ is a unit-capacitated multi-graph} instead of a capacitated simple graph. 
Recall from our definition of capacitated graphs in \cref{sec:prelim} that the capacities are bounded by $n^2$, and thus after replacing each capacitated edge with multiple parallel edges the graph contains $m \leq n^4$ edges. We need this just so that $\log m = O(\log n)$.

\paragraph{Charging of DAG-edges.}
We begin by showing that to bound the $\Bw_{\cH}$-length of any path in $G$, it suffices to bound the contribution of expanding edges to its $\Bw_{\cH}$-length.
\ifnum\cameraready=0
To be more flexible for usage also in \cref{sec:sparse-cut},
\else
To be more flexible for usage also in \cite[Section 6]{BernsteinBST24},
\fi
we consider a slightly more general setting.
The following lemma shows that the hierarchy allows us to charge the weight of DAG edges to non-DAG edges.

\begin{lemma}
  \label{lemma:dag-edge-length}
  Suppose $G = (V,E)$ is a graph, $D\subseteq E$ is a DAG, 
  $\Btau$ a topological order of $D$,
  and $\Bw$ is a weight function
  such that the weight of an edge $e = (u,v)$ satisfies $\Bw(e) \ge |\Btau_{v}-\Btau_{u}|$, with equality if $e\in D$.
  Then for any path $P$ in $G$ it holds that
  \[
    \Bw(P\cap D) \le n + \Bw(P\setminus D).
  \]
\end{lemma}

\begin{proof}
Suppose we walk along $P$ where $P=(v_{1},\dots,v_{k})$. Let $\Phi^{(i)}=\Btau_{v_{i}}$ be the potential that keeps track of how much we proceed in the topological order $\Btau$. The \emph{net} potential increase is $\Phi^{(k)}-\Phi^{(1)}\le n$. Whenever we walk through a DAG-edge $e$, the potential increases by $\Bw(e)$. So the \emph{total} potential increase is $\sum_{i:\Phi^{(i+1)}>\Phi^{(i)}}\Phi^{(i+1)}-\Phi^{(i)}\ge\Bw(P\cap D)$. The \emph{total} potential decrease is $\sum_{i:\Phi^{(i+1)}<\Phi^{(i)}}\Phi^{(i)}-\Phi^{(i+1)}\le\Bw(P\setminus D)$ because only non-DAG-edge $e$ may decrease the potential and it decreases by at most $\Bw(e)$. Therefore, we conclude 
\ifnum\cameraready=0
\begin{align*}
n & \ge\Phi^{(k)}-\Phi^{(1)}\\
 & =\left(\sum_{i:\Phi^{(i+1)}>\Phi^{(i)}}\Phi^{(i+1)}-\Phi^{(i)}\right)-\left(\sum_{i:\Phi^{(i+1)}<\Phi^{(i)}}\Phi^{(i)}-\Phi^{(i+1)}\right)\\
 & \ge\Bw(P\cap D)-\Bw(P\setminus D)
\end{align*}
\else
\begin{align*}
n & \ge\Phi^{(k)}-\Phi^{(1)}\\
 & =\left(\sum_{i:\Phi^{(i+1)}>\Phi^{(i)}}\Phi^{(i+1)}-\Phi^{(i)}\right) \\&-\left(\sum_{i:\Phi^{(i+1)}<\Phi^{(i)}}\Phi^{(i)}-\Phi^{(i+1)}\right)\\
 & \ge\Bw(P\cap D)-\Bw(P\setminus D)
\end{align*}
\fi
The lemma concludes by rearranging.
\end{proof}

\para{Routing Short Flow in Expanders}
\Cref{lemma:dag-edge-length} allow us to focus on bounding the length on the expanding edges. In the following sequence of lemmas, we show how to route a flow within a graph so that it uses only few expanding edges.
In particular, the lemma below shows that if an edge set $F$ is $\phi$-expanding in $G$, then for any routable demand we can almost reroute it in such a way that each flow path uses at most $O(\log m/\phi)$ edges in $F$.
\begin{lemma}
  Consider a routable  flow instance $\cI = (G, \Bsource, \Bsink)$ for a strongly connected $m$-edge $G$ in which $F \subseteq E(G)$ is $\phi$-expanding.
  Given any $\eps > 0$, there is a feasible integral $(\Bsource, \Bsink)$-flow $\Bf$ with value $|\Bf| \geq (1-\eps)\|\Bsource\|_1$ such that $\sum_{e \in F}\Bf(e) \leq |\Bf| \cdot O\left(\frac{\log m}{\eps\phi}\right)$.
  \label{lemma:expander-routing}
\end{lemma}

\begin{proof}
  Let $\ell \defeq \frac{4\log{m}}{\eps \phi} = O\left(\frac{\log m}{\eps \phi}\right)$ be the target $F$-length, and let $\Bf$ be an integral flow in $G$ obtained by repeatedly finding augmenting paths in the residual graph consisting of at most $\ell$ edges in $F$ and send one unit of flow along them until such paths become non-existent.
  We get $\sum_{e \in F}\Bf(e) \leq |\Bf| \cdot \frac{4\log{m}}{\eps \phi}$.
  If $|\Bf| = \|\Bsource\|_1$ then we are done.
  Otherwise, there is at least one unsaturated source $s$ ($\Bsource_{\Bf}(s) > 0$) and one unsaturated $t$ ($\Bsink(t) > \abs_{\Bf}(t)$).
  Let $\dist_F(v)$ be the shortest $F$-distance from an unsaturated source $s$ to vertex $v$ in $G_{\Bf}$, where the $F$-distance is the minimum $F$-length over all such paths $P$.
  Let $S_i \defeq \{v \in V(G): \dist_F(v) = i\}$.
  Note that all unsaturated sinks $t$ must have $\dist_F(t) > \ell$.
  It suffices to show that there is an $0 \leq i \leq \ell$ such that
  \begin{equation}
    \left|E_{G_{\Bf}}(S_{\leq i}, \overline{S_{\leq i}})\right| < \eps \cdot \left|E_{G}(S_{\leq i}, \overline{S_{\leq i}})\right|
    \label{eq:saturate-level-cut}
  \end{equation}
  because of the following claim.

  \begin{claim}
    If there is a cut $S_{\leq i}$ satisfying \eqref{eq:saturate-level-cut}, then $|\Bf| \geq (1-\eps)\|\Bsource\|_1$.
  \end{claim}

  \begin{proof}
    The existence of such a cut $S_{\leq i}$ implies $|\Bf| \geq \Bfout(S_{\leq i}) \geq (1-\eps)|E_G(S_{\leq i}, \overline{S_{\leq i}})|$.
    Consider the maximum $(\Bsource_{\Bf}, \Bsink_{\Bf})$-flow $\Bf^\prime$ in $G_{\Bf}$ for which by \cref{fact:flow-in-residual-graph} and that $\cI$ is routable implies that $|\Bf| + |\Bf^\prime| = \|\Bsource\|_1$.
    However, by the max-flow min-cut theorem (\cref{fact:maxflow-mincut}), we have
    \[
      |\Bf^\prime| \leq \Bc_{\Bf}(S_{\leq i}, \overline{S_{\leq i}}) + \Bsource_{\Bf}(\overline{S_{\leq i}}) + \Bsink_{\Bf}(S_{\leq i}) = \Bc_{\Bf}(S_{\leq i}, \overline{S_{\leq i}})
    \]
    by definition of $S_{\leq i}$.
    As such, we have $|\Bf^\prime| \leq \frac{\eps}{1-\eps}|\Bf|$ and therefore $|\Bf| \geq (1-\eps)(|\Bf| + |\Bf^\prime|) = (1-\eps)\|\Bsource\|_1$.
  \end{proof}
  
  Now, assume for contradiction that no $S_{\leq i}$ satisfying \eqref{eq:saturate-level-cut} exists. The following proof is a standard ball-growing argument.
  Note that by definition of $\dist_F(\cdot)$, it holds that $E_{G_{\Bf}}(S_{\leq i}, \overline{S_{\leq i}}) = E_{G_{\Bf}}(S_i, S_{i+1}) \subseteq F$ for every $i$.
  If $\vol_{F}(S_{\leq \ell/2}) \leq \vol_{F}(\overline{S_{\leq \ell/2}})$, then we have
  \ifnum\cameraready=0
  \begin{align*}
    \vol_F(S_{\leq i}) &\geq \vol_{F}(S_{\leq i-1}) + \left|E_{G_{\Bf}}(S_{\leq i}, \overline{S_{\leq i}})\right| \\ &\stackrel{(i)}\geq \vol_{F}(S_{\leq i-1}) + \eps \left|E_{G}(S_{\leq i}, \overline{S_{\leq i}})\right| \stackrel{(ii)}{\geq} (1+\eps\phi) \cdot \vol_{F}(S_{\leq i-1})
  \end{align*}
  \else
  \begin{align*}
    \vol_F(S_{\leq i}) &\geq \vol_{F}(S_{\leq i-1}) + \left|E_{G_{\Bf}}(S_{\leq i}, \overline{S_{\leq i}})\right| \\ &\stackrel{(i)}\geq \vol_{F}(S_{\leq i-1}) + \eps \left|E_{G}(S_{\leq i}, \overline{S_{\leq i}})\right| \\ &\stackrel{(ii)}{\geq} (1+\eps\phi) \cdot \vol_{F}(S_{\leq i-1})
  \end{align*}
  \fi
  for $0 < i \leq \ell/2$,
  where (i) follows from $S_{\leq i}$ not satisfying \eqref{eq:saturate-level-cut} and (ii) follows from $F$ being $\phi$-expanding and thus $S_{\leq i}$ is not a $\phi$-sparse cut with respect to $F$.
  Since $\vol_F(S_0) > 0$, we have
  \[
    \vol_F(S_{\leq \ell/2}) \geq (1+\eps\phi)^{\frac{2\log {m}}{\eps \phi}} \geq m^2
  \]
  as $(1+x)^{1/x} \geq 2$ for $0 < x \leq 1$,
  which is a contradiction.
  Similarly, if instead it is the case that $\vol_{F}(S_{\leq \ell/2}) > \vol_{F}(\overline{S_{\leq \ell/2}})$, then
  \ifnum\cameraready=0
  \begin{align*}
    \vol_F(\overline{S_{\leq i}}) &\geq \vol_{F}(\overline{S_{\leq i+1}}) + \left|E_{G_{\Bf}}(S_{\leq i}, \overline{S_{\leq i}})\right| \\ &\geq \vol_{F}(\overline{S_{\leq i+1}}) + \eps \left|E_{G}(S_{\leq i}, \overline{S_{\leq i}})\right| \geq (1+\eps\phi) \cdot \vol_{F}(\overline{S_{\leq i+1}})
  \end{align*}
  \else
  \begin{align*}
    \vol_F(\overline{S_{\leq i}}) &\geq \vol_{F}(\overline{S_{\leq i+1}}) + \left|E_{G_{\Bf}}(S_{\leq i}, \overline{S_{\leq i}})\right| \\ &\geq \vol_{F}(\overline{S_{\leq i+1}}) + \eps \left|E_{G}(S_{\leq i}, \overline{S_{\leq i}})\right| \\ &\geq (1+\eps\phi) \cdot \vol_{F}(\overline{S_{\leq i+1}})
  \end{align*}
  \fi
  for $\ell/2 \leq i \leq \ell$.
  Since $\vol_F(\overline{S_{\leq \ell}}) > 0$, it follows that
  \[
    \vol_F(\overline{S_{\leq \ell/2}}) \geq (1+\eps\phi)^{\frac{2\log{m}}{\eps \phi}} \geq m^2,
  \]
  which is also a contradiction.
\end{proof}

The lemma above only returns a flow that partially routes a demand.
Next, we show that we can fully route any demand by paying a small congestion factor by repeatedly routing the remaining demand.
For a subset of edges $F \subseteq E(G)$, we call a demand pair $(\Bsource, \Bsink)$ is \emph{$r$-respecting with respect to $F$} for $r \in \N$ if $\Bsource(v), \Bsink(v) \leq r \cdot \deg_F(v)$ for each $v \in V$.

\begin{lemma}
  Given a flow instance $\cI = (G, \Bsource, \Bsink)$ where $\|\Bsource\|_1 = \|\Bsink\|_1$ for a strongly connected $m$-edge $G$ in which $F \subseteq E(G)$ is $\phi$-expanding such that $(\Bsource, \Bsink)$ is $r$-respecting with respect to $F$ for $\phi \geq \frac{1}{\poly(m)}$ and $r \leq \poly(m)$, there is a integral $(\Bsource, \Bsink)$-flow $\Bf$ with congestion $O\left(\frac{r}{\phi}\log m\right)$ and value $\|\Bf\| = \|\Bsource\|_1$ such that $\sum_{e \in F}\Bf(e) \leq |\Bf| \cdot O\left(\frac{\log m}{\phi}\right)$.
  \label{lemma:expander-routing-respecting}
\end{lemma}

\begin{proof}
  By the max-flow min-cut theorem (\cref{fact:maxflow-mincut}), the $r$-respecting demand $(\Bsource, \Bsink)$ is routable in $G^{(\kappa)}$ for $\kappa \defeq \left\lceil \frac{r}{\phi}\right\rceil$.\footnote{For any cut $S \subseteq V$ in $G^{(\kappa)}$, we have $|E_{G^{(\kappa)}}(S,\overline{S})| \geq r \cdot \min\{\vol_F(S), \vol_F(\overline{S})\} \geq \min\{\Bsource(S), \Bsink(\overline{S})\}$ by definition. By the max-flow min-cut theorem (\cref{fact:maxflow-mincut}), the maximum flow in $G^{(\kappa)}$ has value
  \ifnum\cameraready=0
  \[ \min_S \left|E_{G^{(\kappa)}}(S,\overline{S})\right| \geq \min_S \left\{\min\{\Bsource(S),\Bsink(\overline{S})\} + \Bsource(\overline{S}) + \Bsink(S)\right\} \geq \|\Bsource\|_1 = \|\Bsink\|_1. \]
  \else
  \begin{align*}
  \min_S \left|E_{G^{(\kappa)}}(S,\overline{S})\right| &\geq \min_S \left\{\min\{\Bsource(S),\Bsink(\overline{S})\} + \Bsource(\overline{S}) + \Bsink(S)\right\} \\ &\geq \|\Bsource\|_1 = \|\Bsink\|_1.
  \end{align*}
  \fi
  Therefore, $(\Bsource,\Bsink)$ is routable in $G^{(\kappa)}$.}
  Observe that $F^{(\kappa)}$ is $\phi$-expanding in $G^{(\kappa)}$.
  As such, by \cref{lemma:expander-routing} with $\eps \defeq 1/2$ there is an integral flow $\Bf_1$ in $G^{(\kappa)}$ such that $\|\Bf_1\| \geq \frac{1}{2}\|\Bsource\|_1$ and $\sum_{e \in F^{(\kappa)}}\Bf_1(e) \leq |\Bf_1| \cdot O\left(\frac{\log m}{\phi}\right)$ since $\kappa \leq \poly(m)$.
  The residual demand $(\Bsource_{\Bf_1}, \Bsink_{\Bf_1})$ satisfies $\|\Bsource_{\Bf_1}\|_1 \leq \frac{1}{2}\|\Bsource\|_1$ and is clearly also $r$-respecting with respect to $F$ and thus routable in $G^{(\kappa)}$.
  Applying \cref{lemma:expander-routing} again with $\eps \defeq \frac{1}{2}$ on the $(G^{(\kappa)}, \Bsource_{\Bf_1}, \Bsink_{\Bf_1})$, we get an integral flow $\Bf_2$ in $G^{(\kappa)}$ such that $\|\Bf_2\|_1 \geq \frac{1}{2}\|\Bsource_{\Bf_1}\|$ and $\sum_{e \in F^{(\kappa)}}\Bf_2(e) \leq |\Bf_2| \cdot O\left(\frac{\log m}{\phi}\right)$.
  Because $\|\Bsource\|_1 \leq \poly(m)$, repeating this argument $O(\log m)$ times, we get $O(\log m)$ integral flows $\Bf_1, \ldots, \Bf_{O(\log m)}$ such that the sum of them $\Bf \defeq \Bf_1 + \cdots + \Bf_{O(\log m)}$ has value $|\Bf| = \|\Bsource\|_1$ and $\sum_{e \in F^{(\kappa)}}\Bf(e) \leq |\Bf| \cdot O\left(\frac{\log m}{\phi}\right)$ with congestion $O(\log m)$ in $G^{(\kappa)}$, which can be mapped back to an integral flow in $G$ with congestion $O(\kappa \log m) = O\left(\frac{r}{\phi} \log m\right)$ with the same guarantee.
\end{proof}

\para{Rerouting Long Flow to Short Flow}
If we directly use \cref{lemma:expander-routing-respecting} for rerouting each flow path, we might get a congestion blow-up of $\widetilde{\Theta}(1/\phi)$ which is too expensive.
The crucial idea to avoid this is as follows:
for each long flow path, we reroute the flow starting at the set of first $\widetilde{O}(1/\phi)$ edges of the path to the set of last $\widetilde{O}(1/\phi)$ edges.
This idea leads to cancellation in congestion and allows us to control the congestion blow-up to be at most $(1+1/\eta)$ factor, which is only $O(\eta)$ factor after accumulation over all $\eta$ levels.

We begin by defining what we mean by \emph{rerouting}.
Given an integral flow $\Bf$ decomposable into paths $P_1, \ldots, P_{k}$, we can \emph{reroute} $\Bf$ at $(s_i, t_i)$ for each $1 \leq i \leq k$, where $s_i$ and $t_i$ are vertices on $P_i$ (with $s_i$ occurring before $t_i$) with a flow $\Bf_{\mathrm{route}}$ routing the demand 
\ifnum\cameraready=0
\[ \Bsource(v) \defeq \left|\left\{1 \leq i \leq k: s_i = v\right\}\right|,\qquad \Bsink(v) \defeq \left|\left\{1 \leq i \leq k: t_i = v\right\}\right| \]
\else
\begin{align*}
    &\Bsource(v) \defeq \left|\left\{1 \leq i \leq k: s_i = v\right\}\right|,\\
    &\Bsink(v) \defeq \left|\left\{1 \leq i \leq k: t_i = v\right\}\right|
\end{align*}
\fi
getting the flow $\widetilde{\Bf}$ given by
\(
  \widetilde{\Bf}(e) \defeq \Bf(e) + \Bf_{\mathrm{route}}(e) - \sum_{i = 1}^{k}\Bf_{P_i[s_i, t_i]}(e),
\)
where $\Bf_{P_i[s_i, t_i]}$ is the notation for a flow that sends one unit flow along the path $P_i[s_i, t_i]$. We note that we \emph{do \textbf{not} use multi-commodity flow} when rerouting; that is, $\Bf_{\mathrm{route}}$ does not necessarily consist of $(s_i,t_i)$-paths. Indeed, for our purposes we only need that each $s_i$ is paired up with some $t_j$, i.e., that $\Bf_{\mathrm{route}}$ routes the same demand as the flow $\sum_{i=1}^{k}\Bf_{P_{i}[s_i,t_i]}$.

\begin{observation} \label{obs:rerouting}
  The following facts about such a rerouted flow hold.
  \begin{enumerate}[(1)]
    \item\label{item:rerouting-same-demand}$\widetilde{\Bf}$ routes the same demands as $\Bf$ does, i.e., $\widetilde{\Bf}$ and $\Bf$ are equivalent.
    \item\label{item:rerouting-preserve} If $\Bf_{\mathrm{route}}$ is a flow in a subgraph $H$ of $G$, then $\widetilde{\Bf}(e) \leq \Bf(e)$ for all $e \not\in E(H)$.
    \item\label{item:rerouting-congestion} If $\Bf_1 + \Bf_2$ has congestion $\kappa$ and $\Bf_1$ is rerouted by a flow $\Bf_{\mathrm{route}}$ with congestion $\kappa^\prime$, resulting in $\widetilde{\Bf_1}$, then the flow $\widetilde{\Bf_1} + \Bf_2$ has congestion $\kappa + \kappa^\prime$.
  \end{enumerate}
\end{observation}

Let $c_{\ref{lemma:expander-routing-respecting}} \in \N$ be the constant hidden in the $O(\cdot)$ notation of the congestion guarantee of \cref{lemma:expander-routing-respecting}.
In other words, the flow from \cref{lemma:expander-routing-respecting} has congestion at most $\frac{r}{\phi} \cdot c_{\ref{lemma:expander-routing-respecting}}\log{m}$.
To recall, $\eta$ is the height of the given hierarchy $\cH$. Now we are ready to prove our main rerouting lemma that incurs only a very small congestion blow-up.

\begin{lemma}
  Given a flow $\Bf$ with congestion $\kappa$ in $G$, a target level $i$, and a level-$i$ expander $C$ of $G$, there is an equivalent flow $\Bf^\prime$ in $G$ with congestion $\left\lceil\kappa\right\rceil\left(1 + \frac{1}{\eta}\right)$ such that
  \[
    \sum_{e \in F}\Bf^\prime(e) \leq \left|\Bf^\prime\right| \cdot O\left(\frac{\eta \log n}{\phi}\right),
  \]
  where $F \defeq X_i \cap E(C)$ is the level-$i$ expanding edge set in $C$.
  Additionally, it holds that $\Bf^\prime(e) \leq \Bf(e)$ for all $e \not\in E(C)$.
  \label{lemma:induction}
\end{lemma}

\begin{proof}
  Let us first suppose that $\Bf$ is an integral flow (therefore we may assume $\kappa \in \N$), and let $\xi \defeq \left\lceil\frac{\eta \cdot c_{\ref{lemma:expander-routing-respecting}}\log m}{\phi}\right\rceil \in \N$.
  Let $\P_{\Bf}$ be a decomposition of $\Bf$ into flow paths.
  Let $\Plong \defeq \{P \in \P_{\Bf}: \left|P \cap F\right| \geq 2\xi\}$ be the paths which are long with respect to $F$, i.e., uses at least $2\xi$ edges in $F$ (note that we may assume that the paths $P$ are simple and thus cannot use the same edge in $F$ multiple times).
  For each $P \in \Plong$, let $S_P \defeq \left(s_P^{(1)}, \cdots, s_P^{(\xi)}\right)$ be the endpoints of the first $\xi$ edges from $F$ on $P$. Similarly, let $T_P \defeq \left(t_P^{(1)}, \ldots, t_P^{(\xi)}\right)$ be the endpoints of the last $\xi$ edges from $F$ on $P$.
  Let $\Bf_{\text{long}} \defeq \sum_{P \in \Plong}\Bf_P$ and $\Bf_{\text{short}} \defeq \Bf - \Bf_{\text{long}} = \sum_{P \in \P_{\Bf} \setminus \Plong}\Bf_P$ be the flow corresponding to long and short flow paths, respectively.
  Let $\Bf^{(\xi)}_{\text{long}} \defeq \Bf_{\text{long}} \cdot \xi$ be a flow in $G$ and $\Plong^{(\xi)}$ be the decomposition of $\Bf^{(\xi)}_{\text{long}}$ corresponding to $\Plong$, i.e., $\Plong^{(\xi)} \defeq \bigcup_{P \in \Plong}\left\{P^{(\xi)}_1, \ldots, P^{(\xi)}_\xi\right\}$ where $P^{(\xi)}_i$ is the $i$-th duplicate of the path $P \in \Plong$.

  We now reroute $\Bf_{\text{long}}^{(\xi)}$ at $\left\{\left(s_P^{(i)}, t_P^{(i)}\right)\right\}_{P_i^{(\xi)} \in \Plong^{(\xi)}}$.
  In other words, for the $i$-th copy of $P \in \Plong$, we attempt to reroute it from the $i$-th edge in $S_P$ to the $i$-th edge in $T_P$. This is the main idea which allow us to avoid the congestion blow-up, since, although $\Bf^{(\xi)}_{\mathrm{long}}$ has congestion $\xi \kappa$, the
  demand $(\Bsource, \Bsink)$ corresponding to this rerouting is $\kappa$-respecting on $F$. This is since each start-vertex $s_{P}^{(i)}$ and end-vertex $t_{P}^{(i)}$ in the rerouting can be charged (a single time) to the corresponding edge in the flow path $P$, and $\Bf$ has congestion $\kappa$.
  Therefore by \cref{lemma:expander-routing-respecting}, $\Bf^{(\xi)}_{\mathrm{long}}$ can be routed in $C$ with congestion $\frac{\kappa}{\phi} \cdot c_{\ref{lemma:expander-routing-respecting}}\log{m}$ by a flow $\Bf_{\mathrm{route}}$. %
  Let $\widetilde{\Bf_{\text{long}}^{(\xi)}}$ be the rerouted $\Bf_{\text{long}}^{(\xi)}$.
  It then follows that the flow
  \[ \Bf^\prime \defeq \Bf_{\text{short}} + \frac{\widetilde{\Bf^{(\xi)}_{\text{long}}}}{\xi} \]
  routes the same demand as $\Bf$ does by \cref{obs:rerouting}\labelcref{item:rerouting-same-demand} and has congestion
  \[ \kappa + \frac{\frac{\kappa}{\phi} \cdot c_{\ref{lemma:expander-routing-respecting}}\log{m}}{\xi} = \kappa\left(1 + \frac{1}{\eta}\right) \]
  by \cref{obs:rerouting}\labelcref{item:rerouting-congestion}.
  The total amount of flow on $F$-edges can be bounded by
  \ifnum\cameraready=0
  \begin{align*}
    \sum_{e \in F}\Bf^\prime(e) = \sum_{e \in F}\Bf_{\text{short}}(e) + \frac{\sum_{e \in F}\widetilde{\Bf_{\text{long}}^{(\xi)}}(e)}{\xi} &\leq \left|\Bf_{\text{short}}\right| \cdot 2\xi + \left|\Bf_{\text{long}}\right| \cdot 2\xi + \left|\Bf_{\text{long}}\right| \cdot O\left(\frac{\log{m}}{\phi}\right) \\
    &\leq \left|\Bf^\prime\right| \cdot O\left(\frac{\eta \log n}{\phi}\right),
  \end{align*}
  \else
  \begin{align*}
    \sum_{e \in F}\Bf^\prime(e) &= \sum_{e \in F}\Bf_{\text{short}}(e) + \frac{\sum_{e \in F}\widetilde{\Bf_{\text{long}}^{(\xi)}}(e)}{\xi} \\
    &\leq \left|\Bf_{\text{short}}\right| \cdot 2\xi + \left|\Bf_{\text{long}}\right| \cdot 2\xi + \left|\Bf_{\text{long}}\right| \cdot O\left(\frac{\log{m}}{\phi}\right) \\
    &\leq \left|\Bf^\prime\right| \cdot O\left(\frac{\eta \log n}{\phi}\right),
  \end{align*}
  \fi
  where we use the fact that the rerouting happens at the first $\xi$ and the last $\xi$ edges on flow paths in $\Plong$.
  The property that $\Bf^\prime(e) \leq \Bf(e)$ for all $e \not\in E(C)$ also follows from \cref{obs:rerouting}\labelcref{item:rerouting-preserve}.
  
  If instead the flow $\Bf$ is $\frac{1}{z}$-integral for some $z \in \N$, then we may assume $\kappa \in \frac{1}{z} \cdot \N$ and the demand $(\Bsource, \Bsink)$ it routes to be in $\left(\frac{1}{z} \cdot \N\right)^V$.
  We can then treat $\Bf$ as an integral flow in $G^{(z)}$ routing demand $(z\cdot \Bsource, z\cdot \Bsink)$ with congestion $\left\lceil \kappa \right\rceil$, i.e., for each edge $e$, we put a total of $\Bf(e) \cdot z \leq \kappa z$ units of flow on the duplicates of $e$ in $G^{(z)}$, distributed evenly among the $z$ duplicates so that each of them receives at most $\left\lceil\kappa\right\rceil$ units of flow.
  Applying the same argument as before in $G^{(z)}$ proves the lemma for this case.
  
  Finally, note that the rerouted flow $\Bf^\prime$ is $\frac{1}{z\xi}$-integral so $\Bf^\prime \in \Q_{\geq 0}^E$ (recall that our definition of flow requires rational values).
\end{proof}

\begin{corollary}
  Given a flow $\Bf$ with congestion $\kappa$ in $G$ and a target level $i$, there is a flow $\Bf^\prime$ in $G$ routing the same demand as $\Bf$ does with congestion $\left\lceil\kappa\right\rceil\left(1 + \frac{1}{\eta}\right)$ such that%
  \[
    \sum_{e \in X_i}\Bf^\prime(e)\Bw_{\cH}(e) \leq |\Bf^\prime| \cdot O\left(n \cdot \frac{\eta \log n}{\phi}\right).
  \]
  Additionally, it holds that $\Bf^\prime(e) \leq \Bf(e)$ for all $e \in X_{>i}$.
  \label{cor:induction}
\end{corollary}

\begin{proof}
  The corollary simply follows by applying \cref{lemma:induction} to every level-$i$ expander in an arbitrary order, using the fact that the vertex-sizes of the level-$i$ expanders sum up to $n$ and a level-$i$ expanding edge $e$ has weight $\Bw_{\cH}(e) \le |C|$ for $C$ being the level-$i$ expander $e$ belongs to.
\end{proof}

\cref{lemma:good-weight-function-congested} can now be proved.

\GoodWeightFunctionCongested*

\begin{proof}
  Let $\Bf_{\eta+1} \defeq \Bf$ be the given integral flow with congestion $\kappa_{\eta+1} = \kappa$.
  From $i = \eta$ to $1$, we apply \cref{cor:induction} on $\Bf_{i+1}$ and $\kappa_{i+1}$ on target level $i$, getting a flow $\Bf_{i}$ with congestion $\kappa_i \leq \left\lceil \kappa_{i+1}\right\rceil \left(1 + \frac{1}{\eta}\right)$. By induction, it is easy to see that $\kappa_{i}\leq \kappa \cdot (\eta + 2 - i)$.
  The returned flow $\Bf^{\prime\prime}$ is then set to $\Bf_1$, which has congestion $\kappa(\eta+1)$, with the property that
  \[ \sum_{e \in E \setminus D}\Bf^{\prime\prime}(e)\Bw_{\cH}(e) \leq \eta \cdot \left|\Bf^{\prime\prime}\right| \cdot O\left(n\cdot \frac{\eta \log n}{\phi}\right) \]
  using that each rerouting does not affect (or can only decrease) flows on higher-level edges by \cref{obs:rerouting}\labelcref{item:rerouting-preserve}.
  The lemma then follows from \cref{lemma:dag-edge-length} which asserts that the weights of DAG edges on a path are bounded by the weights of expanding edges on it, up to an additive factor of $n$ which is dominated.
\end{proof}

The above lemma shows that there is a 
\emph{fractional} flow which is also short. While this is good enough to guarantee that our push relabel algorithm returns an approximate flow, we note in the following corollary that, by paying another $\log(n)$-factor in congestion, we can assume the short flow is \emph{integral}.
\ifnum\cameraready=0
This observation will be useful later in \cref{sec:low-diameter-pruning}.
\else
This observation will be useful later in the full version~\cite[Section 6.2]{BernsteinBST24}.
\fi
\begin{corollary}
  Given a flow instance $\cI = (G,\Bsource,\Bsink)$ routable with congestion $\kappa\in \N$ in a graph $G$ equipped with a $\phi$-expander hierarchy $\cH$ of height $\eta$, there is an integral flow $\Bf$ routing $\cI$ with congestion $O(\kappa\eta\log{n})$ such that $\Bw_{\cH}(\Bf) = |\Bf|\cdot O\left(n \cdot \frac{\eta^{2}\log n}{\phi}\right)$.
  \label{cor:good-weight-function-integral}
\end{corollary}

\begin{proof}
  The existence of a fractional such a flow is given by \cref{lemma:good-weight-function-congested}.
  By \cref{cor:fractional-implies-integral}, we get a short integral flow $\Bf_1$ routing $\frac{1}{6}$ fraction of $\|\Bsource\|_1$.
  On the residual demand $(\Bsource_{\Bf}, \Bsink_{\Bf})$ we may apply the same argument again in $G$ (not in $G_{\Bf_1}$) and get a short integral $\Bf_2$ routing $\frac{1}{6}$ fraction of $\|\Bsource_{\Bf_{1}}\|$
  Repeating this $O(\log{n})$ times until the demand becomes empty, we get $O(\log{n})$ flows $\Bf_1, \ldots, \Bf_{O(\log{n})}$ in $G$, each with congestion $O(\kappa\eta)$ and $\Bw_{\cH}(\Bf_i) = O\left(|\Bf_i| \cdot \frac{n\eta^2\log n}{\phi}\right)$, for which $\Bf \defeq \Bf_1 + \cdots + \Bf_{O(\log{n})}$ routes $\cI$.
  We also have $\Bw_{\cH}\left(\Bf\right) \leq \Bw_{\cH}\left(\Bf_1\right) + \cdots + \Bw_{\cH}\left(\Bf_{O(\log{n})}\right) \leq O\left(|\Bf| \cdot \frac{n\eta^2\log n}{\phi}\right)$, proving the corollary.
\end{proof}

\section{The Sparse-Cut Algorithm}\label{sec:sparse-cut}

A central building block in constructing expander decompositions or even expander hierarchies in general is to either solve a flow problem or find a sparse cut in the graph.
In this section we provide such a subroutine using our push-relabel algorithm. Our algorithm to build the expander hierarchy (in \cref{sec:nested-expander-decomposition}) will heavily rely on the following theorem which we prove here.
\begin{theorem}\label{thm:flow}
  Given a diffusion instance $\cI = (G, \Bc, \Bsource, \Bsink)$ on a strongly connected $n$-vertex graph $G$, %
   a $\phi$-expander hierarchy $\cH$ of $(G \setminus F, \Bc)$
   of height $\eta$, and some $\kappa\in \N$ with $1/\phi, \kappa \le n$, there is an $\widetilde{O}(n^2\cdot \frac{\kappa\eta^4}{\phi^{2}})$ time algorithm \emph{\alg{SparseCut}{$\cI, \kappa, F, \cH$}} that finds a flow $\Bf$ with congestion $\kappa$ and, if $|\Bf| < \|\Bsource\|_1$, a cut $\emptyset \neq S \subsetneq V$ with $\abs_{\Bf}(S) = \Bsink(S)$ and $\ex_{\Bf}(S) = \ex_{\Bf}(V)$ such that
  \begin{equation}
    \Bc(E_G(S, \overline{S})) \leq \frac{O(|\Bf|) +\min\{\vol_{F, \Bc}(S), \vol_{F, \Bc}(\overline{S})\}}{\kappa}.
    \label{eq:flow-guarantee}
  \end{equation}
\end{theorem}

\begin{remark}
When $F = \emptyset$ and $\kappa = 1$, \cref{thm:flow} is an $O(1)$-approximate maximum flow algorithm because it either routes all the source, otherwise there is a cut $S$ where 
$\Bc(E_G(S, \overline{S})) = O(|\Bf|)$, which certifies that $\Bf$ is an $O(1)$-approximation. 
One can view this theorem as a generalization of our approximate maximum flow algorithm, as explained in the proof of \cref{sec:weight,thm:max-flow}, where we do not quite have a $\phi$-expander hierarchy of the full graph, but only of $G\setminus F$ for some edge set $F$. The quality of the flow (and cut) we can find here will depend on the edge set $F$ (see the $\vol_{F}$-terms in the theorem statement). As we will see later in \cref{sec:nested-expander-decomposition}, the guarantees here are good enough for the sparse-cut subroutines we need when building the expander hierarchy: in particular, using \cref{thm:flow} we can either certify that $F$ is $\widetilde{\Theta}(\frac{1}{\kappa})$-expanding in $G$ or else find a sparse cut with respect to $F$.
\end{remark}

\paragraph{The Algorithm.}
We first describe the algorithm for \cref{thm:flow} whose pseudocode is given in \cref{alg:sparse-cut}.

\begin{algorithm}[!ht]
  \caption{\alg{SparseCut}{$\cI = (G, \Bc, \Bsource, \Bsink), \kappa, F, \cH$}} \label{alg:sparse-cut}
  
  \SetEndCharOfAlgoLine{}
  \SetKwInput{KwData}{Input}
  \SetKwInput{KwResult}{Output}
  \SetKwProg{KwProc}{function}{}{}
  \SetKwFunction{Relabel}{Relabel}
   \SetKwFor{Loop}{main loop}{}{}

  Let  
  $h \defeq \left\lceil\frac{4\eta^4 \cdot c_{\ref{lemma:low-diameter-expander-new}} \cdot \log^7 n \cdot \kappa}{\phi^2}\cdot n\right\rceil$, $\Bc^{\kappa} \defeq \kappa \cdot \Bc$, and
  $\Bw_G(e) \defeq  \begin{cases}\Bw_{\cH}(e) & \text{for $e \in E \setminus F$ (see \labelcref{eq:weight-function})}\\ n & \text{for $e \in F$}\end{cases}$.\;

  Run \alg{PushRelabel}{$G,\Bc^{\kappa},\Bsource, \Bsink,\Bw_{G},h$} (\cref{thm:push-relabel-main-theorem}) to get a flow $\Bf$.\;
  \lIf{$|\Bf| = \|\Bsource\|_1$}{
      \Return $\Bf$
  }
  \Else{
    Let $\Bw_{\Bf}$ be $\Bw_{G}$ extended to $G_{\Bf}$, except set $\Bw_{\Bf}(\forward{e}) \defeq 0$ for $e \in D_{\cH}$.\;
  
    Let $S_0 = \{s\in V : \Bsource_{\Bf}(s)>0\}$.\;
    Compute $\Bw_{\Bf}$-distance levels $S_i \defeq \left\{v \in V: \dist_{G_{\Bf}}^{\Bw_{\Bf}}(S_0, v) = i\right\}$ in the residual graph $G_{\Bf}$.
    
    \Return $\Bf$ and the cut $(S_{\leq i}, \overline{S_{\leq i}})$ minimizing $\Bc^{\kappa}_{\Bf}(E_{G_{\Bf}}(S_{\leq i}, \overline{S_{\leq i}})) - \min\{\vol_F(S_{\leq i}),\vol_{F}(\overline{S_{\leq i}})\}$.
}
  
\end{algorithm}

The main idea is to run our push-relabel algorithm to try to route the demand $(\Bsource, \Bsink)$. If it fails to find a large enough flow, we will show how to extract a ``sparse cut'' from the residual graph.
In order to run our push-relabel algorithm (\cref{thm:push-relabel-main-theorem,alg:push-relabel}), we need to supply it with a weight function $\Bw$. However, we do not yet have a $\phi$-expander hierarchy of the whole graph $G$, but only of $G\setminus F$.
A natural idea is to extend the weight function $\Bw_{\cH}$ to all of $G$, assigning edges in $F$
a large weight.

Let $\cH$ be the given hierarchy for $G \setminus F$.
Let %
\begin{equation}
  h \defeq \left\lceil\frac{4\eta^4 \cdot c_{\ref{lemma:low-diameter-expander-new}} \cdot \log^7 n \cdot \kappa}{\phi^2}\cdot n\right\rceil = O\left(\frac{n\cdot \eta^4 \log^7 n \cdot \kappa}{\phi^2}\right),
  \label{eq:def-h}
\end{equation}
for a constant $c_{\ref{lemma:low-diameter-expander-new}}$ that will be defined later in \eqref{eq:low-diameter-constant}.
Let $\Bc^{\kappa} \defeq \kappa \cdot \Bc$ as in \cref{alg:sparse-cut}.
We apply \cref{thm:push-relabel-main-theorem} on the flow instance $\cI^{\kappa} \defeq (G, \Bc^{\kappa}, \Bsource, \Bsink)$ to height $h$ and weight function $\Bw_{G}$ where $\Bw_G(e) \defeq \Bw_{\cH}(e)$ for $e \in E \setminus F$ and $\Bw_G(e) \defeq n$ for $e \in F$.
Let $\Bf$ be the flow returned by \cref{thm:push-relabel-main-theorem}.
If $|\Bf| = \|\Bsource\|_1$, we are done. Otherwise
$|\Bf| < \|\Bsource\|_1$, in which case by \cref{thm:push-relabel-main-theorem}\labelcref{item:push-relabel:invariant} we have $\dist_{G_{\Bf}}^{\Bw_G}(s, t) > 3h$ for any $\Bsource_{\Bf}(s) > 0$ and $\Bsink_{\Bf}(t) > 0$. In this case we need to find a sparse cut.

\paragraph{Running Time.}
The weight function $\Bw_{G}$ can be computed in $O(m\eta)$ time by \cref{lemma:compute-respecting-topo}.
The distance layers can be computed with a standard shortest path algorithm (e.g., Dijkstra's algorithm~\cite{Dijkstra59}) in $\tO(m)$ time.
  By \cref{thm:push-relabel-main-theorem,claim:sum-edge-weights-inv}, the running time of the $\alg{PushRelabel}{}$ call is $\widetilde{O}\left(n^{2}\cdot\frac{\kappa \eta^4}{\phi^2}\right)$.

\paragraph{Analysis in a Unit-Capacitated Multi-Graph.}
By the equivalence between the uncapacitated and capacitated definitions of $\phi$-expanding (see \cref{fact:equivalence}), we will assume (without loss of generality), for the remainder of this section, in our analysis that $G$ is a unit-capacitated multi-graph instead of a capacitated simple graph. 
That is, $\Bc = \Bone$ and $\Bc^{\kappa} = \kappa \cdot \Bone$.
Recall that the capacities are bounded by $n^2$ (\cref{sec:prelim}), and thus after replacing each capacitated edge with multiple parallel edges the graph contains $m \leq n^4$ edges.

\paragraph{Finding a Sparse Cut.}
To locate a sparse cut when $G\setminus F$ is the empty graph (that is, when we want to build the first level of expander decomposition) the following strategy is standard (see~e.g.~\cite{HenzingerRW17,SaranurakW19}) and sufficient for us: \emph{Let $S_0 = \{s : \Bsource_{\Bf}(s) > 0\}$ and compute the distance layers $S_{i} = \{v : \dist^{\Bw_{G}}_{G_{\Bf}}(S_0, v) = i\}$. Now at least one of the level cuts $E_{G}(S_{\leq i}, \overline{S_{\leq i}})$ must be sparse.} The proof of this strategy follows from a simple ball-growing argument.

Unfortunately, even when the underlying graph $G\setminus F$ is a DAG, the above strategy fails. The problem is that there might be too many DAG-edges crossing the level cuts. To solve this, we will modify the weight function slightly by setting all forward DAG edges to have weight $0$.
In particular, we let $\Bw_{\Bf}$ be the weight function on $E(G_{\Bf})$, where $\Bw_{\Bf}(e) = \Bw_G(e)$ for all $e$ except for $\Bw_{\Bf}(\forward{e}) = 0$ for $e \in D_{\cH}$.

As we will see in the remainder of this section, if we compute the distance layers with respect to $\Bw_{\Bf}$, at least one of the level cuts must be sparse. This means that the algorithm to find such a sparse cut is quite simple: just compute the distances, and output the sparsest of the level cuts.
While the algorithm itself is simple, showing the \emph{existence} of such a sparse level cut turns out to be nontrivial. There are essentially three types of edges we want to argue are sparse in most level cuts.

\begin{description}
  \item[DAG edges of $\cH$.]
The modification to the weight function makes it so that only DAG edges used in the flow can be in a level-cut, of which there are on average $O(|\Bf|)$ crossing each level cut.
\item[Edges in $F$.] These edges can be handled by a ball-growing argument (see proof of \cref{lemma:exists-sparse-cut}), similar to the case when constructing a single level expander decomposition. This shows that most level cuts $(S_{\leq i}, \overline{S_{\leq i}})$ have at most $\min\{\vol_{F}(S_{\leq i}), \vol_{F}(\overline{S_{\leq i}})\}$ edges from $F$ crossing them.
However, our modification of setting some weights to zero might have reduced the number of layers. So we must argue that we still have enough level cuts left in the graph, or, equivalently, we want the distance from any source to any sink in the residual graph to still be $\Omega(h)$. We argue this in \cref{lemma:many-levels}.
\item[Expanding edges of $\cH$.] These are arguably the trickiest edges to handle and is thus the focus of the majority of our analysis. We want to argue that most level cuts have few expanding edges of $\cH$ in them. In the original graph $G$, each expander in $\cH$ has a low diameter. If we can say that this is also the case in the residual graph $G_{\Bf}$, we can argue that each expander will only span a few level cuts, so most level cuts do not have any expanding edges at all. Using the properties of expanders and how the residual graph is constructed by reversing short augmenting paths, we show something in this direction. We prove in \cref{sec:low-diameter-pruning} a ``low-diameter expander pruning lemma'' which states that a large portion of each expander in $\cH$ remain intact and of low diameter also in the residual graph $G_{\Bf}$, and that there are only a few ``pruned'' edges which cannot contribute too much to the size of all level cuts.
\end{description}

\paragraph{The Modified Weight Function.}
We begin by showing that although the weights of some edges are set to $0$ in $\Bw_{\Bf}$, the distance in the residual graph remains large.
Overloading notation, let
\[ \dist_{G_{\Bf}}^{\Bw}(v) \defeq \min_{\Bsource_{\Bf}(s) > 0} \dist_{G_{\Bf}}^{\Bw}(s, v). \]

\begin{lemma}
  \label{lemma:many-levels}
  If $\dist_{G_{\Bf}}^{\Bw_G}(v) > 3h$, then $\dist_{G_{\Bf}}^{\Bw_{\Bf}}(v) > h$.
\end{lemma}

\begin{proof}
  Consider a vertex $v$ and let $P$ be the shortest path with respect to $\Bw_{\Bf}$ from an unsaturated source to $v$ in $G_{\Bf}$.
  Thus, the $\Bw_{\Bf}$-weight of $P$ is  $\dist_{G_{\Bf}}^{\Bw_{\Bf}}(v) = \Bw_{G}(P \setminus \forward{D})$ because the weight $\Bw_{\Bf}$ is the same as $\Bw_{G}$ except that the weights of all forward DAG-edges are set to zero.  
  
  We note that $\Bw_{G}$ satisfies the assumption of \cref{lemma:dag-edge-length} (in the graph $G_{\Bf}$, with the DAG $\forward{D}$), i.e., that $\Bw_{G}(e) \ge |\Btau_{u}-\Btau_{v}|$ for any edge $e = (u,v)$ since $\Bw_{G}(e) = n$ for $e\in F$ and otherwise it follows from the definition \labelcref{eq:weight-function} of $\Bw_{\cH}$.
  Hence we have
  \[
  \Bw_{G}(P \cap \forward{D}) \le n + \Bw_{G}(P \setminus \forward{D})
  = n + \dist_{G_{\Bf}}^{\Bw_{\Bf}}(v).
  \]
  Since $\dist^{\Bw_G}_{G_{\Bf}}(v)$ is the shortest distance to $v$ (with respect to $\Bw_{G}$), we have
  \[
    3h < \dist^{\Bw_G}_{G_{\Bf}}(v)
    \leq \Bw_{G}(P)
      = \Bw_{G}(P \cap \forward{D}) + \Bw_{G}(P \setminus \forward{D})
      \leq n + 2\dist_{G_{\Bf}}^{\Bw_{\Bf}}(v).
  \]
  Rearranging, we see that $\dist_{G_{\Bf}}^{\Bw_{\Bf}}(v)> \frac{3h-n}{2}\ge h$, as $h \ge n$.
\end{proof}

\paragraph{Level Cuts.}
Let
\[ S_i \defeq \left\{v \in V: \dist_{G_{\Bf}}^{\Bw_{\Bf}}(v) = i\right\} \] 
be the distance levels in the residual graph with respect to this reduced weight function $\Bw_{\Bf}$.
By \cref{lemma:many-levels}, we know that $S_{\leq h} \neq V$. %
\cref{thm:flow} now directly follows from the below key lemma that establishes the existence of a sparse level cut.
In the remainder of the section
we prove \cref{lemma:exists-sparse-cut}.

\begin{restatable}{lemma}{ExistsSparseCut}
  \label{lemma:exists-sparse-cut}
  There exists a level cut $S_{\leq i}$ with $0 \leq i \leq h$ such that
  \begin{equation}
  \Bc^{\kappa}_{\Bf}(E_{G_{\Bf}}(S_{\leq i}, \overline{S_{\leq i}})) \leq O(|\Bf|) + \min\{\vol_F(S_{\leq i}), \vol_F(\overline{S_{\leq i}})\}.
  \label{eq:sparse-cut-cf}
  \end{equation}
\end{restatable}
\begin{proof}[Proof of \cref{thm:flow}]
We take $(S_{\leq i}, \overline{S_{\leq i}})$ as the output cut $(S, \overline{S})$.
By definition, \cref{thm:push-relabel-main-theorem}\labelcref{item:push-relabel:invariant}, and \cref{lemma:many-levels}, we have $S_0 = \{s: \ex_{\Bf}(s) > 0\}$ and $S_{\leq h} \cap \{t: \abs_{\Bf}(t) < \Bsink(t)\} = \emptyset$, and therefore $\ex_{\Bf}(S_{\leq i}) = \ex_{\Bf}(V)$ and $\abs_{\Bf}(S_{\leq i}) = \Bsink(S_{\leq i})$ hold.

What remains is to show that a cut $(S,\overline{S})$ satisfying \labelcref{eq:sparse-cut-cf} (which, by \cref{lemma:exists-sparse-cut} our algorithm will find whenever $|\Bf|<\|\Bsource\|_1$) also satisfies the output requirement \labelcref{eq:flow-guarantee} of \cref{thm:flow}, i.e., 
$\Bc^{\kappa}(E_{G}(S,\overline{S})) \le O(|\Bf|) + \min(\vol_{F}(S),\vol_{F}(\overline{S}))$. Indeed this is the case since
$\Bc_{\Bf}^{\kappa}(E_{G_{\Bf}}(S,\overline{S})) = 
\Bc^{\kappa}(E_{G}(S,\overline{S})) - \Bf^{\mathrm{out}}(S) \ge
\Bc^{\kappa}(E_{G}(S,\overline{S})) - |\Bf|$ by \cref{fact:flow}.
\end{proof}

\subsection{Existence of Sparse Level Cuts}
To prove \cref{lemma:exists-sparse-cut}, we show that each expander, while in the residual graph, has a relatively large portion that still has a low diameter.
Fixing a level $\ell$ in the hierarchy, let $\left\{C^{(1)}_{\ell}, C^{(2)}_{\ell}, \ldots, C^{(k)}_{\ell}\right\}$ be the strongly connected components of $(G\setminus F) \setminus X_{>\ell}$. That is,
$C^{(i)}_\ell$ is a level-$\ell$ expander and denote by $X^{(i)}_\ell = X_\ell\cap E(C^{(i)}_{\ell})$ the set of level-$\ell$ expanding edges in $C^{(i)}_\ell$.

We now argue that except for a small subset of ``pruned'' edges $P_{\ell}^{(i)}$, the edges of the expander remain well-connected and more importantly stay relatively close to each other.

\begin{restatable}{lemma}{LowDiameterExpanderNew}
  There exists a subset $P_{\ell}^{(i)} \subseteq X_{\ell}^{(i)}$ such that
  \begin{enumerate}[(1)]
    \item\label{low-diameter-expander-new:item1} for each pair $e_1, e_2 \in X_{\ell}^{(i)} \setminus P_{\ell}^{(i)}$, we have
    $\dist_{G_{\Bf}}^{\Bw_{\Bf}}(\forward{e_1}, \forward{e_2}) \leq O\left(\frac{\left|C^{(i)}_{\ell}\right|\eta^3\log^7 n}{\phi^{2}}\right)$,
    and
    \item\label{low-diameter-expander-new:item2}  $\left|P_{\ell}^{(i)}\right| \leq O\left(\frac{\eta \log^6 n}{\kappa \phi}\right) \cdot |\Bf|$.
  \end{enumerate}
  \label{lemma:low-diameter-expander-new} 
\end{restatable}

With \cref{lemma:low-diameter-expander-new} (whose proof we defer to \cref{sec:low-diameter-pruning}) we can now prove \cref{lemma:exists-sparse-cut}. First we prove the below intermediary lemma.
Recall that $\forward{F}$ is the set of forward edges of $F$ in the residual graph.
Let
\begin{equation}
c_{\ref{lemma:low-diameter-expander-new}} \geq 1\;\text{be (an upper bound on) the constant hidden in the}\;O(\cdot)\;\text{notation in \cref{lemma:low-diameter-expander-new}\labelcref{low-diameter-expander-new:item1}}.
\label{eq:low-diameter-constant}
\end{equation}

\begin{lemma}
  There are $g \geq \frac{h}{4}$ level cuts $S_{\leq i_1}, S_{\leq i_2}, \ldots, S_{\leq i_g}$ with $0 \leq i_1 < i_2 < \cdots < i_g \leq h$ such that
  \[ \sum_{1 \leq j \leq g}\Bc^{\kappa}_{\Bf}\left(E_{G_{\Bf}}(S_{\leq i_j}, \overline{S_{\leq i_j}}) \setminus \forward{F}\right) \leq O\left(|\Bf| \cdot h\right). \]
  \label{lemma:good-level-cuts}
\end{lemma}

\begin{proof}
  Let $P_{\mathrm{all}} \defeq \bigcup_{\ell}\bigcup_{i}P_{\ell}^{(i)}$ where the $P_{\ell}^{(i)}$'s are obtained from \cref{lemma:low-diameter-expander-new}.
  Let
  \[ \cS_{\mathrm{bad}} \defeq \left\{0 \leq i \leq h: E_{G_{\Bf}}(S_{\leq i}, \overline{S_{\leq i}}) \cap \forward{\left(\bigcup_{\ell}X_{\ell} \setminus P_{\mathrm{all}}\right)} \neq \emptyset\right\} \]
  be the set of level cuts that contain at least one expanding edge not in $P_{\mathrm{all}}$, and let 
  $\cS_{\mathrm{good}} \defeq \left\{0, 1, \ldots, h\right\} \setminus \cS_{\mathrm{bad}}$.
  By \cref{lemma:low-diameter-expander-new}\labelcref{low-diameter-expander-new:item1}, we know that $|\cS_{\mathrm{bad}}| \leq c_{\ref{lemma:low-diameter-expander-new}} \cdot \frac{n \eta^3 \log^7 n}{\phi^{2}} \cdot \eta \leq \frac{h}{4}$ since there are $\eta$ levels in the hierarchy and our choice of $h$ in \labelcref{eq:def-h}.
  This means that there are still $|\cS_{\mathrm{good}}| = h - |\cS_{\mathrm{bad}}| \geq \frac{h}{4}$ ``good'' level cuts $S_i$.

  By definition of $\dist_{G_{\Bf}}^{\Bw_f}$, an edge $e$ with $\Bc^{\kappa}_{\Bf}(e) > 0$ can be in at most $\Bw_{\Bf}(e)$ level cuts.
  There are only a few types of edges that can contribute to the size of a good ($i\in \cS_{\mathrm{good}}$) level cut $\Bc^{\kappa}_{\Bf}(E_{G_{\Bf}}(S_{\leq i}, \overline{S_{\leq i}}))$:
  \begin{enumerate}[(i)]
  \item \ul{Backward edges $\backward{e}$.} These have residual capacities $\Bc^{\kappa}_{\Bf}(\backward{e}) = \Bf(e)$. The contribution of these (across all good level cuts) can be bounded by $\sum_{\backward{e}\in \backward{E}} \Bc^{\kappa}_{\Bf}(\backward{e}) \Bw_{G}(e) = \sum_{e\in E} \Bf(e) \Bw_{G}(e) = \Bw_{G}(\Bf)$.
  \item \ul{Forward edges $\forward{e}$ from $F$.} These we do not care about in this lemma and will handle later.
  \item \ul{Forward DAG edges $\forward{e}$.} We have set $\Bw(\forward{e}) = 0$, so they cannot cross a level cut.
  \item \ul{Forward edges $\forward{e}$, where $e$ is a level-$\ell$ expanding edge inside some level-$\ell$ strongly connected component $C_{\ell}^{(i)}$.} By the definition of $\cS_{\mathrm{good}}$, we know that $e\in P^{(i)}_{\ell}$, so there are not too many of these edges. Note that $\Bw_{G}(\forward{e}) \le |C_{\ell}^{(i)}|$ and that these have residual capacity $\Bc^{\kappa}_{\Bf}(e) \le \Bc^{\kappa}(\forward{e}) = \kappa$ (recall that for the purpose of the analysis, we assume a unit-capacitated multi-graph, i.e., $\Bc = \Bone$).
  \end{enumerate}
  As a result, we can bound%
  \begin{align*}
    \sum_{i \in \cS_{\mathrm{good}}}\Bc^{\kappa}_{\Bf}\left(E_{G_{\Bf}}(S_{\leq i}, \overline{S_{\leq i}}) \setminus \forward{F}\right) &\leq \Bw_G(\Bf) + \kappa \cdot \left(\sum_{\ell}\sum_{i}\left|P_{\ell}^{(i)}\right| \cdot \left|C_{\ell}^{(i)}\right|\right) \\
    &\leq O\left(|\Bf| \cdot h\right) + \kappa \cdot \eta \cdot O\left(\frac{\eta \log^6 n}{\kappa\phi} \cdot |\Bf| \cdot n\right) \leq O(|\Bf| \cdot h),
  \end{align*}
  where we used \cref{lemma:low-diameter-expander-new}\labelcref{low-diameter-expander-new:item2} and $\Bw_G(\Bf) = O(|\Bf| \cdot h)$ by \cref{thm:push-relabel-main-theorem}\labelcref{item:push-relabel:short-flow}.
  \cref{lemma:good-level-cuts} follows by letting $\left\{i_1, \ldots, i_g\right\} \defeq
  \cS_{\mathrm{good}}$.
\end{proof}

We can now do a similar ball-growing argument as in \cref{lemma:expander-routing} to prove \cref{lemma:exists-sparse-cut}.

\ExistsSparseCut*

\begin{proof}
  Let $S_{\leq i_1}, \ldots, S_{\leq i_g}$ be the level cuts given by \cref{lemma:good-level-cuts}, and let
  \[ Z \defeq \sum_{1 \leq j \leq g}\Bc^{\kappa}_{\Bf}\left(E_{G_{\Bf}}(S_{\leq i_j}, \overline{S_{\leq i_j}}) \setminus \forward{F}\right) \leq O(|\Bf| \cdot h). \]
  By an averaging argument, at least half of the $S_{\leq i_j}$'s satisfy
  \begin{equation}
    \Bc^{\kappa}_{\Bf}\left(E_{G_{\Bf}}(S_{\leq i_j}, \overline{S_{\leq i_j}}) \setminus \forward{F}\right) \leq \frac{2Z}{g} \leq O(|\Bf|).
    \label{eq:other-edges-sparse}
  \end{equation}
  Let $i^{*}_1 < \cdots < i^{*}_{g/2}$ be indices satisfying \eqref{eq:other-edges-sparse}, and let $U_{j} \defeq S_{\leq i^{*}_{j \cdot n}} \setminus S_{\leq i^{*}_{(j-1) \cdot n}}$ for each $1 \leq j \leq \left\lfloor \frac{g}{2n}\right\rfloor$.
  Let $k \defeq \left\lfloor \frac{g}{2n}\right\rfloor \geq \frac{g}{4n}$.
  That is, we first split the distance levels into $g/2$ blocks at $i_1^{*}, \ldots, i_{g/2}^{*}$, and then merge every $n$ consecutive blocks to form the $U_j$'s.
  Observe that since $i^{*}_{j \cdot n} \geq i^{*}_{(j-1)\cdot n} + n$, we must have
  \begin{equation}
    \dist_{G_{\Bf}}^{\Bw_F}(U_j, U_{j+2}) > n
    \label{eq:dist-gap}
  \end{equation}
  for every $j$.
  We will now only consider level cuts that are between some $U_j$ and $U_{j+1}$ and bound the contribution of edges from $\forward{F}$ to them using a ball-growing argument.
  Note that if $\vol_F(U_{\leq 1}) = 0$ or $\vol_F(\overline{U_{\leq k}}) = 0$, then the lemma is vacuously true, and therefore we assume otherwise.
  We show that there exists a $1 \leq j \leq k$ such that
  \begin{equation}
    \Bc^{\kappa}_{\Bf}\left(E_{G_{\Bf}}(U_{\leq j}, \overline{U_{\leq j}}) \cap \forward{F}\right) \leq \min\{\vol_{F}(U_{\leq j}), \vol_F(\overline{U_{\leq j}})\},
    \label{eq:sparse-cut}
  \end{equation}
  which proves the lemma.
  Assume for contradiction that none of the $U_j$ satisfies \eqref{eq:sparse-cut}.
  Because of \eqref{eq:dist-gap} and that the weight of any edge is bounded by $n$
  we know that all edges in $E_{G_{\Bf}}(U_{\leq j}, \overline{U_{\leq j}})$ with positive capacities must be in $E_{G_{\Bf}}(U_j, U_{j+1})$.
  Let $\vol^\kappa_F(S) \defeq \kappa\vol_F(S)$.
  If $\vol_F(U_{\leq k/2}) \leq \vol_F(\overline{U_{\leq k/2}})$ then we have
  \[ \vol_{F}^\kappa(U_{\leq j}) \geq \vol_{F}^\kappa(U_{\leq j-1}) + \Bc^{\kappa}_{\Bf}\left(E_{G_{\Bf}}(U_{\leq j}, \overline{U_{\leq j}}) \cap \forward{F}\right) \]
  for $1 \leq j \leq k/2$ which with the assumption of \eqref{eq:sparse-cut} implies that
  \[
    \vol_F^\kappa(U_{\leq j}) \geq \left(1+\frac{1}{\kappa}\right)\vol_F^\kappa(U_{\leq j-1}) \implies
    \vol_F^\kappa(U_{\leq k/2}) \geq \left(1 + \frac{1}{\kappa}\right)^{k/2-1} > n^6
  \]
  since $k/2 - 1 \geq k/4 \geq \frac{h}{64n}$ (by \cref{lemma:good-level-cuts} we have $g \geq h/4$) and that $h \geq 1000 n\kappa \log n$ by \eqref{eq:def-h}.
  This is a contradiction because the $\vol_F^\kappa(S)$ of any $S$ should always be bounded by $2\kappa m \leq n^6$, where recall that $m \leq n^4$ is the total capacities of the input graph and $\kappa \leq n$ is required by \cref{thm:flow}.
  Similarly, if $\vol_F(U_{\leq k/2}) > \vol_F(\overline{U_{\leq k/2}})$, then we have
  \[ \vol_{F}^\kappa(\overline{U_{\leq j}}) \geq \vol_{F}^\kappa(\overline{U_{\leq j+1}}) + \Bc^{\kappa}_{\Bf}\left(E_{G_{\Bf}}(U_{\leq j}, \overline{U_{\leq j}}) \cap \forward{F}\right) \]
  for $k/2 < j < k$ and thus
  \[
    \vol_F^\kappa(\overline{U_{\leq j}}) \geq \left(1+\frac{1}{\kappa}\right)\vol_F^\kappa(\overline{U_{\leq j+1}}) \implies
    \vol_F^\kappa(U_{\leq k/2+1}) \geq \left(1 + \frac{1}{\kappa}\right)^{k/2-1}. 
  \]
  In both cases we have arrived at a contradiction, proving the lemma.
\end{proof}

\subsection{Robustness of Directed Expander Hierarchy under Flow Augmentation}
\label{sec:low-diameter-pruning}

In this section we prove \cref{lemma:low-diameter-expander-new}.
There are two main ingredients to this (which are independent of each other), each of which we believe might be of independent interest.
\begin{enumerate}[(a)]
\item\label{item:low-dia}
We show a generalization of the classic fact that expanders have low diameters. In particular, in \cref{lemma:expander-low-diameter} we show that for any weight function $\Bw \ge \Bzero$,  if $\Bnu$ satisfies $\Bnu(v) \geq \sum_{e\in \delta_{G}(v)} \Bw(e)$ and is $\sigma$-expanding in $G$, then the graph has $\Bw$-diameter $\tO(1/\sigma)$. This indeed generalizes the unweighted pure-expander setting when $\Bw = \Bone$ and $\Bnu = \deg_{E}$.
\item\label{item:exp-prune}
We show an expander pruning lemma saying that (directed) expanders are robust to path reversals (as well as some other updates, like increasing $\Bnu$). Indeed, a path reversal changes the size of any (directed) cut by at most one, similar to what happens when deleting an edge. This allows us to show \cref{lemma:pruning-reversal}, with similar guarantees as standard expander pruning, but which supports path reversals instead of edge deletions.
\end{enumerate}

In order to prove \cref{lemma:low-diameter-expander-new}, for each level-$i$ expander $C$ in $\cH$, we set up an appropriate $\Bnu \in \R_{\geq 0}^{V}$ and $\sigma \approx 1/n$ such that the fact that $\Bnu$ is $\sigma$-expanding  is a certificate that $C$ is initially of low-diameter $\tO(1/\sigma) = \tO(n)$ with respect to edge weights $\Bw_{G}$, via \labelcref{item:low-dia}.
We then show that throughout the run of the push relabel algorithm, a large part of $C$ remains $\sigma$-expanding (with respect to $\Bnu$). Indeed, every time we find an augmenting path in the push relabel algorithm, the residual graph changes by reversing the augmenting path, so we can apply \labelcref{item:exp-prune}. We have to be slightly careful here and use the additional fact that the augmenting paths found by our push relabel algorithm are short (\cref{lem:pr-no-shortcut}) in order to not blow up the diameter. At the end of the algorithm, \labelcref{item:low-dia} will imply that, except for a small pruned part of $C$, the expanding edges in $C$ remain of low diameter.

\subsubsection{Diameter of Expanders with Weighted Edges}
We begin by showing \labelcref{item:low-dia} in the lemma below, a generalization of the standard fact that expanders have low diameters.
Indeed, when $\Bnu(v) = \deg(v)$ and $\Bw(e) = 1$ it recovers the unweighted case.
Note that we are using $\sigma$ instead of $\phi$ to avoid confusion with the $\phi$ in the $\phi$-expander hierarchy:
One should think of $\sigma$ as being very small so that $1/\sigma$ corresponds to a certain notion of diameter induced by the weight function $\Bw_G$.
In particular, $\sigma$ can be as small as $\widetilde{O}(1/n)$, while the value $\phi$ for the expander hierarchy will be set to $1/n^{o(1)}$.

\begin{lemma}
  \label{lemma:expander-low-diameter}
  Suppose $\Bnu\in \R_{\ge 0}^{V}$ is $\sigma$-expanding in $H = (V, E)$ and edge weights $\Bw \in \N^{E}$ such that for all $v\in V$, $\Bnu(v) \geq \sum_{e \in \delta_H(v)}\Bw(e)$. Then for any $s, t \in V$ such that $\Bnu(s), \Bnu(t) > 0$ we have $\dist_{H}^{\Bw}(s, t) \leq O(\log(\Bnu(V)) /\sigma)$.
\end{lemma}

\begin{proof}
  The proof follows a standard ball-growing argument.
  Note that $U \defeq \{v \in V: \Bnu(v) > 0\}$ is strongly connected; otherwise, there will be a sparse cut.
  Let $s, t$ be the vertices with $\Bnu(s),\Bnu(t) > 0$ such that $D \defeq \dist_{H}^{\Bw}(s, t)$ is maximized, and assume for contradiction that $D > 16\left\lceil \frac{\log 4\Bnu(V)}{\sigma} \right\rceil = O(\log(\Bnu(V))/\sigma)$.
  Let $L_i \defeq \{v \in U: \dist_{H}^{\Bw}(s, v) = i\}$.
  
  Let $\Bnu^\prime \in \R_{\geq 0}^{\{0, \ldots, D\}}$ be defined as follows:
  First, we add $\Bnu(L_i)$ to $\Bnu^\prime(i)$.
  Then, for each $e \in E$ such that $e = (u, v)$ with $\dist_{H}^{\Bw}(s, u) < \dist_{H}^{\Bw}(s, v)$, we add $r_e \defeq \frac{\Bw(e)}{\dist_{H}^{\Bw}(s, v) - \dist_{H}^{\Bw}(s, u) + 1}$ 
  to $\Bnu^\prime(i)$  for each $\dist_{H}^{\Bw}(s, u) \leq i \leq \dist_{H}^{\Bw}(s, v)$.
  Observe that $r_e \geq 1/2$ by the fact that $\dist_{H}^{\Bw}(\cdot,\cdot)$ is the shortest-distance function and $\Bw(e) \geq 1$.
  Moreover,
  \[
    \Bnu(L_{\leq i}) \leq \sum_{0 \leq j \leq i}\Bnu^\prime(j) \leq 2\Bnu(L_{\leq i})\quad\text{and}\quad\Bnu(L_{\geq i}) \leq \sum_{D \geq j \geq i}\Bnu^\prime(j) \leq 2\Bnu(L_{\geq i})
  \]
  hold because $\Bnu(v) \geq \sum_{e \in \delta_H(v)}\Bw(e)$ for all $v$.
  By design, for each $0 \leq i < D$ we have
  \[ \min\{\Bnu^\prime(i), \Bnu^\prime(i + 1)\} \geq \sum_{e \in E_H(L_{\leq i}, \overline{L_{\leq i}})}r_e \geq \frac{1}{2}\left|E_H(L_{\leq i}, \overline{L_{\leq i}})\right|. \]
  Also, by the expansion guarantee of $H$, we have
  \[
     |E_H(L_{\leq i}, \overline{L_{\leq i}} )| \geq \sigma \cdot \min\left\{\Bnu(L_{\leq i}), \Bnu(\overline{L_{\leq i}})\right\} \geq \frac{\sigma}{2} \cdot \min\left\{\sum_{j \leq i}\Bnu^\prime(j), \sum_{j \geq i + 1}\Bnu^\prime(j)\right\}.
  \]
  With these we can now do a standard ball-growing argument.
  Let $h \defeq \lfloor D/2\rfloor$.
  If $\sum_{j \leq h}\Bnu^\prime(j) \leq \sum_{j > h}\Bnu^\prime(j)$, then
  \[
    \sum_{j \leq i + 1}\Bnu^\prime(j) \geq \left(1 + \frac{\sigma}{4}\right) \cdot \sum_{j \leq i}\Bnu^\prime(j)
  \]
  holds for all $0 \leq i \leq h$ and therefore
  \[
    \sum_{j \leq h}\Bnu^\prime(j) \geq \left(1+\frac{\sigma}{4}\right)^{h} \cdot \Bnu^\prime(0) \geq \left(1+\frac{\sigma}{4}\right)^{h} \geq 4\Bnu(V),
  \]
  by $\Bnu^\prime(0) \geq \Bnu(s) \geq 1$, which is a contradiction.
  On the other hand, if $\sum_{j \leq h}\Bnu^\prime(j) > \sum_{j > h}\Bnu^\prime(j)$, then similarly
  \[
    \sum_{j \geq i}\Bnu^\prime(j) \geq \left(1 + \frac{\sigma}{4}\right) \cdot \sum_{j \geq i + 1}\Bnu^\prime(j)
  \]
  holds for all $h \leq i \leq D$ and therefore
  \[
    \sum_{j \geq h}\Bnu^\prime(j) \geq \left(1 + \frac{\sigma}{4}\right)^{D-h} \cdot \Bnu^\prime(D) \geq \left(1 + \frac{\sigma}{4}\right)^h \geq 4\Bnu(V),
  \]
  a contradiction as well.
  This proves the lemma.
\end{proof}

\subsubsection{Expander Pruning under Path-Reversals}
Now we show \labelcref{item:exp-prune} in the lemma below. Note that for our purposes, we only need an existential expander pruning lemma, so we do not care about making it algorithmically efficient.
There are a few different types of updates we support, the main ones being reversing a path and adding some volume to $\Bnu$, tailored for our use later in this section. We note that the lemma should seamlessly extend to also support edge deletions (as is the usual goal of expander pruning) with the same guarantees as the path reversals, but we do not need it for our purposes, hence we skip it. Since many flow and cut algorithms work via reversing paths, we believe our expander pruning lemma might be of independent interest.

\begin{fact}
  Let $G = (V, E)$ be a graph and $G^\prime$ be obtained from $G$ by reversing a path in it.
  Then, we have for each $S \subseteq V$ that $\Big||E_G(S, V \setminus S)| - |E_{G^\prime}(S, V \setminus S)|\Big| \leq 1$.

  \label{fact:path-reversal}
\end{fact}

\begin{lemma}[Expander Pruning under Path-Reversals]\label{lemma:pruning-reversal}
  Given $\Bnu\in \R_{\ge 0}^{V}$ such that $\Bnu$ is $\sigma$-expanding in $G = (V, E)$, one can (inefficiently) maintain pruned sets $\emptyset = P^{(0)} \subseteq P^{(1)} \subseteq \cdots \subseteq P^{(k)} \subseteq V$ while $G$ undergoes $k$ updates, the $i$-th of which either
  \begin{enumerate}[(1)]
    \item\label{op:add-vertex} adds a vertex $v_i$ to $G$ with volume $\Bnu(v_i) \defeq 0$,
    \item\label{op:add-edge} adds an edge $e_i$ to $G$ whose endpoints are not in $P$,
    \item\label{op:add-vol} adds $\Delta_i \in \N$ to $\Bnu(v_i)$ for some vertex $v_i$, or
    \item\label{op:reverse} reverses a path $R_i$ in $G$ that does not intersect $P$,
  \end{enumerate}
  where $P$ denotes the current pruned set,
  such that
  \begin{itemize}
    \item if the $i$-th update is of type (1) or (2), then\footnote{It is natural that the weight of the prune set is independent from the number of operations (1) and (2), because if $\Bnu$ was $\sigma$-expanding before, then it remains so after such an operation.}
    $P^{(i)} = P^{(i-1)}$, and
    \item
    $\Bnu^{(i)}$ is $\frac{\sigma}{8}$-expanding in
    $G^{(i)} \setminus P^{(i)}$ and $\Bnu^{(i)}(P^{(i)}) \leq O(k_i/\sigma + \sum_{j \leq i}\Delta_j)$, with $G^{(i)}$ and $\Bnu^{(i)}$ denoting the graph and vertex weights after the $i$-th update and $k_i$ denoting the number of path reversals in the first $i$ updates.
  \end{itemize}
\end{lemma}
\begin{proof}
  We describe how the pruned set $P^{(i)}$ is obtained from $P^{(i-1)}$.
  We maintain for all $i$ the invariant that $P^{(i)}$ can be written as the disjoint union of two sets $P^{(i)}_{+}$ and $P^{(i)}_{-}$ such that\footnote{Recall that with $E_{H}(A,\overline{B})$ and $E_{H}(\overline{B}, A)$, we mean $\overline{B} = H\setminus B$.}
  \begin{equation}
    \Bnu^{(i)}(P^{(i)}) \geq \frac{8}{\sigma}\left(
      \left|E_{G^{(i)}}(P^{(i)}_{+}, \overline{P^{(i)}})\right| +
      \left|E_{G^{(i)}}(\overline{P^{(i)}}, P^{(i)}_{-})\right| +
      \left|E_{G^{(i)}}(P^{(i)}_{+}, P^{(i)}_{-})\right|\right).
    \label{eq:invariant}
  \end{equation}
  Given that $P^{(i-1)}$ satisfies \eqref{eq:invariant} in $G^{(i-1)}$, we observe that if we initialize $P^{(i)} \gets P^{(i-1)}$, then it satisfies \eqref{eq:invariant} in $G^{(i)}$ since we are not adding edges or reversing paths intersecting $P$, and increasing vertex weights also only makes the left-hand side larger.

  After $P^{(i)}$ is initialized, we repeat the following procedure:
  As long as there is a cut $S$ that is $\frac{\sigma}{8}$-sparse with respect to $\Bnu^{(i)}$ in $G^{(i)} \setminus P^{(i)}$, we include it to $P^{(i)}$ by setting $P^{(i)} \gets P^{(i)} \cup S$.
  To see that this does not break the invariant, we assume without loss of generality that $S$ is out-sparse, i.e.,
  \[
    \Bnu^{(i)}(S) \leq \Bnu^{(i)}(\overline{S})\quad\text{and}\quad \left|E_{G^{(i)} \setminus P^{(i)}}(S,\overline{S})\right| < \frac{\sigma}{8} \cdot \Bnu^{(i)}(S),
  \]
  where $\overline{S} \defeq (V(G^{(i)}) \setminus P^{(i)}) \setminus S$.
  Then, we show that we can add $S$ into $P^{(i)}_{+}$ while preserving the invariant.
  To see this, we compute %
  \begin{align*}
    \Bnu^{(i)}(P^{(i)} &\cup S)
    =
    \Bnu^{(i)}(P^{(i)}) + \Bnu^{(i)}(S)\\
    &\geq \frac{8}{\sigma}\left(
      \left|E_{G^{(i)}}(P^{(i)}_{+}, \overline{P^{(i)}})\right| +
      \left|E_{G^{(i)}}(\overline{P^{(i)}}, P^{(i)}_{-})\right| +
      \left|E_{G^{(i)}}(P^{(i)}_{+}, P^{(i)}_{-})\right|
      \right) + 
      \frac{8}{\sigma}\left|E_{G^{(i)} \setminus P^{(i)}}(S, \overline{S})\right| \\
    &\ge \frac{8}{\sigma}\left(
      \left|E_{G^{(i)}}(P^{(i)}_{+} \cup S, \overline{P^{(i)} \cup S})\right| +
      \left|E_{G^{(i)}}(\overline{P^{(i)} \cup S}, P^{(i)}_{-})\right| +
      \left|E_{G^{(i)}}(P^{(i)}_{+} \cup S, P^{(i)}_{-})\right|\right),
  \end{align*}
  which is precisely \eqref{eq:invariant} for the new $P^{(i)}$ and its partition $(P^{(i)}_{+}, P^{(i)}_{-})$.
  The last inequality follows from that the edges counted in both lines are exactly the same, except for $E_{G^{(i)}}(P^{(i)}_{+}, S)$ which is counted in the former but not the latter.
  The case where $S$ is an in-sparse cut can be shown symmetrically (except that it will now be added to $P^{(i)}_{-}$).

  We further show that \eqref{eq:invariant} implies the desired upper bound on $\Bnu^{(i)}(P^{(i)})$.
  Overloading notation, we extend the vertex set of $V \defeq V(G^{(0)})$ to be $V(G^{(i)})$ by adding isolated vertices with weights $0$.
  Note that $\Bnu^{(0)}$ remains $\sigma$-expanding in $G^{(0)}$ after this extension.
  From this point of view, we see that $G^{(i)}$ is obtained from $G^{(0)}$ by adding edges, increasing vertex weights, and reversing paths; no vertex addition is involved now.
  We first assume that $\Bnu^{(0)}(P^{(i)}) \leq 2\Bnu^{(0)}(V)/3$, i.e., $\Bnu^{(0)}(P^{(i)}) \leq 2\min\{\Bnu^{(0)}(P^{(i)}), \Bnu^{(0)}(\overline{P^{(i)}})\}$.
  We know by the expansion guarantee of $G^{(0)}$ that both $P^{(i)}_{+}$ and $P^{(i)}_{-}$ are not sparse in $G^{(0)}$ and therefore
  \begin{align*}
    \Bnu^{(0)}(P^{(i)}) &= \Bnu^{(0)}(P^{(i)}_{+}) + \Bnu^{(0)}(P^{(i)}_{-})
    \leq \frac{2}{\sigma}\left(\left|E_{G^{(0)}}(P^{(i)}_{+}, \overline{P^{(i)}_{+}})\right| + \left|E_{G^{(0)}}(\overline{P^{(i)}_{-}}, P^{(i)}_{-})\right|\right) \\
    &= \frac{2}{\sigma}\left(\left|E_{G^{(0)}}(P^{(i)}_{+}, \overline{P^{(i)}})\right| + \left|E_{G^{(0)}}(\overline{P^{(i)}}, P^{(i)}_{-})\right| + 2\left|E_{G^{(0)}}(P^{(i)}_{+}, P^{(i)}_{-})\right|\right),
  \end{align*}
  where the first inequality was based on
  \[ \min\left\{\Bnu^{(0)}(P^{(i)}_{+}), \Bnu^{(0)}(\overline{P^{(i)}_{+}})\right\} \geq \frac{1}{2}\Bnu^{(0)}(P^{(i)}_{+})\quad\text{and}\quad \min\left\{\Bnu^{(0)}(P^{(i)}_{-}), \Bnu^{(0)}(\overline{P^{(i)}_{-}})\right\} \geq \frac{1}{2}\Bnu^{(0)}(P^{(i)}_{-}) \]
  by our assumption.
  As such, we have
  \begin{align*}
    \Bnu^{(i)}(P^{(i)})
    &\leq \frac{2}{\sigma}\left(\left|E_{G^{(0)}}(P^{(i)}_{+}, \overline{P^{(i)}})\right| + \left|E_{G^{(0)}}(\overline{P^{(i)}}, P^{(i)}_{-})\right| + 2\left|E_{G^{(0)}}(P^{(i)}_{+}, P^{(i)}_{-})\right|\right) + \sum_{j \leq i}\Delta_j \\
    &\stackrel{(i)}{\leq} \frac{2}{\sigma}\left(\left|E_{G^{(i)}}(P^{(i)}_{+}, \overline{P^{(i)}})\right| + \left|E_{G^{(i)}}(\overline{P^{(i)}}, P^{(i)}_{-})\right| + 2\left|E_{G^{(i)}}(P^{(i)}_{+}, P^{(i)}_{-})\right| + 2k_i\right) + \sum_{j \leq i}\Delta_j \\
    &\leq \frac{2}{\sigma}\left(2\left|E_{G^{(i)}}(P^{(i)}_{+}, \overline{P^{(i)}})\right| + 2\left|E_{G^{(i)}}(\overline{P^{(i)}}, P^{(i)}_{-})\right| + 2\left|E_{G^{(i)}}(P^{(i)}_{+}, P^{(i)}_{-})\right|\right) + \frac{4k_i}{\sigma} + \sum_{j \leq i}\Delta_j \\
    &\stackrel{(ii)}{\leq} \frac{\Bnu^{(i)}(P^{(i)})}{2} + \frac{4k_i}{\sigma} + \sum_{j \leq i}\Delta_j,
  \end{align*}
  where (i) follows from
  \[ E_{G^{(r)}}(P^{(i)}_{+}, \overline{P^{(i)}}) \cup E_{G^{(r)}}(P^{(i)}_{+}, P^{(i)}_{-}) = E_{G^{(r)}}(P^{(i)}_{+}, \overline{P^{(i)}_{+}}) \]
  and
  \[ E_{G^{(r)}}(\overline{P^{(i)}}, P^{(i)}_{-}) \cup E_{G^{(r)}}(P^{(i)}_{+}, P^{(i)}_{-}) = E_{G^{(r)}}(\overline{P^{(i)}_{-}}, P^{(i)}_{-}) \]
  for $r \in \{0, i\}$ with \cref{fact:path-reversal},
  and (ii) follows from \eqref{eq:invariant}.
  This implies $\Bnu^{(i)}(P^{(i)}) \leq O(k_i/\sigma + \sum_{j \leq i}\Delta_j)$.
  
  The case when $\Bnu^{(0)}(P^{(i)})> 2\Bnu^{(0)}(V)/3$ can be argued similarly:
  Consider the moment when we added $S$ to $P^{(j)}$ for some $j \leq i$ such that $\Bnu^{(0)}(P^{(j)}) \leq 2\Bnu^{(0)}(V)/3$ but $\Bnu^{(0)}(P^{(j)} \cup S) > 2\Bnu^{(0)}(V)/3$.
  Applying the calculation above we know that $\Bnu^{(j)}(P^{(j)}) \leq O(k_j/\sigma + \sum_{k \leq j}\Delta_j)$ at this moment, before $S$ is included.
  This gives that
  \begin{align*}
    \frac{2}{3}\Bnu^{(0)}(V)
    &< \Bnu^{(0)}(P^{(j)}) + \Bnu^{(0)}(S) \leq \Bnu^{(j)}(P^{(j)}) + \Bnu^{(j)}(S) \\
    &\leq O(k_j/\sigma + \sum_{k \leq j}\Delta_k) + \frac{\Bnu^{(j)}(V)}{2} \leq O(k_j/\sigma + \sum_{k \leq j}\Delta_k) + \frac{\Bnu^{(0)}(V)}{2} + \sum_{k \leq j}\Delta_k,
  \end{align*}
  which implies $\Bnu^{(0)}(V) \leq O(k_j/\sigma + \sum_{k \leq j}\Delta_k) \leq O(k_i/\sigma + \sum_{j \leq i}\Delta_j)$.
  Since $\Bnu^{(i)}(P^{(i)})$ can be trivially upper bounded by $\Bnu^{(i)}(V)$, we get $\Bnu^{(i)}(P^{(i)}) \leq \Bnu^{(0)}(V) + \sum_{j \leq i}\Delta_j \leq O(k_i/\sigma + \sum_{j \leq i}\Delta_j)$, which is the bound we wanted.
  This completes the proof of \cref{lemma:pruning-reversal}.
\end{proof}

\subsubsection{Proof of \cref{lemma:low-diameter-expander-new}}
Now we have the two ingredients \labelcref{item:low-dia,item:exp-prune} in order to prove our \cref{lemma:low-diameter-expander-new}.

\LowDiameterExpanderNew*

Let us focus on a level-$\ell$ expander $C$.
We first describe the high-level strategy of the proof.

\paragraph{Constructing Initial Expander.}
By our analysis in \cref{sec:weight}, we know that all expanding edges in $C$ are only $\tO(|C|/\phi)$ far away from each other, with respect to $\Bw_{G}$-distance. This lets us find an appropriate $\Bnu \in \R_{\geq 0}^{H}$ and $\sigma \approx \frac{\phi^{2}}{|C|}$, in \cref{clm:nu-and-sigma}, such that $\Bnu(v)\ge \sum_{e\in \delta_{H}(v)}\Bw_{G}(e)$ and $\Bnu$ is $\sigma$-expanding in $H$, where $H$ is some subgraph of $C^{\kappa}$.\footnote{Recall that $C^{\kappa}$ is the graph with all edges duplicated $\kappa$ times---indeed, the flow algorithm will work in this graph.} By our generalized ``expanders have low diameter''  argument in \cref{lemma:expander-low-diameter}, $(H, \Bnu, \sigma)$ is now a certificate/witness that the expanding edges in $H$ are of low-diameter $\tO(|C|/\phi^{2})$. We will set $H$ to the graph we would get if we run a cut-matching game on $C$, where in each round we find a \emph{short} matching. The graph $H$ will precisely consist of edges certifying that $C$ is of low diameter, but not include irrelevant parts of $C$ which might be far from all expanding edges in $C$. %

\paragraph{Handling Path-Reversal.}
Next we will consider how the graph develops when we run our push-relabel augmenting paths algorithm. Throughout, we will maintain $(H,\Bnu,\sigma)$ and a small pruned set $P$ as a certificate/witness that most edges from $C$ are still of low-diameter, via our expander pruning \cref{lemma:pruning-reversal}. In particular, $\Bnu$ will be $\Theta(\sigma)$-expanding in $H\setminus P$.

\begin{enumerate}[(a)]
\item \ul{Truncating the Path.}
In particular, consider what happens when we want to reverse an augmenting path $R$. Let $R'$ be the subpath from the first vertex in $H\setminus P$ to the last vertex in $H\setminus P$. Note that when focusing on $H$, we do not care about how the path $R$ looks like outside of the subpath $R'$. Nevertheless, note that it is still possible for $R^\prime$ to go outside of $H \setminus P$ (to $V \setminus H$ or $P$).

\item \ul{Bounding Path Length.}
We first notice that, since our push relabel algorithm almost finds shortest paths (\cref{lem:pr-no-shortcut}), it must be the case that $R'$ is of $\Bw_{G}$-length $\tO(|C|/\phi^{2})$ since $H\setminus P$ is still of low diameter.

\item \ul{Adding New Vertices to $H$.}
We add all vertices on $R'$ which are not already in $H\setminus P$ as ``fresh'' vertices to $H$ (in particular, if $R'$ intersects the pruned set $P$ we add back new copies of these vertices), and add all the edges of $R'$ not already in $H\setminus P$ to $H$ (using operations \labelcref{op:add-vertex,op:add-edge} in \cref{lemma:pruning-reversal}).
Note that $H$ might no longer be a subgraph of $C$ (as $R'$ can move outside of $C$). In fact, this is necessary: the expanding edges of $C$ will at the end of the push relabel algorithm be of low diameter inside of $G^{\kappa}_{\Bf}$, but not necessarily inside the induced subgraph $G^{\kappa}_{\Bf}[C]$.
We remark that technically $H$ may now contain multiple copies of the same vertex in $G$, but only one of these copies will be ``active'' and the others will be in $P$.%

\item \ul{Performing Path-Reversal.}
We then reverse the path $R'$ (using operation \labelcref{op:reverse} in \cref{lemma:pruning-reversal}), which might increase the pruned set $P$ a bit.

\item \ul{Increasing Vertex Weights.}
Additionally we must increase $\Bnu(v)$ for the vertices $v$ incident to $R'$ a bit, to maintain that $\Bnu(v)\ge \sum_{e\in \delta_H(v)} \Bw_{G}(e)$ after we added some edges to $H$ (using operation \labelcref{op:add-vol} in \cref{lemma:pruning-reversal}). This again might increase the pruned set $P$ a bit, and thus we will increase $\Bnu(v)$ proportional to the $\Bw_{G}$-length of $R'$ (which we argued above is not too long) in order to control this blow-up.
\end{enumerate}

At the end, after all augmenting paths, $\Bnu$ is still $\Theta(\sigma)$-expanding in $H\setminus P$, and the pruned set $P$ is small, which lets us conclude \cref{lemma:low-diameter-expander-new}.

\ifnum\cameraready=0

\begin{figure}[!bt]
    \centering

    \begin{subfigure}[b]{0.4\textwidth}
        \includegraphics[width=\textwidth]{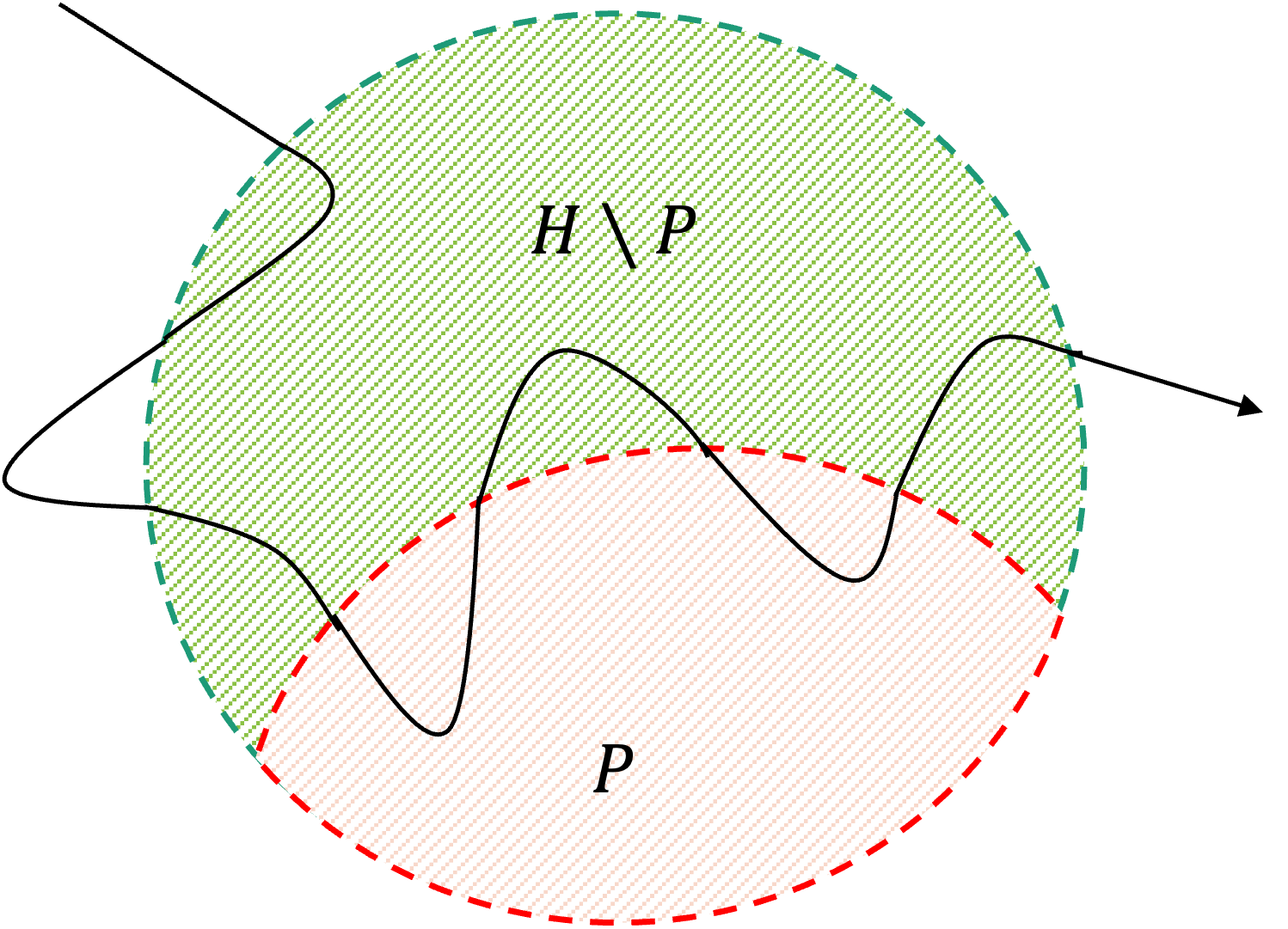}        
    \end{subfigure}
    \begin{subfigure}[b]{0.4\textwidth}
        \includegraphics[width=\textwidth]{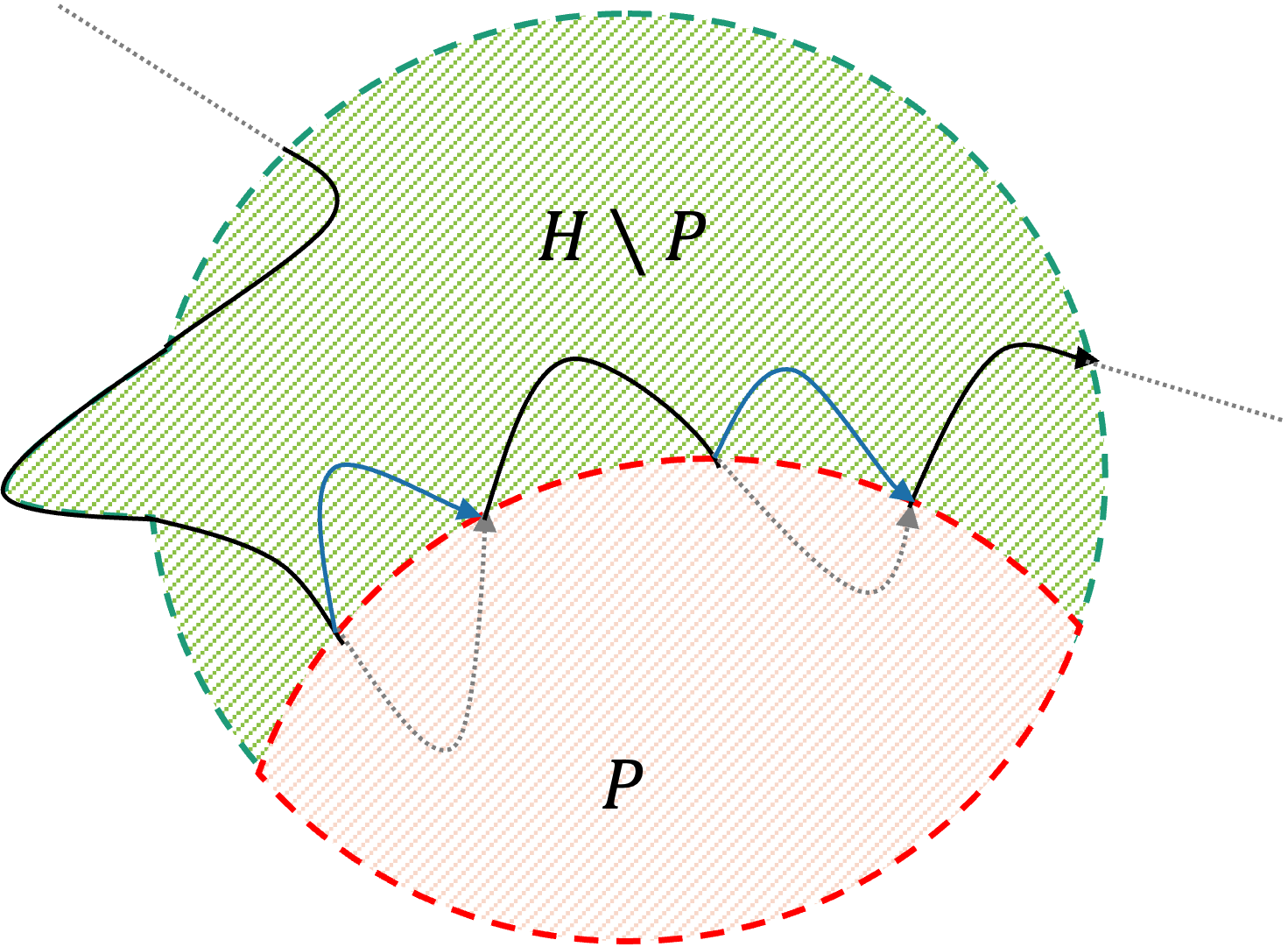}        
    \end{subfigure}
    \caption{Illustration of how a path reversal is handled.
    After truncating the path from the first intersection with $H\setminus P$ to the last, we add vertices on the path not in $H$ into the graph. If the path goes into the pruned set $P$, we also add the corresponding ``fresh'' vertices to the graph and ``reroute'' the path segment inside the un-pruned part which creates the blue segments. The final path that we reverse (via \cref{lemma:pruning-reversal}\labelcref{op:reverse}) at the end consists of the black and blue path segments.}
    \label{fig:reverse}
\end{figure}

\else

\begin{figure*}[!bt]
    \centering

    \begin{subfigure}[b]{0.4\textwidth}
        \includegraphics[width=\textwidth]{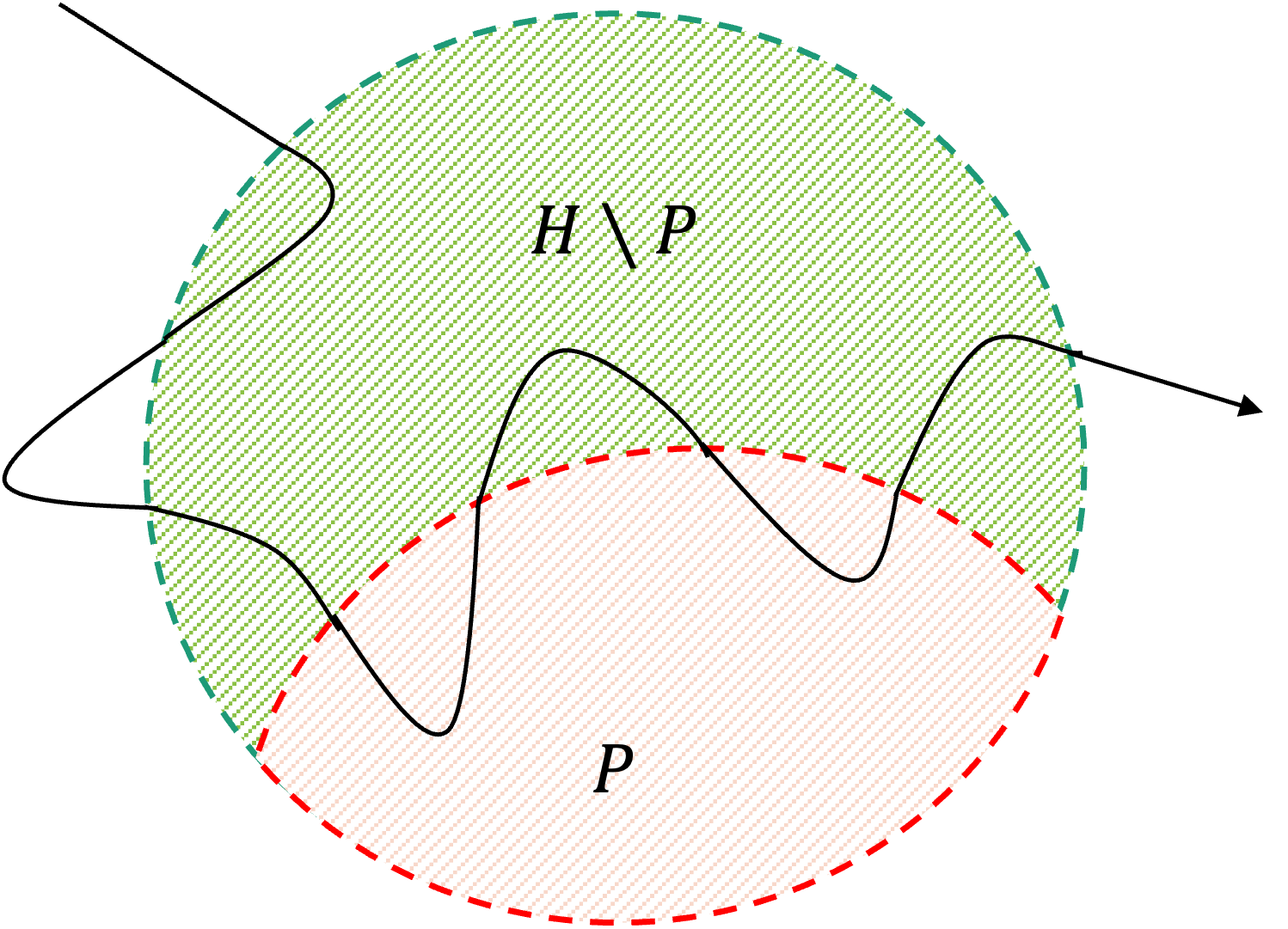}        
    \end{subfigure}
    \begin{subfigure}[b]{0.4\textwidth}
        \includegraphics[width=\textwidth]{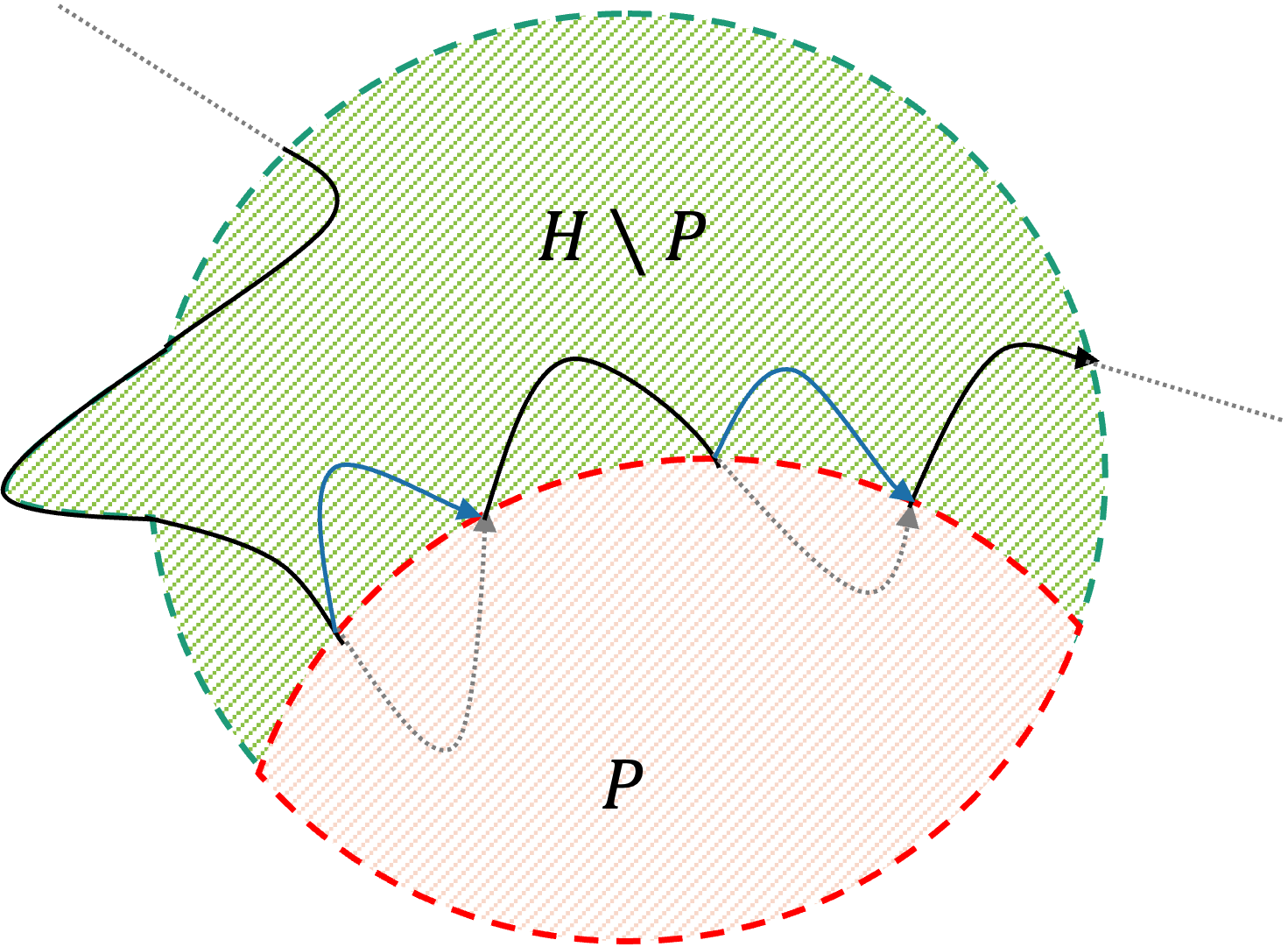}        
    \end{subfigure}
    \caption{Illustration of how a path reversal is handled.
    After truncating the path from the first intersection with $H\setminus P$ to the last, we add vertices on the path not in $H$ into the graph. If the path goes into the pruned set $P$, we also add the corresponding ``fresh'' vertices to the graph and ``reroute'' the path segment inside the un-pruned part which creates the blue segments. The final path that we reverse (via \cref{lemma:pruning-reversal}\labelcref{op:reverse}) at the end consists of the black and blue path segments.}
    \label{fig:reverse}
\end{figure*}

\fi

We begin with this useful claim, which will allow us to set up the appropriate vertex volume $\Bnu$.
\begin{claim}
  \label{claim:was-routable}
  For every $1$-respecting demand $(\Bsource, \Bsink)$ on $X_{\ell}^{(i)}$%
  we can route it by an integral flow in $C_{\ell}^{(i)}$ with congestion $O(\eta \log n / \phi)$ such that each flow path has $\Bw_G$-length at most $O\left(\left|C_{\ell}^{(i)}\right|\eta^{2}\log n /\phi\right)$.
\end{claim}

\begin{proof}
  Every 1-respecting demand on $X^{(i)}_{\ell}$, by definition of expansion and the max-flow min-cut theorem \cref{fact:maxflow-mincut}, is routable in $C^{(i)}_{\ell}$ with congestion $\frac{1}{\phi}$, and hence also with integral congestion $\lceil\frac{1}{\phi}\rceil\le \frac{2}{\phi} = O(\frac{1}{\phi})$ (since $\phi\le 1$).
  Note that we may restrict the expander hierarchy $\cH$ to $C^{(i)}_{\ell}$ (only keeping expanding edges of level at most $\ell$). 
  The claim then follows from \cref{cor:good-weight-function-integral} (we do not use edges from $F$ for this routing, so we have $\Bw_{G}(e) = \Bw_{\cH}(e)$).%
\end{proof}

\begin{proof}[Proof of \cref{lemma:low-diameter-expander-new}.]
  Let $X \defeq X_{\ell}^{(i)}$ and $C \defeq C_{\ell}^{(i)}$ (that is, $X$ are the expanding edges inside some level-$\ell$ expander $C$).
  Consider running the cut-matching game\footnote{While the cut-matching game in \cref{thm:directed-cut-matching-game} is randomized and only works with high probability, here we may simply assume that the randomness used in the cut-matching game is such that it succeeds (indeed such random bits exists, and in this section we only need \emph{existence} of the following witness $W$ and embedding).} of \cref{thm:directed-cut-matching-game} on $\deg_{X}$ to construct a witness $W$ (embeddable into $C$) in which $\deg_{X}$ is $\psi_{\CMG}$-expanding (with $\frac{1}{\psi_{\CMG}} = O(\log^2 n)$).
  Every time we are given a bipartition $(\Bnu_{A}, \Bnu_{B})$ of $\deg_{X}$, we apply \cref{claim:was-routable} to find a matching $\overrightarrow{M}$ and $\overleftarrow{M}$ between $\vol_{A}$ and $\vol_{B}$ that are routable in $C$ with congestion $O(\eta\log n /\phi)$ such that each edge is embedded into a path of $\Bw_G$-weight $O\left(\left|C\right|\eta^2\log n /\phi\right)$.
  Overall, after $t_{\CMG} = O(\log^2 n)$ rounds, we get a witness $W$ embeddable into $C$ with congestion $\kappa_W \defeq O(\eta \log^3 n / \phi)$ where each edge of $W$ is embedded into a path of $\Bw_G$-length $D \defeq O\left(\left|C\right|\eta^2\log n /\phi\right)$.
  We let $H_{0} \subseteq C$ be the image of the embedding (that is, $H_{0}$ consist of the union (after removing duplicates) of edges on all paths in the embedding of $W$ to $C$, and only vertices incident to those edges). %
  
  We will construct vertex volumes $\Bnu_{0} \in \N^{V(H_{0})}$ as follows.
  For each edge $e_{W} = (u, v) \in E(W)$, let $P_{e_{W}} \subseteq H_{0}$ be the embedding path of $(u, v)$ into $H_{0}$.
  For each $e \in P_{e_W}$, we add $\Bw_G(e)$ to the vertex weights of both of its endpoints.
  We then add $D$ to both $\Bnu_{0}(u)$ and $\Bnu_{0}(v)$. (Now, note that $\Bnu_{0}$ is almost a scaled up version of $\deg_X$: in fact $\Bnu_{0} = \Bnu_{0}^a+\Bnu_{0}^b$ where $D\cdot \deg_X \le \Bnu_{0}^a\le t_{\CMG} D\cdot \deg_X$, and $\|\Bnu_{0}^b\|_1 \le 2 \|\Bnu_{0}^a\|_1$).
  This construction guarantees $\Bnu_0(v)\ge \sum_{e\in \delta_{H_0}(v)} \Bw(e)$.

  Since $\deg_X$ is $\psi_{\CMG}$-expanding in $W$ and embeddable into $H_{0}$ with congestion $\kappa_W$, and as we noted before, $\Bnu_{0} \approx D \cdot \deg_X$, the following claim is reasonable.

  \begin{claim}
  \label{clm:nu-and-sigma}
    If $\sigma \defeq \frac{\psi_{\CMG}}{4D\kappa_W} = \Omega\left(\frac{1}{\kappa_W D\log^2 n}\right)$,
  then $\Bnu_{0}$ is $\sigma$-expanding in $H_{0}$.
  \end{claim}

  \begin{proof}
  Consider some cut $S\subseteq V(H_{0})$, with $\Bnu_{0}(S)\le \Bnu_{0}(\overline{S})$. We want to argue that $|E_{H_{0}}(S,\overline{S})|\ge \sigma \Bnu_{0}(S)$ and $|E_{H_{0}}(\overline{S},S)|\ge \sigma \Bnu_{0}(S)$. We argue the former, and the latter is symmetric.

  Consider the multiset of edges
  $E' = \bigcup_{e_W\in E(W)} P_{e_W}$
  from the embedding from $W$ to $H_{0}$. Since this embedding has congestion $\kappa_W$, we know that $|E'(S,\overline{S})| 
\le \kappa_W |E_{H_{0}}(S,\overline{S})|$, so it suffices to show that $|E'(S,\overline{S})|\ge \frac{\psi_{\CMG}}{4D}\Bnu_{0}(S)$.
We write $E'(S,\overline{S}) = E'_1 \cup E'_2$, where $E'_1$ consists of those edges $e$ which comes from embedding paths $P_{e_W}$ where $e_W\in E_W(S,\overline{S})$, and $E'_2$ are the remaining ones.

  We now bound $\Bnu_{0}(S)$ as follows.
  Recall that for each edge $e_W = (u,v)\in E(W)$, we added $D$ to $\Bnu_{0}(u)$ and $\Bnu_{0}(v)$, as well as $\Bw_{G}(e)$ to the two endpoints of $e$ for each $e$ on $P_{e_{W}}$. In particular, the total contribution of $e_W$ to all of $\Bnu_{0}$ (and thus to $\Bnu_{0}(S)$) is at most $4D$.
  If either $u$ or $v$ is in $S$, we can charge this cost of $4D$ to the contribution of $e_W$ to $\vol_{E(W)}(S)$.

  The only volume in $\Bnu_{0}(S)$ we have not accounted for now, is exactly the volume coming from edges $e_W = (u,v)\in E(W)$ where $u,v\in \overline{S}$, but for which the path $P_{e_W}$ intersects $S$. Such edges $e_W$ can contribute at most $2D$ to the volume of $\Bnu_{0}(S)$, and they must also contribute at least one edge in $E'_2$.

  The above reasoning shows that 
  $\Bnu_{0}(S)\le 2D \cdot |E'_2| + 4D \cdot \vol_{E(W)}(S)$.
  Since $W$ is a $\psi_{\CMG}$-expander, we have
  $|E'_1| \ge |E_{W}(S,\overline{S})| \ge \psi_{\CMG} \vol_{E(W)}(S)$ (each edge $e_W\in E_{W}(S,\overline{S})$ must clearly have a counterpart in $E'_1$).
  Thus we conclude $\Bnu_{0}(S)\le 2D\cdot |E'_2| + 4D\cdot |E'_1|/\psi_{\CMG}$, and hence that
  $\frac{\psi_{\CMG}}{4D}\Bnu_{0}(S)\le |E'_1|+|E'_2|$, which proves the lemma.
  \end{proof}

  \paragraph{Setup.}
  Now we can initialize $\Bnu\gets \kappa\cdot \Bnu_0$ and $H \gets (H_{0})^{\kappa}$ (that is $H_{0}$, but with each edge duplicated $\kappa$ times).
  By construction, $H$ is a subgraph of $G^{\kappa}$, and from \cref{clm:nu-and-sigma} we know that $\Bnu$ is $\sigma$-expanding in $H$. Also, by construction, $\Bnu(v)\ge \sum_{e\in \delta_{H}(v)} \Bw(e)$ (and hence by \cref{lemma:expander-low-diameter}, of diameter $O(\frac{\log (n\kappa)}{\sigma})$, and recall that $\kappa \le \poly(n)$). We also know that $\Bnu \ge \kappa \cdot D\cdot \deg_X$.

  Now we consider reversing each flow path in $\Bf$ one at a time, following the order the paths are discovered by the push-relabel algorithm.
  For simplicity we regard $\Bf$ as a flow in the subdivided graph $G^{\kappa}$, and hence that each flow path sends exactly a single unit of flow.
  We are going to maintain the subgraph $H$ of $G^{\kappa}_{\Bf}$ and that $\Bnu$ is $\frac{\sigma}{8}$-expanding in $H\setminus P$ for a pruned set $P$ throughout the reversals via \cref{lemma:pruning-reversal}. In fact, we will sometimes need to add back vertices from $P$ into $H$, and when we do so we will add them as fresh/forked new vertices. Therefore technically speaking $H$ will not necessarily be a subgraph of $G^{\kappa}_{\Bf}$, since some vertices might occur more than once in $H$. However, all but one copy of each vertex will be in $P$, so we still always maintain that $H\setminus P$ is subgraph of $G^{\kappa}_{\Bf}$.%

  \paragraph{Low-Diameter Invariant.}
  After each path reversal, we will increase some of the vertex weights via \cref{lemma:pruning-reversal} to ensure
  that $\Bnu(v) \geq \sum_{e \in E(H \setminus P) \cap \delta(v)}\Bw_G(e)$ holds for all $v \in V(H) \setminus P$. Note that by construction of $H$ and $\Bnu$, this holds initially.
  This together with the fact that $\Bnu$ is $\frac{\sigma}{8}$-expanding in $H \setminus P$ shows that the subgraph $G[V(H)\setminus P]$ has $\Bw_G$-diameter $O(\log{n}/\sigma)$ by \cref{lemma:expander-low-diameter}.

  \paragraph{Dealing with Path Reversal.}
  Suppose we have dealt with the first $j - 1$ flow paths already, and now we are preparing to reverse the $j$-th flow path $R_j$ in $G^{\kappa}_{\Bf_{j-1}}$.
  If $R_j$ does not intersect with $V(H)$, then nothing needs to be done.
  Otherwise, we take the first point $s_j$ and the last point $t_j$ on $R_j$ that intersect $V(H)\setminus P$ and replace $R_j$ with the subpath between them.
  We now ensure to add all vertices from $R_j$ to $H$, which are not already in $H\setminus P$ via \cref{lemma:pruning-reversal}. Importantly, when we add an already pruned vertex $v\in P$, we use a fresh instance of this vertex (so that the newly added $v'$ will be in $V(H)$ but not in $P$, see also \cref{fig:reverse}).
  Indeed \cref{lemma:pruning-reversal} only allows reversing paths that do not intersect $P$.
  Since $s_j, t_j \in V(H) \setminus P$, by the low-diameter invariant above we know that
  \[ \dist_{G^{\kappa}_{\Bf_{j-1}}}^{\Bw_G}(s_j, t_j) = O(\log n/\sigma)\]
  and thus $\Bw_G(R_j) = O(\log{n}/\sigma)$ as well by \cref{lem:pr-no-shortcut} (recall $\Bf$ is obtained by running our push-relabel algorithm of \cref{thm:push-relabel-main-theorem} on $G^{\kappa}$, and the push relabel algorithm will find almost shortest paths).
  We can now for each $e \in R_j$ add $e$ to $H$. %
  After doing so, we go back and for each of the new edge $(u, v)$ added to $H$, we increase $\Bnu(u)$ and $\Bnu(v)$ by $\Bw_G(e)$ to ensure that $\Bnu(v) \geq \sum_{e \in E(H \setminus P) \cap \delta(v)}\Bw_G(e)$ holds for all $v \in V(H) \setminus P$.
  The reason why we first add all edges and then do the vertex weight increments is to make sure the pruned set does not grow while adding edges.
  Note that these edge additions are valid as none of them are incident to $P$ due to us using fresh vertices.
  By $\Bw_G(R_j) = O(\log{n}/\sigma)$, we also conclude that the total amount we just added to the vertex weights is $O(\log{n}/\sigma)$.
  
  Finally, we reverse $R_j$ in $H$, via \cref{lemma:pruning-reversal}.
  The pruned set $P$ may grow according to \cref{lemma:pruning-reversal} after each vertex volume increment and path reversal.
  One can verify that all the invariants are maintained. 

  \paragraph{Summary.}
  In total, we have $|\Bf|$ path reversals, and each also increased $\|\Bnu\|_1$ by $O(\log n/\sigma)$.
  At the end, by \cref{lemma:pruning-reversal}, the total volume of the pruned set $P$ is bounded by:
  \begin{align*}
  \Bnu(P) &=
  O\left(|\Bf|\left(\frac{1}{\sigma} + \frac{\log n}{\sigma}\right) \right)
  \\&= O\left(|\Bf|\cdot\kappa_W D\log^{3} n \right)
  \\&= O\left(|\Bf|\cdot D\cdot \frac{\eta \log^{6} n}{\phi}\right)
  \end{align*}
  Since we have $\Bnu \ge \kappa \cdot D \cdot \deg_X$ initially (and this never changes as $\Bnu$ only grows), $\vol_X(P) \le O(\frac{|\Bf|\eta \log^6 n}{\kappa \phi})$, and hence $P$ can be incident to at most $O(\frac{|\Bf|\eta \log^{6} n}{\kappa \phi})$ many edges from $X$.
  The rest of $X$, by the low-diameter invariant, are reachable from each other in $H$ (and hence in $G^{\kappa}_{\Bf}$) by a path of $\Bw_G$-length $O(\log{n}/\sigma) = O\left(\frac{|C|\eta^3 \log^7 n}{\phi^{2}}\right)$.
  This completes the proof of \cref{lemma:low-diameter-expander-new}.
\end{proof}

\section{Building an Expander Hierarchy}\label{sec:nested-expander-decomposition}

\newcommand{\prev}{\mathrm{prev}}

In this section, we show how to construct an expander hierarchy of the input graph that was used earlier in this paper for deriving the weight function needed by our push-relabel algorithm in \cref{sec:weight}.

\NestedExpanderHierarchyCorollary

In particular, we show the following \cref{thm:nested-expander-hierarchy}, from which \cref{cor:nested-expander-hierarchy} immediately follows if we choose, e.g., $\phi = \exp\left(-\frac{\log n}{(\log \log n)^{1/3}}\right)$.

\begin{restatable}{theorem}{NestedExpanderHierarchy}
  Given an $n$-vertex simple capacitated graph $(G, \Bc)$ and a parameter $0 < \phi < 2^{-\omega\left(\frac{\log n}{\sqrt{\log \log n}}\right)}$ sufficiently small, there is a randomized $\frac{n^{2+o(1)}}{\phi^3}$ time algorithm that with high probability constructs a $\phi/n^{o(1)}$-expander hierarchy $\cH = (D, X_1, \ldots, X_{\eta})$ of $(G, \Bc)$ with $\eta = O(\log n)$.
  \label{thm:nested-expander-hierarchy}
\end{restatable}

In fact, our \cref{cor:nested-expander-hierarchy,thm:nested-expander-hierarchy} achieve an additional property that each $X_i$ is a separator of $G \setminus X_{>i}$, where $X$ is a \emph{separator} of $H$ if none of the edges in $F$ has both endpoints in the same strongly connected component of $H$.
This is because we will construct each $X_i$ by repeatedly finding cuts in $G$ and removing edges from one of the directions (i.e., $E_G(S,\overline{S})$ or $E_G(\overline{S}, S)$ for some $S$) which disconnects the two sides of the cuts.

\subsection{Overview and Setup} \label{subsec:overview}

We first give a high-level overview of the algorithm, where for simplicity we assume the graph is unit-capacitated, since both the analysis and the algorithm itself extend seamlessly to the capacitated setting.
Note that the first level of the expander hierarchy is easy to compute via standard expander decomposition techniques.
Recall that $G_i \defeq G \setminus X_{>i}$ for any $i$.
That is, we can get three edge sets $D, X_1, X_2$ such that $D$ is a DAG, $X_1$ is $\phi$-expanding in $G_1$, and $|X_2|$ is small (on the order of $\phi m$).
To construct the second level and onward, one immediate idea is to simply do expander decomposition with respect to the \emph{volume induced by $X_2$}, which is in fact doable by incorporating our sparse-cut algorithm of \cref{thm:flow} into the framework of e.g.,~\cite{NanongkaiSW17,BernsteinGS20,HuaKGW23}.
If the returned edge set $X_3$ happens to be a subset of $X_2$, then we can set $X_2 \gets X_2 \setminus X_3$ and continue to run expander decomposition on $X_3$.
As the number of edges in the terminal set decreases roughly by a factor of $\phi$ each time, after $O(\log_{1/\phi} n)$ iterations we will get the desired expander hierarchy.

The issue is that the edge set $X_3$ we need to cut when doing expander decomposition with respect to the volume induced by $X_2$ may not be a subset of $X_2$ and hence, it ``cuts through'' the strong components of $G_1$. Indeed, it might necessarily be the case that $X_3$ includes edges from $X_1$ or even from $D$;
in general given any $F \subseteq E$ there might not be a separator contained in $F$ that makes $F$ expanding in the remaining graph.
Having $X_3 \subsetneq X_2$ would result in a \emph{non-nested} expander hierarchy which is incompatible with our sparse-cut algorithm once we recurse on both sides of the cut.

To further understand why this breaks the previous layers, notice that as $X_3 \subsetneq X_2$ the graph $G_1 = G \setminus X_{>1}$ in which $X_1$ is expanding changes.
This would potentially decrease the well-connectivity of $G_1$ and make $X_1$ no longer expanding.
To overcome this, we apply a seemingly na\"ive approach:
Whenever we find $X_3$, we immediately remove $X_3 \setminus X_2$ from $G_1$ (note that it suffices to remove $X_3 \setminus X_2$).
We then try to further refine the strongly connected components of $G_1$ into smaller pieces so that $X_1$ is still expanding in this graph.
As a result, some edges that were previously in $X_1$ got removed from $G_1$, and to accommodate them we further add these edges
into $X_2$.
Since the volume induced by $X_2$ increases, it may no longer be expanding in $G_2=G \setminus X_3$, and to fix it we similarly refine $G_2$ by putting more edges into $X_3$, which in turn may result in us moving even more volume to $X_2$, and so on.
While this creates a loop between the two steps that seemingly takes $\Omega(n)$ rounds, we show that with careful analysis and algorithmic implementation, this number can actually be bounded.

\subsubsection{Intuition of Analysis}

\paragraph{Ideal Scenario.}
To see why the number of iterations can be bounded, let us first consider the ideal case that (1) if we remove $D$ edges from $G_1$, then we can find a set $D^\prime$ of $O(D)$ edges in $G_1$ to be further removed so that $X_1$ remains expanding in it, and (2) if we add $A$ new edges to $X_2$, then we can find a set $A^\prime$ of $O(\phi A)$ edges to be removed from $G_2$ (hence added to $X_3$) so that $X_2$ remains expanding in $G_2$.
In this case, we can easily see that the number of edges to be removed from $G_1$ and the number of edges to be added to $X_2$ in fact decrease by an $O(\phi) < 1/2$ factor each round, which means that the number of rounds is bounded by $O(\log n)$.
If we further consider the third level $X_3$, then we can see that there are $O(\log n)$ rounds of interaction between $X_2$ and $X_3$, each of which generates another $O(\log n)$ rounds of interaction between $X_1$ and $X_2$.
Therefore, we can bound the total number of iterations of the algorithm by $O(\log n)^{\eta}$ where $\eta$ is the height of the final hierarchy we construct.
To this end, notice that the number of terminal edges is reduced by roughly a factor of $O(\phi)$ in each level, and thus we can bound $\eta$ by roughly $O(\log_{1/\phi} m)$.
Choosing $\phi < 1/n^{o(1)}$ sufficiently small (for instance, $\phi = 2^{-\sqrt{\log n}}$), this shows that the algorithm will terminate in $O(\log n)^{O(\log_{1/\phi} m)} = n^{o(1)}$ iterations, which is what we are aiming for.

The question thus now becomes: Is this ideal scenario achievable?
Existentially, by standard expander arguments, such edge sets always exist.
Thus, we may hope to generalize and apply existing expander pruning algorithms (e.g.,~\cite{SaranurakW19,BernsteinGS20,HuaKGW23,SulserP24}) to locate them by replacing the flow algorithm used by these frameworks with the sparse-cut subroutine we developed in \cref{sec:sparse-cut}.
Note that it is NP-hard to compute expander decomposition/pruning exactly, but as in most standard approaches we can afford some multiplicative approximation as long as the number of edges still goes down by half each round.

\paragraph{Fixing Hierarchy with Few Cuts.}
However, none of these algorithms locate the entire edge sets in one shot.
Instead, they work by repeatedly finding sparse cuts in the graph and recurse on both sides $U_1$ and $U_2$ of the cut.
But notice that our \cref{thm:flow} when running on subgraph $G[U]$ requires an expander hierarchy of $G[U] \setminus X_i$ (when we are at level $i$ trying to build expander decomposition with respect to $X_i$).
Although at the beginning of the algorithm, we have obtained from the previous layers a hierarchy $\cH_{\prev}$ of $G\setminus X_{i}$, if the first sparse cut we found is not contained in $X_{i}$, then we can no longer extract a valid hierarchy of $G[U] \setminus X_i$ from $\cH_{\prev}$.
In this case, we need to first go down to the pervious layers and fix their hierarchy before coming back to level $i$ and continue locating sparse cuts.
This invalidates our previous ideal analysis.

Fortunately for us, some of these previous algorithms (specifically~\cite{BernsteinGS20,HuaKGW23}) follow the framework established by \cite{NanongkaiS17,Wulff-Nilsen17} which allows one to argue that we can locate all these edges in $n^{\eps}$ calls for some $\eps = o(1)$ to the sparse-cut subroutines in total.\footnote{This is not technically accurate as we do still need to recurse on the smaller side of the cuts, but in this case we get a size reduction and all the recursions are vertex-disjoint.}
As a result, the number of times we need to go back to the previous level is bounded by roughly $n^{\eps}$ (the $O(\log n)$ factor induced by the reduction of edges is overwhelmed by this term).
Choosing $\phi$ to be even smaller (yet still $1/n^{o(1)}$), we can ensure that the total of calls to the sparse-cut subroutines throughout the whole construction is $n^{\eps \log_{1/\phi} n} = n^{o(1)}$.

\paragraph{Amortized vs Expected Worst-Case Recourse.}
Another issue with applying previous approaches is that these algorithms only have \emph{amortized} recourse guarantees.
For example, in the first case where we remove $D$ edges from $G_1$, instead of always returning an edge set $D^\prime$ of $O(D)$ edges, the amortized guarantee only ensures that if we remove $k$ batches $D_1, \ldots, D_k$ edges from $G_1$, then the algorithm returns $D_1^\prime, \ldots, D_k^\prime$ such that $|D_1^\prime| + \cdots + |D_i^\prime| = O(D_1 + \cdots + D_i)$ for every $i \in [k]$.
Unfortunately, amortized guarantees would break the above analysis of having the size of the edge sets reduced by half each iteration, as we might have a single large update early, and then a large number of small updates later with no decrease in size.

To overcome this, we observe that while worst-case output recourse might be algorithmically challenging to achieve in these algorithms, our analysis works if we can obtain a weaker \emph{expected} worst-case recourse.
Indeed, consider again the ideal scenario except that the size of $D^\prime$ and $A^\prime$ is only $O(D)$ and $O(\phi A)$ \emph{in expectation} respectively.
By the law of total expectation, we can argue that the expected number of edges needed to be fixed still decreases by half each iteration.
Even though we can no longer conclude that the number of iterations is bounded by exactly $\log m$, notice that after $100\log m$ rounds the expected number of edges needed to be fixed drops to at most $m^{-99}$.
By Markov's inequality, this still shows that with high probability the interaction between the two levels is bounded by $100\log m$.
In \cref{subsubsec:overview-HKGW} we describe how we modify previous algorithms to achieve the expected worst-case output recourse.

We remark that this is the only place in our analysis that requires randomness and the only reason why our final algorithm is not deterministic---indeed, the randomized cut-matching game of \cite{Louis10} can be easily replaced with a deterministic one~\cite{BernsteinGS20}.

\subsubsection{A Data Structure Point-of-View.}
To formally capture the interaction between the current level and the previous levels, we employ a data structure perspective and define the following.
In the remainder of the section, let $m \leq n^4$ be the total capacities of the edges in $G$.

\begin{definition}
  A \emph{$(k,\alpha,\beta,\phi, T)$-hierarchy maintainer} $\cM$ is a randomized data structure that maintains a subgraph $G_{\cM} \subseteq G$ of a capacitated graph $(G, \Bc_G)$ such that after each of the following operations it provides a $\phi$-expander hierarchy $\cH_{\cM}$ of $(G_{\cM}, \Bc_G)$ with height $\eta(\cH) \le k$.
  \begin{itemize}
    \item \alg{Init}{$G, \Bc_G$}: Given an $n$-vertex simple capacitated graph $(G, \Bc_G)$, the data structure in expected $T(n)$ time computes a separator $X \subseteq E$ of expected capacity $\expect[\Bc_G(X)] \leq \alpha m$ and initializes $G_{\cM} \gets G \setminus X$.
      The output of the subroutine is $X$.
    \item \alg{Cut}{$D$}: Let $\{U_1, \ldots, U_k\}$ be the SCCs of $G_{\cM}$.
      The input is a separator $D$ of $G_{\cM}$ such that for each $U_i$ either $D \cap G_{\cM}[U_i] = \emptyset$ or $D \cap G_{\cM}[U_i] = E_{G_{\cM}[U_i]}(S_i, \overline{S_i})$ for some $S_i \subseteq U_i$.
      The adversary removes $D$ from $G_{\cM}$, i.e., it sets $G_{\cM} \gets G_{\cM} \setminus D$.
      Let $U_D$ be the union of SCCs that intersect with $D$.
      In response, the data structure in expected $T(|U_D|)$ time computes a separator $A \subseteq G_{\cM}[U_D]$ of the new $G_{\cM}$ of expected capacity $\expect[\Bc_G(A)] \leq \beta \Bc_G(D)$ and update $G_{\cM} \gets G_{\cM} \setminus A$.
      The output of the subroutine is $A$.
  \end{itemize}
  \label{def:hierarchy-maintainer}
\end{definition}

Our main technical result in this section is the following lemma which says that we can design a $(k+1,\alpha n^{o(1)} \phi,\cdot,\cdot,\cdot)$-hierarchy maintainer
using a  $(k,\alpha,\cdot,\cdot,\cdot)$-hierarchy maintainer, thus reducing the number of separator edges by a factor $n^{o(1)}\phi \ll \phi^{\Omega(1)}$ for $\phi$ sufficiently small. The blow-up in the other parameters are carefully set to be manageable.

\begin{restatable}{lemma}{Boosting}\label{lemma:hierarchy-maintainer}
  Given a $(k, \alpha_{\prev}, \beta_{\prev}, \phi_{\prev}, T_{\prev})$-hierarchy maintainer $\cM_{\prev}$ for an $n$-vertex simple capacitated graph $(G, \Bc)$, for any $L \in \N$ there exists some $\delta_L \leq (\log n)^{L^{O(L)}}$ such that for any $\phi < O\left(\frac{1}{\delta_L L \beta_{\prev} n^{O(1/L)}}\right)$ sufficiently small we can construct a $(k + 1, \alpha, \beta, \phi^{\prime}, T)$-hierarchy maintainer $\cM$ with
  \begin{equation}
  \begin{split}
    \alpha &\leq \phi \cdot \delta_L n^{O(1/L)} \cdot \alpha_{\prev}, \\
    \beta &\leq \delta_L n^{O(1/L)}, \\
    \phi^{\prime} &\geq \min\left\{\phi_{\prev}, \frac{\phi}{\delta_L}\right\},
    \\
    T(n) &= \delta_L n^{O(1/L)} \cdot \widetilde{O}\left(\frac{n^2}{\phi\phi_{\prev}^2} + T_{\prev}(n)\right).
  \end{split}
  \label{eq:parameters}
  \end{equation}
\end{restatable}

Note that the value of $\beta_{\prev}$ only affects the value of $\phi$ we can choose but not the new $\beta$ for $\cM$.
We first show that \cref{lemma:hierarchy-maintainer} in fact already implies \cref{thm:nested-expander-hierarchy}, restated below.

\NestedExpanderHierarchy*

\begin{proof}[Proof of \cref{thm:nested-expander-hierarchy}]
  We choose $L = \Theta(\sqrt{\log \log n})$ for which
  \[
    \delta_L \leq (\log n)^{L^{O(L)}} \leq 2^{\log \log n \cdot \Theta(\sqrt{\log \log n})^{\Theta(\sqrt{\log \log n})}} \leq 2^{(\log \log n)^{\Theta(\sqrt{\log \log n})}} = n^{o(1)}.
  \]
  Observe that there is a trivial $(0, 1, 0, 1, O(n^2))$-hierarchy maintainer $\cM_0$ which on \alg{Init}{} simply returns every edge.
  Since $\phi < o\left(2^{-\frac{\log n}{\sqrt{\log \log n}}}\right)$ is sufficiently small, we have $\phi \leq \left(\frac{1}{\delta_L n^{2/L}}\right)^2$ and $\phi < O\left(\frac{1}{\delta_L L \beta n^{O(1/L)}}\right)$ for $\beta = \delta n^{O(1/L)} \leq n^{o(1)}$.
  Therefore, starting from $\cM_0$, for each $k > 0$ we can apply \cref{lemma:hierarchy-maintainer} on $\cM_{k-1}$ to get a $\left(k, \sqrt{\phi}^{k}, n^{o(1)}, \phi/n^{o(1)}, T_{k}\right)$-hierarchy maintainer $\cM_k$, where $T_{k}(n)\defeq\delta_L^{O(k)} \cdot n^{O(k/L)} \cdot \widetilde{O}(n^2/\phi^3)$. %
  As such, for $\eta = 2\log_{1/\phi} (4n^4) < L$, by calling $\cM_{\eta}.\alg{Init}{G}$ we get a $\phi/n^{o(1)}$-expander hierarchy of $G \setminus X$ of height $\eta \leq O(\log n)$ for some edge set $X$ with expected size $\expect[\Bc_G(X)] \leq \Bc_G(E) \cdot (\sqrt{\phi})^{\eta} \leq 1/4$.
  Thus, by Markov's inequality, with probability at least $3/4$ the set $X$ is empty, meaning that the hierarchy we got is indeed a $\phi/n^{o(1)}$-expander hierarchy of $G$.
  The expected running time of the algorithm is
  \[ \log n^{L^O(L)} \cdot n^{O(\log_{1/\phi} n)/L} \cdot \widetilde{O}(n^2/\phi^3) = \widehat{O}(n^2/\phi^3) \]
  time since $O(\log_{1/\phi} n) / L \leq o(\sqrt{\log \log n})/L = o(1)$. %
  Repeating this $O(\log n)$ time we succeed in worst-case time with high probability.
  This proves the theorem.
\end{proof}

We now briefly sketch on how we prove \cref{lemma:hierarchy-maintainer}. To boost the quality of the maintainer $\cM_{\prev}$, let $F$ be the result of running $\cM_{\prev}.\alg{Init}{G}$.
Our algorithm essentially takes the $k$-level hierarchy maintained by $\cM_{\prev}$ and constructs the $(k+1)$-th level of it in order to reduce the number of non-expanding edges by roughly a factor of $\phi$.
Thus, we start with the terminal set $F$ and run expander decomposition in $G$ with respect to $F$.
In other words, the goal is to compute a separator $X$ such that $F$ is $\phi$-expanding in $G \setminus X$.
If we then have an expander hierarchy $\cH = (D, X_1, \ldots, X_k)$ of $G \setminus (F \cup X)$ then we can set $X_{k+1} = F$ and obtain an expander hierarchy of $G \setminus X$.
Needless to say, we will use $\cM_{\prev}$ to maintain such $\cH$, and thus throughout the algorithm we need to ensure $\cH$ is a hierarchy of $G \setminus (X \cup F)$ by properly calling $\cM_{\prev}.\alg{Cut}{}$.

To compute an expander decomposition with respect to $F$, we start with an empty $X$ and let $\cU$ be the SCCs of $G_{\cM} = G \setminus X$ which is initially $\cU = \{V\}$.\footnote{Here we assume the graph is initially strongly connected.}
For each $U \in \cU$, we attempt to locate a sparse cut in $G[U]$ via the sparse-cut algorithm we developed in \cref{sec:sparse-cut}.
Note that as $X$ is a separator of $G$ and $U$ is strongly connected, $G[U]$ is the same as $G_{\cM}[U]$.
If we find a sparse cut $D$, then we include $D$ into $X$ which effectively splits $U$ into (at least) two SCCs on which we recurse our construction.
In addition, to ensure that $\cM_{\prev}$ holds a hierarchy of $G \setminus (F \cup X)$, we need to call $\cM_{\prev}.\alg{Cut}{D}$ to remove $D$ from $G_{\cM_{\prev}}$.
Observe that since $D$ is a cut in $G[U]$ and the SCCs of $G_{\cM_{\prev}}$ form a refinement of $\cU$, the input requirement of $G_{\cM_{\prev}}.\alg{Cut}{D}$ is satisfied.
After $\cM_{\prev}$ further refines its SCCs and outputs an $A \subseteq G_{\cM_{\prev}}[U]$, meaning that now it maintains a hierarchy of $G \setminus (F \cup A \cup X)$, we add $A$ into $F$ to preserve our invariant.
Note that doing all these also ensures that we have an expander hierarchy of $G[U] \setminus F$ at all times, which is required by our sparse-cut algorithm.
Indeed, if we take the SCCs of $G_{\cM_{\prev}}$ contained in $U$ and restrict $\cH$ to these SCCs, then we get an expander hierarchy of $G[U]$.

\begin{fact}
  For a $\phi$-expander hierarchy $\cH = (D, X_1, \ldots, X_{\eta})$ of $(G,\Bc_G)$ and $U \subseteq V$ such that for each $W \in \SCC(G)$ either $W \subseteq U$ or $W \cap U = \emptyset$, the sequence $\cH[U] \defeq (D \cap G[U], X_1 \cap G[U], \ldots, X_{\eta} \cap G[U])$ is a $\phi$-expander hierarchy of $(G[U],\Bc_G)$.
\end{fact}

Having the overall picture, it now remains to implement the steps efficiently and achieve the desired expected guarantee.
For this we adapt and generalize the framework of \cite{HuaKGW23} which maintains expander decomposition by repeatedly finding sparse cuts and thus fits our purposes well.\footnote{It is worth mentioning that there is a recent work of \cite{SulserP24} which improves the almost-linear running time of \cite{HuaKGW23} to near-linear one. They sidestepped the multi-level approach of \cite{NanongkaiSW17,BernsteinGS20,HuaKGW23} by a novel \emph{push-pull-relabel} flow algorithm that allowed them to implement the trimming strategy of \cite{SaranurakW19}. We leave adapting their framework or our use case as an interesting open direction.}
We give an overview of their framework below.

\subsubsection[Overview of {[HKPW23]}]{Overview of \texorpdfstring{\cite{HuaKGW23}}{[HKPW23]}}\label{subsubsec:overview-HKGW}

To certify expansion and locate sparse cuts, \cite{HuaKGW23} employed a celebrated approach of embedding a witness into each $G[U]$.
Informally speaking, a witness is an $\widetilde{\Omega}(1)$-expander with approximately the same degree profile as that of $F$ and is embeddable into $G[U]$ with congestion $\widetilde{O}(1/\phi)$.
It may be hard to construct such a witness directly in few calls to the sparse-cut algorithm since we may repeatedly find rather unbalanced sparse cuts which makes the number of iterations $\Omega(n)$.
To overcome this, instead of a single witness for $G[U]$, \cite{HuaKGW23} used a series of witnesses that contain additional \emph{fake} edges that are counted toward the expansion guarantee of the witness yet are not embeddable into $G[U]$.

\paragraph{Witness with Fake Edges.}
While a witness with fake edges does not immediately certify the expansion of the graph, it provides fairly useful information that the graph does not contain a \emph{balanced} sparse cut.
Let us start with a threshold $R$ and attempt to construct a witness with $R$ fake edges.
In each attempt, we either find a balanced sparse cut of size $\widetilde{\Omega}(R)$ which we then recurse on both sides with a decent size reduction or we certify that no such balanced sparse cut exists in the graph.
In the next iteration, we decrease the parameter $R$ to be $R^\prime$ and then repeat the above loop until we certify there is no $\widetilde{\Omega}(R^\prime)$-balanced sparse cut either.
Crucially, even though as $R^\prime$ decreases we might not get a large size reduction when recursing on both sides of the cut anymore, by setting up the expansion parameter properly in each level, we can in fact guarantee that if we found much more than $\widetilde{\Omega}(R/R^\prime)$ sparse cuts in the graph, then there would have been a $\widetilde{\Omega}(R)$-balanced sparse cut in the graph that we start with which contradicts with our $R$-witness constructed.
By setting the $R$ value of a level-$\ell$ witness to be $n^{\ell/L}$ for some $L = \omega(1)$, the algorithm only needs to run the sparse-cut algorithm $n^{1/L} = n^{o(1)}$ time when computing and maintaining expander decomposition.

Note that there are two notions of levels in our algorithm:
One is the level of the expander hierarchy, and we maintain each level of hierarchy with $L$ levels of witnesses.
To avoid confusion, in the remainder of this overview the term level means the level of witnesses except when we deliberately use the term \emph{hierarchy level}.

\paragraph{Reparing Witnesses.}
This idea of using fake edges with decreasing expansion and balance parameters was initiated in \cite{NanongkaiS17,Wulff-Nilsen17} and has been later used either explicitly or implicitly in many other expander decomposition/pruning algorithms~\cite{NanongkaiSW17,BernsteinGS20}.
What \cite{HuaKGW23} differs from previous work is the explicit usage of witnesses and the way they set up and repair them which allows for maintaining expander decomposition under updates by directly finding sparse cuts.
To be more specific, consider a strongly connected component $U$ for which we want to maintain its expansion.
Their algorithm maintains $L$ witnesses $W_0,\ldots,W_{L}$ with $R_0, \ldots, R_L$ fake edges, where $R_\ell$ is supposed to be roughly $|U|^{\ell/L}$.
For each update to the graph, the algorithm first checks for each witness whether there are edges in it that are embedded into an updated part.
If so, these real edges are replaced by fake edges.
This increases the number of fake edges in the witnesses which might break the invariant of $R_{\ell} \approx |U|^{\ell/L}$ (and in particular, if $R_0 \gg 0$, then we have failed to certify the expansion of $G[U]$).
Their algorithm thus attempts to repair such an invalid witness from the higher-level witness $W_{\ell+1}$.
This is done by setting up a flow problem which corresponds to embedding a sufficiently large number of fake edges in $W_{\ell+1}$ into $G[U]$ so that it becomes valid for level $\ell$, and if it fails to do so, the algorithm finds a sparse cut.
A careful charging argument is then used in \cite{HuaKGW23} to bound the total update time throughout a sequence of updates.

\paragraph{Challenges in Adaptation.}
To adapt their framework to our use case, we replace the standard push-relabel/blocking flow algorithm used in \cite{HuaKGW23} with the weighted push-relabel algorithm we developed.
An immediate challenge for this is that our flow algorithm is not \emph{local} in the sense that it cannot be used to only explore a small neighborhood around the sources.
This is in contrast to the classic unweighted push-relabel algorithm which can explore a region with $k$ edges in $O(k)$ time.
While the idea of having witnesses with decreasing number of fake edges in principle, as we have touched upon above, allows one to at least intuitively argue that the number of times we need to call sparse-cut algorithm is small, \cite{HuaKGW23} used a more direct potential-based analysis which heavily relies on their local running time.
Thus, we need to apply a different analysis strategy than theirs and prove additional stability properties of their witnesses (see \cref{subsec:stability}) which then allow us to incorporate the conceptual guarantee of the high-level fake-edges framework into the specific ways \cite{HuaKGW23} maintained their witnesses.

Another modification we made to \cite{HuaKGW23} is the rebuilding strategy.
Previously, \cite{HuaKGW23} attempts to repair a witness whenever it contains too many fake edges.
However, this only gives an amortized output recourse which is insufficient for our analysis when interacting with previous hierarchy levels.
To overcome this we use a fairly standard strategy: Instead of fixing every witness whenever possible, we set a larger grace period for them and consider a witness valid as long as the number of fake edges it contains falls into its corresponding grace period.
We then for each update sample a \emph{random} witness level to be rebuilt.
By appropriately setting the sampling probability, we can ensure that (1) each witness will be rebuilt with high probability before it becomes invalid (this is due to the larger grace period we set) and (2) we achieve a worst-case output recourse \emph{in expectation}.
While this random rebuilding approach is commonly used, due to the interaction with previous hierarchy levels, during a repair of a witness we might need to abort the current repair, re-sample a higher witness level, and start from there instead.
This complication makes the analysis of our expected guarantee fairly cumbersome (see \cref{subsec:recourse-and-runtime}).
We remark that the stability properties we mentioned earlier also play a role in achieving the expected worst-case guarantee.

\paragraph{Organization.}
In the remainder of the section we present our modification to \cite{HuaKGW23} with new constructs and analyses which ultimately lead to a proof of \cref{lemma:hierarchy-maintainer}.
In particular, in \cref{subsec:repair-witness} we apply our flow algorithm of \cref{thm:flow} to construct and repair witnesses.
In \cref{subsec:stability} we establish certain stability properties of the witnesses that are key to our analyses. 
In \cref{subsec:maintain-exp-decomp} we present the main algorithm of maintaining expander decomposition while interacting with the previous layers via $\cM_{\prev}$.
In \cref{subsec:recourse-and-runtime} we prove that the described algorithm has the desired expected guarantee.
Finally, in \cref{subsec:everything} we put everything together and arrive at a proof of \cref{lemma:hierarchy-maintainer} (and hence \cref{cor:nested-expander-hierarchy}).

\subsection{Constructing and Repairing Witnesses} \label{subsec:repair-witness}

Let $(G, \Bc_G)$ be the input capacitated graph for which we want to construct an expander hierarchy.
Throughout the section, we let $m \leq n^4$ denote the sum of capacities of edges in $G$.
We start by giving a formal definition of a witness adapted from \cite[Definition 2.1]{HuaKGW23} and generalized straightforwardly to the capacitated setting.

\begin{definition}[$R$-Witness]
  For a capacitated simple graph $(W, \Bc)$ with $V(W) = V(G)$, a vector $\Br \in \N^V$, and an embedding $\Pi_{W \to G}$ from $W$ to $G$, the tuple $(W, \Bc, \Br, \Pi_{W \to G})$ is an \emph{$(R, \phi, \psi)$-out-witness} of $(G, \Bc_G, F)$ for $F \subseteq V \times V$ \emph{with respect to $\Bgamma \in \N^V$} if
  \begin{enumerate}[(1)]
    \item\label{item:witness:num-fake-edges} $\|\Br\|_1 \leq R$,
    \item\label{item:witness:degree} $\deg_{F,\Bc_G}(v) \leq \deg_{W,\Bc}(v) + \Br(v) \leq \frac{1}{\psi}\deg_{F,\Bc_G}(v)$ holds for all $v \in V$,
    \item\label{item:witness:expansion} for every cut $(S, \overline{S})$ with $\Bgamma(S) \leq \Bgamma(\overline{S})$ we have $\Bc(E_W(S, \overline{S})) + \Br(S) \geq \psi(\vol_{W,\Bc_W}(S) + \Br(S))$, and
    \item\label{item:witness:congestion} $\Pi_{W \to G}$ embeds $(W,\Bc)$ into $(G,\Bc_G)$ with congestion $\frac{1}{\phi\psi}$.
  \end{enumerate}
  If $(W, \Bc, \Br, \Pi_{W \to G})$ is an $(R, \phi, \psi)$-out-witness of $(G, \Bc_G, F)$ and $(\rev{W}, \Bc, \Br, \Pi_{\rev{W} \to G})$ is an $(R, \phi, \psi)$-out-witless of $(\rev{G}, \Bc_G, \rev{F})$, both with respect to $\Bgamma$, then $(W, \Bc, \Br, \Pi_{W \to G})$ is an \emph{$(R, \phi, \psi)$-witness} of $(G, \Bc_G, F)$ with respect to $\Bgamma$.
  \label{def:witness}
\end{definition}

Note that the $\Br$ vector corresponds to the concept of fake edges introduced in \cref{subsubsec:overview-HKGW}.\footnote{That is, each vertex $v$ has $\Br(v)$ incident fake edges. The $\Br$ vector is easier to maintain when the graph is updated while suffices for the purpose of \cite{HuaKGW23}. Indeed, when some vertices $S$ are removed from the graph, it is unclear where in the remaining graph one should add fake edges corresponding to those incident to $S$; instead, with the $\Br$ vector one can simply increase $\Br(v)$ for each remove edge incident to $v \in V \setminus S$.}
It is easy to see that if $R = 0$, then $F$ is $\Omega(\phi\psi^2)$-expanding in $G$, and we prove a more general version of this below which says this is even the case for $R$ sufficiently small \cref{claim:expanding-if-has-witness}.
Oftentimes $\Bc$, $\Br$, and $\Pi_{W \to G}$ will be clear from context, in which case we may simply refer to $W$ as the witness of $(G, \Bc_G, F)$.
In the remainder of the paper we will assume both $R$ and $\psi$ are reasonably bounded by some polynomials in $n$.
More specifically, we assume $R \leq n^{10}$ and $\psi \geq 1/n$.
Indeed, the choice of $\psi$ will be made explicit in \eqref{eq:psi}, and $R$ is going to be upper-bounded by the total volume of the graph which by capacity scaling is at most $n^{10}$.

The $\Br$ vector can be seen as some ``fake'' edges of the witness which do not exist in the real graph $G$.
While having a witness with many fake edges does not certify that $F$ is expanding in $G$, it does still shows that there is no \emph{balanced} sparse cut in $G$ with respect to $F$.
A cut $S$ in $(G, \Bc_G)$ is \emph{$B$-balanced} with respect to $F$ if $\min\{\vol_{F,\Bc_G}(S), \vol_{F,\Bc_G}(\overline{S})\} \geq B$.

\begin{claim}
  If there is a $(R, \phi, \psi)$-witness $(W, \Bc, \Br, \Pi_{W \to G})$ of $(G, \Bc_G, F)$ with respect to any $\Bgamma$, then there is no $\frac{2R}{\psi}$-balanced $\frac{\phi\psi^2}{2}$-sparse cut in $(G, \Bc)$ with respect to $F$.
  \label{claim:no-balanced-sparse-cut}
\end{claim}

\begin{proof}
  Consider any cut $S$ with $\frac{2R}{\psi} \leq \vol_{F,\Bc_G}(S) \leq \vol_{F,\Bc_G}(V \setminus S)$.
  If $\Bgamma(S) \leq \Bgamma(V \setminus S)$, then we have
  \[ \min\{\Bc(E_W(S, V\setminus S)), \Bc(E_W(V \setminus S))\} + \Br(S) \geq \psi(\vol_{W,\Bc}(S) + \Br(S)) \geq \psi\vol_{F,\Bc_G}(S) \geq 2R \]
  by \cref{def:witness}\labelcref{item:witness:degree,item:witness:expansion}.
  This implies $\min\{\Bc(E_W(S,V \setminus S)), \Bc(E_W(V \setminus S, S))\} \geq \frac{\psi}{2}\vol_{F,\Bc_G}(S)$, and by the fact that $(W, \Bc)$ embeds into $(G, \Bc_G)$ via $\Pi_{W \to G}$ with congestion $\frac{1}{\psi\phi}$ by \cref{def:witness}\labelcref{item:witness:congestion}, we have $\min\{\Bc_G(E_G(S,V \setminus S)),\Bc_G(E_G(V \setminus S, S))\} \geq \frac{\phi\psi^2}{2}\vol_{F,\Bc_G}(S)$.
  A symmetric argument applied to the case where $\Bgamma(S) > \Bgamma(V \setminus S)$ shows that $\min\{\Bc_G(E_G(S,V \setminus S)),\Bc_G(E_G(V \setminus S, S))\} \geq \frac{\phi\psi^2}{2}\vol_{F,\Bc_G}(V \setminus S) \geq\frac{\phi\psi^2}{2}\vol_{F,\Bc_G}(S)$ as well.
  Therefore, such an $S$ can not be a $\frac{\phi\psi^2}{2}$-sparse cut.
\end{proof}

\begin{claim}
  If a (not necessarily strongly connected) subgraph $G[U]$ has volume $\vol_{F,\Bc_G}(G[U]) < 1/\phi$, then $F$ is $\phi$-expanding in $(G[U],\Bc_G)$.
  \label{claim:expanding-in-tiny-component}
\end{claim}

\begin{proof}
  It suffices to consider the case when $G[U]$ is strongly connected by the definition of $\phi$-expanding.
  Since $\Bc_G(E_{G[U]}(S,U\setminus S)) \geq 1$ for all $S \subseteq U$ and $\vol_{F,\Bc_G}(S), \vol_{F,\Bc_G}(\overline{S}) < 1/\phi$, we have $F$ is $\phi$-expanding in $(G[U],\Bc_G)$.
\end{proof}

\begin{claim}
  If there is a $(R, \phi, \psi)$-witness for $(G, \Bc_G, F)$ with respect to any $\Bgamma$ where $R < 1/\phi$, then $F$ is $\frac{\phi\psi^2}{2}$-expanding in $(G,\Bc_G)$.
  \label{claim:expanding-if-has-witness}
\end{claim}

\begin{proof}
  The existence of a $(R, \phi, \psi)$-witness by \cref{claim:no-balanced-sparse-cut} implies that there is no $\frac{2R}{\psi}$-balanced $\frac{\phi\psi^2}{2}$-sparse cut in $(G,\Bc_G)$ with respect to $F$.
  This suggests that $G$ can only have one strongly connected component with volumes at least $\frac{2R}{\psi} \leq \frac{2}{\phi\psi}$, otherwise there is a cut with no edges that separate two such components which would lead to a contradiction.
  By \cref{claim:expanding-in-tiny-component}, $F$ is $\frac{\phi\psi}{2}$-expanding in those components with small volume.
  On the other hand, consider the only component $U$ in $G$ that has volume at least $\frac{2}{\phi\psi}$.
  Note that \cref{claim:no-balanced-sparse-cut} also suggests that there is no $\frac{2R}{\psi}$-balanced $\frac{\phi\psi^2}{2}$-sparse cut in $(G[U],\Bc_G)$ with respect to $F$.
  Indeed, the same cut would have been $\frac{2R}{\psi}$-balanced $\frac{\phi\psi^2}{2}$-sparse in $(G,\Bc_G)$ as well if it existed.
  However, if a cut in $(G[U],\Bc_G)$ has volume less than $\frac{2}{\phi\psi}$, then it can never be $\frac{\phi\psi^2}{2}$-sparse since $G[U]$ is strongly connected.
  This shows that $F$ is $\frac{\phi\psi^2}{2}$-expanding in $(G[U],\Bc_G)$ as well, hence in $(G,\Bc_G)$.
\end{proof}

\paragraph{Algorithms for Constructing $R$-Witnesses.}
We will have two primitives which tries to construct $R$-witnesses for some set of terminal edges $F$, both which are based on our sparse cut algorithm from \cref{sec:sparse-cut}.
\begin{itemize}
\item \textsc{CutOrEmbed} which either finds a $\widetilde{\Omega}(R)$-balanced sparse cut or an $R$-witness. This is done by running a standard cut-matching game.
\item \textsc{PruneOrRepair} which takes an $R$-witness as input, and either finds a $\widetilde{\Omega}(R^\prime)$-balanced sparse cut or an $R'$-witness for some $R'\le R$. 
This is done by attempting to embed the fake edges in the $R$-witness further into the graph.
\end{itemize}
Let $c_{\ref{thm:flow}} \in \N$ be a universal constant such that the cut \cref{thm:flow} returns satisfies
\[ \Bc(E_G(S,\overline{S})) \leq \frac{c_{\ref{thm:flow}} \cdot |\Bf| + \vol_{F,\Bc}(S)}{\kappa}. \]
Let $z \defeq 20\log n$ be fixed throughout the rest of the section.

\begin{lemma}[{Analogous to \cite[Lemma 3.1]{HuaKGW23}}]
  Given an $n$-vertex strongly connected simple capacitated graph $(G, \Bc_G)$, terminal edge set $F \subseteq E$, vectors $\Br, \Bgamma \in \Z_{\geq 0}^{V}$, an $(R, \phi, \psi)$-witness $(W, \Bc, \Br, \Pi_{W \to G})$ of $(G, \Bc_G, F)$ with respect to $\Bgamma$, a parameter $R^\prime \geq 0$ such that $R^\prime \leq R \leq \frac{\psi}{8z} \vol_{F, \Bc_G}(V)$, and a $\phi^\prime$-expander hierarchy $\cH$ of $(G \setminus F, \Bc_G)$ with height $O(\log n)$, there is an algorithm \alg{PruneOrRepair}{$G, F, W, \Bc, \Br,$ $\Pi_{W \to G}, \phi, \psi, R^\prime, \cH$} that either outputs
  \begin{enumerate}
    \item\label{item:repair:sparse-cut} a set $S \subseteq V$ with $\frac{\psi}{16z} \cdot R^\prime \leq \vol_{F,\Bc_G}(S) + \Br(S) \leq \frac{8z}{\psi} \cdot R$ such that $\Bc_G(E_G(S, \overline{S})) < \frac{\phi\psi^3}{256} \cdot (\vol_{F,\Bc_G}(S) + \Br(S))$ or
    \item\label{item:repair:repair} an $(R^\prime, \phi, \psi^\prime)$-out-witness $(W^\prime, \Bc^\prime, \Br^\prime, \Pi_{W^\prime \to G})$ of $(G, \Bc_G, F)$ with respect to $\Bgamma$, where $\psi^\prime \defeq \frac{\psi^4}{2048 c_{\ref{thm:flow}} z^2} = \Omega\left(\frac{\psi^4}{\log^2 n}\right)$.
  \end{enumerate}
  The algorithm runs in time $\widetilde{O}\left(\frac{n^2}{\phi{\phi^\prime}^2\psi^4}\right)$.\footnote{Note that as in \cite{HuaKGW23}, the algorithm does not need to take $\Bgamma$ as an input.}

  \label{lemma:prune-or-repair}
\end{lemma}

\begin{proof}
  We follow essentially the same proof strategy as \cite[Lemma 3.1]{HuaKGW23} but with parameters tailored to our needs.
  We set up a diffusion instance $\cI = (G, \Bc_G, \Bsource, \Bsink)$ with $\Bsource \defeq z \cdot \frac{8}{\psi} \cdot \Br$ and $\Bsink \defeq \deg_{F, \Bc_G} + \Br$.
  Note that $\cI$ is a diffusion instance as $R \leq \frac{\psi}{8z}\vol_{F,\Bc}(V)$.
  Also note that $z \geq \log \|\Bsource\|_1$.
  We are going to compute a flow $\Bf^{*}$ routing $(\Bsource, \Bsink)$ in $G$ with congestion $\kappa z$ for $\kappa \defeq \frac{1024 \cdot c_{\ref{thm:flow}} \cdot z}{\phi\psi^4}$ by invoking \cref{thm:flow} $z$ times.
  Let $\Bf^{*} \defeq \Bzero$.
  While $\|\Bsource\|_1 > R^\prime$, we run \alg{SparseCut}{$\cI, \kappa, F, \cH$} in \cref{thm:flow} to find a flow $\Bf$.
  If $\Bf$ routes half of the demand, i.e., $|\Bf| \geq \frac{1}{2}\|\Bsource\|_1$, then we update $\Bf^{*} \gets \Bf^{*} + \Bf$ and set $\Bsource \gets \Bsource_{\Bf}$ to be the residual supply while keeping the sink intact and then repeat.
  Otherwise, we have $\ex_{\Bf}(V) > \frac{1}{2}\|\Bsource\|_1 > \frac{1}{2}R^\prime$ and $|\Bf| < \frac{1}{2}\|\Bsource\|_1$.
  In this case we will find a sparse cut and output it in Case \ref{item:repair:sparse-cut}.
  Let $S$ be the cut outputted by \cref{thm:flow} on this invocation.
  We have $\vol_{F,\Bc_G}(S) + \Br(S) \geq \Br(S) \geq \frac{\psi}{8z} \ex_{\Bf}(S) \geq \frac{\psi}{16z} R^\prime$ and $\vol_{F,\Bc_G}(S) + \Br(S) = \Bsink(S) = \abs_{\Bf}(S) \leq \frac{8z}{\psi}R$.
  Also, \cref{thm:flow} asserts that
  \begin{align*}
    \Bc_G(E_G(S, \overline{S})) &\leq \frac{c_{\ref{thm:flow}} \cdot |\Bf| + \vol_{F,\Bc_G}(S)}{\kappa} < \frac{c_{\ref{thm:flow}} \cdot \frac{4z}{\psi}\Br(S) + \vol_{F,\Bc_G}(S)}{\kappa} \leq \frac{\phi\psi^3}{256} \cdot (\vol_{F,\Bc_G}(S) + \Br(S))
  \end{align*}
  where we used that $\vol_{F,\Bc_G}(S) + \Br(S) \geq \frac{\psi}{8z}\Bsource(S)$ and $|\Bf| < \frac{1}{2}\|\Delta\|_1 \leq \frac{4z}{\psi}\Br(S)$.
  This shows that $S$ indeed satisfies the output requirement of \cref{lemma:prune-or-repair}.

  If none of the calls to \cref{thm:flow} routes less than half of the demand, we end up with a flow $\Bf^{*}$ with $\ex_{\Bf^{*}}(V) \leq R^\prime$ and congestion $\kappa z$.
  In this case we can construct a new witness $(W^\prime, \Bc^\prime, \Br^\prime, \Pi_{W^\prime \to G})$ in Case \labelcref{item:repair:repair} as follows.
  Initialize $W^\prime$ as $W$, (and $\Pi_{W^\prime \to G}$ as $\Pi_{W \to G}$ consequently), $\Bc^\prime$ as $\Bc$, and $\Br^\prime$ as $\frac{8}{\psi} \Br$.
  By a standard flow decomposition argument, we can in $\widetilde{O}(n^2)$ time decompose $\Bf^{*}$ into at most $n^2$ flow paths $P_i$ that sends $c_i$ units of flow from $(u_i, v_i)$.
  For each such flow path $P_i$, we add an edge $(u_i, v_i)$ with capacity $\Bc^\prime(u_i, v_i) = c_i$ to $W^\prime$, merging parallel edges if exists.
  We then decrease $\Br^\prime(u_i)$ by $c_i$.
  We argue that $(W^\prime, \Bc^\prime, \Br^\prime, \Pi_{W^\prime \to G})$ forms an $(R^\prime, \phi, \psi^\prime)$-out-witness of $(G, \Bc_G, F)$ with respect to $\Bgamma$.

  \paragraph{Properties \labelcref{item:witness:num-fake-edges,item:witness:congestion}.}
  The fact that $\|\Br^\prime\|_1 \leq R^\prime$ is by definition.
  The new embedding $\Pi_{W^\prime \to G}$ has congestion at most $\frac{1}{\phi\psi} + z \cdot \frac{1024 \cdot c_{\ref{thm:flow}} \cdot z}{\phi\psi^4} \leq \frac{1}{\phi\psi^\prime}$.

  \paragraph{Property \labelcref{item:witness:degree}.}
  Observe that $\deg_{W^\prime,\Bc^\prime}(v) + \Br^\prime(v)$ is initialized to $\deg_{W,\Bc}(v) + \frac{8z}{\psi} \cdot \Br(v) \in \Big[\deg_{F,\Bc_G}(v),$ $ \frac{8z}{\psi^2} \cdot \deg_{F,\Bc_G}(v)\Big]$.
  Also note that each $v$ absorbs at most $z \cdot \Bsink(v) = z(\deg_{F,\Bc_G}(v) + \Br(v))$ units of demand.
  For each flow path $P_i$ from $u_i$ to $v_i$ with $c_i$ units of flow, $\deg_{W^\prime, \Bc^\prime}(u_i) + \Br^\prime(u_i)$ stays the same while $\deg_{W^\prime,\Bc^\prime}(v) + \Br^\prime(v)$ increases by $c_i$.
  As a result, we have
  \begin{equation}
    \deg_{W^\prime,\Bc^\prime}(v) + \Br^\prime(v) \leq \frac{8z}{\psi^2} \cdot \deg_{F,\Bc_G}(v) + z (\deg_{F,\Bc_G}(v) + \Br(v)) \stackrel{(*)}{\leq} \frac{10z}{\psi^2} \cdot \deg_{F,\Bc_G}(v) \leq \frac{1}{\psi^\prime} \deg_{F,\Bc_G}(v).
    \label{eq:degree}
  \end{equation}
  Note that as in \cite{HuaKGW23} we will use the stronger bound of $(*)$ later.

  \paragraph{Property \labelcref{item:witness:expansion}.}
  Consider a cut $S$ where $\Bgamma(S) \leq \Bgamma(\overline{S})$ for which we have $\Bc(E_{W}(S, \overline{S})) + \Br(S) \geq \psi(\vol_{W,\Bc}(S) + \Br(S))$.
  \begin{itemize}
    \item If $\Bc(E_W(S, \overline{S})) \geq \Br(S)$: We have $\vol_{W^\prime,\Bc^\prime}(S) + \Br^\prime(S) \leq \frac{10z}{\psi^2}\vol_{F,\Bc_G}(S) \leq \frac{10z}{\psi^2}(\vol_{W,\Bc}(S) + \Br(S))$ by \eqref{eq:degree}.
      This implies
      \begin{align*}
        \Bc^\prime(E_{W^\prime}(S, \overline{S})) + \Br^\prime(S)
        &\geq \Bc(E_W(S, \overline{S})) \geq \frac{1}{2}(\Bc(E_W(S, \overline{S})) + \Br(S)) \\
        &\geq \frac{\psi}{2}(\vol_{W,\Bc}(S) + \Br(S)) \geq \frac{\psi^3}{20z}(\vol_{W^\prime,\Bc^\prime}(S) + \Br^\prime(S)).
      \end{align*}
    \item If $\Bc(E_W(S, \overline{S})) < \Br(S)$ and $\Br^\prime(S) > \frac{1}{2}\Br(S)$:
      We have
      \begin{align*}
        \Bc^\prime(E_{W^\prime}(S, \overline{S})) + \Br^\prime(S)
        &\geq \frac{1}{2}(\Bc(E_W(S, \overline{S})) + \Br(S)) \geq \frac{\psi}{2}(\vol_{W,\Bc}(S) + \Br(S)) \\
        &\geq \frac{\psi^3}{20z}(\vol_{W^\prime,\Bc^\prime}(S) + \Br^\prime(S)).
      \end{align*}
    \item If $\Bc(E_W(S, \overline{S})) < \Br(S)$ and $\Br^\prime(S) \leq \frac{1}{2}\Br(S)$:
    The flow paths with tails in $S$ send precisely $\frac{8z}{\psi}\Br(S) - \Br^\prime(S) \geq \frac{4z}{\psi}\Br(S)$ units of flow, within which at most $z \cdot \Bsink(S) = z(\vol_{F,\Bc_G}(S) + \Br(S))$ units are absorbed in $S$
    Since each path $P_i$ with $c_i$ units of flow from $S$ to $\overline{S}$ turn into an edge of capacity $c_i$ in $E_{W^\prime}(S, \overline{S})$, we have $\Bc^\prime(E_{W^\prime}(S, \overline{S})) \geq \frac{4z}{\psi}\Br(S) - z(\vol_{F,\Bc_G}(S) + \Br(S))$.
    We can bound $\vol_{F,\Bc_G}(S) \leq \vol_{W,\Bc}(S) + \Br(S) \leq \frac{1}{\psi}(\Bc(E_W(S, V \setminus S)) + \Br(S)) \leq \frac{2}{\psi}\Br(S)$, which then gives $\Bc^\prime(E_{W^\prime}(S, V \setminus S)) \geq \frac{2z}{\psi}\Br(S) \geq z \cdot \vol_{F,\Bc_G}(S) \geq \frac{\psi^2}{10z^2}(\vol_{W^\prime,\Bc^\prime}(S) + \Br^\prime(S))$.
  \end{itemize}

  As $\psi^\prime \leq \min\left\{\frac{\psi^3}{20z}, \frac{\psi^2}{10z^2}\right\}$, Property \labelcref{item:witness:expansion} is preserved.
  Since the running time of the algorithm is dominated by $O(\log n)$ calls to \cref{thm:flow} which runs in $\widetilde{O}\left(\frac{n^2}{\phi{\phi^\prime}^2\psi^4}\right)$, the proof is completed.
\end{proof}

By running the cut-matching game of \cref{thm:directed-cut-matching-game} with our sparse-cut algorithm of \cref{thm:flow}, we can also obtain the following lemma which constructs an initial witness.
As the proof is fairly standard, we defer it to \cref{appendix:omitted-proofs}.

\begin{restatable}{lemma}{TopLevelWitness}
  Given an $n$-vertex strongly connected graph $G = (V,E)$, terminal edge set $F \subseteq E$, parameters $\phi, R$, and a $\phi^\prime$-expander hierarchy $\cH$ of $G \setminus F$ with height $O(\log n)$, there is an algorithm \emph{\textsc{CutOrEmbed($G, \Bc_G, F, \phi, R^\prime$)}} that either output
  \begin{enumerate}
    \item\label{item:case:sparse-cut} a set $S \subseteq V$ such that $\min\{\Bc_G(E_G(S, \overline{S})), \Bc_G(E_G(\overline{S}, S))\} < \phi \cdot \vol_{F,\Bc_G}(S)$ and $\frac{1}{4t_{\CMG}}R \leq \vol_{F,\Bc_G}(S) \leq \frac{1}{2}\vol_{F,\Bc_G}(V)$ or
    \item a $\Bgamma \in \N^V$ and an $(R, \phi, \widetilde{\psi})$-witness $(W, \Bc, \Br, \Pi_{W \to G})$ of $(G, \Bc_G, F)$ where $\widetilde{\psi} = \Omega\left(\frac{1}{\log^3 n}\right)$ with respect to $\Bgamma$.
  \end{enumerate}
  The algorithm runs in time $\widetilde{O}\left(\frac{n^2}{\phi{\phi^\prime}^2}\right)$
  \label{lemma:cut-matching}
\end{restatable}

\subsection{Stability of Witnesses} \label{subsec:stability}

Following the framework of \cite{HuaKGW23}, our algorithm maintains for each strongly connected component $U$ of the current graph and each $\ell \in \{0, \ldots, L\}$ a witness $(W_{U,\ell}, \Bc_{U,\ell}, \Br_{U,\ell}, \Pi_{W_{U,\ell} \to G[U]})$.
To maintain these witnesses, we will repeatedly find sparse cuts in the graph and remove them until we have certified such cuts do not exist.
After each cut is found, we may also update the volume with which the witness needs to certify in response to $\cM_{\prev}.\alg{Cut}{}$.
We handle these updates through the subroutine \alg{UpdateWitness}{$U, S, A$} implemented in \cref{alg:remove-edges}.
The subroutine removes the cut $S$ from $U$ and then increases the terminal set from $F$ to $F \cup A$, while making sure that all the $W_{U,\ell}$'s remain valid witnesses.

Each cut $S$ that we call \alg{UpdateWitness}{} on will either correspond directly to a call to the \alg{Cut}{} function in \cref{def:hierarchy-maintainer}, in which case we refer to $S$ as an \emph{external} cut, or correspond to the sparse cut found internally when maintaining witnesses (specifically through \cref{lemma:prune-or-repair,lemma:cut-matching}), in which case we refer to $S$ as an \emph{internal} cut.
To handle the updates, the algorithm simply projects each witness $W_{U,\ell}$ onto $U \setminus S$ and removes edges in it that are no longer embedded into the new $G[U\setminus S]$.
It also makes necessary increases to the value of $\Br_{U,\ell}(v)$ when new edges are added to $F$ to ensure the validity of the witnesses.
Note that we essentially ignore $S$ and are not projecting $W_{U,\ell}$ onto it.
This is because once $\alg{UpdateWitness}{U,S,A}$ is called, as we shall see in \cref{alg:maintain-exp-decomp}, our algorithm will immediately reconstruct all the witnesses of $S$ entirely from scratch.

We take a rather modularized approach and guarantee that all witnesses, once constructed or repaired, will only be updated using \alg{UpdateWitness}{}.
Therefore, before giving the full details on how the subroutine is used in \cref{subsec:maintain-exp-decomp} and how the parameters are set, we first establish several stability properties that are key to our analysis later.

\begin{algorithm}
  \caption{Implementation of \alg{UpdateWitness}{}}
  \label{alg:remove-edges}
  
  \SetEndCharOfAlgoLine{}
  \SetKwInput{KwData}{Input}
  \SetKwInput{KwResult}{Output}
  \SetKwProg{KwProc}{function}{}{}
  \SetKwInOut{State}{global}
  \SetKwFunction{UpdateWitness}{UpdateWitness}
  
  \vspace{0.4em}

  \State{the terminal edge set $F$}

  \vspace{0.4em}

  \KwProc{\UpdateWitness{$U, S \subseteq U, A \subseteq G[U]$}} {
    \tcp{remove $S$ from $U$ and add $A$ into the terminal set $F$}
    \For{$\ell \in \{0, \ldots, L\}$} {
      \For{$e \in E_{G[U]}(S,\overline{S}) \cup E_{G[U]}(S,\overline{S})$ and $e^\prime \in \Pi^{-1}_{W_{U,\ell} \to G[U]}(e)$} {
        Let $e^\prime = (u, v)$. Increase $\Br_{U,\ell}(u)$ and $\Br_{U,\ell}(v)$ by $\Bc_{U,\ell}(e^\prime)$ and remove $e^\prime$ from $W_{U,\ell}$.\;
      }
      \For{$e = (u, v) \in A \setminus F$} {
        Increase $\Br_{U,\ell}(u)$ and $\Br_{U,\ell}(v)$ by $\Bc_G(e)$.\;
      }
    }
    Replace $U$ in $\cU$ with $U \setminus S$ and $S$.\;
    Let $W_{U \setminus S,\ell}$ be $W_{U,\ell}[U \setminus S]$, $\Br_{U \setminus S,\ell}$ be $\Br_{U \setminus S,\ell}$ restricted to $U \setminus S$, and $\Bgamma_{U \setminus S}$ be $\Bgamma_U$ restricted to $U \setminus S$.\;
    Update $U \gets U \setminus S$ and $F \gets F \cup A$.\;
  }

\end{algorithm}

First note that each $W_{U,\ell}$ remains a valid witness, albeit with an increase in $\|\Br_{U,\ell}\|_1$.

\begin{claim}
  After each call to \UpdateWitness{}, the tuple $(W_{U,\ell}, \Bc_{U,\ell}, \Br_{U,\ell}, \Pi_{W_{U,\ell} \to G[U]})$ remains a valid $(\infty, \phi, \psi_{\ell})$-witness of $(G[U], \Bc_G, F)$ with respect to $\Bgamma_U$.
  \label{claim:reman-valid-witness}
\end{claim}

\begin{proof}
  Observe that for each edge $e^\prime = (u, v)$ removed from $W_{U,\ell}$, there is a corresponding increase in $\Br_{U,\ell}(u)$ and $\Br_{U,\ell}(v)$ by $\Bc_{U,\ell}(e^\prime)$.
  Likewise, each newly added terminal edge $(u, v)$ has a corresponding increase in $\Br_{U,\ell}(u)$ and $\Br_{U,\ell}(u)$ by $\Bc_G(u, v)$.
  Let $U^\prime$ and $F^\prime$ be the set $U$ and $F$ before the update.
  Let $W^\prime_{U^\prime,\ell}$, $\Bc^\prime_{U^\prime,\ell}$, and $\Br^\prime_{U^\prime,\ell}$ be the old witness.
  We thus have for each $u \in U$ that
  \begin{equation}
    \left(\deg_{W_{U,\ell}, \Bc_{U,\ell}}(v) + \Br_{U,\ell}(v)\right) - \left(\deg_{W^\prime_{U^\prime,\ell},\Bc_{U,\ell}}(v) + \Br^\prime_{U^\prime,\ell}(v)\right) = (\deg_{F,\Bc_G}(v) - \deg_{F^\prime,\Bc_G}(v)).
    \label{eq:degree-change}
  \end{equation}
  Thus, Property~\labelcref{item:witness:degree} is preserved.
  Property \labelcref{item:witness:congestion} is trivially preserved as well since the congestion of $\Pi_{W_{U,\ell} \to G[U]}$ can only decrease.
  For Property~\labelcref{item:witness:expansion}, consider a cut $(T, U \setminus T)$ in the new $U$ with $\Bgamma_U(T) \leq \Bgamma_U(U \setminus T)$.
  As $\Bgamma_{U^\prime}(T) \leq \Bgamma_{U^\prime}((U \setminus T) \cup S)$ clearly holds, we have
  \[
    \Bc_{U,\ell}\left(E_{W^\prime_{U^\prime,\ell}}(T, (U \setminus T) \cup S)\right) + \Br^\prime_{U^\prime, \ell}(T) \geq \psi(\vol_{W^\prime_{U^\prime,\ell},\Bc_{U,\ell}}(T) + \Br^\prime_{U^\prime, \ell}(T))
  \]
  by the properties of the old witness.
  Using the above arguments we can derive
  \begin{align*}
    \Br_{U,\ell}(T) &- \Br^\prime_{U^\prime,\ell}(T) \\
    &\geq  \left(\Bc_{U,\ell}\left(E_{W^\prime_{U^\prime,\ell}}(T, (U \setminus T) \cup S)\right) - \Bc_{U,\ell}\left(E_{W_{U,\ell}}(T, U \setminus T)\right)\right) + (\vol_{F,\Bc_G}(T) - \vol_{F^\prime,\Bc_G}(T))
  \end{align*}
  which implies
  \begin{align*}
    \Bc_{U,\ell}\left(E_{W_{U,\ell}}(T, U \setminus T)\right) + \Br_{U, \ell}(T) &\geq \psi(\vol_{W^\prime_{U^\prime,\ell},\Bc_{U,\ell}}(T) + \Br^\prime_{U^\prime, \ell}(T)) + (\vol_{F,\Bc_G}(T) - \vol_{F^\prime,\Bc_G}(T)) \\ &\geq \psi(\vol_{W_{U,\ell},\Bc_{U,\ell}}(T) + \Br_{U,\ell}(T))
  \end{align*}
  when combined with \eqref{eq:degree-change}.
  This proves that $(W_{U,\ell}, \Bc_{U,\ell}, \Br_{U,\ell}, \Pi_{W_{U,\ell} \to G[U]})$ is an $(\infty, \phi, \psi)$-out-witness of $(G[U], \Bc_G, F)$.
  The proof on the reversed graph follows analogously.
\end{proof}

We further argue that the increase in $\|\Br_{U,\ell}\|_1$ is relatively stable with respect to internal cuts and is mostly dominated by external cuts.
In particular, we consider the following scenario.

\begin{scenario}
Suppose at some moment $(W_{U,\ell}, \Bc_{U,\ell}, \Br_{U,\ell}, \Pi_{W_{U,\ell} \to G[U]})$ is an $(R, \phi, \psi_{\ell})$-witness of $(G[U], \Bc_G, F)$.
Let $U_0$ be the set $U$ and $F_0$ be the set $F$ at this moment.
There is then a sequence of $r$ calls to \alg{UpdateWitness}{$U, S_i, A_i$} for $S_i \subseteq U_{i-1}$ (where $U_i \defeq U_{i-1} \setminus S_i$) and $A_i \subseteq G[U_{i-1}]$ such that the cut $S_i$ has $\delta_i \defeq \min\{\Bc_G(E_{G[U_{i-1}]}(S_i, \overline{S_i})), \Bc_G(E_{G[U_{i-1}]}(\overline{S_i}, S_i))\}$ boundary capacities and there are at most $\Delta_i = 2\Bc_G(A_i)$ units of volume added after the $i$-th update.
Let $F_i \defeq F_{i-1} \cup A_i$ be the terminal edge set after the $i$-th update.
\label{scenarion:stability}
\end{scenario}

\begin{remark}
Note that \cref{scenarion:stability} models the case when these are the only changes made to $W_{U,\ell}$.
Later in \cref{subsec:maintain-exp-decomp} we will perform periodic reconstruction of $W_{U,\ell}$ which is not characterized by \cref{scenarion:stability}, and we will ensure that the stability properties established in this section are used only when \cref{scenarion:stability} applies.
\end{remark}

We first bound how much $\|\Br_{U,\ell}\|_1$ can grow in terms of the $\delta_i$'s assuming there is no terminal addition at all, i.e., $A_i = \emptyset$ for all $i \in [r]$.
We make the following observation regarding the boundary edges of a union of cuts.

\begin{observation}
  Suppose there is a sequence of cuts $S_1, \ldots, S_k$ in a graph $G = (V, E)$ with $S_i \subseteq V_{i-1}$ where $V_0 \defeq V$ and $V_i \defeq V_{i-1} \setminus S_i$ and consider $S \defeq \bigcup_{j \in \cJ}S_j$ for some $\cJ \subseteq [k]$.
  Then, we have
  \[
    E_G(S, \overline{S}) \subseteq \bigcup_{i \in \cJ}E_{G[V_{i-1}]}(S_i, \overline{S_i}) \cup \bigcup_{i \not\in \cJ}E_{G[V_{i-1}]}(\overline{S_i}, S_i).
  \]
  \label{obs:boundary-of-union-of-cuts}
\end{observation}

\begin{lemma}
  In \cref{scenarion:stability}, suppose $A_i = \emptyset$ for all $i \in [r]$, then after the $r$ updates we have $\|\Br_{U_r,\ell}\|_1 \leq \frac{3R}{\psi_{\ell}} + \frac{4}{\psi_{\ell}^2\phi} \sum_{j \in [r]}\delta_j$.
  \label{lemma:blow-up}
\end{lemma}

\begin{proof}
  Let us call a cut $S_i$ \emph{out-sparse} if $\Bc_G(E_{G[U_{i-1}]}(S_i,\overline{S_i})) \leq \Bc_G(E_{G[U_{i-1}]}(\overline{S_i},S_i))$ and \emph{in-sparse} otherwise.
  Let
  \[ S_{\mathrm{out}} \defeq \bigcup_{\text{$S_j$ is out-sparse}}S_j\quad\text{and}\quad S_{\mathrm{in}} \defeq \bigcup_{\text{$S_j$ is in-sparse}}S_j, \]
  be cuts in $G[U_0]$, for which by \cref{obs:boundary-of-union-of-cuts} we have
  \begin{equation}
    \Bc_G\left(E_{G[U_0]}(S_{\mathrm{out}}, \overline{S_{\mathrm{out}}})\right) \leq \sum_{j \in [r]}\delta_j\quad\text{and}\quad\Bc_G\left(E_{G[U_0]}(\overline{S_{\mathrm{in}}}, S_{\mathrm{in}})\right) \leq \sum_{j \in [r]}\delta_j.
    \label{eq:in-out}
  \end{equation}
  Let $D \defeq E(G[U_0]) \setminus E(G[U_r])$ be the set of edges that are deleted from $G[U]$ after the $r$ updates.
  Note that the each increase in $\Br_{U_r,\ell}(v)$ from the initial $\Br_{U_0,\ell}(v)$ for $v \in U_r$ corresponds to an edge $e$ incident to $v$ in $W_{U_0,\ell}$ that embeds into a deleted edge, i.e., $\Pi_{W_{U_0,\ell} \to G[U_0]}(e) \cap D \neq \emptyset$.
  Let $D^{-1} \subseteq E(W_{U_0,\ell})$ be the set of such edges.
  We can analyze the size of $|D^{-1}|$ by considering an $e \in D^{-1}$.
  \begin{itemize}
    \item If both endpoints of $e$ are in $U_k$, then $\Pi_{W_{U_0,\ell}}(e)$ must use one edge in $E_{G[U_0]}(S_{\mathrm{out}}, \overline{S_{\mathrm{out}}}) \cup E_{G[U_0]}(\overline{S_{\mathrm{in}}}, S_{\mathrm{in}})$ since it must enters and leaves one of $S_{\mathrm{out}}$ and $S_{\mathrm{in}}$.
    \item If $e \in E_{W_{U_0,\ell}}(S_{\mathrm{out}}, \overline{S_{\mathrm{out}}})$, then $\Pi_{W_{U_0,\ell} \to G[U_0]}(e) \cap E_{G[U_0]}(S_{\mathrm{out}}, \overline{S_{\mathrm{out}}}) \neq \emptyset$; similarly, if $e \in E_{W_{U_0,\ell}}(\overline{S_{\mathrm{in}}}, S_{\mathrm{in}})$, then $\Pi_{W_{U_0,\ell} \to G[U_0]}(e) \cap E_{G[U_0]}(\overline{S_{\mathrm{in}}}, S_{\mathrm{in}}) \neq \emptyset$.
    \item Otherwise, we have $e \in E_{W_{U_0,\ell}}(\overline{S_{\mathrm{out}}}, S_{\mathrm{out}})$ or $e \in E_{W_{U_0,\ell}}(S_{\mathrm{in}}, \overline{S_{\mathrm{in}}})$.
  \end{itemize}
  This gives us the bound of
  \begin{align*}
    \Bc_{U,\ell}(D^{-1}) \leq \frac{1}{\psi_{\ell}\phi} \cdot \big(\Bc_G\left(E_{G[U_0]}(S_{\mathrm{out}}, \overline{S_{\mathrm{out}}})\right) &+ \Bc_G\left(E_{G[U_0]}(\overline{S_{\mathrm{in}}}, S_{\mathrm{in}})\right)\big) \\
    &+ \Bc_{U,\ell}\left(E_{W_{U_0,\ell}}(\overline{S_{\mathrm{out}}}, S_{\mathrm{out}})\right) + \Bc_{U,\ell}\left(E_{W_{U_0,\ell}}(S_{\mathrm{in}}, \overline{S_{\mathrm{in}}})\right)
  \end{align*}
  by the congestion of $\Pi_{W_{U_0,\ell} \to G[U_0]}$ from Property \labelcref{item:witness:congestion}.
  To bound the right-hand side, we prove the following claim.
  \begin{claim}
    For any $(R, \phi, \psi)$-witness $(W, \Bc, \Br, \Pi_{W\to G})$ we have $\Bc(E_W(\overline{S}, S)) \leq \frac{1}{\psi}(\Bc(E_W(S, \overline{S})) + R)$ for all $S \subseteq V$.
    \label{claim:two-directions}
  \end{claim}
  
  \begin{proof}
    Note that $\Bc(E_W(\overline{S}, S)) \leq \min\{\vol_{W,\Bc}(S), \vol_{W,\Bc}(\overline{S})\}$ and thus it suffices to bound the latter.
    If $\Bgamma(S) \leq \Bgamma(\overline{S})$, then since $W$ is an out-witness we have $\vol_{W,\Bc}(S) \leq \vol_{W,\Bc}(S) + \Br(S) \leq \frac{1}{\psi}(\Bc(E_W(S, \overline{S})) + \Br(S)) \leq \frac{1}{\psi}(\Bc(E_W(S, \overline{S})) + R)$.
    Similarly, if $\Bgamma(S) > \Bgamma(\overline{S})$, then since $\rev{W}$ is an out-witness we have $\vol_{W,\Bc}(\overline{S}) \leq \vol_{W,\Bc}(\overline{S}) + \Br(\overline{S}) \leq \frac{1}{\psi}(\Bc(E_{\rev{W}}(\overline{S}, S)) + \Br(\overline{S})) \leq \frac{1}{\psi}(\Bc(E_W(S, \overline{S})) + R)$.
  \end{proof}

  Following \cref{claim:two-directions}, we have
  \begin{align*}
    \Bc(D^{-1}) &\leq \frac{2}{\psi_{\ell}\phi} \cdot \sum_{j \in [r]}\delta_j + \frac{1}{\psi_{\ell}}\left(\Bc_{U_0,\ell}\left(E_{W_{U_0,\ell}}(S_{\mathrm{out}}, \overline{S_{\mathrm{out}}})\right) + R\right) + \frac{1}{\psi_{\ell}}\left(\Bc_{U_0,\ell}\left(E_{W_{U_0,\ell}}(\overline{S_{\mathrm{in}}}, S_{\mathrm{in}})\right) + R\right) \\
    &\leq \frac{4}{\psi_{\ell}^2\phi}\sum_{j \in [r]}\delta_j + \frac{2R}{\psi_{\ell}}.
  \end{align*}
  The lemma follows by adding the initial value of $\|\Br_{U_0,\ell}\|_1 \leq R$ to the above quantity.
\end{proof}

Now, we consider the more general case where $A_i$ may be non-empty.
Moreover, we would like to derive a bound in terms only of external cuts.
Let $\delta_{\mathrm{ext}}$ be the sum of $\delta_i$'s for which $S_i$ is an external cut.
Let $\Delta \defeq \Delta_1 + \cdots + \Delta_r$.
We derive a bound when the following conditions are met.

\begin{condition}
  The following holds.
  \begin{enumerate}[(i)]
    \item\label{item:sparse} Each internal cut $S_i$ satisfies $\delta_i \leq \frac{\phi\psi_{\ell}^2}{128}\vol_{F_{i-1},\Bc_G}(S_i)$.
    \item\label{item:small-side} Each cut $S_i$ satisfies either $\vol_{F_{i-1},\Bc_G}(S_i) \leq \frac{1}{4}\vol_{F_{0},\Bc_G}(U_0)$ or $\vol_{F_{i-1},\Bc_G}(S_i) \leq \frac{1}{2}\vol_{F_{i-1},\Bc_G}(U_{i-1})$ when it is found. %
    \item\label{item:parameters} $\psi_{\ell} \leq \frac{1}{16}$, $R, \Delta \leq \frac{\psi_{\ell}}{64}\vol_{F_{0},\Bc_G}(U_0)$, and $\delta_{\mathrm{ext}} \leq \frac{\phi\psi_{\ell}^2}{800}\vol_{F_0,\Bc_G}(U_0)$.
  \end{enumerate}
  \label{cond:blow-up}
\end{condition}

Observe that in either case of \labelcref{item:small-side}, we have by \labelcref{item:parameters} that
\begin{equation}
  \vol_{F_{r},\Bc_G}(S_i) \leq \vol_{F_{i-1},\Bc_G}(S_i) + \Delta \leq \max\left\{\frac{1}{4}\vol_{F_0,\Bc_G}(U_0), \frac{1}{2}\vol_{F_{i-1},\Bc_G}(U_{i-1})\right\} + \Delta \leq \frac{5}{8}\vol_{F_{r},\Bc_G}(U_0).
  \label{eq:real-cond}
\end{equation}
We will later show that \cref{cond:blow-up} indeed holds (with high probability) throughout our algorithm for expander decomposition maintenance.
For now we assume this is the case and prove the following \cref{lemma:bound-from-external-cuts} using \cref{lemma:union-of-sparse-cut} whose proof is deferred to \cref{appendix:omitted-proofs}.

\begin{restatable}{lemma}{UnionOfSparseCuts}
  Given a graph $G = (V, E)$ and a sequence of cuts $S_1, \ldots, S_k$ where $S_i \subseteq V_{i-1}$ with $V_i \defeq V_{i-1} \setminus S_i$ and $V_0 \defeq V$ satisfies
  \begin{equation}
    \sum_{i \in [k]}\min\left\{\Bc_G(E_{G[V_{i-1}]}(S_i,\overline{S_i})), \Bc_G(E_{G[V_{i-1}]}(\overline{S_i},S_i))\right\} < \phi \cdot \sum_{i \in [k]}\vol_{F,\Bc_G}(S_i)
    \label{eq:union-sparse}
  \end{equation}
  and
  \[ \sum_{i \in [k]}\vol_{F,\Bc_G}(S_i) \leq \alpha \cdot \vol_{F,\Bc_G}(V), \]
  there is a $\left(\min\left\{\frac{\alpha}{2}, 1-\alpha\right\}\vol_{F,\Bc_G}(V)\right)$-balanced $\left(2\phi \min\left\{1, \frac{\alpha}{1-\alpha}\right\}\right)$-sparse cut in $(G,\Bc)$ with respect to $F$.
  \label{lemma:union-of-sparse-cut}
\end{restatable}

Essentially, the above lemma says that if one can successively carve out many ``small'' sparse cuts $S_i$, then the original graph must have contained a ``large'' sparse cut. Now, in the following lemma, we establish that most of the change in $\Br$ and $\vol_F$ comes from the external cuts.

\begin{lemma}
  In \cref{scenarion:stability}, if \cref{cond:blow-up} holds, then we have $\|\Br_{W_{U_r,\ell}}\|_1 \leq \Gamma$ and $\vol_{F_r,\Bc_G}(U_r) \geq \vol_{F_r,\Bc_G}(U_0) - \Gamma$ for $\Gamma \defeq \frac{4(R+\Delta)}{\psi_{\ell}} + \frac{8}{\psi_{\ell}^2\phi}\delta_{\mathrm{ext}}$.
  \label{lemma:bound-from-external-cuts}
\end{lemma}

\begin{proof}
  Observe that we may imagine there is $0$-th update with an empty cut $S_0$ with $A_0 \defeq A_1 \cup \cdots \cup A_r$, and after running \alg{UpdateWitness}{$U, S_0, A_0$} we have $\|\Br_{U,\ell}\|_1 \leq R + \Delta$ and $W_{U,\ell}$ being a witness of $(G[U], \Bc_G, F_r)$ at which point we start considering \cref{scenarion:stability} with $A_i = \emptyset$ for all $i \in [r]$.
  Note that \cref{cond:blow-up}\labelcref{item:sparse} still holds in this case as well as \eqref{eq:real-cond}.
  Also note that $\vol_{F_{r},\Bc_G}(U_0) \leq 2\vol_{F_0,\Bc_G}(U_0)$ by the bound on $\Delta$ and thus \cref{cond:blow-up}\labelcref{item:parameters} implies $R, \Delta \leq \frac{\psi_{\ell}}{32}\vol_{F_{r},\Bc_G}(U_0)$ and $\delta_{\mathrm{ext}} \leq \frac{\phi\psi_{\ell}^2}{400}\vol_{F_{r},\Bc_G}(U_0)$.
  
  Let $B_i \defeq \vol_{F_{r},\Bc_G}(U_0) - \vol_{F_{r},\Bc_G}(U_i) = \sum_{j \in [i]}\vol_{F_r,\Bc_G}(S_j)$ be the total volume of the first $i$ cuts.
  We first show that $B_r \leq \frac{1}{8}\vol_{F_{r},\Bc_G}(U_0)$ must hold under the input assumption.
  Otherwise, let $i$ be such that $B_i \leq \frac{1}{8}\vol_{F_{r},\Bc_G}(U_0)$ and $B_{i+1} > \frac{1}{8}\vol_{F_{r},\Bc_G}(U_0)$.
  By $\vol_{F_{r},\Bc_G}(S_{i+1}) \leq \frac{5}{8}\vol_{F_{r},\Bc_G}(U_0)$ with \eqref{eq:real-cond} we know $B_{i+1} \leq \frac{3}{4}\vol_{F_{r},\Bc_G}(U_0)$.
  Let $S_{\mathrm{int}}$ and $S_{\mathrm{ext}}$ be the union of internal and external cuts among the first $i+1$ cuts.
  We have
  \[
    \sum_{j \in [i+1]}\delta_j \leq \frac{\phi\psi_{\ell}^2}{128}\vol_{F_{r},\Bc_G}(S_{\mathrm{int}}) + \delta_{\mathrm{ext}} \leq \frac{\phi\psi_{\ell}^2}{12}B_{i+1}
  \]
  by \cref{cond:blow-up}\labelcref{item:sparse} and that (1) $\vol_{F_{r},\Bc_G}(S_{\mathrm{int}}) \leq B_{i+1}$, (2) $B_{i+1} \geq \frac{1}{8}\vol_{F_{r},\Bc_G}(U_0)$, and (3) $\delta_{\mathrm{ext}} \leq \frac{\phi\psi_{\ell}^2}{400}\vol_{F_{r},\Bc_G}(U_0)$.
  Because $B_{i+1} \leq \frac{3}{4}\vol_{F_{r},\Bc_G}(U_0)$, \cref{lemma:union-of-sparse-cut} with $\alpha \leq 3/4$ implies there is a $\left(\frac{1}{8}\vol_{F_{r},\Bc_G}(U_0)\right)$-balanced $\frac{\phi\psi_{\ell}^2}{2}$-sparse cut in $(G[U_0], \Bc_G)$ with respect to $F_{r}$.
  This is a contradiction to \cref{claim:no-balanced-sparse-cut} with the fact that $\frac{2(R+\Delta)}{\psi_{\ell}} \leq \frac{1}{8}\vol_{F_{r},\Bc_G}(U_0)$.

  As a result, we may assume $B_r \leq \frac{1}{8}\vol_{F_{r},\Bc_G}(U_0)$.
  If $B_r \leq \frac{4(R+T)}{\psi_{\ell}}$, then we have $\sum_{j \in [r]}\delta_j \leq \frac{4(R+T)}{\psi_{\ell}} \cdot \frac{\phi\psi_{\ell}^2}{128} + \delta_{\mathrm{ext}}$ and the lemma follows by applying \cref{lemma:blow-up}.
  Otherwise, letting $S_{\mathrm{int}}^\prime$ be the union of all internal cuts, if $\delta_{\mathrm{ext}} \leq \frac{31\phi\psi_{\ell}^2}{128}\vol_{F_r,\Bc_G}(S_{\mathrm{int}}^\prime)$ then we must have
  \[
    \sum_{j \in [r]}\delta_j \leq \frac{\phi\psi_{\ell}^2}{128}\vol_{F_{r},\Bc_G}(S_{\mathrm{int}}^\prime) + \delta_{\mathrm{ext}} \leq \frac{\phi\psi_{\ell}^2}{4}\vol_{F_r,\Bc_G}(S_{\mathrm{int}}^\prime) \leq \frac{\phi\psi_{\ell}^2}{4}B_r,
  \]
  which again by \cref{lemma:union-of-sparse-cut} with $\alpha \leq 1/4$ implies the existence of a $\frac{B_r}{2}$-balanced $\frac{\phi\psi_{\ell}^2}{2}$-sparse cut which contradicts \cref{claim:no-balanced-sparse-cut}.
  To this end, we have shown that $\delta_{\mathrm{ext}} \geq \frac{31\phi\psi_{\ell}^2}{128}\vol_{F_r,\Bc_G}(S_{\mathrm{int}}^\prime)$ and therefore $\sum_{j \in [r]}\delta_j \leq 2\delta_{\mathrm{ext}}$.
  We can now apply \cref{lemma:blow-up} to conclude bound on $\|\Br_{W_r,\ell}\|_1$.

  As for the bound on $\vol_{F_{r},\Bc_G}(U_r)$, by the discussion above if $B_r > \frac{4(R+T)}{\psi_{\ell}}$ then $\sum_{j \in [r]}\delta_j \leq 2\delta_{\mathrm{ext}}$.
  Since $B_r \leq \frac{1}{4}\vol_{F_{r},\Bc_G}(U_0)$ this means that we must have $2\delta_{\mathrm{ext}} \geq \frac{\phi\psi_{\ell}^2}{4}B_r$, otherwise the same argument above implies the existence of $\frac{B_r}{2}$-balanced $\frac{\phi\psi_{\ell}^2}{2}$-sparse cut that contradicts \cref{claim:no-balanced-sparse-cut}.
  This thus leaves us with $B_r \leq \frac{8}{\phi\psi_{\ell}^2}\delta_{\mathrm{ext}}$ which the lemma statement asserts.

\end{proof}

\subsection{Maintaining Expander Decomposition} \label{subsec:maintain-exp-decomp}

We now present the algorithm for maintaining expander decomposition while interacting with $\cM_{\prev}$ which in turn gives the algorithm for converting from $k$-level hierarchy to a $(k+1)$-level one with significantly fewer cut edges.
The overall structure of our algorithms is similar to \cite[Algorithm 3]{HuaKGW23}, and we use the stability properties established earlier to derive an expected worst-case recourse guarantee.

\paragraph{Setup.}
Given the input $(k, \alpha_{\prev}, \beta_{\prev}, \phi_{\prev}, T_{\prev})$-hierarchy maintainer $\cM_{\prev}$, to maintain the $(k+1)$-th level (and thereby getting a better maintainer $\cM$), we will maintain the graph $G_{\cM}$ in which $F$ is $\phi$-expanding and its collection of strongly connected components
$\cU$.
We start with $F$ being the edge set not handled by $\cM_{\prev}$.
For each $U \in \cU$ we will maintain an estimate $\tau_U$ of $\vol_{F,\Bc_G}(U)$, a vector $\Bgamma_U \in \N^U$, and $L+1$ witnesses $(W_{U,\ell}, \Br_{U,\ell}, \Pi_{W_{U,\ell} \to G[U]})$ for $\ell \in \{0, \ldots, L\}$, where $(W_{U,\ell}, \Br_{U,\ell}, \Pi_{W_{U,\ell} \to G[U]})$ is an $(\infty, \phi, \psi_{\ell})$-witness of $(G[U], \Bc_G, F)$ with respect to $\Bgamma_U$, with parameters
\begin{equation}
  \psi_{L} \defeq \widetilde{\psi}\quad\text{and}\quad\psi_{\ell} \defeq \frac{1}{2} \cdot \left(\frac{\psi_{\ell+1}^2}{2048 \cdot c_{\ref{thm:flow}} \cdot z^{3L}}\right)^L
  \label{eq:psi}
\end{equation}
that satisfy
\begin{equation}
  \frac{1}{\psi_{0}} \leq {\log n}^{L^{O(L)}}\quad\text{and}\quad\frac{\psi_{\ell+1}^2}{\psi_{\ell}^{1/L}} \geq \Omega(\log^{3L} n).
  \label{eq:compare-psi}
\end{equation}
The algorithm will ensure that, with high probability, each $\|\Br_{W,\ell}\|_1$ falls between $\frac{\psi_{\ell}}{10}\tau_U^{\ell/L}$ and $\tau_U^{\ell/L}$.
This can be enforced deterministically by rebuilding a witness whenever $\|\Br_{W,\ell}\|_1$ grows too large.
However, that leaves us with only an amortized guarantee which as we have argued in \cref{subsec:overview} does not suffice for our purposes.
To achieve a stronger expected worst-case recourse, we define the following distribution $\cR_{t,\tau}$ on $\{0, \ldots, L\}$ from which we will sample a random level to rebuild after every update:
\begin{equation}
  \Pr_{x \sim \cR_{t,\tau}}[x \geq \ell] \defeq \min\left\{1, \frac{t}{\psi_0^2 \cdot \tau^{\ell/L}} \cdot c_{\mathrm{rb}}\ln n\right\}
  \label{eq:distribution}
\end{equation}
where $c_{\mathrm{rb}}$ is a fixed constant that controls the exponents in the with-high-probability statements that we will establish later.

\cref{alg:cut-add-terminal} is the implementation of the \textsc{Init()} and \textsc{Cut()} subroutines which given input parameters $L \in \N$ and $\phi \in (0, 1)$ converts $\cM_{\prev}$ into a $(k+1, \alpha, \beta, \phi^\prime, T)$-hierarchy maintainer $\cM$ for parameters $\alpha,\beta,\phi^\prime,T$ that we will establish in the end of the section.
These subroutines rely on the internal subroutine \textsc{MaintainExpander($U,\ell$)} whose implementation is given in \cref{alg:maintain-exp-decomp}.

\begin{algorithm}[ht!]
  \caption{Implementation of $(k + 1, \alpha, \beta, \phi^\prime, T)$-hierarchy maintainer}
  \label{alg:cut-add-terminal}
  
  \SetEndCharOfAlgoLine{}
  \SetKwInput{KwData}{Input}
  \SetKwInput{KwResult}{Output}
  \SetKwProg{KwProc}{function}{}{}
  \SetKwInOut{State}{global}
  \SetKwFunction{MaintainExpander}{MaintainExpander}
  \SetKwFunction{Cut}{Cut}
  \SetKwFunction{PostProcess}{PostProcess}
  \SetKwFunction{Init}{Init}
  
  \vspace{0.4em}

  \State{parameters $L \in \N$ and $\phi \in (0, 1)$}
  \State{the graph $G_{\cM}$ maintained by $\cM$}
  \State{the collection $\cU$ of SCCs of $G_{\cM}$}
  \State{the terminal edge set $F$}

  \vspace{0.4em}

  \KwProc{\Init{$G, \Bc_G$}} {
    Initialize $F \gets \cM_{\prev}.\alg{Init}{G,\Bc_G}$.\;
    Initialize $G_{\cM} \gets G$ and $\cU \gets \{V\}$.\;
    Let $X \gets \MaintainExpander{V, L}$.\;
    Run \PostProcess{$V$} and then \textbf{return} $X$.\;
  }

  \KwProc{\Cut{$D$}} {
    $G_{\cM} \gets G_{\cM} \setminus D$.\;
    $X \gets \emptyset$ and $Q \gets \emptyset$.\;
    Let $\cU_D \gets \{U \in \cU: D \cap G[U] \neq \emptyset\}$ and $U_D \defeq \bigcup_{U \in \cU_D}U$\;
    \For{$U \in \cU_D$} {
     \tcp{see input requirement of Cut()}
      Let $D_U \defeq D \cap G[U]$ and $S_U \subseteq U$ be such that $\vol_{F,\Bc_G}(S_U) \leq \vol_{F,\Bc_G}(U \setminus S_U)$ and either $E_{G[U]}(S_U,\overline{S_U})$ or $E_{G[U]}(\overline{S_U}, S_U)$ equals $D_U$.\;
      \tcp{UpdateWitness() remove $S_U$ from $U$ and add $A_U$ to $F$}
      Run $A_U \gets \cM_{\prev}.\alg{Cut}{D_U}$ and \UpdateWitness{$U, S_U, A_U$}. \label{line:remove-edges-external}\;
      Sample $k \sim \cR_{\Bc_G(D_U)/\phi + 2\Bc_G(A_U),\tau_U}$ and $X \gets X \cup \MaintainExpander{U, k}$ \label{line:sample-cut}.\;
      Add $S_U$ to $Q$.\;
    }
    \lFor{$S \in Q$} {
      $X \gets X \cup \MaintainExpander{S, L}$. \label{line:build-small-pieces-cut}
    }
    Run \PostProcess{$U_D$} and then \textbf{return} $X$.\;
    \textbf{return} $X$.\;
  }

  \KwProc{\PostProcess{$Y$}} {
    \tcp{ensure $\cU$ is exactly the SCCs of $G_{\cM}$}
    \For{$U \in \cU$ such that $U \subseteq Y$} {
      Let $U_1, \ldots, U_k$ be the strongly connected components of $G[U]$.\;
      Replace $U$ in $\cU$ by $U_1, \ldots U_k$.\;
    }
  }
\end{algorithm}

\begin{algorithm}[htp!]
  \caption{Maintaining expander decomposition}
  \label{alg:maintain-exp-decomp}

  \SetEndCharOfAlgoLine{}
  \SetKwInput{KwData}{Input}
  \SetKwInput{KwResult}{Output}
  \SetKwProg{KwProc}{function}{}{}
  \SetKwInOut{State}{global}
  \SetKwFunction{MaintainExpander}{MaintainExpander}
  \SetKwFunction{CutOrEmbed}{CutOrEmbed}
  \SetKwFunction{PruneOrRepair}{PruneOrRepair}
  \SetKwFunction{UpdateWitness}{UpdateWitness}
  \SetKwFor{Loop}{loop}{}{}

  \vspace{0.4em}

  \State{a hierarchy $\cH_{\prev}$ of $G_{\cM} \setminus F$ maintained by $\cM_{\prev}$}

  \vspace{0.4em}

  \KwProc{\MaintainExpander{$U, \ell$}} {
    \lIf{$\vol_{F,\Bc_G}(U) < 1/\phi$} {
      \textbf{return} $\emptyset$.\label{line:too-little-volume}
    }
    $X \gets \emptyset$ and $Q \gets \emptyset$.\;
    \Loop{} { \label{line:while-loop}
      \If{$\ell = L$} {
        $\tau_U \gets \frac{\psi_0^2}{64z} \vol_{F,\Bc_G}(U)$.\label{line:reset}\;
        Run procedure \CutOrEmbed{$G[U], \Bc_G, F, \phi, \frac{\psi_L}{10}\tau_U, \cH_{\prev}[U]$}. \label{line:sparse-cut-L}\label{line:call-cut-matching}\;
      }
      \Else {
        Run procedures \PruneOrRepair{$G[U], \Bc_G, F, \Br_{U,\ell+1}, W_{U,\ell+1}, \Pi_{W_{U,\ell+1}\to G[U]}, \phi, \psi_{\ell+1}, \frac{\psi_{\ell}}{20} \cdot \tau_U^{\ell/L}, \cH_{\prev}[U]$} and \PruneOrRepair{$\rev{G}[U], \Bc_G, F, \Br_{U,\ell+1}, \rev{W}_{U,\ell+1}, \Pi_{\rev{W}_{U,\ell+1}\to \rev{G}[U]}, \phi, \psi_{\ell+1}, \frac{\psi_{\ell}}{20} \cdot \tau_U^{\ell/L}, \cH_{\prev}[U]$}.\label{line:sparse-cut-l}\label{line:call-prune}\;
      }
      \If{a cut $S$ is returned} {
        Let $D$ be edge sets among $E_{G[U]}(S,\overline{S})$ and $E_{G[U]}(\overline{S}, S)$ with smaller total capacities. \label{line:D}\;
        update $X \gets X \cup D$ and $G_{\cM} \gets G_{\cM} \setminus D$\;
        Let $A \gets \cM_{\prev}.\Cut{D}$ and run \UpdateWitness{$U, S, A$}.\label{line:remove-edges}\label{line:call-cut}\tcp*{assert $A \subseteq U$}
        Add $S$ to $Q$.\;
        \If{a sample $k \sim \cR_{2\Bc_G(A),\tau_U}$ satisfies $k > \ell$} {\label{line:sample}
          Update $X \gets X \cup \alg{MaintainExpander}{U,k}$.\label{line:early-return}\;
          \textbf{break}. \label{line:break-l}\;
        }
      }
      \Else { \label{line:else}
        \If{$\ell = L$} {
          Set $(W_{U,L}, \Bc_{U,L},\Br_{U,\ell}, \Pi_{W_{U,\ell} \to G[U]})$ to be the witness returned by \CutOrEmbed{}.\;
          $\Bgamma_U \gets \Bgamma$.\;
        }
        \Else {
          Let $(W_1, \Bc_1, \Br_1, \Pi_{W_1 \to G[U]})$ and $(W_2, \Bc_2, \Br_2, \Pi_{W_2 \to \rev{G}[U]})$ be the witnesses returned by \PruneOrRepair{}.\;
          Set $W_{U,\ell} \gets W_1 \cup \rev{W_2}$, $\Bc_{U,\ell} \gets \Bc_1 + \Bc_2$, and $\Br_{U,\ell} \gets \Br_1 + \Br_2$.\label{line:rebuilt} \tcp*{the witness $W_{U,\ell}$ is rebuilt}
        }
        \textbf{break}. \label{line:break}\;
      }
    }
    \lIf{$\ell > 0$} { \label{line:call-below}
      $X \gets X \cup \alg{MaintainExpander}{U, \ell-1}$.\label{line:recurse}
    }
    \lFor{$S \in Q$} {
      $X \gets X \cup \alg{MaintainExpander}{S, L}$.\label{line:build-small-pieces}
    }
    \textbf{return} $X$.\;
  }

\end{algorithm}

We say that $W_{U,\ell}$ is \emph{rebuilt} if \textsc{MaintainExpander($U,\ell$)} enters Line~\ref{line:rebuilt}.
Recall that one of our goals is to ensure that $\|\Br_{U,\ell}\|_1$ falls between $\frac{\psi_{\ell}}{10}\tau_U^{\ell/L}$ and $\tau_U^{\ell/L}$, and we show that this is indeed the case right after $W_{U,\ell}$ is rebuilt.

\begin{lemma}
  If $W_{U,\ell}$ is rebuilt at Line~\ref{line:rebuilt}, then the new $(W_{U,\ell}, \Bc_{U,\ell}, \Br_{U,\ell}, \Pi_{W_{U,\ell} \to G[U]})$ is a $\left(\frac{\psi_{\ell}}{10} \tau_U^{\ell/L}, \phi, \psi_{\ell}\right)$-witness of $(G[U], \Bc_G, F)$.
  \label{lemma:repair}
\end{lemma}

\begin{proof}
  We have $\|\Br_{U,\ell}\|_1 = \|\Br_1\| + \|\Br_2\| \leq \frac{\psi_{\ell}}{10} \cdot \tau_U^{\ell/L}$.
  Let $\psi_{\ell}^\prime \defeq \frac{\psi_{\ell+1}^4}{2048 \cdot c_{\ref{thm:flow}} \cdot z^2}$ as in \cref{lemma:prune-or-repair}.
  We have $\psi_{\ell} \leq \frac{1}{2}{\psi_{\ell}^{\prime}}^2$ by \eqref{eq:psi}.
  We then have $\deg_{W_{U,\ell},\Bc_{U,\ell}}(v) + \Br_{U,\ell}(v) \leq \frac{2}{\psi_{\ell}^\prime}\deg_{F,\Bc_G}(v) \leq \frac{1}{\psi_{\ell}}\deg_{F,\Bc_G}(v)$.
  Also, the congestion of $\Pi_{W_{U,\ell} \to G[U]}$ is bounded by $2 \cdot \frac{1}{\phi\psi_{\ell}^\prime} \leq \frac{1}{\phi\psi_{\ell}}$.
  It thus remains to verify Property \labelcref{item:witness:expansion}.
  Consider a cut where $\Bgamma_U(S) \leq \Bgamma_U(\overline{S})$.
  We have $\Bc_1(E_{W_1}(S, \overline{S})) + \Br_1(S) \geq \psi_{\ell}^\prime(\vol_{W_1,\Bc_1}(S) + \Br_1(S))$.
  Note that due to Property \labelcref{item:witness:degree}, it holds that $\psi_{\ell}^\prime(\deg_{W_2,\Bc_2}(v) + \Br_2(v)) \leq \deg_{W_1,\Bc_1}(v) + \Br_1(v) \leq \frac{1}{\psi_{\ell}^\prime}(\deg_{W_2,\Bc_2}(v) + \Br_2(v))$.
  Consequently, we have
  \[ \Bc_{U,\ell}(E_{W_{U,\ell}}(S,\overline{S})) \geq \Bc_1(E_{W_1}(S,\overline{S})) \geq \psi^\prime_{\ell}(\vol_{W_1,\Bc_1}(S) + \Br_1(S)) \geq \frac{{\psi_{\ell}^\prime}^2}{2}(\vol_{W_{U,\ell},\Bc_{U,\ell}}(S) + \Br_{U,\ell}(S)). \]
  Similarly, if $\Bgamma_U(S) > \Bgamma_U(\overline{S})$, then 
  \begin{align*}
    \Bc_{U,\ell}(E_{W_{U,\ell}}(S,\overline{S})) &\geq \Bc_2(E_{\rev{W_2}}(S,\overline{S})) = \Bc_2(E_{W_2}(\overline{S}, S)) \\
    &\geq \psi^\prime_{\ell}(\vol_{W_1,\Bc_1}(S) + \Br_1(S)) \geq \frac{{\psi_{\ell}^\prime}^2}{2}(\vol_{W_{U,\ell},\Bc_{U,\ell}}(S) + \Br_{U,\ell}(S)).
  \end{align*}
  This concludes the proof.
\end{proof}

As we have a bound on how large $\Br_{U,\ell}(v)$ can be, the bound of \cref{lemma:prune-or-repair} with respect to $\vol_F(S) + \Br_{U,\ell}(S)$ can be used to bound the actual volume.

\begin{claim}
  Each cut $S$ found at Line~\ref{line:sparse-cut-l} satisfies $\min\{\Bc_G(E_{G[U]}(S,\overline{S})), \Bc_G(E_{G[U]}(\overline{S},S))\} \leq \frac{\phi\psi_{\ell+1}^2}{128}\vol_{F,\Bc_G}(S)$ and $\vol_{F,\Bc_G}(S) \geq \frac{\psi_{\ell+1}^2\psi_{\ell}}{640z} \tau_{U}^{\ell/L}$.
  \label{claim:is-sparse-cut}
\end{claim}

\begin{proof}
  The cut $S$ found by \cref{lemma:prune-or-repair} satisfies
  \[ \min\{\Bc_G(E_{G[U]}(S,\overline{S})), \Bc_G(E_{G[U]}(\overline{S},S))\} \leq \frac{\phi\psi_{\ell+1}^3}{256}(\vol_{F,\Bc_G}(S) + \Br_{U,\ell+1}(S)) \]
  and $\vol_{F,\Bc_G}(S) + \Br_{U,\ell+1}(S) \geq \frac{\psi_{\ell+1}}{16z} \cdot \frac{\psi_{\ell}}{20} \tau_U^{\ell/L}$.
  By \cref{def:witness}\labelcref{item:witness:degree}, we have $\Br_{U,\ell+1}(S) \leq \frac{1}{\psi_{\ell+1}}\vol_{F,\Bc_G}(S)$, and therefore $\vol_{F,\Bc_G}(S) + \Br_{U,\ell+1}(S) \leq \frac{2}{\psi_{\ell+1}}\vol_{F,\Bc_G}(S)$.
  The claim follows.
\end{proof}

Recall that a cut passed to \alg{UpdateWitness}{} is \emph{internal} if it comes from Line~\ref{line:remove-edges} in \cref{alg:maintain-exp-decomp} and \emph{external} if it comes from Line~\ref{line:remove-edges-external} in \cref{alg:cut-add-terminal}.
We further call such an internal cut \emph{level-$(\ell+1)$} as it is found based on $W_{U,\ell+1}$ when running \alg{MaintainExpander}{$X,\ell$}.

\begin{definition}
A witness $W_{U,\ell}$ is \emph{valid} if since the last time it was rebuilt, we only call \alg{RemoveCut}{$U, S, A$} on it with either external or level-$\ell^\prime$ internal cuts with $\ell^\prime \leq \ell$; otherwise, $W_{U,\ell}$ is \emph{invalid}.
\end{definition}

A witness $W_{U,\ell}$ is \emph{being rebuilt} from the moment the algorithm enters \alg{MaintainExpander}{$X, k$} for some $k \geq \ell$ until either it is actually rebuilt or the call to \alg{MaintainExpander}{$X, k$} returns.
By definition, if $W_{U,\ell}$ is currently being rebuilt, then so are all $W_{U,\ell^\prime}$ with $\ell^\prime < \ell$.
Observe that a witness is always valid unless either it is currently being rebuilt or has volume $\vol_{F,\Bc_G}(U) < 1/\phi$ (due to the early return on Line~\ref{line:too-little-volume}).\footnote{Note that this would have been vacuously true if our algorithm does not have the early break on Line~\ref{line:break-l} in \cref{alg:maintain-exp-decomp} (like in \cite[Algorithm 3]{HuaKGW23}). Still, if we break early on Line~\ref{line:break-l}, then the call to \alg{MaintainExpander}{$X, k$} on Line~\ref{line:early-return} will be in charge of rebuilding $W_{U,\ell}$.}
Let us call such a $U$ with $\vol_{F,\Bc_G}(U) < 1/\phi$ \emph{negligible}.

\begin{observation}
  Each witness $W_{U,\ell}$ for a non-negligible $U$ remains valid unless it is currently being rebuilt, in which case it must actually be rebuilt before the call to the corresponding \emph{\alg{MaintainExpander}{$U,k$}} returns.
  Moreover, if the algorithm is currently running \emph{\alg{MaintainExpander}{$U, \ell$}}, then all $W_{U,k}$ where $k > \ell$ are valid.
  \label{obs:will-be-rebuilt}
\end{observation}

\begin{observation}
  If $U$ becomes negligible at some point, then it remains negligible afterward.
  \label{obs:tiny-pieces-remain-tiny}
\end{observation}

Consider two timestamps $t_1 < t_2$ throughout the execution of the algorithm.

\begin{definition}
  A tuple $(U, \ell, t_1, t_2)$ is \emph{active} if (i) $U$ is non-negligible at time $t_2$ (and thus from $t_1$ to $t_2$ by \cref{obs:tiny-pieces-remain-tiny}) and (ii) $W_{U,\ell}$ is not rebuilt nor being rebuilt in any point of time between $t_1$ and $t_2$ (inclusively).
\end{definition}

Consider an active tuple $(U, \ell, t_1, t_2)$.
Let $\delta_{\mathrm{ext}}^{(U,\ell)}(t_1, t_2)$ be the sum of the capacities of $D_U$'s of the \alg{Cut}{$D$} calls that happen from time $t_1$ to $t_2$.
Likewise, let $\Delta^{(U,\ell)}(t_1, t_2)$ be \emph{two times} the sum of the capacities of $A$'s of the \alg{UpdateWitness}{$U, S, A$} calls happened in this period of time.
In other words, $\Delta^{(U,\ell)}$ is an upper bound on the units of volume increased in $U$ from $t_1$ to $t_2$.
We show that for an active tuple $(U, \ell, t_1, t_2)$, with high probability both $\delta_{\mathrm{ext}}^{(U,\ell)}(t_1, t_2)$ and $\Delta^{(U,\ell)}(t_1, t_2)$ are bounded.
Note that as the value of $\tau_U$ only changes in \alg{MaintainExpander}{$U,L$}, during this period of time $\tau_U$ remains unchanged (otherwise $W_{U,\ell}$ would have been rebuilt at some point in between).
Also note that it is important to establish the bound against \emph{any} (possibly adversarial) sequence of \alg{Cut}{} calls, since as we have seen earlier we will encapsulate this hierarchy maintainer into another one that has a better quality.
The update sequence we see now thus depends on our previous outputs (hence the previous randomness used).

\begin{lemma}
  For any (possibly adaptive adversarial) sequence of \emph{\alg{Cut}{}} and any active tuple $(U, \ell, t_1, t_2)$, with high probability, we have $\Delta^{(U,\ell)}(t_1, t_2) \leq \frac{\psi_{\ell}}{10}\tau_U^{\ell/L}$ and $\delta_{\mathrm{ext}}^{(U,\ell)}(t_1, t_2) \leq \frac{\phi\psi_{\ell}^2}{40}\tau_U^{\ell/U}$.
  \footnote{The exponent in the with high probability statement depends on our choice of the constant $c_{\mathrm{rb}}$ in \eqref{eq:distribution}.}
  \label{lemma:good-event}
\end{lemma}

\begin{proof}
  Let $\Delta_i$ be two times the total capacities of $A$ in the $i$-th call of \alg{RemoveCut}{$U, S, A$} from time $t_1$ to $t_2$.
  If $\Delta^{(U,\ell)}(t_1, t_2) > \frac{\psi_{\ell}}{10}\tau_U^{\ell/L}$, then the probability that none these subroutines called \alg{MaintainExpander}{$U, k$} for some $k \geq \ell$ (which scheduled $W_{U,\ell}$ to be rebuilt) is at most
  \[
    \prod_{i}\left(1 - \frac{\Delta_i}{\psi_0^2\tau_U^{\ell/L}} \cdot c_{\mathrm{rb}}\ln n\right) \leq \exp\left(-\sum_{i}\Delta_i \cdot \frac{c_{\mathrm{rb}} \ln n}{\psi_0^2\tau_U^{\ell/L}}\right) \leq \exp\left(-\frac{c_{\mathrm{rb}}\ln n}{10\psi_0}\right) \leq n^{-c_{\mathrm{rb}}/10}.
  \]
  Similarly, let $D_i$ be the total capacities $\Bc_G(D_U)$ of $D_U$ in the $i$-th call to \alg{Cut}{$D$} from time $t_1$ to $t_2$.
  If $\delta_{\mathrm{ext}}^{(U,\ell)}(t_1, t_2) = \sum_{i}D_i > \frac{\phi\psi_{\ell}^2}{40}\tau_U^{\ell/U}$, then the probability that none of the \alg{Cut}{} called \alg{MaintainExpander}{$U, k$} for some $k \geq \ell$ is at most
  \[
    \prod_{i}\left(1 - \frac{D_i}{\phi} \cdot \frac{1}{\psi_0^2\tau_U^{\ell/L}} \cdot c_{\mathrm{rb}}\ln n\right) \leq \exp\left(-\sum_{i}\frac{D_i}{\phi} \cdot \frac{c_{\mathrm{rb}} \ln n}{\psi_0^2\tau_U^{\ell/L}}\right) \leq \exp\left(-\frac{c_{\mathrm{rb}}\ln n}{40}\right) = n^{-c_{\mathrm{rb}}/40}.
  \]
  Consequently, with high probability, the bounds of $\Delta^{(U,\ell)}(t_1, t_2) \leq \frac{\psi_{\ell}}{10}\tau_U^{\ell/L}$ and $\delta_{\mathrm{ext}}^{(U,\ell)}(t) \leq \frac{\phi\psi_{\ell}^2}{40}\tau_U^{\ell/L}$ hold if $W_{U,\ell}$ is not currently scheduled to be rebuilt.
  Observe that we generate new randomness for each of our random choices, and thus this high probability guarantee works against any update sequence.
\end{proof}

In light of \cref{lemma:good-event}, let $\cK$ be the event such that
\begin{equation}
\Delta^{(U,\ell)}(t_1, t_2) \leq \frac{\psi_{\ell}}{10}\tau_U^{\ell/L}\quad\text{and}\quad\delta_{\mathrm{ext}}^{(U,\ell)}(t_1, t_2) \leq \frac{\phi\psi_{\ell}^2}{40}\tau_U^{\ell/L}\;
\label{eq:event-K}
\end{equation}
hold for all active tuples $(U, \ell, t_1, t_2)$.
Observe that there are only $\poly(n)$ effective timestamps throughout any possible execution of the algorithm\footnote{The subroutine \alg{UpdateWitness}{} will be called at most $n$ times and $F \subseteq E$ always hold so there will be at most $O(m)$ terminal additions.} and for each effective timestamp there are only $O(n)$ possible $U$'s at this time since they are vertex-disjoint, and thus by \cref{lemma:good-event} and a union bound $\cK$ happens with high probability.

\begin{lemma}
  Conditioned on $\cK$, for any $U \in \cU$ and time $t$, if $W_{U,L}$ is not currently being rebuilt at time $t$, then we have $\frac{\psi_0^2}{128z}\vol_{F,\Bc_G}(U) \leq \tau_U \leq \frac{\psi_0^2}{32z}\vol_{F,\Bc_G}(U)$; moreover, and for each $\ell \in \{0, \ldots, L\}$, if $W_{U,\ell}$ is not currently being rebuilt at time $t$, then $\|\Br_{U,\ell}\|_1 \leq \tau_U^{\ell/L}$ holds.
  \label{lemma:probability-guarantee}
\end{lemma}

\begin{proof}
  We prove the statement by an induction on the time $t$.
  Fix a $U \in \cU$ and $\ell \in \{0, \ldots, L\}$ for which $W_{U,\ell}$ is not currently being rebuilt.
  Consider the last time $t_{\mathrm{last}}^{(U,\ell)} < t$ that it was rebuilt.
  Let $U_0$ and $F_0$ be the set $U$ and $F$ at time $t_{\mathrm{last}}^{(U,\ell)}$.
  By the inductive hypothesis at time $t_{\mathrm{last}}^{(U,\ell)}$, we have $\tau_U \leq \frac{\psi_0}{32z}\vol_{F_0}(U_0)$.
  Note that the value of $\tau_U$ remains unchanged from $t_{\mathrm{last}}^{(U,\ell)}$ to $t$, otherwise $W_{U,\ell}$ would have currently been being rebuilt.
  By \cref{lemma:repair}, we have $\|\Br_{U_0,\ell}\|_1 \leq \frac{\psi_{\ell}}{10}\tau_U^{\ell/L}$.
  Observe that $(U, \ell, t_{\mathrm{last}}^{(U,\ell)}, t)$ is active, and thus what happens from $t_{\mathrm{last}}^{(U,\ell)}$ to $t$ is modeled by \cref{scenarion:stability}.
  Our goal is thus to apply \cref{lemma:bound-from-external-cuts} on $W_{U,\ell}$ (with $R \leq \frac{\psi_{\ell}}{10}\tau_U^{\ell/L}$) to bound the increase in $\|\Br_{U,\ell}\|_1$.
  For that we need to verify that \cref{cond:blow-up} holds.
  In the remainder of the proof we adapt the notation in \cref{scenarion:stability} (e.g., $U_i$, $S_i$, and $F_i$).

  \paragraph{\cref{cond:blow-up}\labelcref{item:sparse}.}
  Note that by \cref{obs:will-be-rebuilt}, $W_{U,\ell}$ is currently valid, meaning that all the internal cuts $S_i$ for which \alg{UpdateWitness}{$U,S_i,\cdot$} is called are of level-$\ell^\prime$ for $\ell^\prime < \ell$.
  By \cref{claim:is-sparse-cut}, each such cut $S_i$ satisfies $\min\{\Bc_G(E_{G[U_{i-1}]}(S_i,\overline{S_i})), \Bc_G(E_{G[U_{i-1}]}(\overline{S_i}, S_i))\} \leq \frac{\phi\psi_{\ell^\prime+1}^2}{128}\vol_{F_{i-1},\Bc_G}(S_i) \leq\frac{\phi\psi_{\ell}^2}{128}\vol_{F_{i-1},\Bc_G}(S_i)$.
  
  \paragraph{\cref{cond:blow-up}\labelcref{item:small-side}.}
  For external cuts, since we always let $S_U$ be the side with smaller volume on Line~\ref{line:remove-edges-external} in \cref{alg:cut-add-terminal}, we indeed have $\vol_{F_{i-1},\Bc_G}(S_i) \leq \frac{1}{2}\vol_{F_{i-1},\Bc_G}(U_{i-1})$.
  This is the same for internal cuts of level-$L$ by \cref{lemma:cut-matching}.
  For $\ell < L$, by \cref{lemma:prune-or-repair}, every such cut $S_i$ satisfies $\vol_{F_i,\Bc_G}(S_i) + \Br_{U,\ell^\prime+1}(S) \leq \frac{8z}{\psi_{\ell^\prime+1}} R$ for some $\ell^\prime < \ell$ if $(W_{U,\ell^\prime+1}, \Br_{U,\ell^\prime+1}, \Pi_{W_{U,\ell^\prime+1} \to G[U]})$ is an $R$-witness at the time $t_i$ it was found.
  By the inductive hypothesis at time $t_i$, we have $R \leq \tau_U^{(\ell^\prime+1)/L} \leq \tau_U \leq \frac{\psi_0^2}{32z}\vol_{F_0,\Bc_G}(U_0)$, and thus we have $\vol_{F_i,\Bc_G}(S_i) \leq \frac{1}{4}\vol_{F_0,\Bc_G}(U_0)$.

  \paragraph{\cref{cond:blow-up}\labelcref{item:parameters}.}
  The fact that $\delta_{\ell} \leq \frac{1}{16}$ is straightforward.
  That $R \leq \frac{\psi_{\ell}}{64}\vol_{F_0,\Bc_G}(U_0)$ is by \cref{lemma:repair} which shows that right after the rebuild we have $R = \|\Br_{U,\ell}\|_1 \leq \frac{\psi_{\ell}}{10}\tau_U^{\ell/L}$.
  Applying the bound of $\tau_U \leq \frac{\psi_0^2}{32z}\vol_{F,\Bc_G}(U_0)$ establishes the fact.
  For the last two bounds, by the conditioning on $\cK$, we have $\Delta^{(U,\ell)}(t_{\mathrm{last}}^{(U,\ell)}, t) \leq \frac{\psi_{\ell}}{10}\tau_U^{\ell/L}$ and $\delta_{\mathrm{ext}}^{(U,\ell)}(t_{\mathrm{last}}^{(U,\ell)}, t) \leq \frac{\phi\psi_{\ell}^2}{40}\tau_U^{\ell/L}$.
  With $\tau_U \leq \frac{\psi_0}{32z}\vol_{F_0,\Bc_G}(U_0)$, we additionally have $\Delta^{(U,\ell)}(t_{\mathrm{last}}^{(U,\ell)}, t) \leq \frac{\psi_{\ell}}{64}\vol_{F_0,\Bc_G}(U_0)$ and $\delta_{\mathrm{ext}}^{(U,\ell)}(t_{\mathrm{last}}^{(U,\ell)}, t)\leq \frac{\phi\psi_{\ell}^2}{800}\vol_{F_0,\Bc_G}(U_0)$.

  Therefore, against any update sequence, with high probability \cref{cond:blow-up} holds.
  Applying \cref{lemma:bound-from-external-cuts} on $W_{U,\ell}$ (with $R \defeq \frac{\psi_{\ell}}{10}\tau_{U}^{\ell/L}$, $\Delta \defeq \Delta^{(U,\ell)}(t_{\mathrm{last}}^{(U,\ell)}, t)$, and $\delta_{\mathrm{ext}}^{(U,\ell)}(t_{\mathrm{last}}^{(U,\ell)}, t)$), we conclude that $\|\Br_{U,\ell}\|_1 \leq \frac{4(R+\Delta)}{\psi_{\ell}} + \frac{8}{\psi_{\ell}^2\phi} \delta_{\mathrm{ext}} \leq \tau_U^{\ell/L}$.
  
  For the bound on $\tau_U$, we use our previous arguments when $\ell = L$.
  Notice that $\tau_U$ is set to $\frac{\psi_0^2}{64z}\vol_{F,\Bc_G}(U)$ at time $t_{\mathrm{last}}^{(U,L)}$ on Line~\ref{line:reset} in \cref{alg:maintain-exp-decomp}.
  As we have argued above, \cref{lemma:bound-from-external-cuts} applied on $W_{U,L}$ implies that the volume of the current $U$ is at least $\frac{64z}{\psi_0}\tau_U - \left(\frac{4(R+\Delta)}{\psi_{\ell}} + \frac{8}{\psi_{\ell}^2\phi}\right) \geq \frac{64z}{\psi_0^2}\tau_U - \tau_U \geq \frac{32z}{\psi_0^2}\tau_U$ (where $R \leq \frac{\psi_L}{10}\tau_U$, $\Delta \defeq \Delta^{(U,L)}(t_{\mathrm{last}}^{(U,L)}, t)$, and $\delta_{\mathrm{ext}} \defeq \delta_{\mathrm{ext}}^{(U,L)}(t_{\mathrm{last}}^{(U,L)}, t)$).
  This proves the upper bound of $\tau_U$.
  On the other hand, the volume of $U$ can increase by at most $\Delta^{(U,L)}(t_{\mathrm{last}}^{(U,L)}, t) \leq \frac{\psi_L}{10}\tau_U \leq \frac{\psi_0}{10}\tau_U$.
  This shows that $\tau_U \geq \frac{\psi_0^2}{128z}\vol_{F,\Bc_G}(U)$, which completes the proof of the lemma.
\end{proof}

To this end, we can now essentially conclude the correctness of our algorithm, except for the running time and output size guarantees that we will establish in \cref{subsec:recourse-and-runtime}.
Below we prove several useful properties that our algorithm satisfies.

\paragraph{Essential Properties of Algorithm.}

Let $X_0$ be the output of \alg{Init}{$G$} and let $X_i$ be the output of \alg{Cut}{$D_i$} where $D_i$ is the $i$-th update to $\cM$.
We first verify that the outputs of \textsc{Init($G$)} and \textsc{Cut($D$)} are consistent with the graph $G_{\cM}$ that our algorithm maintains internally (see \cref{def:hierarchy-maintainer} for the requirements).

\begin{observation}
  After the $i$-th call to \emph{\alg{Cut}{}}, the graph $G_{\cM}$ is equal to $G \setminus (X_0 \cup D_1 \cup X_1 \cup \cdots \cup D_i \cup X_i)$ with $\cU$ being the collection of its strongly connected components.
  \label{obs:correct-graph}
\end{observation}

\begin{observation}
  The graph $G_{\cM}$ we maintain satisfies $G_{\cM} \supseteq G_{\cM_{\prev}}$, and the terminal set $F$ we maintain satisfies $F \supseteq E(G_{\cM}) \setminus E(G_{\cM_{\prev}})$.
  \label{obs:maintain-terminal}
\end{observation}

\begin{observation}
  Whenever we call \cref{lemma:cut-matching} on Line~\ref{line:call-cut-matching} or \cref{lemma:prune-or-repair} on Line~\ref{line:call-prune} in \cref{alg:maintain-exp-decomp}, the hierarchy $\cH_{\prev}[U]$ is a $\phi_{\prev}$-expander hierarchy of the graph $G[U] \setminus F$.
  \label{obs:has-expander-hierarchy}
\end{observation}

\begin{lemma}
  \cref{alg:cut-add-terminal,alg:maintain-exp-decomp} maintain that $F$ is $\frac{\phi\psi_0^2}{2}$-expanding in $(G_{\cM},\Bc_G)$ after each update ($\cM.\Cut{D}$) against an adaptive adversary.
  \label{lemma:maintain-exp-decomp}
\end{lemma}

\begin{proof}
  After each update, for each $U \in \cU$ and $\ell \in \{0, \ldots, L\}$ we have that $W_{U,\ell}$ is not being rebuilt for all $\ell \in \{0, \ldots, L\}$.
  By \cref{lemma:probability-guarantee}, with high probability, if $\vol_{F,\Bc_G}(U) \geq 1/\phi$, then $\|\Br_{U,0}\| \leq 1$ which by \cref{claim:expanding-if-has-witness} implies $F$ is $\frac{\phi\psi_0^2}{2}$-expanding in $(G[U],\Bc_G)$ (note that by \cref{claim:reman-valid-witness} $W_{U,0}$ is indeed a valid witness of $(G[U], \Bc_G, F)$).
  On the other hand, if $\vol_{F,\Bc_G}(U) < 1/\phi$, then $F$ is also $\frac{\phi\psi_0^2}{2}$-expanding in $(G[U],\Bc_G)$ by \cref{claim:expanding-in-tiny-component}.
  This completes the proof.
\end{proof}

\begin{claim}
  On Line~\ref{line:call-cut} in \cref{alg:maintain-exp-decomp}, the set $D$ satisfies the input requirement of \emph{$\cM_{\prev}.\alg{Cut}{D}$} (see \cref{def:hierarchy-maintainer}), the call runs in $T_{\prev}(|U|)$ time, and it returns an edge set $A \subseteq E(G[U])$.
  \label{claim:cut-time}
\end{claim}

\begin{proof}
  By \cref{obs:maintain-terminal}, the graph $G_{\cM}$ is always a supergraph of $G_{\cM_{\prev}}$ which means that the $\SCC(G_{\cM_{\prev}})$ is a refinement of $\SCC(G_{\cM})$.
  Let $\cU^\prime \subseteq \SCC(G_{\cM_{\prev}})$ be the collection of SCCs of $G_{\cM_{\prev}}$ contained in $U$, where we must have that $U^\prime \defeq \bigcup_{U_i \in \cU^\prime}U_i = U$.
  Since $D$ is a cut $E_{G[U]}(S,U\setminus S)$ for some $S \subseteq U$, for each $U_i \in \cU^\prime$ if $D \cap G[U_i] \neq \emptyset$ it must be that $G \cap G[U_i] = E_{G[U_i]}(S_i, U_i \setminus S_i)$ for some $S_i$.
  Indeed, this will be the case if we let $S_i \defeq S \cap U_i$.
  This shows that the input requirement of $\cM_{\prev}.\alg{Cut}{D}$ is satisfied and also that the set $U_D$ defined in \cref{def:hierarchy-maintainer} is a subset of $U$.
  Therefore, $\cM_{\prev}.\alg{Cut}{D}$ runs in $T_{\prev}(|U|)$ time and returns an edge set $A \subseteq \bigcup_{U_i \in \cU^\prime}G[U_i] \subseteq G[U]$.
\end{proof}

\subsection{Bounding Expected Recourse and Running Time} \label{subsec:recourse-and-runtime}

It now remains to argue both the expected running time and the expected output size of the subroutines \textsc{Init} and \textsc{Cut}.
As the outputs of both these subroutines come from the \textsc{MaintainExpander} algorithm, we focus on establishing the guarantees for this function.

\paragraph{Good vs Bad States of Algorithm.}
In \cref{subsec:maintain-exp-decomp} we have shown that throughout the algorithm, with high probability, we have $\|\Br_{U,\ell}\| \leq \tau_U^{\ell/L}$ for all non-negligible $U$ that is not currently being rebuilt.
Let us define this as a \emph{good} state of the algorithm and consider the following generalizations of this notion.

\begin{definition}
  A state of the algorithm is \emph{good} if for all non-negligible $U \in \cU$ and $\ell \in \{0, \ldots, L\}$ for which $W_{U,\ell}$ is not currently being rebuilt it holds that $\|\Br_{U,\ell}\| \leq \tau_U^{\ell/L}$; and if $W_{U,L}$ is not currently being rebuilt then $\frac{\psi_0^2}{128z}\vol_{F,\Bc_G}(U) \leq \tau_U \leq \frac{\psi_0^2}{32z}\vol_{F,\Bc_G}(U)$ holds.
  The state is \emph{borderline} if $\|\Br_{U,\ell}\| \leq \frac{10}{\psi_{\ell}}\tau_U^{\ell/L}$ holds for these $U$ and $\ell$ instead; and if $W_{U,L}$ is not currently being rebuilt then $\frac{\psi_0^2}{128z}\vol_{F,\Bc_G}(U) \leq \tau_U \leq \frac{\psi_0^2}{32z}\vol_{F,\Bc_G}(U)$ holds.
  The state is \emph{bad} if it is not borderline.
  \label{def:states}\label{def:good-state}
\end{definition}

The reason why we need \cref{def:states} and in particular the definition of borderline states for our analysis is that, while \cref{lemma:good-event,lemma:probability-guarantee} showed that the algorithm is always in a good state with high probability, it is \emph{not true} that if we start from \emph{any} good state, then we always stay in good states with high probability.
Indeed, we can be in a good state in which one of the witnesses has $\|\Br_{U,\ell}\|_1$ being very close to $\tau_U^{\ell/L}$, yet there is a constant probability that we will not rebuild it in the next update (which leads us to a borderline state).
Note that this is not contradictory to \cref{lemma:good-event,lemma:probability-guarantee} because the probability that we are in such good states is small.
While it is not true that a good state remains good, we can generalize \cref{lemma:good-event,lemma:probability-guarantee} and show that a good state never reaches a \emph{bad} state with high probability.
We defer the proof of the following lemma to \cref{appendix:omitted-proofs} since it is essentially identical to that of \cref{lemma:good-event,lemma:probability-guarantee}.

\begin{restatable}{lemma}{RemainBorderline}
  Conditioned on the algorithm being in a good state at the current moment, with high probability, the algorithm will remain in borderline states until termination.
  \label{lemma:remain-borderline}
\end{restatable}

For simplicity of exposition, we will work with the quantities $\Size_{\ell}(b, \lambda)$ and $\Time_{\ell}(b, \lambda)$ that intuitively stand for the size of the cut \alg{MaintainExpander}{$U,\ell$} will return respectively the running time of the call, where $b = \vol_{F,\Bc_G}(U)$ and $\lambda = |U|$.
We formally define them as follows.
Consider $b \in \{0, \ldots, m\}$, $\lambda \in \{0, \ldots, n\}$, and $\ell \in \{0, \ldots, L\}$.
Let $\mathscr{B}_{b,\lambda,\ell}$ be the collection of all possible calls
of \alg{MaintainExpander}{$U, \ell$} that start when the algorithm is in a \emph{borderline} state with $\vol_F(U) \leq b$ and $|U| \leq \lambda$.
For each run $r \in \mathscr{B}_{b, \lambda,\ell}$ there are two random variables $S_{r} \in \{0, \ldots, m\}$\footnote{Recall that $m \leq n^4$ is the total capacities in the graph $G$.} and $T_{r} \in \{0, \ldots, \poly(n)\}$\footnote{It can be checked that the running time of the algorithm \emph{always} runs in polynomial time.} which represent the capacities $\Bc_G(X)$ of $X$ that this \alg{MaintainExpander}{$U, \ell$} call returns respectively the running time of it; the randomness comes from both the sampling of $k \sim \cR_{t,\tau}$ (on Line~\ref{line:sample-cut} in \cref{alg:cut-add-terminal} and Line~\ref{line:sample} in \cref{alg:maintain-exp-decomp}) and the expected guarantee of $\cM_{\prev}.\alg{Cut}{}$ (see \cref{def:hierarchy-maintainer}).
Let $\Omega$ be the space of randomness.
Let $S_r(\omega)$ and $T_r(\omega)$ be the realization of $S_r$ and $T_r$ on randomness $\omega$.
Additionally, we let $\widetilde{S}_r$ and $\widetilde{T}_r$ be random variables that act almost the same as $S_r$ and $T_r$ do, except if for some randomness $\omega$ the run $r \in \mathscr{B}_{b,\lambda,\ell}$ reaches a \emph{bad} state then we define $\widetilde{S}_r(\omega)$ and $\widetilde{T}_r(\omega)$ to both be realized to zero.

\begin{definition}
  We define $\Size_{\ell}(b, \lambda)$ to be the random variable on range $\{0, \ldots, m\}$ such that $\Size_{\ell}(b, \lambda)(\omega) \defeq \max_{r \in \mathscr{B}_{b,\lambda,\ell}}\widetilde{S}_r(\omega)$ for all $\omega \in \Omega$.
  Likewise, $\Time_{\ell}(b, \lambda)$ is a random variable on $\{0, \ldots, \poly(n)\}$ such that $\Time_{\ell}(b, \lambda)(\omega) \defeq \max_{r \in \mathscr{B}_{b,\lambda,\ell}}\widetilde{T}_r(\omega)$ for all $\omega \in \Omega$.
  \label{def:size-time}
\end{definition}

Note that by definition, $\Size_{\ell}(b, \lambda)$ and $\Time_{\ell}(b, \lambda)$ are increasing in both $b$ and $\lambda$.
The reason why we define $\Size_{\ell}(b, \lambda)$ and $\Time_{\ell}(b, \lambda)$ in terms of the rather bizarre-looking $\widetilde{S}_r$ and $\widetilde{T}_r$ is to ensure that we can assume the algorithm to be in a borderline state when doing the calculation which significantly simplifies matters.
To convert a bound on $\Size_{\ell}(b, \lambda)$ and $\Time_{\ell}(b, \lambda)$ to the expected guarantee of the \alg{Cut}{$D$} subroutine, we use \cref{lemma:remain-borderline} to argue that the contribution of bad states to the actual expectation can be made as small as $n^{-100}$ using the following fact since the probability of reaching those states from a good one is inverse polynomially small).

\begin{fact}
  If $Y$ is a random variables in $\{0, \ldots, \poly(n)\}$ and $\cE$ is an event that happens with high probability, then $\expect[Y] \leq \expect[Y \mid \cE] + n^{-100}$.
  \label{fact:expected-cond}
\end{fact}

This together with \cref{lemma:remain-borderline} implies the following.

\begin{lemma}
  If the algorithm is in a good state right before running \MaintainExpander{$U,\ell$}, then the subroutine runs in expected $\expect[\Time_{\ell}(\vol_{F,\Bc_G}(U),|U|)]+n^{-100}$ time and outputs an $X$ of expected size at most $\expect[\Size_{\ell}(\vol_{F,\Bc_G}(U), |U|)] + n^{-100}$.
  \label{lemma:from-size-to-actual-bound}
\end{lemma}

Again, we recall that by \cref{lemma:probability-guarantee}, with high probability, the algorithm is always in a good state when running \alg{Cut}{$D$} even against an adaptive adversary.
Consequently, it remains to bound $\expect[\Size_{\ell}(b, \lambda)]$ and $\expect[\Time_{\ell}(b, \lambda)]$.

\paragraph{An Overestimating Approach.}

To bound the expected value of $\Size_{\ell}(b, \lambda)$ and $\Time_{\ell}(b, \lambda)$, we consider a run of \alg{MaintainExpander}{$U,\ell$} and write down a recurrence that upper bounds them based on \cref{alg:maintain-exp-decomp}.
We often deliberately overestimate the expectation.
More specifically, while we are in a borderline state and might reach a bad state in the recursion which makes the realization of both $\Size_{\ell}(b, \lambda)$ and $\Time_{\ell}(b, \lambda)$ be defined as zero, we compute their values \emph{as if we are always in borderline states}.
Note that since the values of these random variables are non-negative, doing this can only overestimate their true values.
This greatly simplifies setting up the recurrence.%

\paragraph{Setup.}

To bound these quantities, consider the execution of \alg{MaintainExpander}{$U, \ell$} which finds cuts $S_1, \ldots, S_r$ through either \cref{lemma:cut-matching} or \cref{lemma:prune-or-repair} until we exit the while-loop by either successfully constructing/repairing the witness $W_{U,\ell}$ or sampling a $k > \ell$ on Line~\ref{line:sample}.
Observe that the final return set $X$ can be written as the union of four parts $X^{(1)} \cup X^{(2)} \cup X^{(3)} \cup X^{(4)}$, where
\begin{itemize}
  \item $X^{(1)}$ are edges added on Line~\ref{line:remove-edges},
  \item $X^{(2)}$ are edges found by recursively calling \alg{MaintainExpander}{$U,k$} on Line~\ref{line:early-return} (let us call this a \emph{fixing} recursion),
  \item $X^{(3)}$ are edges found by recursively calling \alg{MaintainExpander}{$U, \ell - 1$} on Line~\ref{line:call-below} (let us call this a \emph{downward} recursion),
  \item $X^{(4)}$ are edges found by recursively calling \alg{MaintainExpander}{$S, L$} on Line~\ref{line:build-small-pieces} for all the cuts $S$ (let us call this a \emph{rebuilding} recursion).
\end{itemize}
The running time of \alg{MaintainExpander}{$U, \ell$} can be similarly split into what it takes to compute each of the $X^{(i)}$'s.
We adapt the same notation as in \cref{scenarion:stability}:
We let $F_i$ be the terminal set after finding $S_i$ and running the $F \gets F \cup A$, and in particular $F_0$ is the initial set $F$ when \alg{MaintainExpander}{$U,\ell$} is called.
Let $U_0 \defeq U$ and $U_i \defeq U_{i-1} \setminus S_i$ and let $\Delta_i \leq 2\Bc_G(F_i \setminus F_{i-1})$ be the units of volume added due to $S_i$.
We first prove some basic properties regarding the volume of each $U_i$ and $S_i$.

\begin{lemma}
  For $\phi \leq \frac{1}{4\beta_{\prev}}$, we have
  \begin{enumerate}[(1)]
      \item 
  $\expect[\vol_{F_i,\Bc_G}(U_i) \mid F_{i-1}, U_{i-1}] \leq \vol_{F_{i-1},\Bc_G}(U_{i-1})$ for all $i \in [r]$,
  \item $\expect[\vol_{F_i,\Bc_G}(U_i)] \leq \vol_{F_0,\Bc_G}(U_0)$ for all $i \in [r]$,
  \item $\expect[\vol_{F_r,\Bc_G}(S_1) + \cdots + \vol_{F_r,\Bc_G}(S_r)] \leq (1+8\phi\beta_{\prev})\vol_{F_0,\Bc_G}(U_0)$, and
  \item $\expect[\Bc_G(X^{(1)})] \leq \phi \cdot \expect[\vol_{F_{0},\Bc_G}(S_1) + \cdots + \vol_{F_{r-1},\Bc_G}(S_r)]$.
  \end{enumerate}
  \label{lemma:expected-bounds}
\end{lemma}

\begin{proof}
  Let $D_i$ be the set $D$ defined on Line~\ref{line:D} when $S_i$ is found.
  By \cref{lemma:cut-matching,claim:is-sparse-cut}, we have $\Bc_G(D_i) \leq \phi \cdot \vol_{F_{i-1},\Bc_G}(S_i)$.
  Let $A_i$ be the set $A$ obtained by $A \gets \cM_{\prev}.\alg{Cut}{D_i}$ after $S_i$ is found.
  By the guarantee of $\cM_{\prev}$, we have $\expect[\Bc_G(A_i) \mid S_i] \leq \beta_{\prev} \cdot \Bc_G(D_i) \leq \beta_{\prev}\phi \cdot \vol_{F_{i-1},\Bc_G}(S_i)$.
  Since we are adding at most $2\Bc_G(A_i)$ units of volume to $U_{i-1}$ in total, we have
  \begin{equation}
    \expect[\vol_{F_i,\Bc_G}(U_i) \mid U_{i-1}, F_{i-1}, S_i] \leq \vol_{F_{i-1},\Bc_G}(U_{i-1}) - \vol_{F_{i-1},\Bc_G}(S_i) + 2\beta_{\prev}\phi \cdot \vol_{F_{i-1},\Bc_G}(S_i)
    \label{eq:vol-U}
  \end{equation}
  and
  \begin{equation}
    \expect[\vol_{F_i,\Bc_G}(S_i) \mid U_{i-1}, F_{i-1}] \leq \vol_{F_{i-1},\Bc_G}(S_i) + 2\beta_{\prev}\phi \cdot \vol_{F_{i-1},\Bc_G}(S_i).
    \label{eq:vol-S}
  \end{equation}
  As $2\beta_{\prev}\phi \leq 1/2$, we see that in expectation $\vol_{F_i,\Bc_G}(U_i)$ is smaller than $\vol_{F_{i-1},\Bc_G}(U_{i-1})$.
  This proves (1), and (2) simply follows from (1) and the law of total expectation.
  
  For (3), we add a $\left(\frac{1-2\beta_{\prev}\phi}{1+2\beta_{\prev}\phi}\right)$-multiple of \eqref{eq:vol-S} to \eqref{eq:vol-U} and get
  \[
    \expect\left[\left(\frac{1-2\beta_{\prev}\phi}{1+2\beta_{\prev}\phi}\right)\vol_{F_i,\Bc_G}(S_i) + \vol_{F_{i},\Bc_G}(U_{i})  - \vol_{F_{i-1},\Bc_G}(U_{i-1}) \mathrel{\Big|} U_{i-1}, F_{i-1}, S_i\right] \leq 0
  \]
  which implies
  \[
    \expect\left[\sum_{i \in [r]}\vol_{F_i,\Bc_G}(S_i)\right] \leq \frac{1+2\beta_{\prev}\phi}{1-2\beta_{\prev}\phi} \cdot \vol_{F_0,\Bc_G}(U_0) \leq (1+8\beta\phi)\vol_{F_0,\Bc_G}(U_0)
  \]
  for $2\beta_{\prev}\phi \leq \frac{1}{2}$.
  The bound on $\expect[\vol_{F_r,\Bc_G}(S_1) + \cdots + \vol_{F_r,\Bc_G}(S_r)]$ then follows from the observation that $\vol_{F_r,\Bc_G}(S_i) = \vol_{F_i,\Bc_G}(S_i)$ for all $i \in [r]$, since the $A_j \subseteq U_{j-1}$ which is disjoint from $S_i$ for $j > i$. %
  Finally, (4) follows simply from \cref{claim:is-sparse-cut} and the linearity of expectation.
\end{proof}

Let us use uppercase letters (e.g., $B, \Lambda$) to denote random variables and lowercase letters (e.g., $b, \lambda$) to denote their realizations.
We first derive a recurrence of $\Size_{\ell}(b, \lambda)$ and $\Time_\ell(b, \lambda)$ for the topmost layer $\ell = L$.

\begin{lemma}
  For $\phi < O\left(\frac{\psi_0^3}{\beta_{\prev}\log^3 n}\right)$ sufficiently small, we have
  \begin{equation}
  \begin{split}
  \expect[\Size_L(b, \lambda)] \leq 3\phi b &+
  \max_{\cD_1}\expect_{B^\prime \sim \cD_1}\left[\expect[\Size_{L-1}(B^\prime, \lambda) \mid B^\prime]\right] \\
  &+ \max_{r \in \{0, \ldots, n\}}\max_{\cD_{2,r}}\max_{\lambda_1,\ldots,\lambda_r}\expect_{(B_1, \ldots, B_r) \sim \cD_{2,r}}\left[\sum_{i \in [r]} \expect[\Size_L(B_i, \lambda_i) \mid B_i]\right]
  \label{eq:rec-size-L-complex}
  \end{split}
  \end{equation}
  and
  \begin{equation}
  \begin{split}
  \expect[\Time_L(b, \lambda)] &\leq \widetilde{O}\left(\frac{1}{\psi_0^2}\right) \cdot \left(\widetilde{O}\left(\frac{\lambda^2}{\phi\phi^\prime}\right) + T_{\prev}(\lambda)\right) \\ &+ \max_{\cD_1}\expect_{B^\prime \sim \cD_1}[\expect[\Time_{L-1}(B^\prime, \lambda) \mid B^\prime]] \\ &+ \max_{r \in \{0, \ldots, n\}}\max_{\cD_{2,r}}\max_{\lambda_1, \ldots, \lambda_r}\expect_{(B_1,\ldots,B_r) \sim \cD_{2,r}}\left[\sum_{i \in [r]}\expect[\Time_L(B_i, \lambda_i) \mid B_i]\right],
  \end{split}
  \label{eq:rec-time-L-complex}
  \end{equation}
  where $\cD_{1}$ iterates over all distributions on $\{0, \ldots, m\}$ with $\expect_{B^\prime \in \cD_{1}}[B^\prime] \leq b$,
  $\cD_{2,r}$ iterates over all distributions on $\{0, \ldots, m\}^r$ such that $\expect_{B_1, \ldots, B_r \sim \cD_{2,r}}[B_i] \leq \frac{3}{4}b$ and $\expect_{B_1, \ldots, B_r \sim \cD_{2,r}}[B_1 + \cdots + B_r] \leq (1+8\phi\beta_{\prev}) B$, and
  $\lambda_1, \ldots, \lambda_r$ iterate over all such sequences with $\lambda_1 + \cdots + \lambda_r < \lambda$.
  \label{lemma:rec-L-complex}
\end{lemma}

To avoid the cumbersome expressions as in \cref{lemma:rec-L}, in the remainder of the section we will slightly overload the notation and simply write, e.g., $\expect_{B^\prime}$ without the preceding $\max_{\cD_1}$, and later in the description of lemma specifies the properties that $\cD_1$ has to satisfy.
This means that we in fact consider all qualifying distributions and pick the one that maximizes the expression.
As a concrete example, we will rewrite \cref{lemma:rec-L-complex} in the following form. %

\begin{lemma}[\cref{lemma:rec-L-complex} restated]
  For $\phi < O\left(\frac{\psi_0^3}{\beta_{\prev}\log^3 n}\right)$ sufficiently small, we have
  \begin{equation}
  \begin{split}
  \expect[\Size_L(b, \lambda)] \leq 3\phi b &+
  \expect_{B^\prime}\left[\expect[\Size_{L-1}(B^\prime, \lambda) \mid B^\prime]\right] \\
  &+ \max_{r, \lambda_1,\ldots,\lambda_r}\expect_{B_1, \ldots, B_r}\left[\sum_{i \in [r]} \expect[\Size_L(B_i, \lambda_i) \mid B_i]\right]
  \label{eq:rec-size-L}
  \end{split}
  \end{equation}
  and
  \begin{equation}
  \begin{split}
  \expect[\Time_L(b, \lambda)] &\leq \widetilde{O}\left(\frac{1}{\psi_0^3}\right) \cdot \left(\widetilde{O}\left(\frac{\lambda^2}{\phi\phi_{\prev}^2}\right) + T_{\prev}(\lambda)\right) \\ &+ \expect_{B^\prime}[\expect[\Time_{L-1}(B^\prime, \lambda) \mid B^\prime]]  + \max_{r, \lambda_1, \ldots, \lambda_r}\expect_{B_1,\ldots,B_r}\left[\sum_{i \in [r]}\expect[\Time_L(B_i, \lambda_i) \mid B_i]\right],
  \end{split}
  \label{eq:rec-time-L}
  \end{equation}
  where $B^\prime$ is a random variable satisfying $\expect[B^\prime] \leq b$,
  $B_1, \ldots, B_r$ are random variables satisfying $\expect[B_i] \leq \frac{3}{4}b$ and $\expect[B_1 + \cdots + B_r] \leq (1+8\phi\beta_{\prev})b$, and
  $\lambda_1, \ldots, \lambda_r$ iterate over all such sequences with $\lambda_1 + \cdots + \lambda_r < \lambda$.
  \label{lemma:rec-L}
\end{lemma}

\begin{proof}
  We first bound $\Size_L(b, \lambda)$, and by the linearity of expectation it suffices to bound each of the $X^{(1)}$, $X^{(2)}$, $X^{(3)}$, and $X^{(4)}$.
  \begin{itemize}
  \item
  By \cref{lemma:expected-bounds}, $X^{(1)}$ has expected total capacities at most $3\phi b$.
  \item
  The edge set $X^{(2)}$ is empty for \alg{MaintainExpander}{$U, L$} as the $k$ sampled on Line~\ref{line:sample} is never greater than $L$.
  \item
  The expected total capacities of $X^{(3)}$ is $\expect_{B^\prime}[\expect[\Size_{L-1}(B^\prime, \lambda) \mid B^\prime]]$ where $B^\prime$ a random variable indicating the volume of the set $U$ when calling \alg{MaintainExpander}{$U, L - 1$} on Line~\ref{line:call-below}.
  By \cref{lemma:expected-bounds}, we have $\expect[B^\prime] \leq B$ (the volume might, with low probability, increase since the set of terminals `$F$' increase).
  \item
  Finally, for $X^{(4)}$, its total capacities is equal to $\sum_{i \in [r]}\expect[\Size(b_i, \lambda_i, L)]$, where $b_i$ is the volume of the cut $S_i$ and $\lambda_i$ is the number of vertices in $S_i$.
  By \cref{lemma:expected-bounds}, if we let $B_i$ be the random variables for $b_i$, then we have $\expect_{B_1,\ldots,B_r}[B_1 + \cdots + B_r] \leq (1+8\phi\beta)B$.
  Moreover, as the $S_i$'s are vertex-disjoint and are proper cuts, we have $\lambda_1 + \cdots + \lambda_r < \lambda$.
  It remains to prove the expected marginal of each $B_i$.
  By \eqref{eq:vol-S} and that $2\beta\phi \leq \frac{1}{2}$, we know that using the volume upper-bound of \cref{lemma:cut-matching}
  \[
    \expect[\vol_{F_i,\Bc_G}(S_i) \mid U_{i-1}, F_{i-1}] \leq \frac{3}{2}\vol_{F_{i-1},\Bc_G}(S_i) \leq \frac{3}{4}\vol_{F_{i-1},\Bc_G}(U_{i-1}).
  \]
  By the law of total expectation and \cref{lemma:expected-bounds} this implies $\expect[\vol_{F_i,\Bc_G}(S_i)] \leq \frac{3}{4}\vol_{F_0,\Bc_G}(U_0)$ or equivalently $\expect[B_i] \leq \frac{3}{4}B$.
  \end{itemize}
  This proves \eqref{eq:rec-size-L}.
  We likewise bound each of the four parts of $\Time_L(b, \lambda)$ and then apply the linearity of expectation to derive \eqref{eq:rec-time-L}.
  Note that among them it suffices to bound the time spent to compute $X^{(1)}$ as the other parts follow the same recurrences as their counterparts in \eqref{eq:rec-size-L}.
  For this we bound the number of iterations the while-loop in \alg{MaintainExpander}{$U, L$} takes and then use the fact that the time spent in each iteration is dominated by calling \alg{CutOrEmbed}{} (which by \cref{lemma:cut-matching} takes $O(\frac{\lambda^2}{\phi\phi^\prime})$ time) and running $\cM_{\prev}.\alg{Cut}{D}$ (which by \cref{claim:cut-time} takes $T_{\prev}(\lambda)$ time).\footnote{Indeed, since a witness on $U$ has at most $|U|^2$ edges, the time it takes to update the witnesses in \alg{UpdateWitness}{} can be easily charged to the time spent in constructing/repairing them.}
  By \cref{lemma:cut-matching}, we have $\vol_{F_{i-1},\Bc_G}(S_i) \geq \frac{1}{4t_{\CMG}} \cdot \frac{\psi_L}{10}\tau_U$ where $\tau_U = \frac{\psi_0^2}{64z}\vol_{F_{i-1},\Bc_G}(U_{i-1})$ according to Line~\ref{line:call-cut-matching} in \cref{alg:maintain-exp-decomp}.
  This shows that $\vol_{F_{i-1},\Bc_G}(S_i) \geq \Omega\left(\frac{\psi_0^3}{\log^3 n}\right)\vol_{F_{i-1},\Bc_G}(U_{i-1})$,
  which by \eqref{eq:vol-U} gives
  \[
    \expect[\vol_{F_i,\Bc_G}(U_i) \mid U_{i-1}, F_{i-1}] \leq \left(1-\Omega\left(\frac{\psi_0^3}{\log^3 n}\right)\right)\vol_{F_{i-1},\Bc_G}(U_{i-1})
  \]
  since $\phi < O\left(\frac{\psi_0^3}{\beta_{\prev} \log^3 n}\right)$ is sufficiently small.
  By the law of total expectation, this shows that
  \[
    \expect[\vol_{F_t,\Bc_G}(U_t)] \leq \left(1-\Omega\left(\frac{\psi_0^3}{\log^3 n}\right)\right)^t \cdot  b
  \]
  and in particular $\expect[\vol_{F_t,|Bc_G}(U_t)] \leq 1/\lambda$ for some $t = \widetilde{\Theta}\left(\frac{1}{\psi_0^3}\right)$.
  By Markov's inequality, this means that $\Pr[\vol_{F_t,\Bc_G}(U_t) > 0] \leq 1/\lambda$ (note that $\vol_{F_t,\Bc_G}(U_t)$ is a nonnegative integer).
  As the number of iterations is always bounded by $\lambda$ because each cut removes at least one vertex from $U$, the expected number of iterations is $\widetilde{\Theta}\left(\frac{1}{\psi_0^3}\right) + \frac{1}{\lambda} \cdot \lambda \leq \widetilde{\Theta}\left(\frac{1}{\psi_0^3}\right)$.
  The expected time spent in computing $X^{(1)}$ is thus, by the linearity of expectation, bounded by $\widetilde{O}\left(\frac{1}{\psi_0^3}\right) \cdot \left(\widetilde{O}\left(\frac{\lambda^2}{\phi{\phi_{\prev}}^2} + T_{\prev}(\lambda)\right)\right)$.
  This proves \eqref{eq:rec-time-L}.
\end{proof}

We can similarly write a recurrence for $\Size_{\ell}(b, \lambda)$ and $\Time_{\ell}(b, \lambda)$ for $\ell < L$.
Recall from the beginning of this subsection that the recurrence we establish will deliberately overestimate some quantities.
In other words, in case that we recurse to a bad state of the algorithm, we are still going to write the recurrence as if we are in a borderline state.

\begin{lemma}
  For $\phi < O\left(\frac{\psi_0^3}{\beta_{\prev}\log^3 n}\right)$ sufficiently small and $\ell < L$, we have
  \begin{equation}
  \begin{split}
  \expect[\Size_{\ell}(b,\lambda)] \leq &O\left(\frac{\phi}{\psi_0^2}\right) b^{(\ell+1)/L}
    + \widetilde{O}\left(\frac{1}{\psi_0^4}\right)  \cdot \expect_{\Delta}\left[\sum_{k = \ell+1}^{L}\min\left\{1, \frac{\Delta}{b^{k/L}}\right\} \cdot \expect[\Size_k(b+\Delta, \lambda - 1) \mid \Delta]\right] \\
    &+ \expect_{B^\prime}\left[\expect[\Size_{\ell-1}(B^\prime, \lambda) \mid B^\prime]\right] + \max_{r,\lambda_1,\ldots,\lambda_r}\expect_{B_1,\ldots,B_r}\left[\sum_{i \in [r]}\expect[\Size_{L}(B_i, \lambda_i) \mid B_i]\right]
  \end{split}
  \label{eq:rec-size-l}
  \end{equation}
  and
  \begin{equation}
  \begin{split}
  \expect[\Time_{\ell}(b, \lambda)] &\leq \widetilde{O}\left(\frac{1}{\psi_0^5}\right) \cdot m^{1/L} \cdot \left(\widetilde{O}\left(\frac{\lambda^2}{\phi\phi_{\prev}^2\psi_0^4}\right) + T_{\prev}(\lambda)\right) \\
  &+ \widetilde{O}\left(\frac{1}{\psi_0^3}\right)  \cdot \expect_{\Delta}\left[\sum_{k = \ell+1}^{L}\min\left\{1, \frac{\Delta}{b^{k/L}}\right\} \cdot \expect[\Time_k(b+\Delta, \lambda - 1) \mid \Delta]\right] \\
  &+ \expect_{B^\prime}\left[\expect[\Time_{\ell-1}(B^\prime, \lambda) \mid B^\prime]\right] + \max_{r,\lambda_1, \ldots, \lambda_r}\expect_{B_1,\ldots,B_r}\left[\sum_{i \in [r]}\expect[\Time_L(B_i, \lambda_i) \mid B_i]\right],
  \end{split}
  \label{eq:rec-time-l}
  \end{equation}
  where $\Delta$ is a random variable with $\expect[\Delta] \leq O\left(\frac{\phi\beta_{\prev}}{\psi_0^2}\right) b^{(\ell+1)/L}$, $B^\prime$ is a random variable with $\expect[B^\prime] \leq b$, $B_1, \ldots, B_r$ satisfy $\expect[B_1 + \cdots + B_r] \leq O\left(\frac{1}{\log^{3L} n}\right)b^{(\ell+1)/L}$, and $\lambda_1, \ldots, \lambda_r$ satisfy $\sum_{i \in [r]}\lambda_i < \lambda$.
  \label{lemma:rec-l}
\end{lemma}

\begin{proof}
  Recall that $\Delta_i$ is the units of volume added after the $i$-th $S_i$ is found in the while-loop of \alg{MaintainExpander}{$U, \ell$}.
  Conditioned on the total number $r$ of the cuts found and the values of $\Delta_1, \ldots, \Delta_r$, let $\Delta \defeq \Delta_1 + \cdots + \Delta_r$.
  If we call \alg{MaintainExpander}{$U, k$} for some $k > \ell$ on Line~\ref{line:early-return}, then we can upper-bound the current volume of $U$ by $b + \Delta$.
  Therefore, the expected total capacities of $X^{(2)}$ is bounded by
  \begin{equation}
  \begin{split}
    \sum_{k = \ell+1}^{L}\min\Bigg\{1, \frac{\Delta_r}{\psi_0^2\tau_U^{\ell/L}} &\cdot c_{\mathrm{rb}}\ln n\Bigg\} \cdot \expect[\Size_k(b+\Delta, \lambda - 1) \mid \Delta] \\
      &\leq \sum_{k = \ell+1}^{L}\min\left\{1, \frac{\Delta}{\psi_0^2\tau_U^{\ell/L}} \cdot c_{\mathrm{rb}}\ln n\right\} \cdot \expect[\Size_k(b+\Delta, \lambda - 1) \mid \Delta]
  \end{split}
  \label{eq:rec-higher-level}
  \end{equation}
  since the new $U$ contains at most $\lambda - 1$ vertices after removing at least one cut from it.
  By \cref{def:size-time}, the algorithm is in a borderline state when it enters the current \alg{MaintainExpander}{$U,\ell$} which by \cref{def:good-state} means that $\frac{\psi_0^2}{128z} b \leq \tau_U \leq \frac{\psi_0^2}{32z} b$.
  Therefore, \eqref{eq:rec-higher-level} is further upper-bounded by (if we move the $\psi_0$-term in the denominator out to the beginning)
  \begin{equation}
    \widetilde{O}\left(\frac{1}{\psi_0^4}\right) \cdot \sum_{k = \ell+1}^{L}\min\left\{1, \frac{\Delta}{b^{k/L}}\right\} \cdot \expect[\Size_k(b+\Delta, \lambda - 1) \mid \Delta].
    \label{eq:rec-higher-level-new}
  \end{equation}
  
  As the bounds on $\Bc_G(X^{(3)})$ and $\Bc_G(X^{(4)})$ are the same as in \cref{lemma:rec-L}, it remains to (i) bound $\Bc_G(X^{(1)})$ which by lcref{lemma:expected-bounds} is at most $\phi \cdot \expect[\vol_{F_{0},\Bc_G}(S_1) + \cdots + \vol_{F_{r-1},\Bc_G}(S_r)]$ in expectation and (ii) use the bound to prove the expected value of $\Delta$.
  
  By \cref{def:good-state}, we know that $\|\Br_{U,\ell+1}\|_1 \leq \frac{10}{\psi_{\ell+1}}\tau_U^{(\ell+1)/L}$ holds in the beginning of this current run of \alg{MaintainExpander}{$U, \ell$}.
  Observe that what happened before we recurse on some \alg{MaintainExpander}{$X, k$} is modeled by \cref{scenarion:stability} with $\delta_{\mathrm{ext}} = 0$.
  Similar to the proof of \cref{lemma:probability-guarantee}, conditioned on there being $r$ cuts (i.e., none of the first $r - 1$ cuts sampled a $k$ larger than $\ell$ on Line~\ref{line:sample}), the probability that $\Delta_1 + \cdots + \Delta_{r-1} > \tau_U^{(\ell+1)/L}$ is bounded by
  \[
    \prod_{i \in [r-1]}\left(1 - \min\left\{1, \frac{\Delta_i}{\psi_0^2 \tau_U^{(\ell+1)/L}} \cdot c_{\mathrm{rb}}\ln n\right\}\right) \leq \exp\left(-\sum_{i \in [r-1]}\Delta_i \cdot \frac{c_{\mathrm{rb}}\ln n}{\psi_0^2 \tau_U^{(\ell+1)/L}}\right) \leq n^{-c_{\mathrm{rb}}}.
  \]
  In particular, with high probability, we can apply \cref{lemma:bound-from-external-cuts} on $W_{U,\ell+1}$ at the moment right after finding $S_r$ but before we added the $\Delta_r$ units of volume.
  In other words, we pretend that the last cut $S_r$ adds zero units of volume to $U$.
  Formally speaking, letting $\widetilde{\Delta}_i = \Delta_i$ for $i < r$ and $\widetilde{\Delta}_r = 0$, we can verify that \cref{cond:blow-up} holds with $R = \frac{10}{\psi_{\ell+1}}\tau_U^{(\ell+1)/L} \leq \frac{10\psi_0}{32z}\vol_{F_0,\Bc_G}(U_0) \leq \frac{\psi_{\ell+1}}{64}\vol_{F_0,\Bc_G}(U_0)$ and $\widetilde{\Delta} \defeq \widetilde{\Delta}_1 + \cdots + \widetilde{\Delta}_r = \Delta_1 + \cdots + \Delta_{r-1} \leq \tau_U^{(\ell+1)/L} \leq \frac{\psi_{0}^2}{32z}\vol_{F_0,\Bc_G}(U_0) \leq \frac{\psi_{\ell+1}}{64}\vol_{F_0,\Bc_G}(U_0)$.
  Thus, \cref{lemma:bound-from-external-cuts} shows that $\vol_{F_{r-1},\Bc_G}(S_1) + \cdots + \vol_{F_{r-1},\Bc_G}(S_r) \leq \frac{80}{\psi_{\ell+1}^2}\tau_U^{(\ell+1)/L}$.
  As a result, conditioned on $r$ and $\Delta_1 + \cdots + \Delta_{r-1} \leq \tau_U^{(\ell+1)/L}$, the size of $X^{(1)}$ is bounded by $\phi \cdot O\left(\frac{1}{\psi_0^2}\right) \cdot \tau_U^{(\ell+1)/L} \leq \phi \cdot O\left(\frac{1}{\psi_0^2}\right) \cdot b^{(\ell+1)/L}$ in expectation.
  As $\Bc_G(X^{(1)}) \leq m$ always hold, this implies via the following \cref{fact:expected-cond} that after removing the conditioning of $\Delta_1 + \cdots + \Delta_{r-1}$ we still have $\expect[\Bc_G(X^{(1)}) \mid r] \leq \phi \cdot O\left(\frac{1}{\psi_0^2}\right) \cdot b^{(\ell+1)/L} + n^{-100}$.
  By the law of total expectation, this implies unconditionally $\expect[\Bc_G(X^{(1)})] \leq \phi \cdot O\left(\frac{1}{\psi_0^2}\right) \cdot b^{(\ell+1)/L} + n^{-100}$.
  
  Moreover, notice that by \eqref{eq:vol-S} we have $\expect[\vol_{F_i,\Bc_G}(S_i)] \leq \frac{3}{2}\vol_{F_{i-1},\Bc_G}(S_i)$ and since $\vol_{F_r,\Bc_G}(S_i) = \vol_{F_i,\Bc_G}(S_i) \leq \frac{3}{2}\vol_{F_{i-1},\Bc_G}(S_i)$ we further have
  \begin{align*}
    \expect[\vol_{F_r,\Bc_G}(S_1) + \cdots + \vol_{F_r,\Bc_G}(S_r)] \leq \frac{120}{\psi_{\ell+1}^2}\tau_U^{(\ell+1)/L} &\leq 120\frac{(\psi_0 b)^{(\ell+1)/L}}{\psi_{\ell+1}^2} \\ &\leq 120b^{(\ell+1)/L} \cdot \frac{\psi_{0}^{1/L}}{\psi_{\ell+1}^2} \leq O\left(\frac{1}{\log^{3L} n}\right) b^{(\ell+1)/L}
  \end{align*}
  by \eqref{eq:compare-psi}.
  This proves the bound on $\expect[B_1 + \cdots + B_r]$.
  
  It remains to bound the expected value of $\Delta$ to finish our bound on $\Bc_G(X^{(2)})$ via the expression \eqref{eq:rec-higher-level-new} that we developed.
  For this we use the guarantee of $\cM_{\prev}$ and that we add at most two units of volume per edge returned by $\cM_{\prev}.\alg{Cut}{D}$ which implies $\expect[\Delta] \leq 2\beta_{\prev} \expect[\Bc_G(X^{(1)})] \leq O\left(\frac{\phi\beta_{\prev}}{\psi_0^2}\right)b^{(\ell+1)/L} + n^{-99} \leq O\left(\frac{\phi\beta_{\prev}}{\psi_0^2}\right)b^{(\ell+1)/L}$.
  This proves \eqref{eq:rec-size-l}.

  For the recurrence \eqref{eq:rec-time-l} of $\Time_{\ell}(b, \lambda)$, similar to \cref{lemma:rec-L}, it suffices to bound the number of iterations in the while-loop in \alg{MaintainExpander}{$U, \ell$}.
  Again, conditioned on $\Delta_1 + \cdots + \Delta_{r-1} \leq \tau_U^{(\ell+1)/L}$, we have $\vol_{F_{0},\Bc_G}(S_1) + \cdots + \vol_{F_{r-1},\Bc_G}(S_{r-1}) \leq \frac{8}{\psi_{\ell+1}^2}\tau_U^{(\ell+1)/L}$.
  Yet, \cref{claim:is-sparse-cut} suggests that $\vol_{F_{i-1},\Bc_G}(S_i) \geq \frac{\psi_{\ell+1}^2\psi_{\ell}}{640z}\tau_U^{\ell/L}$ holds for each $i \in [r]$.
  Therefore, the number of iterations $r$ can be bounded by $\widetilde{O}\left(\frac{1}{\psi_0^5}\right) \cdot m^{1/L}$, where each iteration by \cref{lemma:prune-or-repair} takes $\widetilde{O}\left(\frac{\lambda^2}{\phi\phi_{\prev}^2\psi_0^4}\right) + T_{\prev}(\lambda)$ time.
  This proves \eqref{eq:rec-time-l}.
\end{proof}

\subsubsection{Solving the Recurrences}

Having established \cref{lemma:rec-L,lemma:rec-l}, we can now solve the recurrences.
In all of the recurrences, the only part not solvable by a simple induction is when we recurse on \alg{MaintainExpander}{$S, L$} (e.g., the $\max_{r,\lambda_1,\ldots,\lambda_r}\expect_{B_1,\ldots,B_r}\left[\sum_{i \in [r]}\expect[\Size_L(B_i, \lambda_i)]\right]$ term in \eqref{eq:rec-size-L}) since the expected total volume of the cuts can grow larger than what we start with.
We abstract this tricky part of the recurrence and handle it by proving the following lemma which exploits the fact that $\expect[B_i] \leq \frac{3}{4}b$.

\begin{lemma}
  For random functions $f, g: \{0, \ldots, \poly(n)\}^2 \to \{0, \ldots, \poly(n)\}$ with $f(0, \cdot) = f(\cdot, 0) = g(0, \cdot) = g(\cdot, 0) = 0$ that admit a recurrence relationship of the form
  \begin{equation}
    \expect[f(b, \lambda)] \leq \expect[g(b, \lambda)] + \max_{r,\lambda_1,\ldots,\lambda_r}\expect_{B_1,\ldots,B_r}\left[\sum_{i \in [r]}\expect[f(B_i, \lambda_i) \mid B_i]\right]
    \label{eq:rec}
  \end{equation}
  where $B_1, \ldots, B_r$ are random variables satisfying $\expect[B_i] \leq \frac{3}{4}b$, $\expect[B_1 + \cdots + B_r] \leq (1+\gamma) b$ for some $\gamma < O\left(\frac{1}{\log n}\right)$ sufficiently small and $\lambda_1, \ldots, \lambda_r$ satisfy $\lambda_1 + \cdots + \lambda_r < \lambda$, we have
  \begin{equation}
    \expect[f(b, \lambda)] \leq \expect[g(b, \lambda)] + \max_{p, \widetilde{\lambda}_1, \ldots, \widetilde{\lambda}_p} \expect_{\widetilde{B}_1, \ldots, \widetilde{B}_p}\left[\sum_{i \in [p]}\expect[g(\widetilde{B}_i, \widetilde{\lambda}_i) \mid \widetilde{B}_i]\right] + n^{-100},
    \label{eq:rec-simp}
  \end{equation}
  where $\widetilde{B}_1, \ldots, \widetilde{B}_p$ are random variables satisfying $\expect[\widetilde{B}_1 + \cdots + \widetilde{B}_p] \leq O(\log n)b$ and $\widetilde{\lambda}_1, \ldots, \widetilde{\lambda}_p$ satisfy $\widetilde{\lambda}_1 + \cdots + \widetilde{\lambda}_p \leq O(\log n) \lambda$ and $\widetilde{\lambda}_i < \lambda$.\footnote{In principle we could put the $\expect[g(b,\lambda)]$ part into the max expression by setting $\widetilde{B}_0 = b$ deterministically and $\widetilde{\lambda}_0 = \lambda$. The reason why there is a standalone term is to ensure that we can later apply induction on $\lambda$ for the terms in the max expression. In some future cases this may not be required (e.g., when we already have a parameter decrease that facilitates induction), we may for simplicity move the standalone term into the max expression.}
  \label{lemma:simplify-rec}
\end{lemma}

\begin{proof}
  Let us expand the recurrence of \eqref{eq:rec} and consider its recursion tree.
  We mark each node $v$ in the tree with the values of $b$ and $\lambda$, denoted by $b_v$ and $\lambda_v$, that it corresponds to.
  Note that $b_v$ is not deterministic; rather it is a realization of a random variable that we denote by $B_v$.
  Observe that it suffices to consider the part of the tree that corresponds to recursion from $f$ to $f$ and then in the end sum over the $g(b_v, \lambda_v)$ value for all nodes in the tree. 
  Furthermore, we can ignore all nodes with either $b_v = 0$ or $\lambda_v = 0$ as both functions evaluate to zero on them.
  
  By $\expect[B_i] \leq \frac{3}{4}b$ and the law of total expectation, we know that for some $d = \Theta(\log n)$ all the depth-$d$ nodes have $\expect[B_v] \leq (3/4)^d \cdot b \leq n^{-c}$ where $c > 0$ is an arbitrarily large but fixed constant.
  As the sum of the $\lambda_v$'s decreases by at least one in each level, the number of nodes in the whole tree is bounded by $\lambda^2 \leq \poly(n)$ (there are at most $\lambda$ levels, each consisting of at most $\lambda$ nodes).
  Therefore, by Markov's inequality and a union bound, with high probability all depth-$d$ nodes have $b_v = 0$, and we can safely ignore them and truncate the tree to have depth $d - 1$.
  On the other hand, by $\expect[B_1 + \cdots + B_r] \leq (1+\gamma)b$ and again the law of total expectation, the expected value of the sum of $B_v$'s of all depth-$t$ nodes are bounded by $(1+\gamma)^t b$.
  For $\gamma < O\left(\frac{1}{\log n}\right)$ sufficiently small and $t < d$, the quantity $(1+\gamma)^t b$ is bounded by $2b$.
  Therefore, the expected value of the sum of $B_v$'s in the whole recursion tree is bounded in expectation by $O(\log n)b$.
  Similarly, the sum of $\lambda_v$'s is bounded by $O(\log n)\lambda$.
  This shows that conditioned on the event that the tree has depth at most $d = O(\log n)$ (which happens with high probability) we have
  \[
    \expect[f(b, \lambda) \mid \text{tree has depth $d = O(\log n)$}] \leq \expect[g(b, \lambda)] + \max_{p, \widetilde{\lambda}_1, \ldots, \widetilde{\lambda}_p} \expect_{\widetilde{B}_1, \ldots, \widetilde{B}_p}\left[\sum_{i \in [p]}\expect[g(\widetilde{B}_i, \widetilde{\lambda}_i) \mid \widetilde{B}_i]\right],
  \]
  where each $\widetilde{B}_i$ and $\lambda_i$ corresponds to the $B_v$ and $\lambda_v$ for some non-root $v$ in the tree.
  The lemma now follows by \cref{fact:expected-cond} since both $f$ and $g$ are polynomially bounded.
\end{proof}

We can now simplify \cref{lemma:rec-L} via \cref{lemma:simplify-rec}.
Indeed, the value of $\gamma \defeq 8\phi\beta_{\prev}$ as in \cref{lemma:rec-L} is sufficiently smaller than $O\left(\frac{1}{\log n}\right)$ when $\phi < O\left(\frac{\psi_0^3}{\beta_{\prev}\log^3 n}\right)$.

\begin{corollary}
  For $\phi < O\left(\frac{\psi_0^3}{\beta\log^3 n}\right)$ sufficiently small, we have
  \begin{equation}
  \begin{split}
  \expect[\Size_L(b, \lambda)] &\leq 3\phi \cdot b
  + \expect_{B^\prime}\left[\expect[\Size_{L-1}(B^\prime, \lambda) \mid B^\prime]\right] \\ 
  &+ \max_{p, \widetilde{\lambda}_1, \ldots, \widetilde{\lambda}_p}\expect_{\widetilde{B}_1, \ldots, \widetilde{B}_p}\left[\sum_{i \in [p]}\expect\left[3\phi \cdot \widetilde{B}_i + \expect_{B_i^\prime}\left[\expect[\Size_{L-1}(B_i^\prime, \widetilde{\lambda}_i) \mid B_i^\prime] \mid \widetilde{B}_i\right]\right]\right] + n^{-100}
  \label{eq:rec-size-L-simp}
  \end{split}
  \end{equation}
  and 
  \begin{equation}
  \begin{split}
  \expect[\Time_L(b, \lambda)] &\leq \widetilde{O}\left(\frac{1}{\psi_0^3}\right) \cdot \left(\widetilde{O}\left(\frac{\lambda}{\phi\phi_{\prev}^2}\right) + T_{\prev}(\lambda)\right) \\
  &+ \max_{p, \widetilde{\lambda}_1, \ldots, \widetilde{\lambda}_p}\expect_{\widetilde{B}_1, \ldots, \widetilde{B}_p}\Bigg[
  \sum_{i \in [p]}\expect\Bigg[\widetilde{O}\left(\frac{1}{\psi_0^3}\right) \cdot \left(\widetilde{O}\left(\frac{\widetilde{\lambda}_i^2}{\phi\phi^\prime}\right) + T_{\prev}(\widetilde{\lambda}_i)\right)
  \\ &+ \expect_{B_i^\prime}\left[\expect[\Time_{L-1}(B_i^\prime, \widetilde{\lambda}_i) \mid B_i^\prime] \mid \widetilde{B}_i \right]\Bigg]\Bigg],
  \end{split}
  \label{eq:rec-time-L-simp}
  \end{equation}
  where $B^\prime$ is a random variable satisfying $\expect[B^\prime] \leq b$, $\widetilde{B}_1, \ldots, \widetilde{B}_p$ are random variables satisfying $\expect[\widetilde{B}_1 + \cdots + \widetilde{B}_p] \leq O(\log n) b$,
  $B_1^\prime, \ldots, B_p^\prime$ are random variables satisfying $\expect[B_i^\prime \mid \widetilde{B}_i] \leq \widetilde{B}_i$, and
  $\widetilde{\lambda}_1, \ldots, \widetilde{\lambda}_p$ satisfy $\widetilde{\lambda}_1 + \cdots + \widetilde{\lambda}_p \leq O(\log n) \lambda$ and $\widetilde{\lambda_i} < \lambda$.\footnote{Note that $\Time_L(b, \lambda)$ is always positive, and thus the $n^{-100}$ term can be absorbed into the big-O expression in the first line. In contrast to this, $\Size_{L}(b,\lambda)$ may be zero, and thus we need to explicitly put the $n^{-100}$ term.}
  \label{cor:rec-L-simp}
\end{corollary}

Similarly, we can expand the term in \eqref{eq:rec-size-l} and \eqref{eq:rec-time-l} that corresponds to recursion of \alg{MaintainExpander}{$S,L$} using \cref{cor:rec-L-simp}.

\begin{corollary}
  For $\phi < O\left(\frac{\psi_0^3}{\beta_{\mathrm{prev}}\log^3 n}\right)$ sufficiently small and $\ell < L$, we have
  \begin{equation}
  \begin{split}
    \expect[\Size_{\ell}(&b,\lambda)] \leq O\left(\frac{\phi}{\psi_0^2}\right) b^{(\ell+1)/L}
    + \widetilde{O}\left(\frac{1}{\psi_0^4}\right)  \expect_{\Delta}\left[\sum_{k = \ell+1}^{L}\min\left\{1, \frac{\Delta}{b^{k/L}}\right\} \cdot \expect[\Size_k(b+\Delta, \lambda-1) \mid \Delta]\right] \\
    &+ \expect_{B^\prime}\left[\expect[\Size_{\ell-1}(B^\prime, \lambda) \mid B^\prime]\right] \\
    &+ \max_{q,\widetilde{\lambda}_1,\ldots,\widetilde{\lambda}_q}\expect_{\widetilde{B}_1,\ldots,\widetilde{B}_q}\left[\sum_{i \in [q]}\expect\left[3\phi \cdot \widetilde{B}_i
    + \expect_{B_i^\prime}\left[\expect[\Size_{L-1}(B_i^\prime, \widetilde{\lambda}_i) \mid B_i^\prime]\right] \mid \widetilde{B}_i\right]\right] + n^{-100}
  \end{split}
  \label{eq:rec-size-l-simp}
  \end{equation}
  and
  \begin{equation}
  \begin{split}
  \expect[&\Time_{\ell}(b, \lambda)] \leq \widetilde{O}\left(\frac{1}{\psi_0^5}\right) \cdot m^{1/L} \cdot \left(\widetilde{O}\left(\frac{\lambda^2}{\phi\phi_{\prev}^2\psi_0^4}\right) + T_{\prev}(\lambda)\right) \\
  &+ \widetilde{O}\left(\frac{1}{\psi_0^4}\right)  \cdot \expect_{\Delta}\left[\sum_{k = \ell+1}^{L}\min\left\{1, \frac{\Delta}{b^{k/L}}\right\} \cdot \expect[\Time_k(b+\Delta, \lambda - 1) \mid \Delta]\right]
  + \expect_{B^\prime}\left[\expect[\Time_{\ell-1}(B^\prime, \lambda) \mid B^\prime]\right] \\
  &+ \max_{q,\widetilde{\lambda}_1,\ldots,\widetilde{\lambda}_q}\expect_{\widetilde{B}_1,\ldots,\widetilde{B}_q}\left[\sum_{i \in [q]}\expect\left[\widetilde{O}\left(\frac{1}{\psi_0^3}\right) \cdot \left(\widetilde{O}\left(\frac{\widetilde{\lambda}_i^2}{\phi\phi_{\prev}^2}\right) + T_{\prev}(\widetilde{\lambda}_i)\right) + \expect_{B_i^\prime}\left[\expect[\Time_{L-1}(B_i^\prime, \widetilde{\lambda}_i) \mid B_i^\prime]\right] \mid \widetilde{B}_i\right]\right],
  \end{split}
  \label{eq:rec-time-l-simp}
  \end{equation}
  where $\Delta$ is a random variable with $\expect[\Delta] \leq O\left(\frac{\phi\beta_{\prev}}{\psi_0^2}\right) b^{(\ell+1)/L}$, $B^\prime$ is a random variable with $\expect[B^\prime] \leq b$, $\widetilde{B}_1, \ldots, \widetilde{B}_q$ are random variables satisfying $\expect[\widetilde{B}_1 + \cdots + \widetilde{B}_q] \leq O\left(\frac{1}{\log^{2L} n}\right)b^{(\ell+1)/L}$, $B_1^\prime, \ldots, B_q^\prime$ are random variables satisfying $\expect[B_i^\prime \mid \widetilde{B}_i] \leq \widetilde{B}_i$, and $\widetilde{\lambda}_1, \ldots, \widetilde{\lambda}_q$ satisfy $\widetilde{\lambda}_1 + \cdots + \widetilde{\lambda}_q \leq O(\log n)\lambda$ and $\widetilde{\lambda}_i < \lambda$.
  
  \label{cor:rec-l-simp}
\end{corollary}

\begin{proof}
  To see \eqref{eq:rec-size-l-simp}, consider applying \cref{cor:rec-L-simp} to each of the $\expect[\Size_L(B_i, \lambda_i) \mid B_i]$ term in \eqref{eq:rec-size-l}.
  This yields
  \begin{equation}
  \begin{split}
    &\max_{r,\lambda_1,\ldots,\lambda_r}\expect_{B_1,\ldots,B_r}\Bigg[\sum_{i \in [r]}\expect[\Size_L(B_i, \lambda_i) \mid B_i]\Bigg] \leq
      \expect_{B_1,\ldots,B_r}\Bigg[\sum_{i \in [r]}\expect\Bigg[3\phi \cdot B_i + \expect_{B_i^\prime}\left[\expect[\Size_{L-1}(B_i^\prime,\lambda_i) \mid B_i^\prime\right] \\
      &+ \max_{p^{(i)}, \widetilde{\lambda}_1^{(i)},\ldots,\widetilde{\lambda}_{p^{(i)}}^{(i)}}\expect_{\widetilde{B}_1^{(i)},\ldots,\widetilde{B}_{p^{(i)}}^{(i)}}\left[\sum_{j \in [p^{(i)}]}\expect\left[3\phi \cdot \widetilde{B}_{j}^{(i)} + \expect_{{B_j^{(i)}}^{\prime}}\left[\Size_{L-1}({B_j^{(i)}}^{\prime},\widetilde{\lambda}_j^{(i)}) \mid {B_j^{(i)}}^\prime\right] \mid \widetilde{B}_j^{(i)}\right] \right]\mathrel{\Big|} B_i\Bigg]\Bigg],
    \label{eq:expand}
  \end{split}
  \end{equation}
  where $B_1^\prime, \ldots, B_r^\prime$ are random variables satisfying $\expect[B_i^\prime \mid B_i] \leq B_i$,
  $\widetilde{B}_1^{(i)}, \ldots, \widetilde{B}_{p^{(i)}}^{(i)}$ are random variables satisfying $\expect[\widetilde{B}_1^{(i)} + \cdots + \widetilde{B}_{p^{(i)}}^{(i)} \mid B_i] \leq O(\log n)B_i$,
  ${B_1^{(i)}}^\prime, \ldots, {B_{p^{(i)}}^{(i)}}^\prime$ are random variables satisfying $\expect[{B_j^{(i)}}^\prime \mid \widetilde{B}_j^{(i)}] \leq \widetilde{B}_j^{(i)}$, and $\widetilde{\lambda}_1^{(i)}, \ldots, \widetilde{\lambda}_{p^{(i)}}^{(i)}$ satisfy $\widetilde{\lambda}_1^{(i)} + \cdots + \widetilde{\lambda}_{p^{(i)}}^{(i)} \leq O(\log n)\lambda_i$ and $\widetilde{\lambda}_{j}^{(i)} < \lambda_i$.
  Observe that we can combine the outer max and expectation with the inner ones, in which case if define $\widetilde{B}_0^{(i)}$ to be a random variable that always realizes to $B_i$ and define $\widetilde{\lambda}_0^{(i)}$ to be $\lambda_i$, then we can rewrite \eqref{eq:expand} as
  \[
    \max_{q,\widetilde{\lambda}_1,\ldots,\widetilde{\lambda}_q}\expect_{\widetilde{B}_1,\ldots,\widetilde{B}_q}\left[\sum_{i \in [q]}\expect\left[3\phi \cdot \widetilde{B}_i
    + \expect_{B_i^\prime}\left[\expect[\Size_{L-1}(B_i^\prime, \widetilde{\lambda}_i) \mid B_i^\prime]\right] \mid \widetilde{B}_i\right]\right]
  \]
  where $q$ corresponds to $r + (p^{(1)} + 1) + \cdots + (p^{(r)} + 1)$, $\widetilde{\lambda}_1, \ldots, \widetilde{\lambda}_q$ correspond to $\widetilde{\lambda}_0^{(1)},\ldots,\widetilde{\lambda}_{p^{(r)}}^{(r)}$, and $\widetilde{B}_1, \ldots, \widetilde{B}_q$ correspond to $\widetilde{B}_0^{(1)},\ldots,\widetilde{B}_{p^{(r)}}^{(r)}$.
  The random variables $\widetilde{B}_1, \ldots, \widetilde{B}_q$ satisfy $\expect[\widetilde{B}_1 + \cdots + \widetilde{B}_q] \leq O(\log n)\expect[B_1 + \cdots + B_r] \leq O\left(\frac{1}{\log^{2L} n}\right)b^{(\ell+1)/L}$, $B_1^\prime,\ldots,B_q^\prime$ satisfy $\expect[B_i^\prime \mid \widetilde{B}_i] \leq \widetilde{B}_i$, and $\widetilde{\lambda}_1,\ldots,\widetilde{\lambda}_q$ satisfy $\widetilde{\lambda}_1+\cdots+\widetilde{\lambda}_q \leq O(\log n)(\lambda_1+\cdots+\lambda_r) \leq O(\log n)\lambda$ and $\widetilde{\lambda}_i \leq \lambda_j < \lambda$.\footnote{We remark that the reason why we can put everything into the max and expectation is because \eqref{eq:rec-size-l} guarantees that $\lambda_j < \lambda$, so even though in \cref{lemma:simplify-rec} we need to have a standalone term $\expect[g(b,\lambda)]$ handling the case when there is no decrease in $\lambda$, here we can simply put $\lambda_i$ (which corresponds to $\widetilde{\lambda}_0^{(i)}$) into the max and expectation expressions.}
  The proof of \eqref{eq:rec-time-l-simp} follows analogously.
  
\end{proof}

We are finally ready to prove an actual bound on $\expect[\Size_{\ell}(b, \lambda)]$ and $\expect[\Time_{\ell}(b, \lambda)]$.

\begin{lemma}
  For $\phi < O\left(\frac{\psi_0^{O(L^2)}}{\beta_{\prev} \cdot L \cdot m^{1/L}}\right)$ sufficiently small, we have $\expect[\Size_{\ell}(b, \lambda)] \leq O\left(\frac{1}{\psi_0^2}\right) \cdot \phi \cdot b^{(\ell+1)/L}$.
  \label{lemma:expected-size-bound}
\end{lemma}

\begin{proof}
  We show that there exists a $c \geq \Theta\left(\frac{1}{\psi_0}\right)$ sufficiently large for which $\expect[\Size_{\ell}(b, \lambda)] \leq c\phi \cdot b^{(\ell+1)/L}$.
  Note that it suffices to consider the case when $b \geq 1/\phi \geq O(\log n)^{O(L)}$ due to the condition on Line~\ref{line:too-little-volume} in \cref{alg:maintain-exp-decomp}.
  We proceed by an induction on $\lambda$ and $\ell$.
  The base case of $\lambda = 0$ is trivial as $\lambda = 0$ implies $b = 0$.
  Consider now $\lambda > 0$ and $\ell = L$.
  From \eqref{eq:rec-size-L-simp} and the inductive hypothesis, we have
  \begin{align*}
    \expect[\Size_{L}(b,\lambda)]
    &\leq
    3\phi \cdot b + \expect_{B^\prime}\left[c\phi B^\prime\right] + \max_{p,\widetilde{\lambda}_1, \ldots, \widetilde{\lambda}_p}\expect_{\widetilde{B}_1,\ldots,\widetilde{B}_p}\left[\sum_{i \in [p]}\expect\left[3\phi \cdot \widetilde{B}_i + \expect_{B_i^\prime}\left[c\phi \cdot B_i^\prime\right] \mid \widetilde{B}_i\right]\right] + n^{-100} \\
    &\leq 3\phi \cdot b + c\phi \cdot b + \max_{p,\widetilde{\lambda}_1, \ldots, \widetilde{\lambda}_p}\expect_{\widetilde{B}_1,\ldots,\widetilde{B}_p}\left[\sum_{i \in [p]}\expect\left[3\phi \cdot \widetilde{B}_i + c\phi \cdot \widetilde{B}_i \mid \widetilde{B}_i\right]\right] + n^{-100} \\
    &\leq (3+c)\phi \cdot b + (3+c)\phi \cdot O(\log n)b + n^{-100},
  \end{align*}
  which is at most $c\phi \cdot b^{1+1/L}$ for $b \geq 1/\phi$ sufficiently large.
  To bound $\expect[\Size_{\ell}(b, \lambda)]$ for $\ell < L$, we first prove the following helper claim.

  \begin{claim}
    Assuming the inductive hypothesis, we have
    \begin{align*}
      \expect_{\Delta}\left[\sum_{k = \ell+1}^{L}\min\Bigg\{1, \frac{\Delta}{b^{k/L}}\right\} \cdot &\expect[\Size_k(b+\Delta, \lambda-1) \mid \Delta]\Bigg] \\ &\leq O(c\phi\log n) \cdot L \cdot \left(O\left(\frac{\phi\beta_{\prev}}{\psi_0^2}\right) \cdot b^{(\ell+2)/L} + b^{(\ell+1)/L-1/L^2}\right),
    \end{align*}
    for $\expect[\Delta] \leq O\left(\frac{\phi\beta_{\prev}}{\psi_0^2}\right)b^{(\ell+1)/L}$.%
    \label{claim:helper}
  \end{claim}

  \begin{proof}
    By the linearity of expectation we can move the summation out of the expectation.
    For $k < L$, we can bound the summand as follows:
    \begin{align*}
      \expect_{\Delta}\Bigg[\min\Bigg\{1, &\frac{\Delta}{b^{k/L}}\Bigg\} \cdot \expect[\Size_k(b+\Delta, \lambda-1) \mid \Delta]\Bigg] \\
      &\leq \expect_{\Delta}\left[\frac{\Delta}{b^{k/L}} \cdot \expect[\Size_k(2b, \lambda-1)] + \expect[\Size_k(2\Delta, \lambda-1) \mid \Delta]\right] \\
      &\leq \frac{1}{b^{k/L}} \cdot \expect[\Size_k(2b, \lambda-1)] \cdot \expect[\Delta] + \underbrace{\expect_{\Delta}\left[\expect[\Size_k(2\Delta, \lambda-1) \mid \Delta]\right]}_{(i)} \\
      &\leq 4c\phi \cdot \left(\frac{1}{b^{k/L}} \cdot b^{(k+1)/L} \cdot O\left(\frac{\phi\beta_{\prev}}{\psi_0^2}\right)b^{(\ell+1)/L} + \left(O\left(\frac{\phi\beta_{\prev}}{\psi_0^2}\right)b^{(\ell+1)/L}\right)^{(k+1)/L}\right) \\
      &\leq 4c\phi \cdot O\left(\frac{\phi\beta_{\prev}}{\psi_0^2}\right) \cdot b^{(\ell+2)/L} + 4c\phi \cdot \max\left\{O\left(\frac{\phi\beta_{\prev}}{\psi_0^2}\right)b^{(\ell+1)/L}, b^{(\ell+1)/L-1/L^2}\right\} \\
      &\leq 8c\phi \cdot O\left(\frac{\phi\beta_{\prev}}{\psi_0^2}\right) \cdot b^{(\ell+2)/L} + 4c\phi \cdot b^{(\ell+1)/L-1/L^2}. \\
    \end{align*}
    For $k = L$, we can no longer substitute the $\expect[(2\Delta)^{(k+1)/L}]$ in (i) by $4\expect[\Delta]^{(k+1)/L}$ since $f(x) = x^{1+\eps}$ is convex for $\eps > 0$ and thus Jensen's inequality does not apply anymore.
    As such We further expand this term using \cref{cor:rec-L-simp}.
    This gives us
    \begin{align*}
      \expect_{\Delta}\left[\expect[\Size_L(2\Delta, \lambda-1) \mid \Delta]\right]
      &\leq \expect\left[\max_{p,\widetilde{\lambda}_1,\ldots,\widetilde{\lambda}_p}\expect_{\widetilde{\Delta}_1,\ldots,\widetilde{\Delta}_p}\left[\sum_{i \in [p]}\expect\left[3\phi \cdot \widetilde{\Delta}_i + \expect_{\Delta^\prime_i}\left[\expect[\Size_{L-1}(\Delta^\prime_i,\widetilde{\lambda_i})\mid \Delta^\prime_i\right] \mathrel{\Bigg|} \widetilde{\Delta}_i \right]\right] \mathrel{\Bigg|} \Delta\right] \\
      &\leq \expect\left[\max_{p,\widetilde{\lambda}_1,\ldots,\widetilde{\lambda}_p}\expect_{\widetilde{\Delta}_1,\ldots,\widetilde{\Delta}_p}\left[\sum_{i \in [p]}\expect\left[3\phi \cdot \widetilde{\Delta}_i + \expect_{\Delta^\prime_i}\left[c\phi \cdot \Delta_i^\prime\right] \mathrel{\Big|} \widetilde{\Delta}_i \right]\right] \mathrel{\Bigg|} \Delta\right]  \\
      &\leq \expect\left[\max_{p,\widetilde{\lambda}_1,\ldots,\widetilde{\lambda}_p}\expect_{\widetilde{\Delta}_1,\ldots,\widetilde{\Delta}_p}\left[\sum_{i \in [p]}\expect\left[(3+c)\phi \cdot \widetilde{\Delta}_i\right]\right] \mathrel{\Big|} \Delta\right] \\
      &\leq (3+c)\phi \cdot O(\log n) \cdot O\left(\frac{\phi\beta_{\prev}}{\psi_0^2}\right)b^{(\ell+1)/L},
    \end{align*}
    where $\expect[\widetilde{\Delta}_1+\cdots+\widetilde{\Delta}_r \mid \Delta] \leq O(\log n)(2\Delta) \leq O(\log n)\Delta$.
    Substituting this back to the above calculation, we get
    \begin{align*}
      \expect_{\Delta}&\Bigg[\min\Bigg\{1, \frac{\Delta}{b}\Bigg\} \cdot \expect[\Size_L(b+\Delta, \lambda-1) \mid \Delta]\Bigg] \\
      &\leq \frac{1}{b} \cdot \expect[\Size_k(2b, \lambda-1)] \cdot \expect[\Delta] + (3+c)\phi \cdot O(\log n) \cdot O\left(\frac{\phi\beta_{\prev}}{\psi_0^2}\right)b^{(\ell+1)/L} \\
      &\leq O(c\phi \log n) \cdot \left(\frac{1}{b} \cdot b^{1+1/L} \cdot O\left(\frac{\phi\beta_{\prev}}{\psi_0^2}\right)b^{(\ell+1)/L} + O\left(\frac{\phi\beta_{\prev}}{\psi_0^2}\right)b^{(\ell+1)/L}\right) \\
      &\leq O(c\phi \log n) \cdot O\left(\frac{\phi\beta_{\prev}}{\psi_0^2}\right) \cdot b^{(\ell+2)/L}.
    \end{align*}
    The claim follows by summing over at most $L$ different values of $k$.
  \end{proof}

  With \cref{claim:helper}, we can now bound $\expect[\Size_{\ell}(b, \lambda)]$ for $b \geq 1/\phi$ via \cref{cor:rec-l-simp} as follows:
  \begin{align*}
    \expect[\Size_{\ell}(&b,\lambda)] \leq O\left(\frac{\phi}{\psi_0^2}\right) b^{(\ell+1)/L}
    + \widetilde{O}\left(\frac{1}{\psi_0^4}\right)  \cdot \expect_{\Delta}\left[\sum_{k = \ell+1}^{L}\min\left\{1, \frac{\Delta}{b^{k/L}}\right\} \cdot \expect[\Size_k(b+\Delta, \lambda-1) \mid \Delta]\right] \\
    &+ \expect_{B^\prime}\left[c\phi \cdot (B^\prime)^{(\ell/L)}\right] + \max_{q,\widetilde{\lambda}_1,\ldots,\widetilde{\lambda}_q}\expect_{\widetilde{B}_1,\ldots,\widetilde{B}_q}\left[\sum_{i \in [q]}\expect\left[3\phi \cdot \widetilde{B}_i + \expect_{B_i^\prime}\left[c\phi \cdot B_i^\prime\right] \mid \widetilde{B}_i\right]\right] \\
    &\leq O\left(\frac{\phi}{\psi_0^2}\right) b^{(\ell+1)/L} + \widetilde{O}\left(\frac{1}{\psi_0^4}\right) \cdot O(c\phi\log n) \cdot L \cdot \left(O\left(\frac{\phi\beta_{\prev}}{\psi_0^2}\right) \cdot b^{(\ell+2)/L} + b^{(\ell+1)/L-1/L^2}\right) \\
    &+ c\phi \cdot b^{\ell/L} + (3+c)\phi \cdot O\left(\frac{1}{\log^{2L} n}\right)b^{(\ell+1)/L}
  \end{align*}
  which is at most $c\phi \cdot b^{(\ell+1)/L}$ for $c \geq \Omega\left(\frac{1}{\psi_0^2}\right)$ sufficiently large, $\phi < O\left(\frac{\psi_0^{O(L^2)}}{\beta_{\prev} \cdot L \cdot m^{1/L}}\right)$ sufficiently small, and $b \geq 1/\phi \geq \widetilde{\Omega}\left(\frac{1}{\psi_0^3}\right)^{O(L^2)}$.
  This proves the lemma.
\end{proof}

The proof of $\expect[\Time_{\ell}(b, \lambda)]$ follows analogously.

\begin{lemma}
  For $\phi < O\left(\frac{\psi_0^{O(L^2)}}{\beta_{\prev} \cdot L \cdot m^{1/L}}\right)$ sufficiently small, we have $\expect[\Time_{\ell}(b, \lambda)] \leq m^{1/L} \cdot b^{1/L} \cdot \widetilde{O}\left(\frac{\lambda^2}{\phi\phi_{\prev}^2\psi_0^4} + T_{\prev}(\lambda)\right)$.
  \label{lemma:expected-time-bound}
\end{lemma}

\begin{proof}
  Let $G(\lambda) = \widetilde{O}\left(\frac{\lambda^2}{\phi\phi_{\prev}^2\psi_0^4} + T(\lambda)\right)$ be an upper bound on the time \alg{MaintainExpander}{$U, \ell$} on an $\lambda$-vertex $U$ spends in the while-loop.
  We proceed by an induction on $\lambda$ and $\ell$ and show that $\expect[\Time_{\ell}(b, \lambda)] \leq 2^\ell \cdot m^{1/L} \cdot b^{1/L} \cdot G(\lambda)$ for $\ell < L$ and $\expect[\Time_L(b, \lambda)] \leq 2^L \cdot O(\log^2 n) \cdot m^{1/L} \cdot b^{1/L} \cdot G(\lambda)$.
  Again, due to the early return on Line~\ref{line:too-little-volume} in \cref{alg:maintain-exp-decomp}, it suffices to prove the bound for $b \geq 1/\phi$ and thus we may assume $b \geq \widetilde{\Omega}\left(\frac{1}{\psi_0^4}\right)^{O(L)}$ is sufficiently large.
  The base case is easy to verify.
  For $\ell = L$, by the inductive hypothesis, \eqref{eq:rec-time-L-simp} can be bounded by 
  \begin{align*}
    \expect[\Time_L(b,\lambda)] \leq &\widetilde{O}\left(\frac{1}{\psi_0^3}\right) \cdot G(\lambda) \\ &+ \max_{p,\widetilde{\lambda}_1,\ldots,\widetilde{\lambda}_p}\expect_{\widetilde{B}_1,\ldots,\widetilde{B}_p}\left[\sum_{i \in [p]}\expect\left[\widetilde{O}\left(\frac{1}{\psi_0^3}\right) \cdot G(\widetilde{\lambda}_i) + \expect_{B_i^\prime}\left[2^{L-1} \cdot m^{1/L} \cdot (B_i^\prime)^{1/L} \cdot G(\widetilde{\lambda_i})\right] \Big| \widetilde{B}_i\right]\right] \\
    &\leq \widetilde{O}\left(\frac{1}{\psi_0^3}\right) \cdot G(\lambda) + \widetilde{O}\left(\frac{1}{\psi_0^3}\right) \cdot O(\log n)G(\lambda) +  2^{L-1} \cdot m^{1/L} \cdot O(\log n)b^{1/L} \cdot O(\log n)G(\lambda) \\
    &\leq 2^{L-1} \cdot O(\log^2 n) \cdot m^{1/L} \cdot b^{1/L} \cdot G(\lambda)
  \end{align*}
  since $m^{1/L} \gg \widetilde{O}\left(\frac{1}{\psi_0^3}\right)$.
  For $\ell < L$, we prove a helper claim similar to \cref{claim:helper}.
  \begin{claim}
    We have
    \begin{align*}
      \expect_{\Delta}\Bigg[\sum_{k=\ell+1}^{L}\min\left\{1, \frac{\Delta}{b^{k/L}}\right\} &\cdot \expect\left[\Time_k(b+\Delta, \lambda-1) \mid \Delta\right]\Bigg] \\ &\leq 2^L \cdot O(\log^2 n) \cdot O\left(\frac{\phi\beta_{\prev}}{\psi_0^2}\right)^{1/L} \cdot L \cdot m^{1/L} \cdot b^{1/L} \cdot G(\lambda)
    \end{align*}
    for $\expect[\Delta] \leq O\left(\frac{\phi\beta_{\prev}}{\psi_0^2}\right) b^{(\ell+1)/L}$.
    \label{claim:helper-time}
  \end{claim}

  \begin{proof}
    We first move the summation out of the expectation and bound each summand as follows:
    \begin{align*}
      \expect_{\Delta}\Bigg[\min\left\{1, \frac{\Delta}{b^{k/L}}\right\} \cdot &\expect\left[\Time_k(b+\Delta, \lambda) \mid \Delta\right]\Bigg] \\
      &\leq \expect_{\Delta}\left[\frac{\Delta}{b^{k/L}} \cdot \expect\left[\Time_k(2b, \lambda)\right] + \expect\left[\Time_k(2\Delta, \lambda) \mid \Delta\right]\right] \\
      &\leq 2^L \cdot O(\log^2 n) \cdot m^{1/L} \cdot \left(\frac{\expect[\Delta]}{b^{k/L}} \cdot b^{1/L} + \expect_{\Delta}[\Delta^{1/L}]\right) \cdot G(\lambda) \\
      &\leq 2^L \cdot O(\log^2 n) \cdot m^{1/L} \cdot \left(O\left(\frac{\phi\beta_{\prev}}{\psi_0^2}\right)b^{(\ell+1-k)/L} \cdot b^{1/L} + O\left(\frac{\phi\beta_{\prev}}{\psi_0^2}\right)^{1/L} \cdot b^{(\ell+1)/L^2}\right) \cdot G(\lambda) \\
      &\leq 2^L \cdot O(\log^2 n) \cdot O\left(\frac{\phi\beta_{\prev}}{\psi_0^2}\right)^{1/L} \cdot m^{1/L} \cdot b^{1/L} \cdot G(\lambda).
    \end{align*}
    The claim follows by summing over at most $L$ values of $k$.
  \end{proof}
  We can now expand \eqref{eq:rec-time-l-simp} using \cref{claim:helper-time} and get
  \begin{align*}
    \expect[\Time_{\ell}(b, \lambda)]
      &\leq \widetilde{O}\left(\frac{1}{\psi_0^5}\right) \cdot m^{1/L} \cdot G(\lambda) + \widetilde{O}\left(\frac{1}{\psi_0^4}\right) \cdot 2^L \cdot O(\log^2 n) \cdot O\left(\frac{\phi\beta_{\prev}}{\psi_0^2}\right)^{1/L} \cdot L \cdot m^{1/L} \cdot b^{1/L} \cdot G(\lambda) \\
      &+ \expect_{B^\prime}\left[2^{\ell-1} \cdot m^{1/L} \cdot (B^\prime)^{1/L} \cdot G(\lambda)\right] \\
      &+ \underbrace{\max_{q,\widetilde{\lambda}_1,\ldots,\widetilde{\lambda}_q}\expect_{\widetilde{B}_1,\ldots,\widetilde{B}_q}\left[\sum_{i \in [q]}\expect\left[\widetilde{O}\left(\frac{1}{\psi_0^3}\right) \cdot G(\widetilde{\lambda}_i) + \expect_{B_i^\prime}\left[2^L \cdot m^{1/L} \cdot (B_i^\prime)^{1/L} \cdot G(\widetilde{\lambda}_i)\right] \mid \widetilde{B}_i \right]\right]}_{(i)},
  \end{align*}
  where we can bound the last term (i) by
  \begin{align*}
    \widetilde{O}\left(\frac{1}{\psi_0^3}\right) &\cdot G(\lambda) + 2^L \cdot m^{1/L} \cdot \max_{q,\widetilde{\lambda}_1,\ldots,\widetilde{\lambda}_q}\expect_{\widetilde{B}_1,\ldots,\widetilde{B}_q}\left[\sum_{i \in [q]}\expect[\widetilde{B}_i]^{(1/L)} \cdot G(\widetilde{\lambda}_i)\right] \\
    &\leq \widetilde{O}\left(\frac{1}{\psi_0^3}\right) \cdot G(\lambda) + 2^L \cdot m^{1/L} \cdot O(\log n) \cdot G(\lambda) \cdot O\left(\frac{1}{\log^{2L} n}\right) b^{1/L} \\ &\leq \widetilde{O}\left(\frac{1}{\psi_0^3}\right) \cdot G(\lambda) + \frac{m^{1/L}}{\Omega(\log n)} \cdot G(\lambda) \cdot b^{1/L}.
  \end{align*}
  Substituting this back into the above calculation, we can see that $\expect[\Time_{\ell}(b, \lambda)]$ is at most $2^{\ell} \cdot m^{1/L} \cdot G(\lambda)$ for $\phi < O\left(\frac{\psi_0^{O(L^2)}}{\beta_{\prev} \cdot L \cdot m^{1/L}}\right)$ sufficiently small and $b \geq 1/\phi \geq \widetilde{\Omega}\left(\frac{1}{\psi_0^5}\right)^L$.
  Observe that $\frac{1}{\psi_0} \gg 2^L$ by \eqref{eq:compare-psi}.
  This proves the lemma.
\end{proof}

To this end, we can establish the expected guarantee of $\cM.\alg{Init}{}$ and $\cM.\alg{Cut}{D}$ implemented in \cref{alg:cut-add-terminal}.
Recall in \cref{def:hierarchy-maintainer} that $U_D$ is the union of $U$'s that intersect with the input cut $D$.

\begin{lemma}
  For $\frac{1}{n} < \phi < O\left(\frac{\psi_0^{O(L^2)}}{\beta_{\prev} \cdot L \cdot m^{1/L}}\right)$ sufficiently small, the subroutine \Init{$G$} runs in expected $m^{2/L} \cdot \widetilde{O}\left(\frac{n^2}{\phi\phi_{\prev}^2\psi_0^4} + T_{\prev}(n)\right)$ and outputs a set $X$ of expected size $\phi \cdot O\left(\frac{1}{\psi_0^2}\right) \cdot m^{O(1/L)} \cdot \alpha_{\prev} \cdot m$.
  \label{lemma:init}
\end{lemma}

\begin{proof}
  We initialize $F$ as the output of $\cM_{\prev}.\Init{G}$ which by its guarantee satisfies $\expect[\Bc_G(F)] \leq \alpha_{\prev} \cdot m$.
  Therefore, the initial volume of $V$ on which we run \MaintainExpander{$V,L$} is at most $2\alpha_{\prev} \cdot m$ in expectation.
  Observe that the algorithm is always in a good state in the beginning, and thus by \cref{lemma:from-size-to-actual-bound} the expected output size of \Init{$G$} is at most
  \begin{align*}
    \expect_{F}\left[\expect[\Size_L(2\Bc_G(F), n)] \mid F\right]
    &\leq \expect_{F}\left[O\left(\frac{1}{\psi_0^2}\right) \cdot \phi \cdot (2\Bc_G(F))^{1+1/L}\right]
    &\leq \phi \cdot O\left(\frac{1}{\psi_0^2}\right) \cdot m^{1/L} \cdot \alpha_{\prev} \cdot m
  \end{align*}
  by \cref{lemma:expected-size-bound} and it runs in expected
  \begin{align*}
    \expect_{F}\left[\expect[\Time_L(2\Bc_G(F), n)] \mid F)\right]
    &\leq \expect_{F}\left[m^{2/L} \cdot \widetilde{O}\left(\frac{n^2}{\phi\phi_{\prev}^2\psi_0^4} + T_{\prev}(n)\right)\right] \\
    &\leq m^{2/L} \cdot \widetilde{O}\left(\frac{n^2}{\phi\phi_{\prev}^2\psi_0^4} + T_{\prev}(n)\right)
  \end{align*}
  time by \cref{lemma:expected-time-bound}.
  Note that \PostProcess{$V$} runs in $O(n^2)$ time using \cite{Tarjan72} and is therefore negligible.
\end{proof}

\begin{lemma}
  For $\frac{1}{n} < \phi < O\left(\frac{\psi_0^{O(L^2)}}{\beta_{\prev} \cdot L \cdot m^{1/L}}\right)$ sufficiently small, the subroutine \emph{\alg{Cut}{$D$}} in expected $m^{2/L} \cdot \widetilde{O}\left(\frac{|U_D|^2}{\phi\phi_{\prev}^2\psi_0^4} + T_{\prev}(|U_D|)\right)$ outputs a set $X$ of expected total capacities $\expect[\Bc_G(X)] \leq O\left(\frac{1}{\psi_0^5}\right) \cdot m^{O(1/L)} \cdot |D|$.
  \label{lemma:cut}
\end{lemma}

\begin{proof}
  Recall the implementation of \alg{Cut}{$D$} in \cref{alg:cut-add-terminal}, where we visit each $U \in \cU$ that intersects with $D$ and remove the corresponding cut from $U$.
  Observe that the running time and output size can be computed for each $U$ individually by the linearity of expectation.
  Fix a $U \in \cU$.
  Note that $U$ contributes to $X$ and the running time in two ways:
  one is when running \alg{MaintainExpander}{$U, k$} on Line~\ref{line:sample-cut}, and
  the other is when running \alg{MaintainExpander}{$S, L$} with $S = S_U$ on Line~\ref{line:build-small-pieces-cut}.
  We bound these two terms separately.
  Recall that $D_U \defeq D \cap G[U]$ is the set of edges that are removed from $G[U]$.
  Let $\Delta_U$ be \emph{two times} the total capacities of the edge set $A$ output by $\cM_{\prev}.\alg{Cut}{D \cap G[U]}$ which upper bounds the units of volume added to both $U \setminus S_U$ and $S_U$.
  By the guarantee of $\cM_{\prev}$, the expected value of $\Delta_U$ is upper-bounded by $2\beta_{\prev} \Bc_G(D_U)$.
  Note that the $k$ we sampled on Line~\ref{line:sample-cut} is from the distribution $\cR_{\Bc_G(D_U)/\phi + \Delta_U}$.
  Let $b_U$ be the initial volume of $U$ and $\lambda_U$ be the number of vertices in $U$.

  \paragraph{Line~\ref{line:sample-cut}.}
  Note that when we call \alg{MaintainExpander}{$U, k$} on Line~\ref{line:sample-cut}, the volume of $U$ is upper-bounded by $b_U + \Delta_U$.
  By \cref{lemma:expected-size-bound}, conditioned on the event $\cK$, we have $\tau_U \geq \Omega(\psi_0^2 b_U)$ and thus we can bound the expected output size of this call by
  \begin{align*}
    \expect_{\Delta_U}\Bigg[\sum_{k=0}^{L}&\min\left\{1, \frac{\Bc_G(D_U)/\phi + \Delta_U}{\psi_0^2 \cdot \tau_U^{k/L}}\right\} \cdot \expect[\Size_k(b_U + \Delta_U, \lambda_U) \mid \Delta_U]\Bigg] \\
      &\leq O\left(\frac{1}{\psi_0^4}\right) \cdot \expect_{\Delta_U}\left[\sum_{k=0}^{L}\min\left\{1, \frac{\Bc_G(D_U)/\phi + \Delta_U}{b_U^{k/L}}\right\} \cdot \expect[\Size_k(b_U + \Delta_U, \lambda_U) \mid \Delta_U]\right].
  \end{align*}
  Moving the expectation into the summation, we can bound each summand by
  \begin{align*}
    \expect_{\Delta_U}\Bigg[\min\Bigg\{1, &\frac{\Bc_G(D_U)/\phi + \Delta_U}{b_U^{k/L}}\Bigg\} \cdot \expect[\Size_k(b_U + \Delta_U, \lambda_U) \mid \Delta_U]\Bigg]  \\
    &\leq \expect_{\Delta_U}\left[\frac{\Bc_G(D_U)/\phi + \Delta_U}{b_U^{k/L}} \cdot \expect[\Size_k(2b_U, \lambda_U)] + \expect[\Size_k(2\Delta_U, \lambda_U) \mid \Delta_U]\right] \\
    &\leq O\left(\frac{1}{\psi_0^2}\right) \cdot \phi \cdot \expect_{\Delta_U}\left[\frac{\Bc_G(D_U)/\phi + \Delta_U}{b_U^{k/L}} \cdot b_U^{(k+1)/L}\right] + \underbrace{\expect_{\Delta_U}\left[\Delta_U^{(k+1)/L}\right]}_{(i)} \\
    &\leq O\left(\frac{1}{\psi_0^2}\right) \cdot \phi \cdot \left(m^{1/L} \cdot \left(\Bc_G(D_U)/\phi + 2\beta_{\prev} \Bc_G(D_U)\right) + \beta_{\prev} \cdot \Bc_G(D_U)^{(k+1)/L}\right) \\
    &\leq O\left(\frac{1}{\psi_0^2}\right) \cdot m^{1/L} \cdot \Bc_G(D_U)
  \end{align*}
  when $k < L$.
  For $k = L$, as in \cref{claim:helper} we likewise expand (i) using \cref{cor:rec-L-simp} and get
  \begin{align*}
    \expect_{\Delta_U}&\left[\expect[\Size_L(2\Delta_U, \lambda_U) \mid \Delta_U]\right] \\
    &\leq \expect\left[\max_{p,\widetilde{\lambda}_1,\ldots,\widetilde{\lambda}_p}\expect_{\widetilde{\Delta}_1,\ldots,\widetilde{\Delta}_p}\left[\sum_{i \in [p]}\expect\left[3\phi \cdot \widetilde{\Delta}_i + \expect_{\Delta_i^\prime}\left[\expect[\Size_{L-1}(\Delta_i^\prime,\widetilde{\lambda}_i) \mid \Delta_i^\prime]\right] \mid \widetilde{\Delta}_i \right]\right] + n^{-100}\Bigg| \Delta_U\right] \\
    &\leq \phi \cdot O(\log n) \cdot \expect[\Delta_U] + O(\log n) \cdot O\left(\frac{1}{\psi_0^2}\right) \cdot \phi \cdot \expect[\Delta_U] \leq \phi \cdot O\left(\frac{\log n}{\psi_0^2}\right)\cdot \beta_{\prev}\Bc_G(D_U).
  \end{align*}
  Plugging this back into the above calculation we can conclude that
  \begin{align*}
    \expect_{\Delta_U}\Bigg[\frac{\Bc_G(D_U)/\phi + \Delta_U}{b_U} \cdot &\expect[\Size_L(2b_U, \lambda_U)] + \expect[\Size_L(2\Delta_U, \lambda_U) \mid \Delta_U]\Bigg] \\
    &\leq O\left(\frac{1}{\psi_0^2}\right) \cdot \phi \cdot \expect_{\Delta_U}\left[\frac{\Bc_G(D_U)/\phi + \Delta_U}{b_U} \cdot B_U^{1+1/L}\right] + O\left(\frac{\phi \log n}{\psi_0^2}\right) \cdot (\beta \Bc_G(D_U)) \\
    &\leq O\left(\frac{1}{\psi_0^2}\right) \cdot m^{1/L} \cdot \Bc_G(D_U).
  \end{align*}
  Summing over $O(L)$ values of $k$, we get that conditioned on the event $\cK$, the expected contribution to $X$ of Line~\ref{line:sample-cut} is bounded by $O\left(\frac{L}{\psi_0^6}\right) \cdot m^{1/L} \cdot \Bc_G(D_U)$.
  By \cref{fact:expected-cond} this is asymptotically the same as the unconditional expectation.
  The expected running time is easily bounded using \cref{lemma:expected-time-bound} by $m^{2/L} \cdot \widetilde{O}\left(\frac{\lambda_U^2}{\phi\phi_{\prev}^2\psi_0^4} + T(\lambda_U)\right)$.
  
  \paragraph{Line~\ref{line:build-small-pieces-cut}.}
  For the contribution of Line~\ref{line:build-small-pieces-cut}, note that since $F$ is $\frac{\phi\psi_0^2}{2}$-expanding in $(G[U],\Bc_G)$ by \cref{lemma:maintain-exp-decomp} before the this run of $\cM.\Cut{D}$, the volume of $S_U$ before adding those $\Delta_U$ units is at most $\frac{2\Bc_G(D_U)}{\phi\psi_0^2}$.
  As such, the expected output size of the \alg{MaintainExpander}{$S_U, L$} call on Line~\ref{line:build-small-pieces-cut} is at most
  \begin{align*}
    \expect_{\Delta_U}\Bigg[\Size_L&\Bigg(\frac{2\Bc_G(D_U)}{\phi\psi_0^2} + \Delta_U, \lambda_U\Bigg) \Big| \Delta_U\Bigg]
    \leq \expect_{\Delta_U}\left[\Size_L\left(\frac{4\Bc_G(D_U)}{\phi\psi_0^2}, \lambda_U, L\right) + \Size_L\left(2\Delta_U, \lambda_U\right) \Big| \Delta_U\right] \\
    &\leq O\left(\frac{\phi}{\psi_0^2}\right)\left(\left(\frac{\Bc_G(D_U)}{\phi \psi_0^2}\right)^{1+1/L} + (\beta_{\prev} \Bc_G(D_U))^{1+1/L}\right) \leq O\left(\frac{1}{\psi_0^6}\right) \cdot m^{O(1/L)} \cdot \Bc_G(D_U)
  \end{align*}
  since $1/\phi \leq n$.
  The expected running time of this part, again, by \cref{lemma:expected-time-bound} is $m^{2/L} \cdot \widetilde{O}\left(\frac{\lambda_U^2}{\phi\phi_{\prev}^2\psi_0^4} + T(\lambda_U)\right)$.
  Since the sum of $\Bc_G(D_U)$'s among all $U \in \cU$ is at most $\Bc_G(D)$ and the sum of $\lambda_U$'s for which $D_U \neq \emptyset$ is $|U_D|$ (recall the definition of $U_D$ in \cref{def:hierarchy-maintainer}, the lemma follows.
\end{proof}

This completes the discussion on expected output size and running time of the subsection.

\subsection{Putting Everything Together} \label{subsec:everything}

To this end, we have developed all the technical pieces needed for proving \cref{lemma:hierarchy-maintainer}.

\Boosting*

\begin{proof}
  The algorithm for the new hierarchy maintainer $\cM$ is \cref{alg:cut-add-terminal} with the given parameters $L$ and $\phi$.
  By \cref{obs:correct-graph}, the graph $G_{\cM}$ the algorithm maintains is equal to what its output indicates (see \cref{def:hierarchy-maintainer}).
  Moreover, by \cref{lemma:maintain-exp-decomp}, with high probability, after every update $F$ is $\frac{\phi\psi_0^2}{2}$-expanding in $(G_{\cM},\Bc_G)$.
  By \cref{obs:maintain-terminal}, the terminal set $F$ that the algorithm maintains is a superset of $G_{\cM} \setminus G_{\cM_{\prev}}$.
  Letting $X \defeq G \setminus G_{\cM}$, this implies if we set $\cH_{\cM} \defeq (D, X_1, \ldots, X_k, F \setminus X)$, where $\cH_{\prev} \defeq (D, X_1, \ldots, X_k)$ is the $k$-level $\phi_{\prev}$-expander hierarchy of $(G_{\cM_{\prev}},\Bc_G)$, then it is easy to see that $\cH_{\cM}$ is an $\min\left\{\phi_{\prev}, \frac{\phi\psi_0^2}{2}\right\}$-expander hierarchy of $(G_{\cM},\Bc_G)$ with height $k+1$.
  On the other hand, if with inverse polynomially small probability $F$ is not expanding in $(G_{\cM},\Bc_G)$, then we output $X = F$ after that update and thus the $\cH_{\cM}$ defined above is still a valid $\phi_{\prev}$-expander hierarchy of $(G_{\cM},\Bc_G)$.
  This only affects the output size by an additive $n^{-100}$ factor in expectation and is thus negligible.
  Note that we can maintain the $\cH_{\cM}$ after each update in time $O(|U_D|^2)$ which is subsumed by the running time of $\cM.\Cut{D}$ we established in \cref{lemma:cut}.
  
  By \cref{lemma:init,lemma:cut}, the output edge set of $\cM.\Init{G}$ and $\cM.\Cut{D}$ has total capacities in expectation bounded by $\left(\phi \cdot O\left(\frac{1}{\psi_0^2}\right)n^{O(1/L)}\right) \cdot \alpha_{\prev} \cdot m$ and $O\left(\frac{1}{\psi_0^4}\right) \cdot m^{O(1/L)} \cdot \Bc_G(D)$.
  Therefore, we have $\alpha \leq \left(\phi \cdot O\left(\frac{1}{\psi_0^2}\right)n^{O(1/L)}\right) \cdot \alpha_{\prev}$ and $\beta \leq O\left(\frac{1}{\psi_0^5}\right) \cdot m^{O(1/L)}$.
  The subroutines run in $T(n)$ and $T(|U_D|)$ time, for $T(n) = n^{O(1/L)} \cdot \frac{1}{\psi_0^4} \cdot \widetilde{O}\left(\frac{n^2}{\phi\phi_{\prev}^2} + T_{\prev}(n)\right)$.
  Letting $\delta_L = \left(\frac{1}{\psi_0}\right)^{\Theta(L^2)} = (\log n)^{L^{\Theta(L)}}$ sufficiently large, these become the bounds stated in \eqref{eq:parameters}.
\end{proof}

\section*{Acknowledgements}
We thank Danupon Nanongkai and Christian Wulff-Nilsen for the helpful discussions during the preliminary stages of this work.

\bibliography{reference}

\appendix

\section{Using Dynamic Trees for Capacitated Push-Relabel} \label{sec:capacitated}
\label{sec:capacitated-push-relabel}

In section \cref{sec:push-relabel}, we showed \cref{alg:push-relabel}.
The running time analysis in \cref{sec:push-relabel} only shows the desired running time $\tO(m+n+\sum_{e\in E}\frac{h}{\Bw(e)})$ (of \cref{thm:push-relabel-main-theorem})
when the edges are of unit capacity $\Bc(e) = 1$.
Here we show that we can implement the same algorithm equally efficiently for any capacities, with the use of dynamic trees \cite{SleatorT83}.

In particular, we assume the following data structure (see \cite{SleatorT83} for details).

\begin{lemma}[Dynamic Trees \cite{SleatorT83,GoldbergT88}]
  There is a data structure which maintains a collection of rooted trees $\cT$ together with values $\Bnu \in \Z^{E}$ on the edges. 
  The data structure supports the following operations, all in amortized $O(\log n)$ update time.
  \begin{itemize}
    \item $\textsc{Link}(e)$: if $e = (u,v)$, add the edge to $\cT$, with $v$ now the parent of $u$. Before this update, $v$ must not have any parent and $u$ cannot be in the same tree as $v$.
    \item $\textsc{Delete}(e)$: remove the edge $e$ from $\cT$.
    \item $\textsc{FindMin}(u)$: find the edge $e$ with the minimum value $\Bnu(e)$ on the path from $u$ to the root of the tree which $u$ is in. In case of ties, return the edge closest to $u$.
    \item $\textsc{Add}(u, x)$: set $\Bnu(e)\gets \Bnu(e) + x$ for all edges $e$ on the path from $u$ to the root of the tree containing $u$.
  \end{itemize}
\end{lemma}

The idea is similar to how a standard push-relabel algorithm can be sped up with dynamic trees (see \cite{GoldbergT88}).
We keep track of a set of rooted trees $\cT$, where for each vertex $u\in V$, we pick an arbitrary admissible out-edge $e = (u,v)$ as the parent-edge in $T$.
Indeed this forms a rooted tree, since the parent $u$ has lower level than $v$.
The values $\Bnu$ which the data structures keeps track of will be the residual capacities $\Bc_{\Bf}$, from which the flow $\Bf$ can implicitly be calculated from.

Whenever an edge is marked admissible or inadmissible, we might need to add/remove it from the tree, and perhaps replace the removed edge with another admissible edge (this takes $O(\log n)$ time via the $\textsc{Link}$ and $\textsc{Delete}$ operations). Whenever an edge is removed from $\cT$, we update the residual capacity of the corresponding reverse edge (which might at this point be outdated; this is okay since only one of $\forward{e}$ and $\backward{e}$ can be admissible at the same point in time, and we only need to maintain the residual capacity correctely for the admissible edge).

When no vertices can be relabeled, this means that each vertex except the unsaturated sinks will have a parent in its tree, and the roots of the trees will thus exactly be the unsaturated sinks $t$ with $\abs_{\Bf}(t) < \Bsink(t)$.
When the algorithm wants to trace a path $P$ from $s$ to some sink $t$, this path $P$ can thus be the path from $s$ in its tree to the corresponding root.
The value of $c^{\mathrm{augment}}$ can be found using the $\textsc{FindMin}$ operation.
Thereafter, the residual capacities on the path $P$ can be adjusted via the $\textsc{Add}$ operation.
We still need to find all the edges on $P$ which now has $\Bc_{\Bf} = 0$, so that we can mark them as inadmissible.
We do this with iteratively calling the $\textsc{FindMin}$ operation climbing the path $P$ as long as the returned edge $e$ has $\Bc_{\Bf} = 0$.

Except for marking the edges on $P$ as inadmissible, we use $O(\log n)$ time per augmenting path, and by \cref{lem:few-augmenting-paths}, there are only $O(n + \sum \frac{h}{\Bw(e)})$ augmenting paths in total over the run of the algorithm.
\Cref{lem:few-augmenting-paths} also says that each edge $e$ appears as a saturated edge in at most $O(\frac{h}{\Bw(e)})$ augmenting paths, so the total cost of marking edges on augmenting paths as inadmissible, over the whole run of the algorithm,
will be $O(\sum \frac{h}{\Bw(e)}\log n)$.

Together with the analysis in \cref{sec:push-relabel}, we conclude that we can implement \cref{alg:push-relabel} in $\tO(m + n + \sum_{e\in E} \frac{h}{\Bw(e)})$ time, thus proving the stated running time bound of \cref{thm:push-relabel-main-theorem}.

\section{Capacity Scaling}\label{appendix:capacity-scaling}

In this section, we recall the folklore capacity scaling argument for maximum flow. In particular we say that with an additional $O(\log U)$ overhead, we can reduce capacities from $\{0,1,2,\ldots, U\}$ to $\{0,1,2,\ldots, n^2\}$.

\begin{lemma}
Given an algorithm $\cA$ that can solve any maximum flow instance $\cI = (G, \Bc, \Bsource, \Bsink)$, for an $n$-vertex $m$-edge simple directed graph $G$, with $\|\Bc\|_{\infty},\|\Bsource\|_{\infty},\|\Bsink\|_{\infty}\le n^2$, in time $T_{\cA}(n,m)$; there is an algorithm $\cA'$ which can solve maximum flow where $\|\Bc\|_{\infty},\|\Bsource\|_{\infty},\|\Bsink\|_{\infty}\le U$ in time $O(T_{\cA}(n,m)\log U)$.
\end{lemma}
\begin{proof}
For an edge $e$, write $\Bc(e)$ in binary as
$\Bc(e) = \sum_{i=0}^{k} 2^{i}\cdot \Bc^{(i)}(e)$, similarly for a vertex $v$ write $\Bsource(v) = \sum_{i=0}^{k} 2^{i}\cdot \Bsource^{(i)}(v)$.
 and $\Bsink(v) = \sum_{i=0}^{k} 2^{i}\cdot \Bsink^{(i)}(v)$, for $k = O(\log U)$.
 We will go from the most significant bit, and add one bit at a time to $\Bc$, $\Bsource$ and $\Bsink$. Let $\Bc^{(\uparrow b)} = \sum_{i=0}^{b} 2^{i}\cdot \Bc^{(k-b+i)}$ be the capacity function, but we only keep the $b$ most significant bits and scale it down. Define $\Bsource^{(\uparrow b)}$ and $\Bsink^{(\uparrow b)}$ similarly.

 If $\Bf$ is a maximum flow of $\cI^{(\uparrow b)} = (G,\Bc^{(\uparrow b)}, \Bsource^{(\uparrow b)}, \Bsink^{(\uparrow b)})$, we can use $\Bf$ as a starting point to compute the maximum flow after we added one extra bit, i.e. for the instance
 $\cI^{(\uparrow b+1)} = (G,\Bc^{(\uparrow b+1)}, \Bsource^{(\uparrow b+1)}, \Bsink^{(\uparrow b+1)})$. The crucial observation is that $2\Bf$ is a feasible flow for $\cI^{(\uparrow b+1)}$, and the maximum flow $\Bf'$ in the residual instance $\cI^{(\uparrow b+1)}_{2\Bf}$ has value at most $|\Bf'|\le n^2$. This is since 
 the flow instance $\cI^{(\uparrow b+1)}_{2\Bf}$ is obtained from $\cI^{(\uparrow b)}_{\Bf}$ (which has no more augmenting paths) by (1) doubling all the capacities, demand, and flow-values; and (2) adding up to $m \le n^2$ unit-capacity edges (and unit-source/unit-sink demand). This means that  when solving $\cI^{(\uparrow b+1)}_{2\Bf}$ (with algorithm $\cA$), we may cap all capacities above by $n^2$, as this will not change the maximum value of the flow.
 We update $\Bf \gets 2\Bf + \Bf'$, which is now a maximum flow of $\cI^{(\uparrow b+1)}$, and proceed to the next bit.
\end{proof}

\section{Omitted Proofs}\label{appendix:omitted-proofs}

\RespectingTopo*

\begin{proof}
  Let $(D,X_1, \ldots, X_\eta) = \cH$, and recall that $X_i$ is a separator of the graph $G_{i} = G \setminus X_{>i}$.
  By design, the collection of strongly connected components of $G_i$ are a refinement of the strongly connected components of $G_{i+1}$, and $G_{0} = (V,D)$ is a DAG (see also \cref{fig:hierarchy example}).
  
  To compute the $\cH$-respecting topological order $\Btau$, we start at the highest level $\eta$ and compute (in $O(m)$ time) the strongly connected components $\{C_1, C_2, \ldots, C_{r}\}$ of $G_{\eta} = G$, together with a topological order of them~\cite{Tarjan72}. We reorder the $C_{i}$'s with respect to this topological order so that for any DAG edges $(u,v)\in D$ with 
 $u\in C_{i}$ and $v\in C_{j}$ we have $i\le j$.
 
 We will assign
 $\{1, \ldots, |C_{1}|\}$ 
 to vertices in $C_1$, 
 $\{|C_{1}|+1, \ldots, |C_{1}|+|C_2|\}$  to vertices in $C_2$ and so on, since this would guarantee that $\Btau$ is contiguous for all level-$\eta$ expanders $C_1, \ldots, C_r$, and that the topological ordering $\Btau$ respects all the DAG edges between two different $C_{i}$'s.

 If $\eta = 0$, we are done, since each $C_{i}$'s are singletons. Otherwise
 we may simply recurse on each strongly connected component $C_{i}$, with the hierarchy $\cH_{i} = (D\cap E[C_{i}], X_1\cap E[C_{i}], \ldots, X_{\eta-1}\cap E[C_{i}])$ of height $\eta(\cH_{i}) = \eta-1$, to figure out the internal ordering of the vertices inside $C_{i}$.

 The total running time will be $O(m\eta)$, since on each level from $\eta$ down to $0$ we will need to find the strongly connected components of some graphs with a total of $m$ edges.
\end{proof}

\TopLevelWitness*

\begin{proof}
  We run the cut-matching game of \cref{thm:directed-cut-matching-game} with input $\Bnu \defeq \deg_{F,\Bc}$.
  Note that $\Bnu(v)$ is bounded by $n^2$ due to capacity scaling.
  In each of the iteration $t_{\CMG} = O(\log^2 n)$ iterations, given $(\Bnu_A^{(i)}, \Bnu_B^{(i)})$, we invoke \cref{thm:flow} on the flow instance $\cI = \left(G, \Bc_G, \Bsource, \Bsink\right)$ with $\kappa \defeq \frac{2 \cdot c_{\ref{thm:flow}}}{\phi}$ where $\Bsource \defeq \Bnu_A^{(i)}$ and $\Bsink \defeq \Bnu_B^{(i)}$.
  Let $\Bf^{*} \defeq \Bzero$.
  If the flow $\Bf$ from \cref{thm:flow} routes half of the demand, i.e., $|\Bf| \geq \frac{1}{2}\|\Bsource\|_1$, then we update $\Bf^{*} \gets \Bf^{*} + \Bf$, $\Bsource \gets \ex_{\Bf}$, and $\Bsink \gets \Bsink - \abs_{\Bf}$, and re-run \cref{thm:flow} until $\|\Bsource\|_1$ becomes less than $\frac{R}{2t_{\CMG}}$ (which will happen in at most $z$ runs of \cref{thm:flow}, where recall that $z \defeq 20\log n$).
  On the other hand, if $|\Bf| < \frac{1}{2}\|\Bsource\|_1$, then $\ex_{\Bf}(V) > \frac{1}{2}\|\Bsource\|_1 \geq \frac{1}{4t_{\CMG}} R$.
  By \cref{thm:flow}, in this case the cut $S$ that it returns satisfies $\ex_{\Bf}(S) = \ex_{\Bf}(V)$ which implies $\vol_{F,\Bc_G}(S) \geq \ex_{\Bf}(S) \geq \frac{1}{4t_{\CMG}}R$.
  Likewise, we have $\vol_{F,\Bc_G}(\overline{S}) \geq \Bsink_{\Bf}(V) \geq \|\Bsink\|_1 - \frac{1}{2}\|\Bsource\|_1 \geq \frac{1}{2}\|\Bsource\|_1 \geq \frac{1}{4t_{\CMG}}R$.
  Therefore, the guarantee of \cref{thm:flow} implies that
  \begin{align*}
    \Bc_G(E_G(S,\overline{S})) &\leq \frac{c_{\ref{thm:flow}} \cdot |\Bf| + \min\{\vol_{F,\Bc_G}(S),\vol_{F,\Bc_G}(\overline{S})\}}{\kappa}
    \leq \frac{2c_{\ref{thm:flow}} \cdot \min\{\vol_{F,\Bc_G}(S), \vol_{F,\Bc_G}(\overline{S})\}}{\kappa}.
  \end{align*}
  Thus, depending on whether $\vol_{F,\Bc_G}(S) \leq \vol_{F,\Bc_G}(\overline{S})$ or not we can return either $S$ or $\overline{S}$ in Case~\ref{item:case:sparse-cut}.
  
  Now, if none of the calls to \cref{thm:flow} routes less than half of the given demand, then by adding capacitated fake edges with total capacities at most $\frac{R}{2t_{\CMG}}$ we have found a $(\Bnu_A^{(i)}, \Bnu_B^{(i)})$-perfect matching $(M_i, \Bc_i)$ in which the non-fake edges are embeddable into $(G, \Bc_G)$ with congestion $\kappa z$.
  By \cref{thm:directed-cut-matching-game}, after $t_{\CMG}$ iterations, we have constructed a $\psi_{\CMG}$-expander $(\widetilde{W}, \Bc_W)$ containing fake edges whose total capacities sum to at most $R/2$ and in which non-fake edges are embeddable into $(G, \Bc_G)$ with congestion $\kappa z t_{\CMG}$.
  Let $E_{\mathrm{fake}} \subseteq \widetilde{W}$ be the set of fake edges.
  If we set $\Br \defeq \deg_{E_{\mathrm{fake}},\Bc_W}$ and $W \defeq \widetilde{W} \setminus E_{\mathrm{fake}}$, then we have $\|\Br\|_1 \leq R$, $\deg_{F,\Bc_G}(v) \leq \deg_{W,\Bc_W}(v) + \Br(v) \leq t_{\CMG}\cdot \deg_{F,\Bc_G}(v)$, and that $(W, \Bc_W)$ embeds into $(G, \Bc_G)$ with congestion $kzt_{\CMG}$.
  Moreover, by the expansion guarantee of $\widetilde{W}$, we have
  \[
    \Bc_W(E_W(S, \overline{S})) + \Br(S) \geq \Bc_W(E_{\widetilde{W}}(S, \overline{S})) \geq \psi_{\CMG}(\vol_{W,\Bc_W}(S) + \Br(S))
  \]
  and
  \[
    \Bc_W(E_W(\overline{S}, S)) + \Br(\overline{S}) \geq \Bc_W(E_{\widetilde{W}}(\overline{S}, S)) \geq \psi_{\CMG}(\vol_{W,\Bc_W}(S) + \Br(S))
  \]
  for every $\vol_{W,\Bc_W}(S) + \Br(S) \leq \vol_{W,\Bc_W}(\overline{S}) + \Br(\overline{S})$.
  As such, $(W, \Bc_W, \Br, \Pi_{W \to G})$ is an $(R, \phi, \widetilde{\psi})$-witness of $(G, \Bc_G, F)$ with respect to $\Bgamma(S) \defeq \vol_{W,\Bc_W}(S) + \Br(S)$ for some $\widetilde{\psi} = \Omega\left(\frac{1}{\log^2 n}\right)$.
      The running time of the algorithm is $\widetilde{O}\left(\frac{n^2 \kappa}{{\phi^\prime}^2}\right)$ for $\kappa = \frac{2 \cdot c_{\ref{thm:flow}}}{\phi}$ which is $\widetilde{O}\left(\frac{n^2}{\phi{\phi^\prime}^2}\right)$.
  This proves the lemma.
\end{proof}
\UnionOfSparseCuts*

\begin{proof}[Proof of \cref{lemma:union-of-sparse-cut}]
  Let us call an $S_i$ out-sparse if $\Bc_G\left(E_{G[V_{i-1}]}(S_i,\overline{S_i})\right) \leq \Bc_G\left(E_{G[V_{i-1}]}(\overline{S_i},S_i)\right)$ and in-sparse otherwise.
  Let $\cI_{\mathrm{out}} \defeq \{i: S_i\;\text{is out-sparse}\}$ and $\cI_{\mathrm{in}} \defeq \{i: S_i\;\text{is in-sparse}\}$.
  Let $S_{\mathrm{out}} \defeq \bigcup_{i \in \cI_{\mathrm{out}}}S_i$ and $S_{\mathrm{in}} \defeq \bigcup_{i \in \cI_{\mathrm{in}}}S_i$.
  Suppose without loss of generality that $\vol_{F, \Bc_G}(S_{\mathrm{out}}) \geq \vol_{F,\Bc_G}(S_{\mathrm{in}})$.
  By \cref{obs:boundary-of-union-of-cuts} we have
  \[
    E_G(S_{\mathrm{out}}, \overline{S_{\mathrm{out}}}) \subseteq \bigcup_{i \in \cI_{\mathrm{out}}}E_{G[V_{i-1}]}(S_i, \overline{S_i}) \cup \bigcup_{i \in \cI_{\mathrm{in}}}E_{G[V_{i-1}]}(\overline{S_i}, S_i)
  \]
  and therefore $\Bc_G\left(E_G(S_{\mathrm{out}}, \overline{S_{\mathrm{out}}})\right) < 2\phi \cdot \vol_{F,\Bc_G}(S_{\mathrm{out}})$.
  Since $\vol_{F,\Bc_G}(S_{\mathrm{out}}) \geq \frac{\alpha}{2}\vol_{F,\Bc_G}(V)$ and $\vol_{F,\Bc_G}(\overline{S_{\mathrm{out}}}) \geq (1-\alpha)\vol_{F,\Bc_G}(V)$, the lemma follows.
\end{proof}

\RemainBorderline*

\begin{proof}
  Recall that $\cK$ is the event that $\Delta^{(U,\ell)}(t_1, t_2) \leq \frac{\psi_{\ell}}{10}\tau_U^{\ell/L}$ and $\delta_{\mathrm{ext}}^{(U,\ell)}(t_1,t_2) \leq \frac{\phi\psi_{\ell}^2}{40}\tau_U^{\ell/L}$ hold for all active tuples $(U, \ell, t_1, t_2)$ which happens with high probability by \cref{lemma:good-event}.
  Let $t_{\mathrm{start}}$ be the moment in the lemma statement when the algorithm is in a good state.
  Note that the random choices that happened after time $t_{\mathrm{start}}$ are completely independent of the condition that the algorithm is in a good state at time $t_{\mathrm{start}}$.
  Therefore, \cref{lemma:good-event} suggests that with high probability $\Delta^{(U,\ell)}(t_1, t_2) \leq \frac{\psi_{\ell}}{10}\tau_U^{\ell/L}$ and $\delta_{\mathrm{ext}}^{(U,\ell)}(t_1,t_2) \leq \frac{\phi\psi_{\ell}^2}{40}\tau_U^{\ell/L}$ hold for all active tuples $(U, \ell, t_1, t_2)$ with $t_1 \geq t_{\mathrm{start}}$.

  Similar to \cref{lemma:probability-guarantee}, we prove by induction on time starting from $t_{\mathrm{start}}$.
  Let $t \geq t_{\mathrm{start}}$ be the current time.
  Fix a $U \in \cU$ and $\ell \in \{0, \ldots, L\}$ for which $W_{U,\ell}$ is not currently being rebuilt.
  Consider the last time $t_{\mathrm{last}}^{(U,\ell)}$ that it was rebuilt, and let $\widetilde{t_{\mathrm{last}}^{(U,\ell)}} \defeq \max\{t_{\mathrm{last}}^{(U,\ell)}, t_{\mathrm{start}}\}$.
  Let $U_0$ and $F_0$ be the set $U$ and $F$ at time $\widetilde{t_{\mathrm{last}}^{(U,\ell)}}$.
  Note that if $t_{\mathrm{last}}^{(U,\ell)} \geq t_{\mathrm{start}}$, then by \cref{lemma:repair} the witness of $U$ and $\ell$ at time $\widetilde{t_{\mathrm{last}}^{(U,\ell)}}$ satisfies $\|\Br_{U_0,\ell}\|_1 \leq \frac{\psi_{\ell}}{10}\tau_U^{\ell/L}$.
  On the other hand, if $t_{\mathrm{last}}^{(U,\ell)} < t_{\mathrm{start}}$, then by the lemma statement that the algorithm is in a good state (where recall the definition of good in \cref{def:good-state}), we have $\|\Br_{U_0,\ell}\|_1 \leq \tau_U^{\ell/L}$ (observe that $W_{U,\ell}$ \emph{cannot} be being rebuilt at time $t_{\mathrm{start}}$ in this case, as that would imply by \cref{obs:will-be-rebuilt} that either it is still currently being rebuilt or $t_{\mathrm{last}}^{(U,\ell)} \geq t_{\mathrm{start}}$).
  In either case, we have $\|\Br_{U_0,\ell}\|_1 \leq \tau_U^{\ell/L}$.

  To likewise apply \cref{lemma:bound-from-external-cuts}, we note again that what happened from $\widetilde{t_{\mathrm{last}}^{(U,\ell)}}$ to the current moment is modeled by \cref{scenarion:stability}.
  Moreover, \cref{cond:blow-up}\labelcref{item:sparse,item:small-side} hold by exactly the same arguments as in the proof of \cref{lemma:probability-guarantee}.
  It thus remains to verify \cref{cond:blow-up}\labelcref{item:parameters}.
  Again, $\delta_{\ell} \leq \frac{1}{16}$ is straightforward.
  That $R \leq \frac{\psi_{\ell}}{64}\vol_{F_0}(U_0)$ is by $R \leq \tau_U^{\ell/L} \leq \frac{\psi_{0}^2}{32z}\vol_{F_0}(U_0) \leq \frac{\psi_{\ell}}{64}\vol_{F_0}(U_0)$ by the inductive hypothesis at time $\widetilde{t_{\mathrm{last}}^{(U,L)}}$.
  The bound on $\Delta \defeq \Delta^{(U,\ell)}(\widetilde{t_{\mathrm{last}}^{(U,\ell)}}, t)$ and $\delta_{\mathrm{ext}} \defeq \delta_{\mathrm{ext}}^{(U,\ell)}(\widetilde{t_{\mathrm{last}}^{(U,\ell)}}, t)$ follow from our discussion in the beginning of this proof.

  Since \cref{cond:blow-up} holds, \cref{lemma:bound-from-external-cuts} applied on $W_{U,\ell}$ (with $R \defeq \tau_U^{\ell/L}$, $\Delta \defeq \Delta^{(U,\ell)}(\widetilde{t_{\mathrm{last}}^{(U,\ell)}}, t)$, and $\delta_{\mathrm{ext}} \defeq \delta_{\mathrm{ext}}^{(U,\ell)}(\widetilde{t_{\mathrm{last}}^{(U,\ell)}}, t)$) shows that $\|\Br_{U,\ell}\|_1 \leq \frac{4(R+\Delta)}{\psi_{\ell}} + \frac{8}{\psi_{\ell}^2\phi}\delta_{\mathrm{ext}} \leq \frac{10}{\psi_{\ell}}\tau_U^{\ell/L}$.

  For the bound on $\tau_U$, we consider when $\ell = L$.
  A similar argument as in the proof of \cref{lemma:probability-guarantee} shows that the current volume of $U$ is at least $\frac{64z}{\psi_0}\tau_U - \frac{10}{\psi_0^2}\tau_U \geq \frac{32z}{\psi_0^2}\tau_U$.
  Likewise, the volume increase by at most $\Delta^{(U,L)}(\widetilde{t_{\mathrm{last}}^{(U,L)}}, t)$, hence the lower bound of $\tau_U \geq \frac{\psi_0^2}{128z}\vol_F(U)$.
\end{proof}

\end{document}